\documentclass[preprint,12pt,3p,authoryear]{elsarticle}
\usepackage[colorlinks, linkcolor=skyblue, citecolor=skyblue]{hyperref}
\usepackage{booktabs}
\usepackage{array}
\usepackage{subcaption}

\usepackage{multirow}

\usepackage{lineno}
\newcommand*\patchAmsMathEnvironmentForLineno[1]{%
  \expandafter\let\csname old#1\expandafter\endcsname\csname #1\endcsname
  \expandafter\let\csname oldend#1\expandafter\endcsname\csname end#1\endcsname
  \renewenvironment{#1}%
     {\linenomath\csname old#1\endcsname}%
     {\csname oldend#1\endcsname\endlinenomath}}%
\newcommand*\patchBothAmsMathEnvironmentsForLineno[1]{%
  \patchAmsMathEnvironmentForLineno{#1}%
  \patchAmsMathEnvironmentForLineno{#1*}}%
\AtBeginDocument{%
\patchBothAmsMathEnvironmentsForLineno{equation}%
\patchBothAmsMathEnvironmentsForLineno{align}%
\patchBothAmsMathEnvironmentsForLineno{flalign}%
\patchBothAmsMathEnvironmentsForLineno{alignat}%
\patchBothAmsMathEnvironmentsForLineno{gather}%
\patchBothAmsMathEnvironmentsForLineno{multline}%
}
\usepackage{amsmath}
\usepackage{amssymb}

\usepackage{amsthm}

\newtheorem{theorem}{Theorem}[section]
\newtheorem{proposition}[theorem]{Proposition}

\newtheorem{corollary}[theorem]{Corollary}
\newtheorem{assumption}{Assumption}
\newtheorem{remark}{Remark}

\usepackage{dashrule}

\graphicspath{
    {figure/}
    {figure/figure1/}
    {figure/figure2/}
    {figure/figure3/}
    {figure/figure4to6/}
    {figure/figure7/}
    {figure/figure8/}
    {figure/figure9/}
    {figure/figure10/}
    {figure/figure11/}
    {figure/figure12/}
    {figure/figure13/}
    {figure/figure14/}
    {figure/figure15/}
}

\usepackage{stackengine}
\usepackage{tikz}
\usepackage{float}

\usepackage{threeparttable}

\newcommand{\myimage}[1]{\includegraphics[width=\linewidth, keepaspectratio]{#1}}
\usepackage{enumitem}
\usepackage{tabularx}
\usepackage{caption}
\usepackage{xurl}
\newcolumntype{C}{>{\centering\arraybackslash}X}

\begin{document}
    \begin{frontmatter}
        \title{The Evolution of Eco-routing under Population Growth: Evidence from Six U.S. Cities}

         \author[label1,label4]{Zhiheng Shi}

         \author[label2,label3,label5]{Xiaohan Xu\corref{cor1}}

         \author[label6]{Wei Ma}

         \author[label3,label4]{Kairui Feng\corref{cor1}}

         \author[label3,label2,label5]{Bin He\corref{cor1}}

         \address[label1]{School of Mathematics, Harbin Institute of Technology, Harbin 150001, China}
         \address[label2]{College of Electronics and Information Engineering, Tongji University, Shanghai 201804, China}
         \address[label3]{State Key Laboratory of Autonomous Intelligent Unmanned Systems, Shanghai 201210, China}
         \address[label5]{Frontiers Science Center for Intelligent Autonomous Systems, Shanghai 201210, China}
         \address[label6]{Department of Civil and Environmental Engineering, The Hong Kong Polytechnic University, Hong Kong 999077, China}
         \address[label4]{Shanghai Innovation Institute, Shanghai 200030, China}

         \cortext[cor1]{Corresponding author: xhxu@tongji.edu.cn; kelvinfkr@tongji.edu.cn; hebin@tongji.edu.cn}

        \begin{abstract}

            Rapid urban population growth drives car travel demand,
            increasing transport carbon emissions and posing a critical challenge to sustainable development.
            Although existing studies have demonstrated that eco-routing can reduce individual emissions, research gaps remain.
            On the one hand, such personal reductions have a negligible impact on overall emissions,
            and cannot be simply aggregated to capture the complex effects of large-scale eco-routing.
            On the other hand, under population growth, the long-term effectiveness of eco-routing,
            as well as the evolution of its efficiency and traveler route choice, remain underexplored.
            To address these limitations, this study proposes Time-Only and Time-Carbon user equilibrium (UE) models,
            integrates them with a demand forecasting method for simulating future network traffic,
            and designs multi-dimensional metrics to characterize urban dynamics.
            Using real-world road networks, commuting origin-destination (OD) demand, and population projections under various shared socioeconomic pathways (SSPs)
            for six representative U.S. cities as a case study,
            we conduct a comprehensive analysis of urban dynamics across different routing strategies and population sizes.
            The results reveal that
            while eco-routing mitigates total emissions,
            emissions in most cities scale superlinearly with population,
            a scaling order that remains invariant regardless of routing and construction strategies.
            Moreover, under population growth,
            travelers using eco-routing tend to increasingly select shorter routes,
            giving rise to carbon bottlenecks.
            A strategy of targeted capacity expansion on these critical bottlenecks (0.46\% of links)
            significantly reduces both emissions (3\%) and travel time (28\%) without compromising eco-routing efficiency.
            This study provides a foundation for formulating low-carbon urban transport planning and emission reduction policies.
        \end{abstract}


    \end{frontmatter}

    \section{Introduction}
    \label{sec:1}
        Cities have become centers of energy consumption, air pollution, and environmental degradation,
        with urban transport emissions (i.e., carbon emissions) contributing a substantial fraction of national greenhouse gas (GHG) emissions,
        specifically accounting for 26.3\% in the U.S., a share that continues to rise annually~\citep{UNFCCC_2021_GHG}.
        In addition, given projections estimating that approximately 68\% of the global population will reside in urban areas by 2050~\citep{desa2018world},
        such population growth and urban sprawl are accompanied by rising private car ownership,
        resulting in growing demand for car travel and increased vehicle miles traveled (VMT)~\citep{glaeser2004sprawl, bento2005effects},
        which will inevitably translate into escalating emissions,
        posing a major challenge to sustainable development~\citep{kazancoglu2021prediction, naganathan2017evaluation, yao2020optimization}.
        Moreover, this challenge is particularly acute in the United States,
        where studies project that approximately 80\% of urban cities in the Northeast and 75\% in the Midwest
        will experience population growth by 2100 under the shared socioeconomic pathway 2 (SSP2) scenario ~\citep{sutradhar2024depopulation}.
        Furthermore, with only 5\% of the world’s population,
        the U.S. possesses 30\% of the world’s automobiles
        and contributes 45\% of global automotive emissions~\citep{us_census_2020, decicco2006global}.
        Given this high car dependency in the U.S., population growth will drive car travel demand even more significantly,
        leading to a substantial surge in emissions.

        To mitigate emissions, researchers have explored multiple approaches, such as alternative fuels~\citep{fouladvand2024does},
        more efficient engine technologies~\citep{dong2022comparative},
        and diverse traffic management strategies~\citep{zong2024congestion}.
        Among these, eco-routing (also known as eco-friendly navigation) remains a research focus~\citep{ortega2021can,zhuang2024dual,liu2023collaborated}.
        Specifically, eco-routing guides travelers to select routes that minimize their fuel consumption and emissions
        by incorporating these environmental costs into the generalized travel cost.
        In contrast, strategies prioritizing minimizing travel time are referred to as time-routing,
        while other routing approaches exist based on various optimization objectives~\citep{ku2021new,wang2025terouting}.
        The effectiveness of eco-routing has been validated by many studies,
        which demonstrate that a single traveler adopting eco-routing can significantly reduce fuel consumption and emissions~\citep{fahmin2024eco,zhou2016review}.
        In practice, Google Maps introduced eco-friendly navigation in the U.S. in 2021
        to guide travelers along the ``lowest carbon footprint'' routes,
        a feature with the potential to reduce over 1 million tons of carbon emissions per year~\citep{reinicke2021routee}.

        Although eco-routing has proven effective, research gaps remain in existing literature.
        Prior studies have almost entirely focused on minimizing emissions at the individual level~\citep{nie2013eco, ericsson2006optimizing}.
        While a single driver utilizing eco-routing can achieve reductions in personal emissions,
        the impact of such individual emission savings on overall urban transport emissions is negligible.
        However, with the widespread adoption of eco-routing, substantial shifts in travel route choice significantly affect network congestion,
        making the collective impact of this large-scale eco-routing complex.
        Moreover, most eco-routing studies formulate generalized travel costs based solely on environmental costs,
        neglecting the time-sensitivity of travelers.
        When both factors are considered,
        the trade-offs between travel time and emissions are expected to influence traveler route choice.
        Additionally, while limited studies have investigated the impact of large-scale eco-routing~\citep{ahn2013network,antunez2023does},
        they have overlooked its long-term effectiveness in the context of population growth,
        failing to assess whether eco-routing can fundamentally curb the upward trend of emissions.
        Meanwhile, the evolution of traveler route choice also remains underexplored.
        As infrastructure development is inevitable, understanding this evolution is critical for guiding sustainable urban planning and road construction.
        Given that emission reductions achieved via eco-routing typically come at the expense of increased travel time~\citep{ahn2008effects,ahn2013network},
        it is necessary to investigate how the relative differences in travel time and emissions between time-routing and eco-routing vary under population growth,
        thereby evaluating the efficiency of eco-routing.
        Furthermore, available studies primarily focus on idealized or small-to-medium-scale networks,
        lacking systematic quantitative analysis of real-world urban road networks and actual origin-destination (OD) demand.

        To address these gaps, this study formulates Time-Only and Time-Carbon UE models,
        integrates them with an OD demand forecasting method to simulate future urban network traffic under different routing strategies.
        Subsequently, several urban metrics are designed,
        including urban transport emissions, UE travel time (UETT), UE travel length (UETL), and volume over capacity (VOC).
        The evolution of these metrics and traveler route choice collectively constitutes the analyzed urban dynamics.
        Finally, through a case study of six representative U.S. cities under various SSPs,
        the results enable the analysis of urban dynamics across different routing strategies and population sizes,
        as well as eco-routing efficiency.
        The research findings can inform policy recommendations for low-carbon urban transport planning, road construction, and emission mitigation.

        The main contributions of this study are summarized as follows:
        \begin{enumerate}[label=(\arabic*), itemsep=0.5pt, parsep=0pt]
            \item \textbf{Time-Only and Time-Carbon UE models.}
                We formulate Time-Only and Time-Carbon UE models,
                establish the existence and uniqueness of their solutions,
                and employ a bi-conjugate Frank-Wolfe algorithm for computation.
                By integrating these models with a demand forecasting method,
                we simulate future network traffic under different routing strategies.

            \item \textbf{Theoretical insights.}
                We design a set of metrics
                and conduct theoretical derivations based on the proposed UE models to explore urban dynamics and the impact
                of bounded link capacity expansion from both macroscopic and microscopic perspectives, thereby providing a rigorous foundation for the corresponding findings presented in the case study.

            \item \textbf{Empirical insights from six U.S. cities.}
                Using real-world road networks and block-level commuting demand, we reveal that
                under population growth:
                (i) emissions in most cities scale superlinearly with population regardless of routing and construction strategies;
                (ii) eco-routing is either already inefficient in certain cities or projected to become inefficient in others;
                (iii) carbon bottlenecks (0.46\% of links) emerge as travelers using eco-routing increasingly select shorter routes,
                    and expanding their capacity significantly reduces both emissions (3\%) and travel time (28\%).
        \end{enumerate}

        The remainder of this paper is organized as follows:
        Section~\ref{sec:2} reviews the relevant literature.
        Section~\ref{sec:3} establishes the methodological framework, including UE models construction, indicator design, and theoretical derivations.
        Section~\ref{sec:data} details the study area, data sources and processing procedures.
        Section~\ref{sec:4} presents the case study analysis and discusses the key findings.
        Finally, Section~\ref{sec:5} summarizes the conclusions of this study.

    \section{Literature Review}
    \label{sec:2}
        \subsection{Population Growth and Urban Transport Emissions}
        \label{sec:2.1}
            Demographic projections indicate that the global population will continue to grow,
            with consensus estimates converging on a peak exceeding 9.7 billion
            in the latter half of this century~\citep{department2024world,vollset2020fertility,samir2017human}.
            In developed economies, urban population dynamics are characterized by the simultaneous phenomena of reurbanization and suburbanization,
            driving both densification and sprawl~\citep{kabisch2011diversifying,oecd2018rethinking,rerat2012new}.
            In the U.S. context, studies forecast widespread urban growth by 2100~\citep{gao2017downscaling,hauer2019population}.
            Specifically, approximately 80\% of urban cities in the Northeast, 75\% of those in the Midwest,
            and nearly 100\% in the West and South are projected to experience population increases under the SSP2 scenario~\citep{sutradhar2024depopulation}.
            Even cities currently experiencing population loss, such as Long Island and those around Chicago, may still achieve future growth driven by immigration.
            Moreover, cities with high vehicle ownership are more likely to grow in population.
            Consequently, demographic expansion acts as a driver of total car travel demand,
            thereby placing stress on road networks and increasing urban transport emissions~\citep{chapman2007transport}.
            Additionally, urban sprawl remains ubiquitous and continues to expand,
            often accompanied by increased private car ownership~\citep{glaeser2004sprawl}.
            These findings align with spatial analyses in the U.S.~\citep{bento2005effects} indicating that urban sprawl results in increased car usage and higher annual VMT,
            thereby leading to higher emissions.

            Few studies have directly explored the relationship between population growth and urban transport emissions,
            as most related research primarily focuses on overall urban emissions.
            To quantify urban dynamics under population growth, \citet{bettencourt2007growth} proposed the urban scaling theory,
            demonstrating that various urban indicators scale with population following a power-law function.
            Specifically, urban metrics like traffic congestion often exhibit superlinear scaling,
            whereas gasoline sales often exhibit sublinear scaling,
            implying that both traffic congestion and emissions increase as population grows~\citep{bettencourt2010urban}.
            Similarly, \citet{mohsin2019integrated} and \citet{phdthsis1} showed that
            population growth leads to increased energy consumption,
            which directly intensifies emissions and accelerates environmental degradation.
            Existing research has clearly shown the growth trend of future population and demonstrated the associated increase in urban transport emissions.
            However, constrained by data availability,
            research gaps remain in exploring the effectiveness of traffic management strategies for mitigating emissions driven by population growth.
            In this study, we utilize urban commuting demand and integrate it with population projections to predict future travel demand.
            By applying the forecasted travel demand to the real-world road network, we simulate future network traffic,
            thereby investigating emissions when implementing city-scale routing strategies in the context of population growth.

        \subsection{Application of Eco-routing in Large-Scale Networks}
        \label{sec:2.2}
            An eco-routing system comprises two fundamental components: (i) an energy consumption model
            to estimate the energy consumption for a single vehicle on a given road segment;
            (ii) a routing algorithm that identifies the most energy-efficient path from an origin to a destination.
            Existing studies often model energy consumption as a function of
            road gradient, road type, vehicle operational characteristics, and vehicle speed~\citep{yu2016improving,wang2016fuel},
            with most research formulating fuel consumption solely based on vehicle speed~\citep{yang2018carbon,antunez2023does,wang2025carbon}.
            Routing algorithms generally fall into two categories: search-based and optimization-based solutions.
            Search-based solutions treat eco-routing as a shortest path problem in graph theory,
            setting environmental costs (e.g., fuel consumption, emissions) as edge weights to identify the minimum cost path via algorithms like
            Dijkstra, $\mathrm{A}^*$, and Bellman-Ford~\citep{hart1968formal,johnson1977efficient}.
            Optimization-based solutions formulate the routing task as a single or multi-objective optimization problem,
            solving for the optimal route using meta-heuristics such as
            Ant Colony Optimization~\citep{jabir2017design}, Particle Swarm Optimization~\citep{hannan2018capacitated}, and Genetic Algorithms~\citep{wang2021voyage}.
            Many studies have validated the effectiveness of eco-routing,
            which demonstrate that individual travelers adopting eco-routing can reduce fuel consumption and emissions~\citep{fahmin2024eco,zhou2016review}.
            In practice, a key application of eco-routing lies in navigation systems that assist users in finding energy-efficient routes.
            For instance, Google Maps has introduced an eco-friendly navigation feature,
            with estimates suggesting that the adoption of this functionality could reduce carbon emissions by over one million tons annually~\citep{reinicke2021routee}.

            Prior research on eco-routing almost entirely focused on minimizing individual emissions.
            Only a limited number of studies have explored the complex impact of large-scale eco-routing,
            demonstrating that it can reduce urban transport emissions in most cities~\citep{ahn2008effects,coloma2019environmental}.
            Specifically, \citet{ahn2013network} integrated eco-routing with incremental UE traffic assignment
            to develop an ECO-SFA routing mechanism,
            concluding that eco-routing can reduce network-wide fuel consumption and emissions in most cases,
            with emission reductions becoming more pronounced as market penetration increases.
            Subsequent research applied eco-routing to city backbone networks and found that as the level of market penetration rises,
            emissions may increase in some cities while decreasing in others,
            highlighting the complex impact of the large-scale eco-routing~\citep{antunez2023does}.
            However, existing studies on eco-routing strategies primarily rely on static travel demand,
            failing to explore the urban dynamics under population growth.
            To bridge this gap, we forecast future travel demand via population projections and historic demand data,
            utilizing them to simulate future network traffic,
            thereby investigating the long-term impact of large-scale eco-routing, as well as the evolution of travel route choice.
            Moreover, most available literature formulates generalized travel costs based solely on environmental costs.
            Given the time-sensitivity of travelers, we incorporate travel time into the optimization objective
            and vary its relative weight against carbon cost to explore how trade-offs between time and emissions influence traveler route choice.
            Furthermore, while eco-routing reduces emissions, it typically comes at the expense of increased travel time~\citep{ahn2008effects,ahn2013network}.
            Consequently, this study offers valuable insights into the efficiency of eco-routing
            by investigating the relative differences in time and emissions between time-routing and eco-routing.


        \subsection{UE Traffic Assignment}
        \label{sec:2.3}
            The UE traffic assignment is widely recognized as the most representative approach in static traffic assignment.
            The principle of UE originates from the seminal work of \citet{wardrop1952road},
            stating that in a UE state, no traveler can reduce their own travel cost by changing their route alone.
            \citet{beckmann1956studies} later gave this principle a strict mathematical formulation,
            which has since become the theoretical foundation for modern road network analysis.
            Solution methods for the UE problem can be traced back to the convex combination algorithm proposed by \citet{frank1956algorithm}.
            Subsequently, based on this foundation, \citet{florian1976application} made improvements and refinements to the algorithmic framework.
            These theories and algorithms form the cornerstone of traffic assignment research
            and are widely used to allocate travel demand under given traffic flow conditions~\citep{van1980most}.
            Many studies have demonstrated that the UE traffic assignment model is effective in reproducing real-world traffic flow patterns,
            and thus serves as an essential basis
            for evaluating the impact of infrastructure projects and traffic management policies~\citep{sheffi1985urban, nagurney2000multiclass}.
            Specifically, in policy research, by modifying the generalized travel cost within the UE model,
            studies can simulate urban network traffic, capture traveler behavior, and investigate urban dynamics under traffic management strategies.

            Extensive literature on congestion pricing utilizes UE models by modifying generalized travel cost in specific areas during certain times,
            successfully capturing significant shifts in traveler behavior and traffic flow patterns
            that lead to reductions in congestion and pollution~\citep{chen2020did, jia2023system,chen2022dynamic}.
            For instance, by combining travel time, route length, and emissions as the travel cost,
            research in Changchun reported a 5\% decrease in car usage and a 23\% reduction in emissions~\citep{zong2024congestion}.
            Similarly, numerous studies on routing strategies integrate UE models by formulating generalized travel cost based on specific objectives,
            also observing a fundamental reshaping of traffic flow patterns and travel route choice~\citep{raith2014solving,yin2006internalizing}.
            Specifically, existing literature applying UE models to explore the effectiveness of eco-routing in urban road networks
            has documented a reduction in VMT and an increase in average travel time~\citep{antunez2023does, ahn2013network}.
            Moreover, \citet{ccolak2016understanding} demonstrated that
            changing the selfish-routing strategy alone, through modifying the generalized travel cost to the social optimum,
            can significantly alter traveler route choice and reduce time losses by 15\%\textnormal{--}30\%,
            where the resulting urban mean congestion exhibited a high correlation ($R^{2}=0.758$) with empirical data from the TomTom Traffic Index.
            These studies have confirmed the efficacy, accuracy, and widespread applicability of the UE model in
            simulating real-world traffic patterns and evaluating traffic management strategies.
            Building on this foundation, we employ Time-Only and Time-Carbon UE models
            to analyze network traffic and urban dynamics under distinct routing strategies and different population sizes.


    \section{Methodology}
    \label{sec:3}
        \subsection{ARIMA-based Prediction of OD Demand}
        \label{sec:3.2}
                We first normalize the OD demand $d_{y}^{s, t}$ to obtain the OD demand distribution $\pi_{y}^{s, t}$, formulated as:
                \begin{equation}
                    \label{eq:Pr}
                    \pi_{y}^{s, t} = \frac{ d_{y}^{s, t} }{\sum_{s, t} d_{y}^{s, t} }
                \end{equation}
                where $d_{y}^{s, t}$ denotes the travel demand from origin $s$ to destination $t$ in year $y$,
                and $\pi_{y}^{s, t}$ represents the corresponding demand distribution.
                Through this normalization, the evolution of $d_{y}^{s, t}$ is decoupled into
                the population-driven change in total demand and the variation in $\pi_{y}^{s, t}$,
                where the latter solely reflects the intrinsic evolution of the OD demand distribution over time.
                To forecast the $\pi_{n}^{s, t}$ for future year $n$,
                we employ the AutoRegressive Integrated Moving Average (ARIMA) algorithm,
                which involves a two-stage process:
                model calibration via historical data fitting and one-step-ahead forecasting, as detailed below:
                \begin{enumerate}[label=(\arabic*), itemsep=0.5pt, parsep=0pt]
                \item \textbf{Model calibration:}
                    The ARIMA model is calibrated on the historical OD demand distribution time series $\pi_y^{s, t}$ where $y \in \{ n-18, \dots, n-1 \}$.
                    We determine the optimal model orders based on the Akaike Information Criterion (AIC), including
                    the autoregressive order $o_1$, the degree of differencing $o_2$, and the moving average order $o_3$;
                    and estimate the model coefficients using the Maximum Likelihood Estimation (MLE) method,
                    a standard approach in time series analysis~\citep{box2015time}.
                    The general form of the ARIMA$(o_1, o_2, o_3)$ model is expressed as:
                    \begin{equation}
                        \label{eq:arima1}
                        \left(1 - \sum_{i=1}^{o_1} \phi_i B^i\right) \nabla^{o_2} \pi_y^{s, t} = c + \left(1 + \sum_{j=1}^{o_3} \theta_j B^j\right) \varepsilon_y
                    \end{equation}
                    where $o_1$ and $o_3$ are determined within a search range of $[0, 3]$;
                    and $o_2$ is fixed at 1.
                    Here, $B$ is the backshift operator defined as $B^i \pi_y^{s, t} = \pi_{y-i}^{s, t}$;
                    $\nabla$ is the differencing operator where $\nabla \pi_y^{s, t} = \pi_y^{s, t} - \pi_{y-1}^{s, t}$;
                    $\phi_i$ and $\theta_j$ are the autoregressive and moving average coefficients, respectively;
                    $c$ is the drift term;
                    and $\varepsilon_y$ represents the white noise error term.
                \item \textbf{One-step-ahead forecasting:}
                    Based on the calibrated parameters and historical data, the one-step-ahead forecast for target year $n$,
                    denoted as $\widehat{\pi}_n^{s, t}$, is given by:
                    \begin{equation}
                    \label{eq:arima2}
                        \widehat{\pi}_n^{s, t} = \pi_{n-1}^{s, t}
                            + \sum_{i=1}^{o_1} \phi_i \left( \pi_{n-i}^{s, t} - \pi_{n-i-1}^{s, t} \right)
                            + \sum_{j=1}^{o_3} \theta_j \widehat{\varepsilon}_{n-j} + c
                    \end{equation}
                    where $\widehat{\varepsilon}_{y} = \pi_{y}^{s, t} - \widehat{\pi}_{y}^{s, t}$ represents the observed residual at year $y$.
                \end{enumerate}
%

                Finally, the predicted OD distribution $\widehat{\pi}_n^{s, t}$
                is normalized to sum to unity over all OD pairs to satisfy the distributional constraint.
                By integrating the future population data with $\widehat{\pi}_n^{s, t}$,
                the OD demand for the future year $n$ is calculated as:
                \begin{equation}
                    \label{eq:P}
                    d_{n}^{s, t} = \widehat{\pi}_n^{s, t} \cdot
                        \frac{ N_{n} \sum_{s,t} d_{n-1}^{s, t} }{ N_{n-1} }
                \end{equation}
                where $N_{n}$ denotes the population size in year $n$.
                Under this formulation, the total OD demand is directly proportional to the annual population size.

        \subsection{Time-Only and Time-Carbon UE models}
        \label{sec:3.3}
            Table~\ref{tab:parameters} summarizes the nomenclature and parameter definitions used in this section.
            \begin{table}[!htb]
                \centering
                \small
                \caption{Nomenclature}
                \label{tab:parameters}
                \begin{tabular}{ll}
                    \toprule 
                    Notation & Definition \\
                    \midrule 
                    $N_{y}$                 & The population size of a given city in year $y$; \\
                    $\alpha$, $\beta$       & Parameters of BPR function; \\
                    $a$                     & Index for a link (road segment); \\
                    $A$                     & Set of all links in the network; \\
                    $W$                     & Set of all OD pairs in the network; \\
                    $p$                     & Index for a route (path); \\
                    $\mathcal{P}^{s, t}$    & Set of all feasible paths connecting the origin $s$ and destination $t$ in the network; \\
                    $V$                     & Set of all nodes in the network; \\
                    $Q$                     & Total OD demand in the network; \\
                    $G$                     & The network $G=(V,A)$; \\
                    $l_{a}$                 & The travel length of link $a$, \ (km); \\
                    $q_{a}$                 & The capacity of link $a$, \ (veh/h); \\
                    $t^{0}_{a}$             & The free flow travel time of link $a$, \ (min); \\
                    $x_{a}$                 & The traffic vehicle flow on link $a$, \ (veh/h); \\
                    $t_{a}(x_{a})$          & The travel time of link $a$ for traffic flow $x_{a}$, (min); \\
                    $v_{a}(x_{a})$          & The travel speed on link $a$ for traffic flow $x_{a}$, (km/h); \\
                    $c_{a}(x_{a})$          & The travel fuel consumption of a single vehicle on link $a$ for traffic flow $x_{a}$, (L); \\
                    $e_{a}(x_{a})$          & The travel carbon emissions of a single vehicle on link $a$ for traffic flow $x_{a}$, (g); \\
                    $C_{a}(x_{a})$          & The generalized travel cost of a single vehicle on link $a$ for traffic flow $x_{a}$, (\$); \\
                    $f^{s, t}_{p}$          & The path flow between the origin $s$ and destination $t$ on route $p$; \\
                    $d^{s, t}$              & The demand volume between the origin $s$ and destination $t$; \\
                    $\delta^{s, t}(p, a)$   & $\delta^{s, t}(p,a)=1$ when link $a$ lies on route $p$, otherwise $\delta^{s, t}(p,a)=0$; \\
                    $h(v)$                  & The fuel consumption factor at speed $v$, \ (L/km); \\
                    $\gamma$                & The emission factor, (g/L); \\
                    $\psi_{1}$, $\psi_{2}$  & The weights for the time cost and carbon cost, respectively; \\
                    $\omega_{1}$, $\omega_{2}$ & The monetary coefficients for time (\$/min) and fuel (\$/L), respectively. \\
                    \bottomrule
                \end{tabular}
                \vspace{-1.0em}
            \end{table}

            \begin{assumption}[Network Topology]
            \label{asp:network}
                Let the road network be modeled as a strongly connected directed graph $G = (V, A)$,
                where $V$ is the set of nodes, and $A$ is the set of links.
                Let $W$ denote the finite set of OD pairs.
                For each OD pair $(s,t) \in W$, let $\mathcal{P}^{s, t}$ denote the set of feasible paths connecting origin $s$ and destination $t$,
                which is assumed to be non-empty and finite.
                Furthermore, the travel demand $d^{s, t}$ is positive and bounded,
                and the total demand in the network is denoted by $Q = \sum_{(s,t) \in W} d^{s, t}$.
            \end{assumption}

            \subsubsection{Time-Only UE Model}
            \label{sec:3.3.1}
                The traffic assignment problem can be formalized by modeling route choice:
                every traveler $i$ selects a route $p$ connecting their origin and destination
                based on a personal utility $\mathrm{cost}_i=\sum_{a\in p}{C_{a}(x_{a})}$,
                where $\mathrm{cost}_i$ is expressed as the sum of costs $C_{a}(x_{a})$ of every link $a$ on the chosen route $p$.
                Under time-routing,
                the travel cost considers only travel time,  $C_{a}(x_{a})= \tau t_{a}(x_{a})$,
                where $t_{a}(x_{a})$ represents the travel time on link $a$ for traffic flow $x_{a}$,
                and $\tau$ denotes the average value of time (VOT) for all purposes, set to \$0.30/min~\citep{usdot2023benefit}.
                The relationship between travel time $t_{a}(x_{a})$ and traffic flow $x_{a}$
                is commonly described by the bureau of public roads (BPR) function~\citep{manual1964bureau}:
                \begin{equation}
                    \label{eq:BPR}
                    t_{a}(x_{a})=t^{0}_{a}\left(1+\alpha \left(\frac{x_{a}}{q_{a}}\right)^{\beta } \right)
                \end{equation}
                where 
                $t^{0}_{a}$ is the free flow time of link $a$;
                $\alpha$ and $\beta$ are the parameters of BPR function.

                The set of link flows that occur
                when every traveler adopts time-routing to minimize their own travel time is referred to as the UE flows.
                In this resulting equilibrium state, no traveler can benefit from deviating from their route.
                This concept is fundamentally a Nash equilibrium in road networks and is captured in Wardrop’s first principle~\citep{wardrop1952road}:
                the journey times on all the used routes for an OD pair are equal and are less than those that would be experienced by a
                single vehicle on any unused route.
                To find this UE flows, we adopt the Time-Only UE model (i.e., the classic UE model),
                whose convex programming formulation is given by:
                \begin{equation}
                    \begin{aligned}
                        \label{eq:Time-Only UE}
                        & \operatorname*{Minimize}_{x_{a} \ \forall a\in A} \quad && \sum_{a \in A}\int_{0}^{x_{a}} \tau t_{a}(x)dx \\
                        & \text{subject to}                                    && \sum_{p}f^{s, t}_{p}=d^{s, t} \\
                        &                                                      && x_{a}=\sum_{s}\sum_{t}\sum_{p}f^{s, t}_{p}\delta^{s, t}(p,a), \quad \forall a \in A \\
                        &                                                      && x_{a}\ge 0, \quad f^{s, t}_{p}\ge0, \quad \forall s,t,p
                    \end{aligned}
                \end{equation}
                where
                $f^{s, t}_{p}$ is the path flow between the origin $s$ and destination $t$ on route $p$;
                and $\delta^{s, t}(p,a)=1$ when link $a$ lies on route $p$, otherwise $\delta^{s, t}(p,a)=0$.
                This optimization problem has been proven to be strictly convex, guaranteeing a unique solution for the UE flows~\citep{sheffi1985urban}.
                However, finding the optimal solution in city-scale road networks is challenging due to the prohibitive computational cost.
                To address this,
                we employ the bi-conjugate Frank-Wolfe algorithm~\citep{mitradjieva2013stiff,zill2019toll}.
                Additionally, in a UE model, although each traveler selects a route solely to minimize their own travel time,
                the total system cost also includes the marginal cost each traveler imposes on all other travelers sharing the links.
                Consequently, due to such selfish choices, the system settles into a suboptimal state~\citep{ccolak2016understanding}.

            \subsubsection{Time-Carbon UE Model}
            \label{sec:3.3.2}
                Car travel accounts for a dominant share of urban trips in the U.S.~\citep{mcguckin2018summary},
                a proportion projected to escalate due to ongoing urbanization and population growth.
                Moreover, research indicates that travelers exhibit strong behavioral inertia,
                maintaining their travel mode despite high travel costs~\citep{li2023joint, kim2023non}.
                Given these factors, this study focuses on modeling car traffic during the peak hour.
                \begin{enumerate}[label=(\arabic*), itemsep=0.5pt, parsep=0pt]
                \item \textbf{Fuel consumption factor}:
                    Consistent with existing studies~\citep{wang2025carbon,zong2024congestion},
                    we model fuel consumption per kilometer as a function of vehicle speed:
                    \begin{equation}
                        \label{eq:fuel consumption}
                        h(v)= {\varphi_{1} \left(v-v_e\right)}^{2} + \varphi_{2}
                    \end{equation}
                    where $h(v)$ denotes the fuel consumption factor, representing the fuel consumption per kilometer (L/km);
                    $v_e$ represents the optimal vehicle speed, at which the $h(v)$ reaches its minimum value $\varphi_{2}$.
                    The parameter $\varphi_{1}$ reflects the sensitivity of $h(v)$ to speed deviations;
                    specifically, a larger $\varphi_{1}$ indicates a more rapid increase in $h(v)$ as the current speed $v$ deviates from optimal speed $v_e$.
                    Following \citet{yang2018carbon},
                    we set $\varphi_{1}=3.968\times {10}^{-5}$, $v_e=73.412$~km/h,
                    and $\varphi_{2}=4.275\times {10}^{-2}$~L/km.
                    Consequently, the fuel consumption of a single vehicle on link $a$, denoted as $c_{a}(x_{a})$, is formulated as:
                    \begin{equation}
                        \label{eq:fuel consumption on links}
                        c_{a}(x_{a})=h(v_{a}) l_{a}
                    \end{equation}
                    where $c_{a}(x_{a})$ is measured in liters (L); $v_{a}$ represents the travel speed on link $a$; and $l_{a}$ is the length of link $a$.

                \item \textbf{Carbon emissions}:
                    Carbon emissions $e_{a}(x_{a})$ of a single vehicle on link $a$
                    are calculated as the product of its fuel consumption $c_{a}(x_{a})$ and a carbon emission factor $\gamma$,
                    as shown in Eq.~\eqref{eq:carbon emission}.
                    Since the factor $\gamma$ depends on gasoline composition and is specific to the study region,
                    we adopt the value from \citet{EPAwebGHG},
                    which states that the combustion of one gallon of gasoline produces 8,887~g of $\mathrm{{CO}_{2}}$,
                    equivalent to approximately 2,350~g of $\mathrm{{CO}_{2}}$ per liter.
                    \begin{equation}
                        \label{eq:carbon emission}
                        e_{a}(x_{a})=\gamma c_a(x_{a})
                    \end{equation}

                \item \textbf{Time-carbon cost}:
                    We formulate the generalized travel cost $C_{a}(x_{a})$ for a single vehicle on link $a$ based on the fuel consumption $c_{a}(x_{a})$,
                    thereby ensuring the cost function with a clear physical interpretation.
                    As shown in Eq.~\eqref{eq:Time-Carbon cost}, $C_{a}(x_{a})$ represents the time-carbon cost,
                    which reduces to the time cost when $\psi_{2}=0$ and to the carbon cost when $\psi_{1}=0$.
                    \begin{equation}
                        \label{eq:Time-Carbon cost}
                        \begin{split}
                            C_{a}(x_{a}) &= \tau \left(\psi_{1}t_{a}(x_{a}) + \psi_{2}c_{a}(x_{a})\right) \\
                                            &= \omega_{1} t_{a}(x_{a}) + \omega_{2} c_{a}(x_{a})
                        \end{split}
                    \end{equation}
                    where $\psi_{1}$ and $\psi_{2}$ represent the weights for the time and carbon costs;
                    and $\omega_{1}$ and $\omega_{2}$ correspond to the monetary coefficients for time (\$/min) and fuel (\$/L), respectively.
                \end{enumerate}

                In this study, we trade off travel time and emissions by adjusting the unit prices of time $\omega_1$ and fuel $\omega_2$,
                which are derived from varying cost weights $\psi_1$ and $\psi_2$.
                Specifically, $\omega_{2}$ is determined based on the VOT $\tau$ and the carbon cost weight $\psi_{2}$,
                representing the generalized fuel price as the sum of the market fuel price per liter, $\omega_{2}^{1}$, and
                the monetized cost of emissions generated by combusting one liter of fuel, $\omega_{2}^{2}$,
                defined as $\omega_{2}^{2} = \tau \psi_2 - \omega_{2}^{1}$.
                This is equivalent to imposing a carbon tax of
                $\omega_{3}=\omega_{2}^{2} / \gamma$ per gram of emissions.

                When every traveler adopts eco-routing to minimize their individual generalized travel cost,
                the objective of the Time-Carbon UE model is to find the UE flows based on the time-carbon cost.
                The convex programming formulation for this problem is given by Eq.~\eqref{eq:Time-Carbon UE}.
                To examine how trade-offs between time and emissions affect the UE flows,
                we vary the weight ratio of $\psi_{1}:\psi_{2}$:
                adopting $\psi_{2}=0$ (the Time-Only UE model) as the baseline;
                applying $\psi_{1}=8$, $\psi_{2}=2$ (i.e., $\psi_{1}=4$, $\psi_{2}=1$) and $\psi_{1}=1$, $\psi_{2}=9$
                to explore the impact of increasing the relative weight of carbon cost;
                and setting $\psi_{1}=0$, $\psi_{2}=10$ as the carbon-only boundary case.
                \begin{equation}
                    \begin{aligned}
                        \label{eq:Time-Carbon UE}
                        & \operatorname*{Minimize}_{x_{a} \ \forall a\in A} \quad && \sum_{a \in A}\int_{0}^{x_{a}} \tau (\psi_{1}t_{a}(x)+\psi_{2}c_{a}(x))dx \\
                        & \text{subject to}                                    && \sum_{p}f^{s, t}_{p}=d^{s, t} \\
                        &                                                      && x_{a}=\sum_{s}\sum_{t}\sum_{p}f^{s, t}_{p} \delta^{s, t}(p,a), \quad \forall a \in A \\
                        &                                                      && x_{a}\ge 0, \quad f^{s, t}_{p}\ge0, \quad \forall s,t,p
                    \end{aligned}
                \end{equation}

                \begin{enumerate}[label=(\arabic*), itemsep=0.5pt, parsep=0pt]
                    \item \textbf{Existence}:
                        Since the feasible region defined by Eq.~\eqref{eq:Time-Carbon UE} is compact (closed and bounded) and convex,
                        and the objective function is continuous on this domain, the existence of an optimal solution is guaranteed by
                        the Weierstrass Extreme Value Theorem.

                    \item \textbf{Uniqueness}:
                        The objective function in Eq.~\eqref{eq:Time-Carbon UE} is strictly convex with respect to the link flow $x_{a}$ when $\psi_{1}>0$ and $\psi_{2}>0$,
                        and the feasible region is convex because the constraints are linear.
                        Therefore, the UE flows is unique, as guaranteed by the Strict Convexity Theorem.

                    \item \textbf{Equivalence}:
                        The objective function in Eq.~\eqref{eq:Time-Carbon UE} represents a minimization problem subject to linear conservation and non-negativity constraints.
                        Consequently, the equilibrium conditions are satisfied at any minimum point
                        if the first-order optimality conditions are equivalent to the equilibrium conditions~\citep{sun2021multi}.

                        First, the Lagrangian function of the objective function in Eq.~\eqref{eq:Time-Carbon UE} is constructed as:
                        \begin{equation}
                        \label{Lagrangian function}
                            \begin{aligned}
                                \mathcal{L} =
                                &\sum_{a \in A} \int_{0}^{x_{a}}
                                \tau \big(\psi_{1} t_{a}(x) + \psi_{2} c_{a}(x)\big) \, dx \\
                                &+ \sum_{s,t} \sigma^{s, t} \Big(d^{s, t} - \sum_{p} f_{p}^{s, t} \Big)
                                + \sum_{a} \mu_{a} \Big(x_{a} - \sum_{s,t,p} f_{p}^{s, t} \delta^{s, t}(p,a) \Big),
                            \end{aligned}
                        \end{equation}
                        where $\sigma^{s, t}$ and $\mu_a$ are the Lagrange multipliers associated with the OD demand and link flow conservation constraints, respectively.

                        Taking the derivative of $\mathcal{L}$ with respect to the path flow $f_p^{s, t}$ yields:
                        \begin{equation}
                        \label{derivative}
                            \frac{\partial \mathcal{L}}{\partial f_p^{s, t}}
                            = \sum_{a} \delta^{s, t}(p,a)
                            \tau \big(\psi_{1} t_a(x_a) + \psi_{2} c_a(x_a)\big)
                            - \sigma^{s, t}.
                        \end{equation}

                        According to the Kuhn-Tucker (KKT) conditions, we have:
                        \begin{equation}
                        \label{KKT}
                            \frac{\partial \mathcal{L}}{\partial f_p^{s, t}}
                            \begin{cases}
                                = 0, & \text{if } f_p^{s, t} > 0,\\[4pt]
                                \ge 0, & \text{if } f_p^{s, t} = 0.
                            \end{cases}
                        \end{equation}

                        Therefore, for each OD pair $(s,t)$, the following holds:
                        \begin{equation}
                        \label{EQU}
                            \sum_{a\in p} C_a(x_a)
                            =
                            \sum_{a\in p} \tau \big(\psi_{1}t_a(x_a)+\psi_{2}c_a(x_a)\big)
                            \begin{cases}
                                =\sigma^{s, t}, & \text{if } f_p^{s, t}>0,\\[4pt]
                                \ge \sigma^{s, t}, & \text{if } f_p^{s, t}=0.
                            \end{cases}
                        \end{equation}

                        This condition is exactly the Wardrop equilibrium condition for the Time-Carbon UE,
                        and the first-order conditions of the optimization problem (Eq.~\eqref{eq:Time-Carbon UE}) are equivalent to
                        the equilibrium conditions of the Time-Carbon UE model.
                        Therefore, the equivalence between the first-order conditions and the Wardrop equilibrium conditions is ensured by the KKT conditions,
                        which provide the necessary optimality conditions for this convex optimization problem.
                \end{enumerate}

        \subsection{Indicator Design}
        \label{sec:Indicator Design}
            \subsubsection{Urban Metrics}
            \label{sec:Results of UE models}
                \noindent(1) \textbf{Urban Transport Emissions}

                Urban transport emissions (i.e., carbon emissions),
                denoted as $E$ and measured in grams (g), are quantified by aggregating the emissions across all links:
                \begin{equation}
                    \label{eq:total_carbon_emissions}
                    E = \sum_{a \in A}{e_{a}(x_{a})x_{a}}
                \end{equation}

                \noindent(2) \textbf{UETT}

                The UETT, denoted as $T$ and measured in minutes (min),
                represents the average travel time in the urban road network, calculated as:
                \begin{equation}
                    \label{eq:UETT}
                    T = \frac{\sum_{a \in A}{t_{a}{(x_{a})x_{a}}}}{Q}
                \end{equation}

                \noindent(3) \textbf{UETL}

                Given that vehicle kilometers traveled (VKT) tends to rise as the population grows,
                we utilize the UETL, denoted as $L$ and measured in kilometers (km), to represent the average travel distance in the urban road network.
                This metric is calculated by normalizing $\mathrm{VKT}$ by the total demand $Q$:
                \begin{equation}
                    \label{eq:VKT}
                    \mathrm{VKT} = \sum_{a \in A}{l_{a}x_{a}}
                \end{equation}

                \begin{equation}
                    \label{eq:UETL}
                    L = \frac{\sum_{a \in A}{l_{a}x_{a}}}{Q}
                \end{equation}

                \noindent(4) \textbf{VOC}

                We utilize the VOC to measure the congestion level of link $a$, denoted as $\mathrm{VOC}_a$ and calculated as:
                \begin{equation}
                    \label{eq:VOC}
                    \mathrm{VOC}_a = \frac{x_a}{q_a}
                \end{equation}

            \subsubsection{Eco-routing Efficiency Metrics}
            \label{sec:Efficiency of Eco-routing}
                \noindent(1) \textbf{Trade-off Percentage and Optimization Percentage}

                The trade-off percentage $P_{t}$ and optimization percentage $P_{c}$ are used to
                quantify the relative changes in emissions and UETT
                between Time-Only and Time-Carbon UE models, as calculated in Eqs.~\eqref{eq:Pt} and~\eqref{eq:Pc}.
                Within a city network, $P_{c}$ is employed to evaluate the optimization performance of eco-routing,
                where a larger $P_{c}$ indicates a better performance.
                \begin{equation}
                    \label{eq:Pt}
                    P_{t}=\frac{ T_\mathrm{TC} - T_\mathrm{TO} }{ T_\mathrm{TO} } \times 100\%
                \end{equation}
                \begin{equation}
                    \label{eq:Pc}
                    P_{c}=\frac{E_\mathrm{TO} - E_\mathrm{TC}}{ E_\mathrm{TO} } \times 100\%
                \end{equation}
                where the subscripts $\mathrm{TO}$ and $\mathrm{TC}$ denote urban metrics under Time-Only and Time-Carbon UE models, respectively.

                \noindent(2) \textbf{Optimization Potential and Cost Potential}

                Examining how $P_{t}$ and $P_{c}$ evolve with respect to population size
                helps to better understand the variation in eco-routing performance.
                Specifically, $P_{t}$ and $P_{c}$ are plotted on the y-axis against population size on the x-axis,
                with the demand distribution implicitly reflected in population.
                The slope of $P_{c}$ curve at a given population is termed the optimization potential ($\mathrm{OP}$),
                while the slope of $P_{t}$ curve is termed the cost potential ($\mathrm{CP}$), calculated as Eq.~\eqref{eq:OP and CP}.
                A larger $\mathrm{OP}$ implies that the $P_{c}$ increases more rapidly as population grows,
                equivalent to a larger improvement in the eco-routing performance,
                signifying that eco-routing possesses a greater emission reduction potential.
                \begin{equation}
                \label{eq:OP and CP}
                    \mathrm{OP} = \frac{\partial P_c}{\partial N }, \quad \mathrm{CP} = \frac{\partial P_t}{\partial N }
                \end{equation}
                where $N$ represents a given population size.

                \noindent(3) \textbf{Optimization Level}

                Ideally, to ensure that eco-routing remains efficient under population growth,
                $\mathrm{OP}$ should increase to enhance future optimization performance,
                while $\mathrm{CP}$ should decrease to moderate the increase in travel time.
                Accordingly, the urban carbon optimization status under eco-routing can be classified into four distinct levels,
                as shown in Fig.~\ref{fig:Optimization Level},
                where eco-routing is considered efficient only if the city is at Level 1 or 2.
                \begin{figure}[!ht]
                    \vspace{-1.0em}
                    \centering
                    \includegraphics{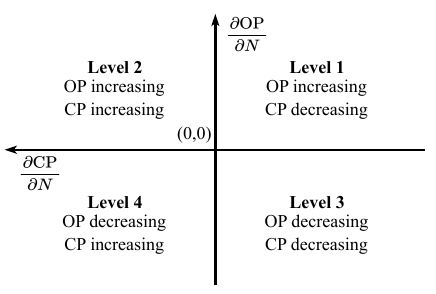}
                    \vspace{-0.7em}
                    \caption{The four optimization levels of eco-routing (with the positive x-axis extending to the left and the positive y-axis extending upwards).}
                    \label{fig:Optimization Level}
                    \vspace{-1.0em}
                \end{figure}

                In subsequent analyses, to facilitate a fair comparison of population sizes across cities,
                we transform the absolute population size into the relative population change
                with respect to a given reference population (e.g., the population under a specific scenario in a specific year), defined as:
                \begin{equation}
                \label{eq:population_percentage}
                    G_{\mathrm{ref}}(\%) = \frac{ N_{u, y} - N_{\mathrm{ref}} }{ N_{\mathrm{ref}} } \times 100\%
                \end{equation}
                where $N_{u, y}$ denotes the population size under scenario $u$ in year $y$;
                $N_{\mathrm{ref}}$ is the reference population;
                and $G_{\mathrm{ref}}$ represents the percentage change in population relative to the reference $N_{\mathrm{ref}}$ across all scenarios and years.
                Here, the scenario $u$ refers to the SSP scenario, as detailed in Section~\ref{sec:population data}.
                This transformation constitutes a strictly monotonic affine mapping of the absolute population size.
                Consequently, by the chain rule,
                the sign and the trend of both $\mathrm{OP}$ and $\mathrm{CP}$ remain invariant.
                Specifically, for any two reference populations $N_{\mathrm{ref}}^{1}$ and $N_{\mathrm{ref}}^{2}$,
                the corresponding $\mathrm{OP}$ values satisfy Eq.~\eqref{eq:chain},
                where $N_{\mathrm{ref}}^{2} / N_{\mathrm{ref}}^{1}$ is strictly positive.
                Therefore, changing the reference population preserves both the sign and the trend of $\mathrm{OP}$ and $\mathrm{CP}$,
                as well as the optimization levels.
                Analogously, we can select reference urban metrics,
                where the transformations are also strictly monotonic affine mappings of the corresponding original urban metrics.
                \begin{equation}
                    \label{eq:chain}
                    \operatorname{OP}^{2}
                    = \operatorname{OP}^{1} \cdot \frac{ N_{\mathrm{ref}}^{2} }{ N_{\mathrm{ref}}^{1} },
                \end{equation}

                The overall analytical framework of this study is illustrated in Fig.~\ref{fig:workflow}.
                \begin{figure}[!htbp]
                    \vspace{-0.7em}
                    \centering
                    \includegraphics[width=0.95\linewidth]{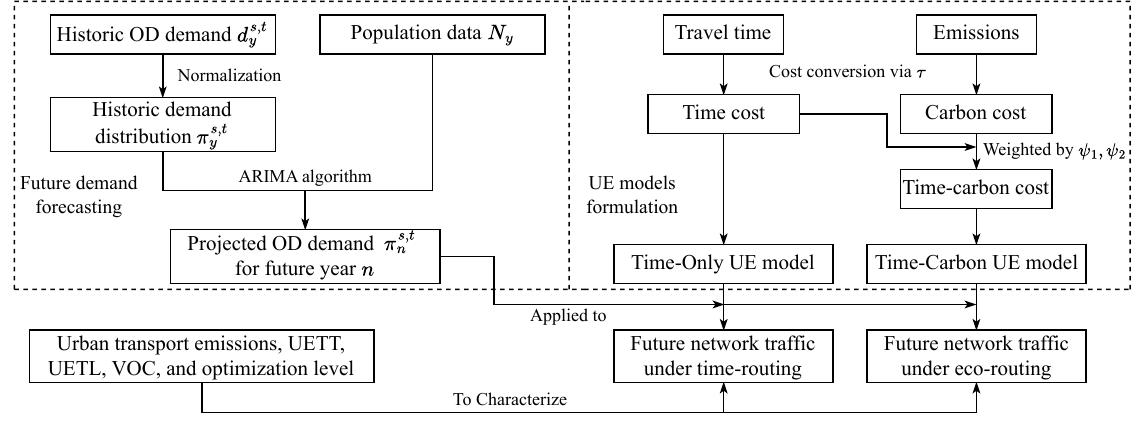}
                    \caption{The proposed analytical framework of this study.}
                    \label{fig:workflow}
                    \vspace{-1.0em}
                \end{figure}

        \subsection{Theoretical Derivations}
        \label{sec:theorem}
            This section derives general theorems based on the Time-Only and Time-Carbon UE models,
            offering theoretical insights that provide robust support for the case study findings.
            \subsubsection{Scaling Relationship between Emissions and Total Demand}
            \label{sec:theorem1}
                \begin{theorem}[Linear Scaling of Emissions with Total Demand]
                \label{thm:scaling}
                    Let $d(1)=\{d^{s, t}(1)\}_{(s,t) \in W}$ denote the base demand.
                    For a demand scaling factor $\lambda > 0$,
                    the scaled demand $d(\lambda)=\{d^{s, t}(\lambda)\}_{(s,t) \in W}$ is defined by $d^{s, t}(\lambda)=\lambda d^{s, t}(1)$;
                    $x(\lambda)=\{x_a(\lambda)\}_{a \in A}$ represents
                    the UE flows corresponding to the scaled demand $d(\lambda)$;
                    and the total carbon emission is given by:
                    \begin{equation}
                        \label{eq:thm1.1}
                        E(\lambda) = \sum_{a \in A} e_a(x_a(\lambda)) x_a(\lambda) = \sum_{a \in A} \gamma l_a h(v_a) x_a(\lambda)
                    \end{equation}

                    Within typical road networks, there exist constants $c_1$, $c_2 > 0$, which are independent of $\lambda$, such that:
                    \begin{equation}
                        \label{eq:thm1.2}
                        c_1 \lambda \leq E(\lambda) \leq c_2 \lambda
                    \end{equation}
                    implying that $E(\lambda) = \Theta(\lambda)$ and emissions scale linearly with total demand.
                \end{theorem}

                \begin{proof}
                    For a demand scaling factor $\lambda > 0$, the feasible region of the Beckmann convex program (Eq.~\eqref{eq:Time-Carbon UE}) is:
                    \begin{equation}
                    \label{eq:thm2.1}
                        K(\lambda) = \left\{ x \in \mathbb{R}^{|A|}_+ :
                            \exists f \geq 0, \; x_a = \sum_{s,t}\sum_{p} f_p^{s, t} \delta^{s, t}(p,a), \; \sum_{p} f_p^{s, t} = d^{s, t}(\lambda) \right\}
                    \end{equation}

                    \noindent \textbf{Step 1: Linear scaling of link flows with demand.}

                    Define the normalized link flows $\xi(\lambda) = \{\xi_a(\lambda)\}_{a \in A}$ by $\xi_a(\lambda) = x_a(\lambda) / \lambda$.
                    Since the UE flows $x(\lambda)$ belong to the feasible set $K(\lambda)$,
                    it follows from the linearity of constraints that the normalized flows $\xi(\lambda)$ must belong to the unit-demand feasible set $K(1)$.
                    Given that $K(1)$ is a non-empty compact convex set, the feasible region for the normalized flows $\xi(\lambda)$ is bounded;
                    consequently, each $\xi_a(\lambda)$ is uniformly bounded with respect to $\lambda$.
                    The relationship $x_a(\lambda) = \lambda \xi_a(\lambda)$ thus implies that the link flows scale linearly with the demand scaling factor $\lambda$,
                    demonstrating a linear scaling property between link flows and demand.

                    \noindent \textbf{Step 2: Boundedness of the fuel consumption factor.}

                    Within typical road networks, the vehicle speed $v_a$ is physically constrained to the range $[0, v_{\mathrm{max}}]$.
                    Consequently, there exist positive constants $h_{\mathrm{min}}$ and $h_{\mathrm{max}}$ such that the fuel consumption factor, $h(v_a)$, satisfies:
                    \begin{equation}
                    \label{eq:thm2.2}
                        h_{\mathrm{min}} = \varphi_2 \leq h(v_a) \leq h_{\mathrm{max}} < \infty
                    \end{equation}

                    Notably, as the demand scaling factor $\lambda \to \infty$, severe congestion causes travel time $t_a \to \infty$ and speed $v_a \to 0$.
                    Under these conditions, $h(v_a)$ remains bounded, rather than diverging to infinity:
                    \begin{equation}
                    \label{eq:thm2.3}
                        h_{\mathrm{max}} = h(0) = \varphi_1 v_e^2 + \varphi_2 < \infty
                    \end{equation}

                    \noindent \textbf{Step 3: Upper and lower bounds of emissions.}

                    \noindent \textit{Upper Bound:}
                    \begin{equation}
                    \label{ea:thm2.4}
                        E(\lambda) = \sum_{a \in A} \gamma l_a h(v_a) x_a(\lambda)
                            \leq \gamma h_{\mathrm{max}} \sum_{a \in A} l_a x_a(\lambda)
                            = \gamma \cdot h_{\mathrm{max}} \cdot \mathrm{VKT}(\lambda)
                    \end{equation}
                    where $\mathrm{VKT}(\lambda) = \sum_a l_a x_a(\lambda) = \lambda \sum_a l_a \xi_a(\lambda) \leq \lambda \cdot L_{\mathrm{max}} \cdot Q(1)$.
                    Here, $L_{\mathrm{max}}$ denotes the maximum feasible path length between OD pairs in the network,
                    and $Q(1)$ represents the total base demand, both of which are bounded and independent of $\lambda$,
                    allowing us to define the positive constant $c_2 = \gamma \cdot h_{\mathrm{max}} \cdot L_{\mathrm{max}} \cdot Q(1)$.

                    \vspace{1.1mm}
                    \noindent \textit{Lower Bound:}
                    \begin{equation}
                    \label{ea:thm2.5}
                        E(\lambda) \geq \gamma \cdot h_{\mathrm{min}} \cdot \mathrm{VKT}(\lambda)
                            \geq \gamma \cdot h_{\mathrm{min}} \cdot \lambda \cdot L_{\mathrm{min}} \cdot Q(1) = c_1 \lambda
                    \end{equation}
                    where $L_{\mathrm{min}}$ denotes the minimum feasible path length between OD pairs in the network,
                    and we define the positive constant $c_1 = \gamma \cdot h_{\mathrm{min}} \cdot L_{\mathrm{min}} \cdot Q(1)$.

                    Combining these results,
                    there exist positive constants $c_1$ and $c_2$ such that $c_1 \lambda \leq E(\lambda) \leq c_2 \lambda$,
                    which implies $E(\lambda) = \Theta(\lambda)$.
                \end{proof}

                \begin{remark}[Local Scaling Characteristics]
                \label{re:Local scaling}
                    Theorem~\ref{thm:scaling} establishes the global linear relationship between emissions and total demand.
                    However, from a local perspective, particularly within the finite range of $\lambda$ considered in this study,
                    $E(\lambda)$ may exhibit sublinear or superlinear growth trends.
                \end{remark}

                \begin{corollary}[Invariant Scaling Order under Routing Strategies]
                \label{cor:eco_coefficient}
                    For both Time-Only $(\psi_2=0)$ and Time-Carbon $(\psi_2>0)$ UE models,
                    the emissions satisfy:
                    \begin{equation}
                        \label{eq:corboth}
                        E^{\mathrm{TO}}(\lambda) = \Theta(\lambda), \quad E^{\mathrm{TC}}(\lambda) = \Theta(\lambda)
                    \end{equation}
                    where the superscripts $\mathrm{TO}$ and $\mathrm{TC}$ denote emissions under Time-Only and Time-Carbon UE models,
                    indicating that their global scaling orders are identical.

                    Furthermore, there exists a constant $\rho \in (0, 1)$ such that:
                    \begin{equation}
                        \label{eq:thm3.2}
                            E^{\mathrm{TC}}(\lambda) \approx \rho \cdot E^{\mathrm{TO}}(\lambda)
                    \end{equation}
                    implying that eco-routing reduces emissions while retaining the local scaling order.
                \end{corollary}

                \begin{proof}
                    The derivation of Theorem~\ref{thm:scaling} demonstrates that the linear scaling relationship between $E(\lambda)$ and $\lambda$
                    is determined by the boundedness of $h(v)$, independent of the cost weights $\psi_1$ and $\psi_2$;
                    consequently, their scaling orders are identical.

                    Furthermore, the empirical results in Table~\ref{tab:Transport emissions} indicate that as $\lambda$ grows,
                    the optimization percentage $P_c$ consistently remains within the range $(0, 1)$ and exhibits only a small increase.
                    This suggests the existence of a constant $0 < \rho < 1$ such that
                    $E^{\mathrm{TC}}(\lambda) \approx \rho \cdot E^{\mathrm{TO}}(\lambda)$,
                    thereby retaining the local scaling order observed under time-routing.
                \end{proof}

                \begin{remark}[Policy Implications]
                \label{rm:fundamental approach}
                    Given that the total demand is directly proportional to the population size according to Eq.~\eqref{eq:P},
                    Theorem~\ref{thm:scaling} and Corollary~\ref{cor:eco_coefficient} reveal that
                    the linear scaling relationship between emissions and population $($total demand$)$
                    is an intrinsic property of urban road networks,
                    thereby remaining invariant regardless of routing strategies.
                    Consequently, the fundamental approach to mitigating emissions driven by future population growth
                    lies in managing the associated rise in total demand,
                    or following a sustainable development pathway characterized by slower population growth.
                \end{remark}

            \subsubsection{Critical Role of Travel Distance Reduction in Emission Mitigation}
            \label{sec:theorem2}
                \begin{theorem}[Critical Role of Travel Distance Reduction in Emission Mitigation]
                \label{thm:short_distance}
                    Define the average fuel consumption factor as $\bar{h} = E / (\gamma \cdot Q \cdot L)$.
                    Consequently, the emissions can be reformulated as $E = \gamma Q L \bar{h}$.
                    Let $\bar{v}$ denote the average travel speed, defined as $\bar{v} = L/T$.

                    The elasticity of fuel consumption factor $h(v)$ with respect to vehicle speed is defined as:
                    \begin{equation}
                    \label{eq:thm6.1}
                        \eta(v) = \frac{v h'(v)}{h(v)} = \frac{2\varphi_1 v (v - v_e)}{\varphi_1 (v-v_e)^2 + \varphi_2}
                    \end{equation}
                    which quantifies the percentage change in fuel consumption factor resulting from a 1\% change in vehicle speed.

                    The variation in emissions corresponding to the transition from
                    the initial state $(Q^1, L^1, \bar{h}^1)$ to the final state $(Q^2, L^2, \bar{h}^2)$ can be decomposed as:
                    \begin{equation}
                    \label{eq:decomposition1.1}
                        \Delta E = E^2 - E^1 =
                        \gamma Q^1 ( \underbrace{\bar{h}^1 \Delta L}_{ \Delta E_\mathrm{dist} }
                        + \underbrace{L^2 \Delta \bar{h}}_{ \Delta E_\mathrm{fuel} } )
                    \end{equation}
                    where $\Delta L = (Q^2 / Q^1) \cdot L^2 - L^1$ and $\Delta \bar{h} = (Q^2 / Q^1) \cdot (\bar{h}^2 - \bar{h}^1)$;
                    $\Delta E_\mathrm{dist}$ and $\Delta E_\mathrm{fuel}$
                    represent the distance effect and the fuel consumption factor effect, respectively.

                    Within typical road networks, the elasticity of $h(v)$ satisfies $| \eta(v) | < 1$,
                    and the following relationship holds:
                    \begin{equation}
                    \label{eq:thm6.2}
                        \frac{\left| \Delta E_\mathrm{fuel} \right|}{\left| \Delta E_\mathrm{dist} \right|} \le
                            | \eta(\bar{v}^1) | \cdot \frac{\left| \frac{\bar{v}^2 - \bar{v}^1}{\bar{v}^1} \right|}{\left| \frac{L^2 - L^1}{L^2} \right|}
                    \end{equation}

                    We analyze the emission variations in the following two scenarios:
                    \begin{enumerate}[label=(\arabic*), itemsep=0.5pt, parsep=0pt]
                        \item Transition from time-routing to eco-routing $(\Delta E < 0)$:
                            The system state shifts from $(Q, L^{\mathrm{TO}}, \bar{h}^{\mathrm{TO}})$ to $(Q, L^{\mathrm{TC}}, \bar{h}^{\mathrm{TC}})$.
                            To maximize emission reduction, $\Delta E$ should be as negative as possible.
                        \item Using eco-routing under population growth $(\Delta E > 0)$:
                            The system state shifts from $(Q(\lambda_1), L^{\lambda_1}, \bar{h}^{\lambda_1})$ to $(Q(\lambda_2), L^{\lambda_2}, \bar{h}^{\lambda_2})$,
                            where $Q(\lambda_{1}) < Q(\lambda_{2})$.
                            For effective emission mitigation, $\Delta E$ should be minimized.
                    \end{enumerate}

                    The inequality in Eq.~\eqref{eq:thm6.2} implies that in both scenarios, to achieve effective emission mitigation,
                    reducing travel distance is significantly more critical than improving $($or maintaining$)$ travel speed.
               \end{theorem}

                \begin{proof}
                    \noindent \textbf{Step 1: Derivation of the formula.}

                    Starting from the definition $E = \gamma Q L \bar{h}$:
                    \begin{equation}
                        \begin{aligned}
                        \label{eq:thm7.1}
                            \Delta E &= \gamma Q^2 L^2 \bar{h}^2 - \gamma Q^1 L^1 \bar{h}^1 \\
                            &= \gamma Q^1 \left[ \frac{Q^2}{Q^1}
                            \left( L^2 \bar{h}^2 - L^2 \bar{h}^1 + L^2 \bar{h}^1 \right)
                            - L^1 \bar{h}^1 \right] \\
                            &= \gamma Q^1 \left[ \bar{h}^1 \left( \frac{Q^2}{Q^1}L^2 - L^1 \right)
                            + L^2 \frac{Q^2}{Q^1} \left( \bar{h}^2 - \bar{h}^1 \right) \right]
                        \end{aligned}
                    \end{equation}

                    Within typical road networks, the average speed generally satisfies $\bar{v} \le 30$ km/h~\citep{TRB2016HCM}.
                    Consequently, given the parameters $\varphi_{1}$, $\varphi_{2}$, and $v_{e}$ configured as in Section~\ref{sec:3.3.2},
                    this implies that $| \eta(\bar{v}) | < 1$.
                    Considering the definition of elasticity, we obtain the Eq.~\eqref{eq:thm7.2},
                    which holds trivially when $Q^1 = Q^2$;
                    in cases where $Q^1 \neq Q^2$, it follows from the physical reality that the relative variation of $L$ is significantly smaller than that of $Q$:
                    \begin{equation}
                    \label{eq:thm7.2}
                        \frac{\left| \Delta E_\mathrm{fuel} \right|}{\left| \Delta E_\mathrm{dist} \right|}
                            = \frac{\left| \frac{\frac{Q^2}{Q^1}(\bar{h}^2 - \bar{h}^1)}{\bar{h}^1} \right|}
                                    {\left| \frac{\frac{Q^2}{Q^1}L^2 - L^1}{L^2} \right|}
                            \approx \frac{\left| \eta(\bar{v}^1) \cdot \frac{\frac{Q^2}{Q^1}(\bar{v}^2 - \bar{v}^1)}{\bar{v}^1} \right|}
                                    {\left| \frac{\frac{Q^2}{Q^1}L^2 - L^1}{L^2} \right|}
                            \le | \eta(\bar{v}^1) | \cdot \frac{\left| \frac{\bar{v}^2 - \bar{v}^1}{\bar{v}^1} \right|}{\left| \frac{L^2 - L^1}{L^2} \right|}
                    \end{equation}

                    \noindent \textbf{Step 2: Critical Role of Travel Distance Reduction.}

%

                    Assuming that variations in $\bar{v}$ and $L$ are independent,
                    a change in $\bar{v}$ affects solely $\Delta E_\mathrm{fuel}$
                    and leaves $\Delta E_\mathrm{dist}$ unchanged.
                    Eq.~\eqref{eq:thm7.2} implies that a $1\%$ change in speed $\bar{v}$,
                    yielding $\left| \Delta E_\mathrm{fuel} \right| > 0$,
                    only requires a change in $L$ of at most $|\eta(\bar{v})|\% < 1\%$
                    to achieve $\left| \Delta E_\mathrm{dist} \right| = \left| \Delta E_\mathrm{fuel} \right| $.
                    Therefore, $\Delta E$ is significantly more sensitive to changes in travel distance $L$ than to changes in speed $\bar{v}$.

                    \vspace{0.5mm}
                    \noindent \textit{From time-routing to eco-routing:}

                    Consider the system transition from $(Q, L^{\mathrm{TO}}, \bar{h}^{\mathrm{TO}})$ to $(Q, L^{\mathrm{TC}}, \bar{h}^{\mathrm{TC}})$,
                    where $\Delta E < 0$.
                    Ideally, a smaller $L^{\mathrm{TC}}$ and a higher $\bar{v}^{\mathrm{TC}}$ (closer to optimal speed $v_e$)
                    would contribute to maximizing the emission reduction (i.e., making $\Delta E$ more negative).
                    However, within typical road networks, a trade-off exists:
                    prioritizing a shorter distance $(L^{\mathrm{TC}} < L^{\mathrm{TO}})$ implies a reduction in speed $(\bar{v}^{\mathrm{TC}} < \bar{v}^{\mathrm{TO}})$,
                    whereas prioritizing higher speed $(\bar{v}^{\mathrm{TC}} > \bar{v}^{\mathrm{TO}})$ comes at the expense of longer travel distance $L^{\mathrm{TC}} > L^{\mathrm{TO}}$.
                    Recalling that $\Delta E$ is significantly more sensitive to travel distance than to speed,
                    we conclude that reducing travel distance is far more critical than improving travel speed for maximizing emission reduction.

                    \vspace{0.5mm}
                    \noindent \textit{Eco-routing under population growth:}

                    Consider the system evolution driven by population growth
                    from $(Q(\lambda_1), L^{\lambda_1}, \bar{h}^{\lambda_1})$ to $(Q(\lambda_2), L^{\lambda_2}, \bar{h}^{\lambda_2})$,
                    where $\Delta E > 0$.
                    Similarly, a smaller $L^{\lambda_2}$ and a higher $\bar{v}^{\lambda_2}$ (closer to optimal speed) would contribute to minimizing the emission growth.
                    However, within typical road networks,
                    population growth leads to increased demand scaling factor $\lambda$,
                    which raises link flows $x(\lambda)$ and inevitably reduces travel speeds $v(\lambda)$.
                    Prioritizing maintaining speed $\bar{v}^{\lambda_2}$ inevitably results in a longer travel distance
                    compared to prioritizing reducing $L^{\lambda_2}$.
                    Analogously, given the higher sensitivity of $\Delta E$ to $L$,
                    prioritizing reducing travel distance is significantly more important than maintaining travel speed for mitigating emissions.
                \end{proof}

                \begin{remark}[Policy Implications]
                \label{re:construction strategies}
                    Theorem~\ref{thm:short_distance} establishes that to minimize emissions,
                    reducing travel distance is significantly more critical than improving $($or maintaining$)$ travel speed,
                    whether shifting from time-routing to eco-routing or using eco-routing under population growth.
                    Consequently, to reduce emissions and construct low-carbon road networks,
                    on the one hand, it is crucial to enhance the jobs-housing balance and promote polycentric development to shorten travel distance.
                    On the other hand, strategies should focus on improving the short routes between OD pairs.

                \end{remark}

            \subsubsection{Impact of Capacity Expansion on Emissions}
            \label{sec:theorem3}
                \begin{theorem}[Impact of Capacity Expansion on Emissions]
                \label{thm:sensitivity}
                    For specific link $k$, the sensitivity of emissions $E$ with respect to capacity $q_k$ is given by:
                    \begin{equation}
                    \label{eq:dCE_dq}
                        \frac{\partial E}{\partial q_k}
                            = x_k \cdot \frac{\mathrm{d} e_k}{\mathrm{d} t_k} \cdot \frac{\partial t_k}{\partial q_k}
                            = - \alpha \beta t_k^0 \cdot \frac{\mathrm{d} e_k}{\mathrm{d} t_k} \cdot \left( \frac{x_k}{q_k} \right)^{\beta+1}
                    \end{equation}
                    where $\mathrm{d} e_k / \mathrm{d} t_k > 0$ and $\partial t_k / \partial q_k < 0$.
                    Consequently, $\partial E / \partial q_k < 0$.
                \end{theorem}

                \begin{proof}
                    \begin{equation}
                    \label{eq:thm8.1}
                        \frac{\partial E}{\partial q_k} = \frac{\partial \sum_{a \in A} x_a e_a(x_a)}{\partial  q_k}
                            = x_k \cdot \frac{\mathrm{d} \left( \gamma l_k h(\frac{l_k}{t_k}) \right) }{\mathrm{d} t_k} \cdot \frac{\partial t_k}{\partial q_k}
                            = - \alpha \beta t_k^0 \cdot \frac{\mathrm{d} e_k}{\mathrm{d} t_k} \cdot \left( \frac{x_k}{q_k} \right)^{\beta+1}
                    \end{equation}

                    Based on the physical interpretation of the fuel consumption factor $h(v_k)$,
                    the link emission $e_k$ increases monotonically with travel time $t_k$,
                    yielding $\mathrm{d} e_k / \mathrm{d} t_k > 0$.
                    Combining this with the fact that $ \partial t_k / \partial q_k < 0$,
                    it follows immediately that $ \partial E / \partial q_k < 0$.
                \end{proof}

                \begin{remark}[Impact of Flow Redistribution]
                \label{re:redistribution}
                    The partial derivative of the generalized cost $C_k$ on link $k$ is negative:
                    \begin{equation}
                    \label{eq:ma1.1}
                        \frac{\partial C_k}{\partial q_k}
                            = \tau \left( \psi_1 \frac{\partial t_k}{\partial q_k} + \psi_2 \frac{\partial c_k}{\partial q_k} \right) < 0
                    \end{equation}
                    indicating that capacity expansion leads to a reduction in $C_k$.
                    In this study, the expanded links are located on the short routes between OD pairs.
                    Consequently, the reduced generalized cost on these links
                    attracts more traffic flow under redistribution,
                    thereby reducing both $\mathrm{UETL}$ and $\mathrm{VKT}$.
                    Given the critical role of distance reduction in emission mitigation,
                    coupled with the fact that capacity expansion enhances average travel speed,
                    the emissions ultimately decrease following the redistribution,
                    indicating that the total derivative of $E$ satisfies $ \mathrm{d} E / \mathrm{d} q_k < 0$.
                \end{remark}

                \begin{proposition}[Properties of Carbon Influence]
                \label{prop:zipf}
                    The carbon influence of link $k$ is defined as:
                    \begin{equation}
                    \label{eq:carbon_influence}
                        \mathcal{I}_k = \left| \frac{\partial E}{\partial q_k} \right|
                    \end{equation}
                    which quantifies the change in emissions resulting from a unit capacity expansion on link $k$.
                    Within a road network, the following properties hold:
                    \begin{enumerate}[label=(\arabic*), itemsep=0.5pt, parsep=0pt]
                        \item Dependence on flow and capacity: $\mathcal{I}_k \propto \left( x_k / q_k \right)^{\beta+1} = \mathrm{VOC}_k^{\beta+1}$.
                        \item Significant heterogeneity across links: The distribution of carbon influence is highly uneven,
                                with the ratio between the maximum and minimum values spanning several orders of magnitude:
                                $ \mathcal{I}_{\mathrm{max}} / \mathcal{I}_{\mathrm{min}} \approx 10^3 \textnormal{--} 10^5$.
                    \end{enumerate}
                \end{proposition}

                \begin{proof}
                    \noindent \textbf{Step 1: Expression for Carbon Influence.}

                    According to Theorem~\ref{thm:sensitivity}, the carbon influence is given by:
                    \begin{equation}
                    \label{eq:prop13.1}
                        \mathcal{I}_k = \left| \frac{\partial E}{\partial q_k} \right|
                        = \left| \frac{d e_k}{d t_k} \cdot  \left( -\alpha \beta t_k^0 \cdot \frac{x_k^{\beta+1}}{q_k^{\beta+1}} \right) \right|
                        \propto \left( \frac{x_k}{q_k} \right)^{\beta+1}
                        = \mathrm{VOC}_k^{\beta+1}
                    \end{equation}

                    \noindent \textbf{Step 2: Significant Heterogeneity of Carbon Influence.}

                    Based on typical urban road functional classification standards~\citep{FHWA2013functional, TRB2016HCM},
                    we adopt the parameters for a representative road network during the peak hour,
                    where Principal Arterials exhibit VOC ranging from 0.6 to 0.85, and Local roads range from 0.1 to 0.3.
                    Setting the standard BPR parameter $\beta = 4$~\citep{sheffi1985urban},
                    we calculate the ratio of carbon influence between the Principal Arterial $(\mathcal{I}_1)$ and the Local road $(\mathcal{I}_2)$:
                    \begin{equation}
                    \label{eq:prop13.2}
                        \frac{\mathcal{I}_1}{\mathcal{I}_2}
                             \approx \left( \frac{\mathrm{VOC}_1}{\mathrm{VOC}_2} \right)^5 \\
                             = \left( \frac{0.75}{0.15} \right)^5
                            = 3125 = O(10^3)
                    \end{equation}

                    Considering more extreme scenarios under population growth
                    where severely congested links exhibit $\mathrm{VOC} > 1$,
                    the resulting ratio can reach the order of $10^4 \textnormal{--} 10^5$.
                \end{proof}

                \begin{remark}[Policy Implications]
                \label{re:dynamic carbon influence}
                    A larger carbon influence $\mathcal{I}_k$ indicates that
                    a unit capacity expansion on link $k$ yields a greater reduction in emissions.
                    Consequently, given the significant heterogeneity of carbon influence revealed by Proposition~\ref{prop:zipf},
                    expansion strategies should prioritize bottleneck links with high $\mathrm{VOC}$ to maximize emission reduction efficiency.
                    Furthermore, effective planning requires extending beyond static indicators
                    to account for the dynamic evolution of network traffic under population growth.
                    Specifically, expansion planning should target
                    links projected to experience significant increases in $\mathrm{VOC}$,
                    thereby ensuring robust emission mitigation for future development.
                \end{remark}

                \begin{corollary}[Invariant Scaling Order under Bounded Capacity Expansion]
                \label{cor:finite_expansion}
                    Consider a bounded capacity expansion scheme, defined as
                    $\mathbf{q}' = \mathbf{q} + \Delta \mathbf{q}$ subject to $\|\Delta \mathbf{q}\| \leq B$,
                    where $\mathbf{q}$ and $\mathbf{q}'$ denote the link capacity vectors before and after expansion, respectively;
                    $\Delta \mathbf{q}$ represents the expansion vector; and $B$ is the budget constraint.
                    Let $E(\lambda; \mathbf{q})$ and $E(\lambda; \mathbf{q}')$ denote the emissions before and after expansion.
                    Then:
                    \begin{enumerate}[label=(\arabic*), itemsep=0.5pt, parsep=0pt]
                        \item Invariant Scaling Order: $E(\lambda; \mathbf{q}) = \Theta(\lambda)$ and $E(\lambda; \mathbf{q}') = \Theta(\lambda)$.
                        \item Proportional Reduction: There exists a constant $\rho_q \in (0, 1]$ such that
                            $E(\lambda; \mathbf{q}') \approx \rho_q \cdot E(\lambda; \mathbf{q})$,
                            indicating that bounded capacity expansion reduces emissions while retaining the local scaling order.
                    \end{enumerate}
                \end{corollary}

                \begin{proof}
                    According to Theorem~\ref{thm:scaling},
                    the scaling order of carbon emissions is fundamentally determined by the boundedness of the fuel consumption factor $h(v)$.
                    While the bounded expansion strategy modifies link capacities, it does not alter this boundedness.
                    Therefore, the scaling order remains invariant at $\Theta(\lambda)$.

                    Proposition~\ref{prop:zipf} establishes that capacity expansion leads to a reduction in carbon emissions.
                    Analogous to Corollary~\ref{cor:eco_coefficient},
                    the results in Table~\ref{tab:Expanding table} indicate the existence of a constant $\rho_q \in (0, 1]$, such that
                    $E(\lambda; \mathbf{q}') \approx \rho_q \cdot E(\lambda; \mathbf{q})$,
                    thereby retaining the local scaling order.
                \end{proof}


                \begin{proposition}[Optimal Bounded Capacity Expansion Scheme and Shadow Price]
                \label{prop:shadow_price}
                    Consider the optimal capacity expansion problem subject to a budget constraint:
                    \begin{equation}
                        \label{eq:prop10.1}
                        \operatorname*{Minimize}_{q'_a \geq q_a \ \forall a \in A} E(\mathbf{q}')
                            \quad \text{s.t.} \quad \sum_{a \in A} \kappa_a (q'_a - q_a) \leq B
                    \end{equation}
                    where $E(\mathbf{q}')$ denotes the emissions after expansion;
                    and $\kappa_a > 0$ represents the unit expansion cost of link $a$.
                    To minimize $E(\mathbf{q}')$ within the budget constraint,
                    for any expanded link $a$ $(q'_a > q_a)$ and for any non-expanded link $b$ $(q'_b = q_b)$,
                    the following optimality condition must hold:
                    \begin{equation}
                        \label{eq:prop10.2}
                        \frac{\mathcal{I}_a}{\kappa_a} = \nu, \quad
                        \frac{\mathcal{I}_b}{\kappa_b} \le \nu
                    \end{equation}
                    where $\nu \geq 0$ is the shadow price of the budget constraint.
                    Consequently, the optimal scheme prioritizes allocating the budget to links with the highest carbon influence per unit cost.
                    Upon budget exhaustion, the carbon influence per unit cost equalizes across all expanded links and corresponds to the shadow price.
                \end{proposition}

                \begin{proof}
                    Construct the Lagrangian function for the constrained minimization problem:
                    \begin{equation}
                        \label{eq:prop11.1}
                        \mathcal{L}(\mathbf{q}', \nu) = E(\mathbf{q}') + \nu \left( \sum_{a \in A} \kappa_a (q'_a - q_a) - B \right)
                    \end{equation}
                    where $\nu$ is the Lagrange multiplier associated with the budget constraint.
                    The KKT conditions with respect to $q'_a$ require:
                    \begin{equation}
                        \label{eq:prop11.2}
                        \frac{\partial \mathcal{L}}{\partial q'_a} = \frac{\partial E(\mathbf{q}')}{\partial q'_a} + \nu \kappa_a \geq 0, \ \textnormal{and} \
                        (q'_a - q_a ) \frac{\partial \mathcal{L}}{\partial q'_a} = 0
                    \end{equation}

                    By definition, the carbon influence $\mathcal{I}_a = \left| \partial E / \partial q'_a \right|$.
                    Since $ \partial E / \partial q'_a < 0$,
                    we have $ \partial E / \partial q'_a = - \mathcal{I}_a$.
                    Substituting this into the conditions yields:
                    \begin{equation}
                        \label{eq:prop11.3}
                        -\mathcal{I}_a + \nu \kappa_a
                        \begin{cases}
                            = 0, & \text{if } q'_a > q_a,\\
                            \geq 0, & \text{if } q'_a = q_a.
                        \end{cases}
                        \quad \Rightarrow \quad \frac{\mathcal{I}_a}{\kappa_a}
                        \begin{cases}
                            = \nu, & \text{if } q'_a > q_a,\\
                            \le \nu, & \text{if } q'_a = q_a.
                        \end{cases}
                    \end{equation}
                    where the multiplier $\nu \geq 0$ can be interpreted as the shadow price of the budget constraint.
                \end{proof}

            \subsection{Summary: Key Findings and Theoretical Support}
            \label{sec:theory_summary}
                The core findings of this section and their theoretical support are summarized in Table~\ref{tab:theory_support}.
                \begin{table}[H]
                \vspace{-1.0em}
                \centering
                \small
                \renewcommand{\arraystretch}{1.5}
                \caption{Key findings and theoretical support}
                \label{tab:theory_support}
                \begin{tabular}{
                    >{\raggedright\arraybackslash}m{0.60\textwidth}
                    >{\raggedright\arraybackslash}m{0.35\textwidth}
                }
                    \toprule
                    \centering Key Findings & \centering Theoretical Support \tabularnewline
                    \midrule

                    Emissions scale linearly with population; \newline
                    neither eco-routing nor bounded capacity expansion fundamentally alters this scaling order.
                    & Theorem~\ref{thm:scaling}; \newline
                      Corollary~\ref{cor:eco_coefficient} \& Corollary~\ref{cor:finite_expansion} \\

                    Emission mitigation is primarily driven by shortening travel distance rather than maintaining travel speed;
                    & Theorem~\ref{thm:short_distance} \\

                    Expanding a few bottlenecks with high $\mathrm{VOC}$ yields significant emission reduction benefits; \newline
                    the optimal expansion scheme under budget constraints is governed by the shadow price.
                    & Theorem~\ref{thm:sensitivity} $\to$ Proposition~\ref{prop:zipf} \newline
                      $\to$ Proposition~\ref{prop:shadow_price} \\

                    Policies should extend beyond routing and construction strategies to manage the rising demand driven by population growth.
                    & Corollary~\ref{cor:eco_coefficient} \& Corollary~\ref{cor:finite_expansion} \\
                    \bottomrule
                \end{tabular}
                \vspace{-1.0em}
                \end{table}


                \begin{remark}[Theoretical Limitations]
                    This study is subject to the following theoretical limitations:
                    \vspace{-0.5em}
                    \begin{enumerate}[label=(\arabic*), itemsep=0.5pt, parsep=0pt]
                        \item All theorems are derived assuming the \textnormal{BPR} link performance function.
                        \item The fuel consumption factor $h(v)$ is assumed to follow a quadratic form,
                            whereas operational emission models in practice, such as the Motor Vehicle Emission Simulator $(\mathrm{MOVES})$,
                            exhibit greater complexity.
                    \end{enumerate}
                \end{remark}

    \section{Study Area and Data}
    \label{sec:data}
        \subsection{Study Area}
        \label{sec:study area}
            To ensure the generalizability of the findings, we considered factors such as
            geographic location, urban scale, network topology, and demographic characteristics to select six representative U.S. cities:
            San Francisco, Dallas, Pittsburgh, Miami, Denver, and Philadelphia.
            Their geographical locations and basic information are presented in Fig.~\ref{fig:1} and Table~\ref{tab:basic city},
            with population and land area data derived from 2020 Census statistics~\citep{USCensus2024}.
            Specifically, these cities are distributed across the United States,
            including larger cities such as San Francisco, Dallas, Miami, and Philadelphia,
            as well as smaller cities like Denver and Pittsburgh.
            Among these cities, Dallas has the largest land area and total population,
            while San Francisco exhibits the highest population density.
            The network topologies of these cities are also distinct.
            For example, San Francisco is divided by bay terrain,
            forcing traffic to rely on key bridges and tunnels.
            In contrast, Dallas features relatively flat and continuous terrain,
            forming grid-like transport systems dominated by surface roads.

            \begin{figure}[!htb]
                \centering
                \includegraphics[width=0.98\linewidth]{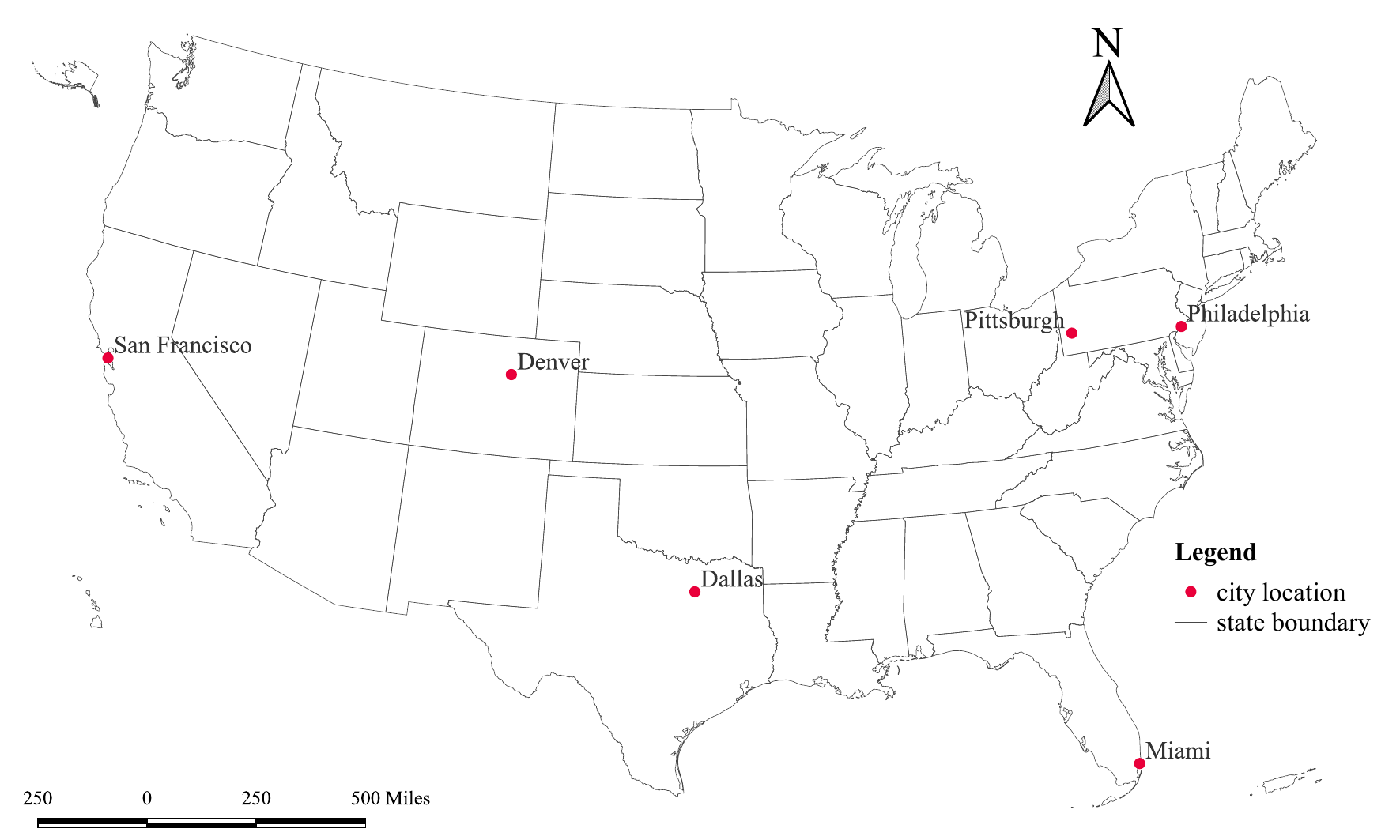}
                \caption{The geospatial distribution of the six representative U.S. cities.}
                \label{fig:1}
                \vspace{-1.0em}
            \end{figure}

            \begin{table}[!htb]
                \centering
                \footnotesize
                \caption{Basic information for the six representative U.S. cities.}
                \label{tab:basic city}
                \begin{tabularx}{0.8\textwidth}{l c*{3}{C}}
                    \toprule
                    City & State & Land Area ($\mathrm{km}^2$) & 2020 Census & Population Density \\
                    \midrule
                    San Francisco   & California & 121.5 & 873,965 & 7,195 \\
                    Dallas          & Texas & 879.6 & 1,304,379 & 1,483 \\
                    Pittsburgh      & Pennsylvania & 143.5 & 302,971 & 2,112 \\
                    Miami           & Florida & 93.2 & 442,241 & 4,743 \\
                    Denver          & Colorado & 396.5 & 715,522 & 1,805 \\
                    Philadelphia    & Pennsylvania & 348.1 & 1,603,797 & 4,607 \\
                    \bottomrule
                \end{tabularx}
                \vspace{-1.0em}
            \end{table}

        \subsection{OD Demand and Road Network Data}
        \label{sec:od dmand data}
            In this study, we constructed OD demand data using the LODES dataset,
            which has been widely adopted in existing studies~\citep{mckinney2021total}.
            This dataset is derived from administrative records where employers report employee work and home locations,
            thereby offering comprehensive block-level workforce commuting patterns across all U.S. states from 2002 to 2019.
            The U.S. Census Bureau collaborates with the state agencies to process this information into detailed OD pairs,
            making the LODES dataset an excellent basis for the initial OD demand matrix~\citep{LODES7}.

            The urban road network data required for traffic assignment are derived from the OpenStreetMap database~\citep{OpenStreetMap},
            an open-source mapping platform providing comprehensive road network data.
            This allows us to extract the nodes, links, network topologies, and road attributes of the six selected U.S. cities.
            As shown in Table~\ref{tab:city parameters},
            Dallas possesses the most complex road network, featuring 21,389 nodes and 77,818 links,
            while Pittsburgh has the smallest, with 3,532 nodes and 13,662 links.

            Notably, the origins and destinations within the initial OD demand matrix do not directly correspond to the nodes in the road network.
            To address this spatial mismatch, we aggregate the OD demand data from the fine-grained block level to the coarser tract level,
            where a tract comprises multiple blocks~\citep{smith2021census}.
            Consistent with existing research~\citep{abdel2013geographical},
            we adopt tracts as traffic analysis zones (TAZs) to serve as the fundamental units for traffic assignment,
            which ensures computational feasibility while minimizing the loss of spatial resolution.
            The number of TAZs across the six U.S. cities is presented in Table~\ref{tab:city parameters},
            where cities with a greater number of nodes and links typically possess more TAZs.

            Regarding the UE model parameters, following the work of \citet{xu2024unified},
            the BPR parameters, road capacities, and free flow speeds for the six selected U.S. cities are summarized in Table~\ref{tab:city parameters}.
            \begin{table}[!htb]
                \centering
                \footnotesize
                \begin{threeparttable}
                    \caption{Parameters of the BPR function, road capacity, free flow speed, and road network data.}
                    \label{tab:city parameters}
                    \setlength{\tabcolsep}{3pt}
                    \begin{tabularx}{\textwidth}{l c *{12}{C}}
                        \toprule
                        \multirow{2}{*}{City} & \multicolumn{2}{c}{BPR function} & \multicolumn{3}{c}{Capacity (veh/h)} & \multicolumn{3}{c}{Free flow speed (km/h)} & \multicolumn{3}{c}{Road network data}\\
                        \addlinespace[-3pt]
                        \cmidrule(lr){2-3} \cmidrule(lr){4-6} \cmidrule(lr){7-9} \cmidrule(lr){10-12}
                        \addlinespace[-3pt]
                        & \hspace{2mm}$\alpha$\hspace{3mm} & $\beta$ & road1 & road2 & road3 & road1 & road2 & road3 & TAZs & nodes & links\\
                        \midrule
                        San Francisco   & 0.5 & 1.8 & 2,200 & 2,000 & 1,400 & 90 & 60 & 40  & 194 & 4,986  & 18,002  \\
                        Dallas          & 0.6 & 1.3 & 2,200 & 2,000 & 1,400 & 90 & 65 & 45  & 328 & 21,389 & 77,818  \\
                        Pittsburgh      & 0.5 & 2.0 & 2,200 & 2,000 & 1,200 & 90 & 60 & 30  & 149 & 3,532  & 13,662  \\
                        Miami           & 0.5 & 1.5 & 1,800 & 1,500 &   900 & 65 & 50 & 35  & 108 & 4,121  & 15,108  \\
                        Denver          & 0.5 & 1.5 & 2,000 & 1,800 & 1,300 & 90 & 60 & 35  & 175 & 9,205  & 34,724  \\
                        Philadelphia    & 0.5 & 1.2 & 2,000 & 1,800 & 1,200 & 90 & 60 & 30  & 389 & 10,410 & 38,641  \\
                        \bottomrule
                    \end{tabularx}
                    \begin{tablenotes}[flushleft]
                        \footnotesize
                        \item where road 1, road 2, and road 3 correspond to Expressway, Highway, and Local roads, respectively.
                    \end{tablenotes}
                \end{threeparttable}
                \vspace{-1.0em}
            \end{table}

        \subsection{SSPs and Population Data}
        \label{sec:population data}
            The population projection data used in this study are grounded in the SSPs framework presented by \citet{o2017roads}.
            Unlike most global scenario exercises driven by inputs or elements,
            the SSPs were developed based on the challenges space, defined primarily by the nature of the outcomes.
            To span this space, the pathways were constructed using a backcasting scenario approach,
            where a particular outcome was first identified,
            and the key elements of society required to reach those states were subsequently determined.
            This process resulted in five distinct narratives (SSP1 to SSP5),
            with their core characteristics and impacts on future population outlined in Table~\ref{tab:SSPs}.
            \begin{table}[!htb]
                \centering
                \small
                \caption{Core characteristics of SSPs and their impacts on future population.}
                \label{tab:SSPs}
                \begin{tabularx}{0.98\textwidth}{l >{\raggedright\arraybackslash}X >{\raggedright\arraybackslash}X >{\raggedright\arraybackslash}X}
                    \toprule 
                    SSP & Alias & Core Characteristics & Impacts on Population \\
                    \midrule 
                    SSP1 & Sustainable Development & Green, inclusive, global cooperation & Low growth: lowest population peak, fastest decline \\
                    SSP2 & Middle of the Road & Continuation of historical trends & Medium growth: medium peak population, slow decline \\
                    SSP3 & Regional Rivalry & Nationalism, conflicts, fragmentation & High growth: highest population peak, longest duration of growth \\
                    SSP4 & Inequality & Unequal wealth, social division & Medium growth but highly unequal: extreme internal disparities \\
                    SSP5 & Fossil-fueled Development & High carbon, globalization, growth maximized & Rapid growth with earlier transition: relatively early peak, faster decline \\
                    \bottomrule
                \end{tabularx}
                \vspace{-1.0em}
            \end{table}

            The future population data for U.S. cities were derived from the dataset developed by \citet{sutradhar2024depopulation},
            providing population projections from 2020 to 2100 under five SSPs.
            The dataset integrates two primary sources:
            Hauer’s data (U.S. County-Level Population Projections, Total and by Sex, Race and Age, based on the SSPs, v1~\citep{hauer2019population}) and
            NCAR data (Global 1-km Downscaled Population Base Year and Projection Grids Based on the SSPs, v1.01~\citep{gao2017downscaling}).
            Specifically, for Hauer’s data, the county-level projections were redistributed to the city level based on the intersected areas between cities and counties.
            For the NCAR data, city-level projections were derived using area weighted zonal statistics from the overlapping areas of the 1-km gridded projected population
            and city administrative boundaries.
            The two city-level projection datasets were finally combined by weighting them based on
            their error contribution in predicting the 2020 Census population.

            To capture a broad range of future uncertainties, we selected SSP1 (Sustainable Development), SSP2 (Middle of the Road),
            and SSP5 (Fossil-fueled Development) as the research scenarios and
            extracted population data across the six selected cities for the period 2030 to 2050 under these SSPs,
            with a specific focus on the decadal milestones of 2030, 2040, and 2050, as shown in Table~\ref{tab:population}.
            The (baseline, 2025) population is forecasted using an ARIMA model calibrated
            on historical population statistics~\citep{CensusIntercensal2010,CensusVINT2019,USCensus2024}.
            In subsequent analysis, for each given city under each SSP scenario,
            we consistently apply the OD demand forecasting method and the UE models established in Section~\ref{sec:3}.
            \begin{table}[!htb]
                \centering
                \footnotesize
                \caption{Population data for the six representative U.S. cities under various SSPs (100,000).}
                \label{tab:population}
                \begin{tabularx}{0.98\textwidth}{l c *{9}{C}}
                    \toprule
                    \multirow{2}{*}{City} & baseline & \multicolumn{3}{c}{SSP1} & \multicolumn{3}{c}{SSP2} & \multicolumn{3}{c}{SSP5} \\
                    \addlinespace[-3pt]
                    \cmidrule(lr){2-2} \cmidrule(lr){3-5} \cmidrule(lr){6-8} \cmidrule(lr){9-11}
                    \addlinespace[-3pt]
                    & 2025 & 2030 & 2040 & 2050 & 2030 & 2040 & 2050 & 2030 & 2040 & 2050 \\
                    \midrule
                    San Francisco   & 8.27 & 9.78 & 10.72 & 11.59 & 9.71 & 10.56 & 11.34 & 10.28 & 11.68 & 13.20 \\
                                    & 0\%  & 18.2\% & 29.5\% & 40.0\% & 17.3\% & 27.7\% & 37.0\% & 24.2\% & 41.2\% & 59.4\% \\[2pt]
                    Dallas          & 13.35 & 15.38 & 17.64 & 19.82 & 15.29 & 17.41 & 19.41 & 16.25 & 19.38 & 22.80 \\
                                    & 0\%   & 15.2\% & 32.1\% & 48.4\% & 14.5\% & 30.4\% & 45.3\% & 21.7\% & 45.2\% & 70.7\% \\[2pt]
                    Pittsburgh      & 3.08 & 3.11 & 3.14 & 3.13 & 3.09 & 3.10 & 3.08 & 3.29 & 3.46 & 3.63 \\
                                    & 0\%  & 1.2\% & 2.1\% & 1.8\% & 0.5\% & 0.8\% & 0.2\% & 6.8\% & 12.3\% & 17.9\% \\[2pt]
                    Miami           & 4.93 & 4.96 & 5.44 & 5.88 & 4.92 & 5.36 & 5.74 & 5.21 & 5.93 & 6.68 \\
                                    & 0\%  & 0.6\% & 10.3\% & 19.1\% & -0.2\% & 8.7\% & 16.4\% & 5.7\% & 20.2\% & 35.5\% \\[2pt]
                    Denver          & 7.31 & 7.98 & 8.72 & 9.40 & 7.92 & 8.60 & 9.22 & 8.38 & 9.51 & 10.73 \\
                                    & 0\%  & 9.0\% & 19.2\% & 28.5\% & 8.3\% & 17.6\% & 26.0\% & 14.6\% & 30.0\% & 46.6\% \\[2pt]
                    Philadelphia    & 15.74 & 17.17 & 18.06 & 18.81 & 17.05 & 17.82 & 18.45 & 18.07 & 19.75 & 21.53 \\
                                    & 0\%   & 9.1\% & 14.8\% & 19.5\% & 8.3\% & 13.2\% & 17.2\% & 14.8\% & 25.5\% & 36.8\% \\
                    \bottomrule
                \end{tabularx}
                \vspace{-1.0em}
            \end{table}

    \section{Case Study}
    \label{sec:4}
        \subsection{Urban Metrics}
        \label{sec:4.1}
            In this section, we utilize urban metrics for the six U.S. cities across all years, SSPs, and UE models
            to analyze the urban dynamics under different routing strategies and population growth from a macroscopic city perspective.
            \subsubsection{Urban Transport Emissions}
            \label{sec:4.1.1}
                Carbon emissions under the Time-Carbon UE models exhibit a clear reduction compared to those under the Time-Only UE model
                across all cities, years, and SSPs, as illustrated in Fig.~\ref{fig:Carbon emission reduction}.
                Within a given year and SSP scenario, the magnitude of emission reduction is inversely
                related to weight ratio of $\psi_{1}:\psi_{2}$.
                A smaller $\psi_{1}:\psi_{2}$ (i.e., a higher relative weight of carbon cost),
                yields a greater emission reduction,
                with the maximum occurring at $\psi_{1}=0$, $\psi_{2}=10$.
                \begin{figure}[!htb]
                    \captionsetup[subfigure]{skip=1pt}
                    \centering
                    \begin{subfigure}[b]{0.33\textwidth}
                        \centering
                        \includegraphics[width=\linewidth, keepaspectratio]{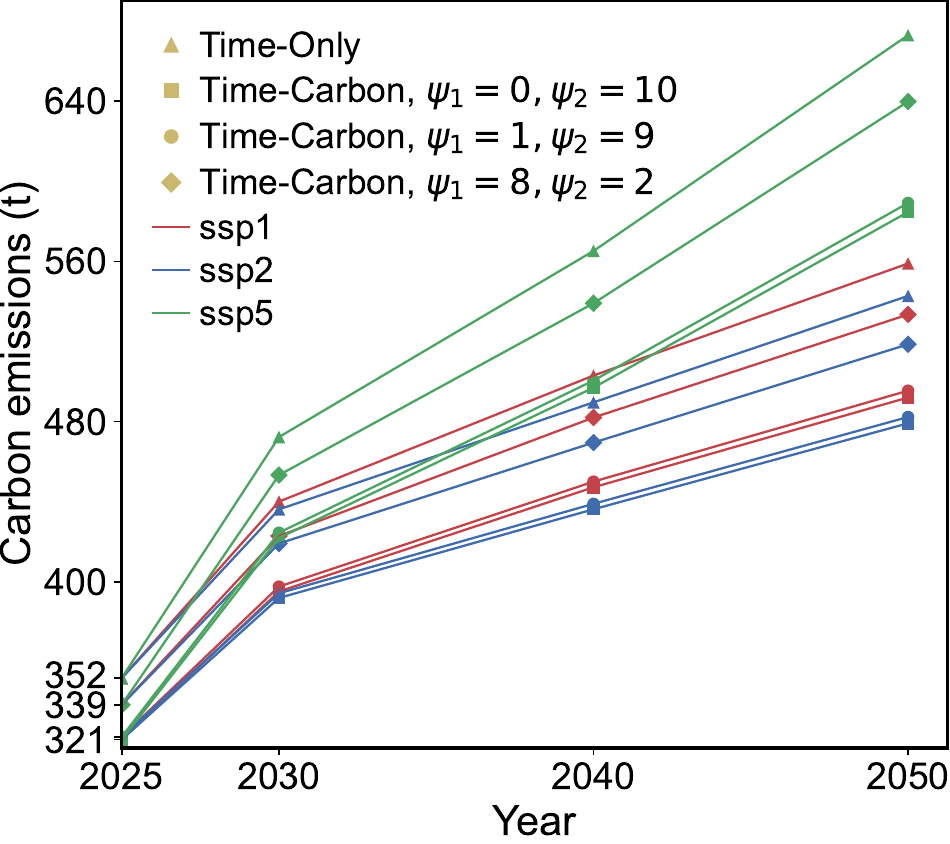}
                        \caption{San Francisco}
                        \label{fig:carbon_sf}
                    \end{subfigure}
                    \hfill
                    \begin{subfigure}[b]{0.33\textwidth}
                        \centering
                        \includegraphics[width=\linewidth, keepaspectratio]{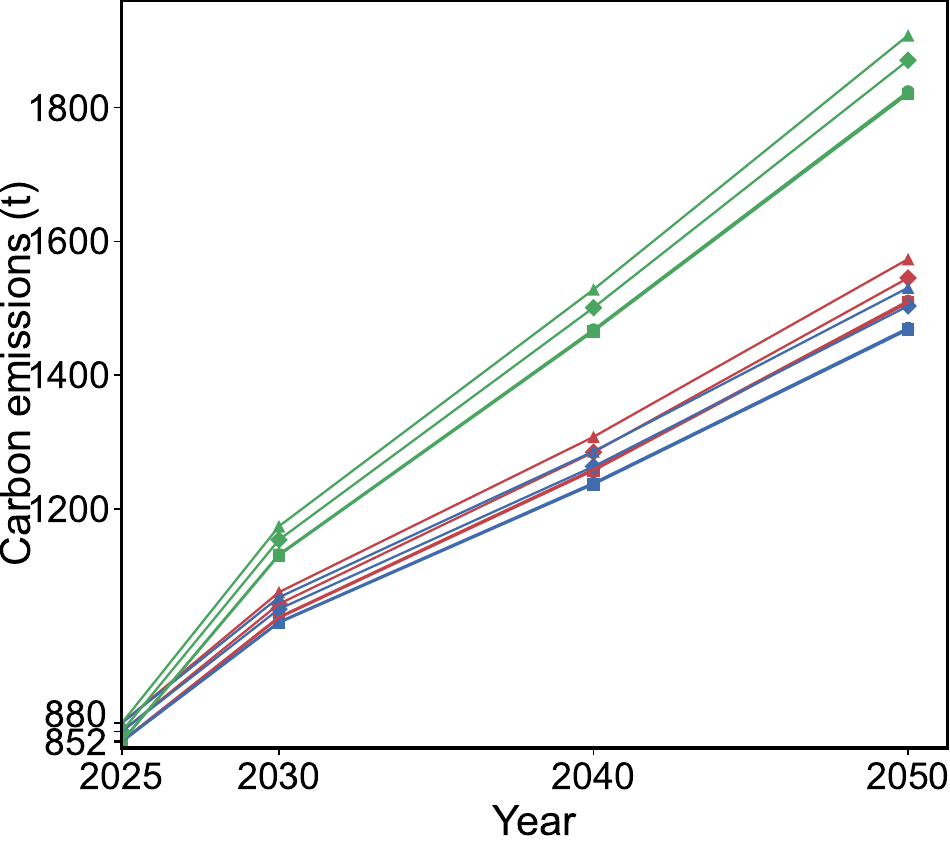}
                        \caption{Dallas}
                        \label{fig:carbon_dallas}
                    \end{subfigure}
                    \hfill
                    \begin{subfigure}[b]{0.33\textwidth}
                        \centering
                        \includegraphics[width=\linewidth, keepaspectratio]{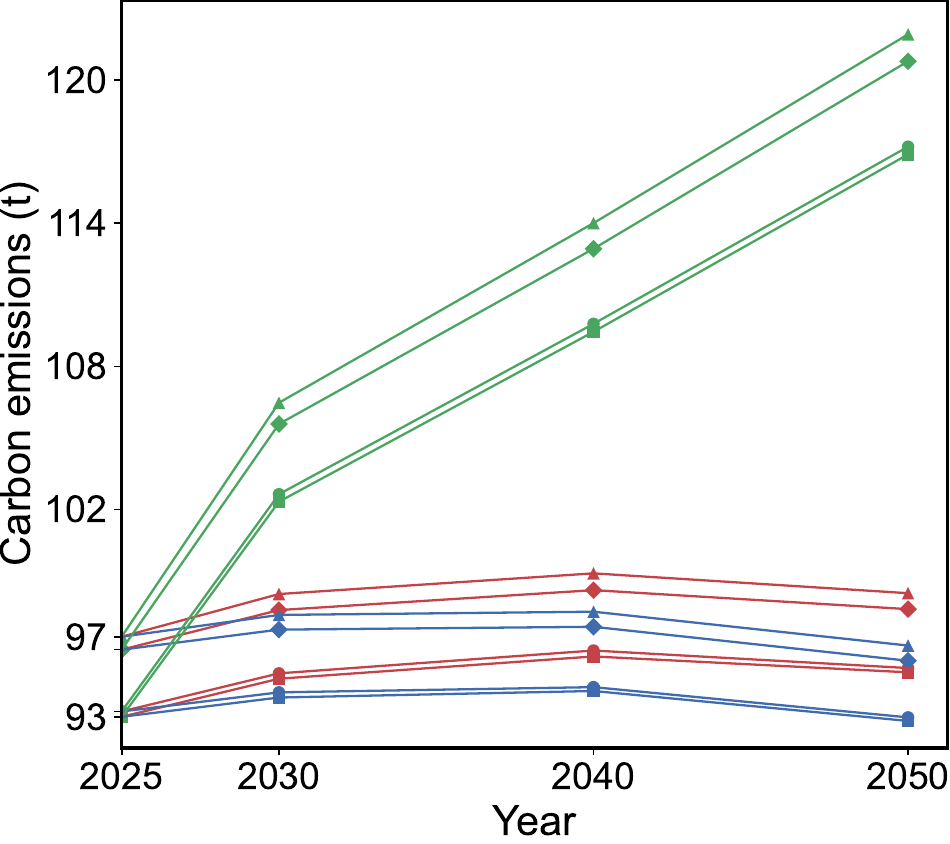}
                        \caption{Pittsburgh}
                        \label{fig:carbon_pittsburgh}
                    \end{subfigure}
                    \\[-1.5pt]
                    \begin{subfigure}[b]{0.33\textwidth}
                        \centering
                        \includegraphics[width=\linewidth, keepaspectratio]{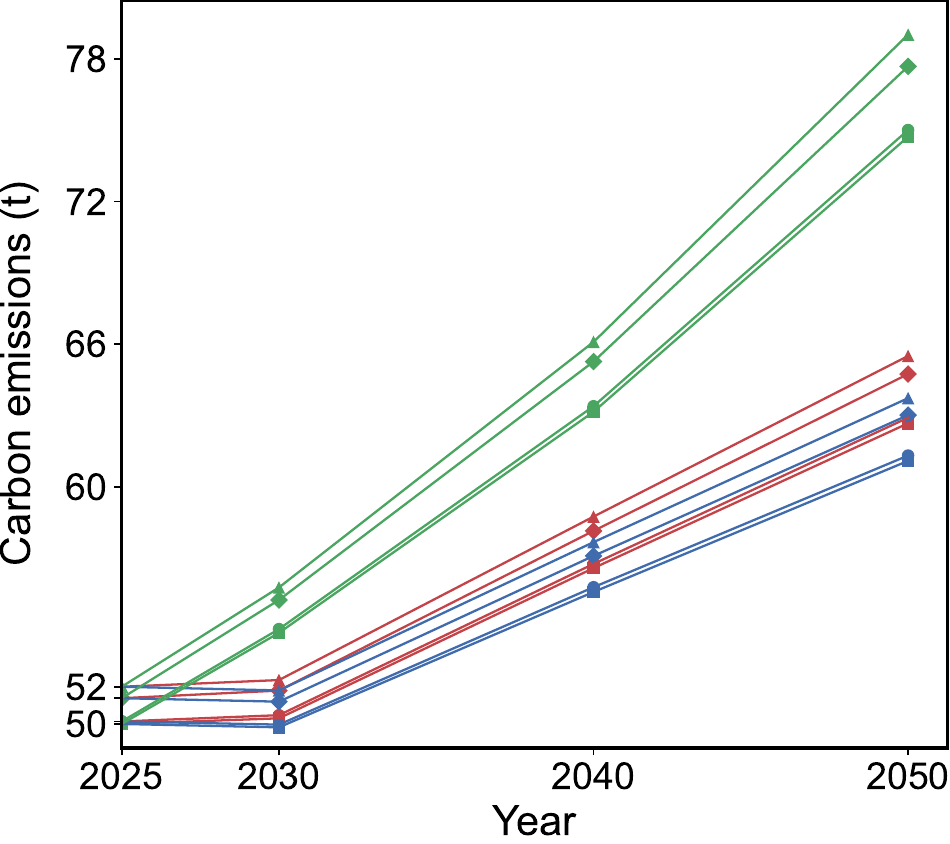}
                        \caption{Miami}
                        \label{fig:carbon_miami}
                    \end{subfigure}
                    \hfill
                    \begin{subfigure}[b]{0.33\textwidth}
                        \centering
                        \includegraphics[width=\linewidth, keepaspectratio]{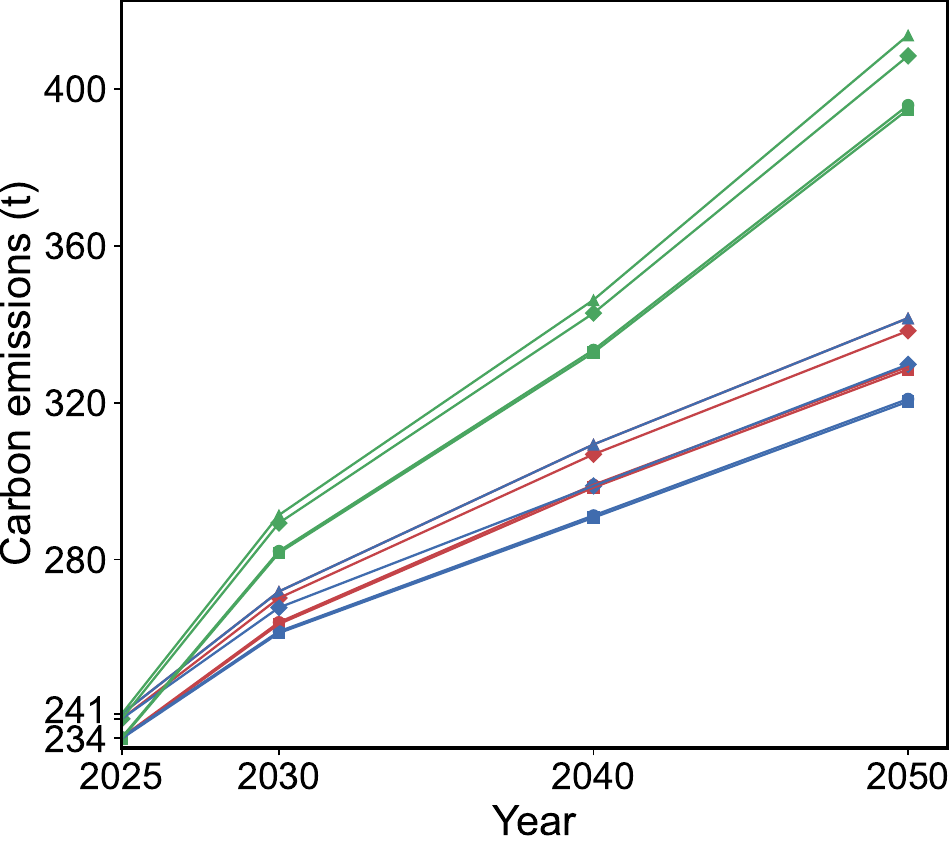}
                        \caption{Denver}
                        \label{fig:carbon_denver}
                    \end{subfigure}
                    \hfill
                    \begin{subfigure}[b]{0.33\textwidth}
                        \centering
                        \includegraphics[width=\linewidth, keepaspectratio]{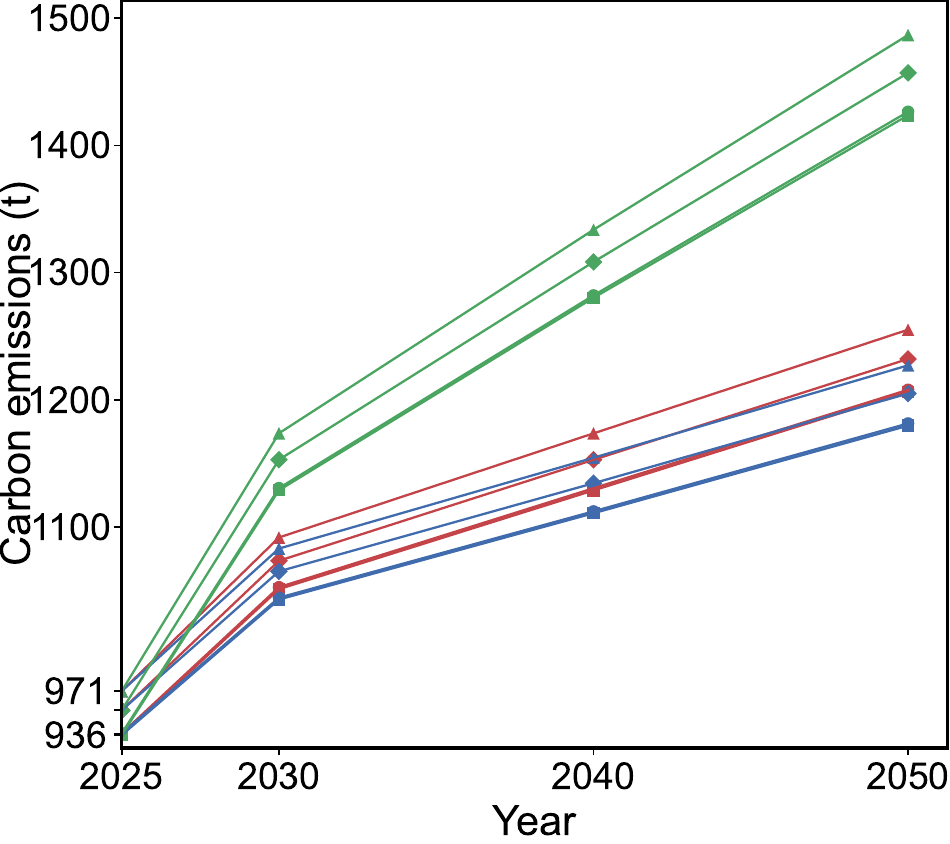}
                        \caption{Philadelphia}
                        \label{fig:carbon_philadelphia}
                    \end{subfigure}
                    \vspace{-0.7em}
                    \caption{Carbon emissions (t) for the six representative U.S. cities across all years, SSPs, and UE models.}
                    \label{fig:Carbon emission reduction}
                    \vspace{-1.0em}
                \end{figure}

                For any city, (SSPs, years) with a larger population size invariably lead to higher emissions across all UE models;
                notably, in Pittsburgh (SSP1/SSP2, 2050) and Miami (SSP2, 2030),
                emissions decrease as their population declines.
                Consequently, although both OD demand distribution $\pi_{y}^{s, t}$ and population size influence carbon emissions,
                emissions consistently exhibit a positive scaling with annual population (Table~\ref{tab:population}) across all UE models,
                indicating that population growth is the dominant driver.
                This observation fully aligns with Theorem~\ref{thm:scaling} and Corollary~\ref{cor:eco_coefficient},
                which state that emissions scale linearly with population regardless of routing strategies.

                San Francisco exhibits the most significant emission reduction across all cities.
                As shown in Table~\ref{tab:Transport emissions}, when $\psi_{1}=1$ and $\psi_{2}=9$,
                its optimization percentage $P_{c}$ ranges from a minimum of 8.4\% in (baseline, 2025)
                to a maximum of 12.4\% in (SSP5, 2050).
                Moreover, for this city in 2030, the Time-Carbon UE model ($\psi_{1}=0$, $\psi_{2}=10$ or $\psi_{1}=1$, $\psi_{2}=9$) under the SSP5 scenario
                yields lower emissions than the Time-Only UE model under the SSP1 and SSP2 scenarios.
                However, in 2040 and 2050, this advantage diminishes and eventually reverses,
                as the emission reduction benefits from eco-routing are overwhelmed by the rapid population growth under SSP5 scenario.
                For all other cities in any given year, regardless of which weight ratio $\psi_{1}:\psi_{2}$ is applied,
                the Time-Carbon UE model under SSP5 scenario consistently yields higher emissions than the Time-Only UE model under any other SSPs.
                These results suggest that the implementation of eco-routing
                cannot reduce the high emissions under SSP5 scenario
                to the lower levels under SSP1 or SSP2 scenarios with time-routing.
                Therefore, low-carbon urban development fundamentally
                requires following a sustainable SSP characterized by slower population growth.
                \begin{table}[!htb]
                \centering
                \footnotesize
                \begin{threeparttable}
                    \caption{Carbon emissions (t) for the six representative U.S. cities under the Time-Only and Time-Carbon ($\psi_{1}=1$ and $\psi_{2}=9$) UE models.}
                    \label{tab:Transport emissions}
                    \setlength{\tabcolsep}{3pt}
                    \begin{tabularx}{0.99\textwidth}{c @{} c *{10}{C}}
                        \toprule
                        \multirow{2}{*}{City} & & baseline & \multicolumn{3}{c}{SSP1} & \multicolumn{3}{c}{SSP2} & \multicolumn{3}{c}{SSP5} \\
                        \addlinespace[-3pt]
                        \cmidrule(lr){3-3} \cmidrule(lr){4-6} \cmidrule(lr){7-9} \cmidrule(lr){10-12}
                        \addlinespace[-3pt]
                         & & 2025 & 2030 & 2040 & 2050 & 2030 & 2040 & 2050 & 2030 & 2040 & 2050 \\
                        \midrule
                        \multirow{3}{*}{San Francisco}  & TO    & 352  & 440 & 503 & 559 & 436 & 490 & 543 & 472 & 565 & 673  \\
                                                        & TC  & 323  & 398 & 450 & 495 & 394 & 439 & 482 & 425 & 501 & 589  \\
                                                        & $P_{c}$      &  8.4\%  & 9.6\% & 10.5\% & 11.4\% & 9.6\% & 10.3\% & 11.1\% & 10.1\% & 11.5\% & 12.4\%  \\[3pt]
                        \multirow{3}{*}{Dallas}         & TO    & 880  & 1076 & 1308 & 1574 & 1068 & 1286 & 1531 & 1174 & 1528 & 1908  \\
                                                        & TC  & 853  & 1039 & 1260 & 1511 & 1031 & 1239 & 1471 & 1132 & 1468 & 1824  \\
                                                        & $P_{c}$      &  3.0\%  & 3.4\% & 3.7\% & 4.0\% & 3.4\% & 3.7\% & 3.9\% & 3.6\% & 3.9\% & 4.4\%  \\[3pt]
                        \multirow{3}{*}{Pittsburgh}     & TO    & 96.7  & 98.5 & 99.3 & 98.5 & 97.6 & 97.7 & 96.3 & 106 & 114 & 122 \\
                                                        & TC  & 93.5  & 95.1 & 96.1 & 95.4 & 94.3 & 94.6 & 93.3 & 103 & 110 & 117  \\
                                                        & $P_{c}$      &  3.2\%  & 3.4\% & 3.3\% & 3.2\% & 3.3\% & 3.2\% & 3.1\% & 3.6\% & 3.7\% & 3.9\%  \\[3pt]
                        \multirow{3}{*}{Miami}          & TO    & 51.6  & 51.9 & 58.7 & 65.5 & 51.4 & 57.7 & 63.7 & 55.8 & 66.1 & 79.0  \\
                                                        & TC  & 50.1  & 50.4 & 56.8 & 62.9 & 50.0 & 55.8 & 61.3 & 54.0 & 63.4 & 75.0  \\
                                                        & $P_{c}$      &  2.8\%  & 2.9\% & 3.4\% & 4.0\% & 2.8\% & 3.3\% & 3.8\% & 3.1\% & 4.1\% & 5.1\%  \\[3pt]
                        \multirow{3}{*}{Denver}         & TO    & 241  & 272 & 309 & 342 & 269 & 301 & 333 & 291 & 346 & 414  \\
                                                        & TC  & 235  & 264 & 299 & 329 & 262 & 291 & 321 & 282 & 333 & 396  \\
                                                        & $P_{c}$      &  2.5\%  & 2.9\% & 3.3\% & 3.6\% & 2.9\% & 3.3\% & 3.5\% & 3.2\% & 3.7\% & 4.3\%  \\[3pt]
                        \multirow{3}{*}{Philadelphia}   & TO    & 971  & 1092 & 1174 & 1255 & 1083 & 1154 & 1227 & 1174 & 1334 & 1487  \\
                                                        & TC  & 937  & 1053 & 1130 & 1209 & 1044 & 1112 & 1181 & 1130 & 1282 & 1427  \\
                                                        & $P_{c}$      &  3.4\%  & 3.6\% & 3.7\% & 3.8\% & 3.6\% & 3.7\% & 3.7\% & 3.7\% & 3.9\% & 4.1\%  \\
                        \bottomrule
                    \end{tabularx}
                    \begin{tablenotes}[flushleft]
                        \footnotesize
                        \item
                        where the $\mathrm{TO}$ and $\mathrm{TC}$ denote metrics under Time-Only and Time-Carbon UE models, respectively.
                    \end{tablenotes}
                \end{threeparttable}
                \vspace{-1.0em}
                \end{table}

                When $\psi_{1}=1$, $\psi_{2}=9$, as shown in Table~\ref{tab:Transport emissions},
                for a given year, SSPs with a larger population size (i.e., higher emissions)
                exhibit a higher optimization percentage $P_{c}$;
                similarly, under a given SSP scenario, years with a larger population size also exhibit a higher $P_{c}$.
                Moreover, as shown in Fig.~\ref{fig:Carbon emission reduction},
                these findings are robust across all Time-Carbon UE models,
                indicating that population growth also exerts a dominant influence on $P_{c}$.
                Therefore, it is essential to further investigate the impact of population growth on $P_{c}$,
                thereby evaluating the efficiency of eco-routing.

            \subsubsection{UETT}
            \label{sec:4.1.2}
                UETT under the Time-Carbon UE models is consistently higher than that under the Time-Only UE model, as shown in Fig.~\ref{fig:UETT}.
                This confirms that while eco-routing can reduce emissions, it generally comes at the expense of increased travel time.
                Moreover, for a given year and SSP scenario, as the weight ratio $\psi_{1}:\psi_{2}$ decreases,
                the Time-Carbon UE model yields a greater emission reduction but incurs a corresponding increase in travel time.
                An exception is observed in Pittsburgh (SSP5),
                where the Time-Carbon UE model ($\psi_{1}=8$, $\psi_{2}=2$) yields both lower UETT and emissions compared to the Time-Only UE model.
                This is possible because the relative weight of the carbon cost is small,
                and the equilibrium state of Time-Only UE model is suboptimal.
                Consequently, this rare case does not undermine the analysis above;
                rather, it demonstrates that cases where both UETT and emissions decrease under the Time-Carbon UE model are extremely limited.
                \begin{figure}[!htb]
                    \captionsetup[subfigure]{skip=1pt}
                    \centering
                    \begin{subfigure}[b]{0.33\textwidth}
                        \centering
                        \includegraphics[width=\linewidth, keepaspectratio]{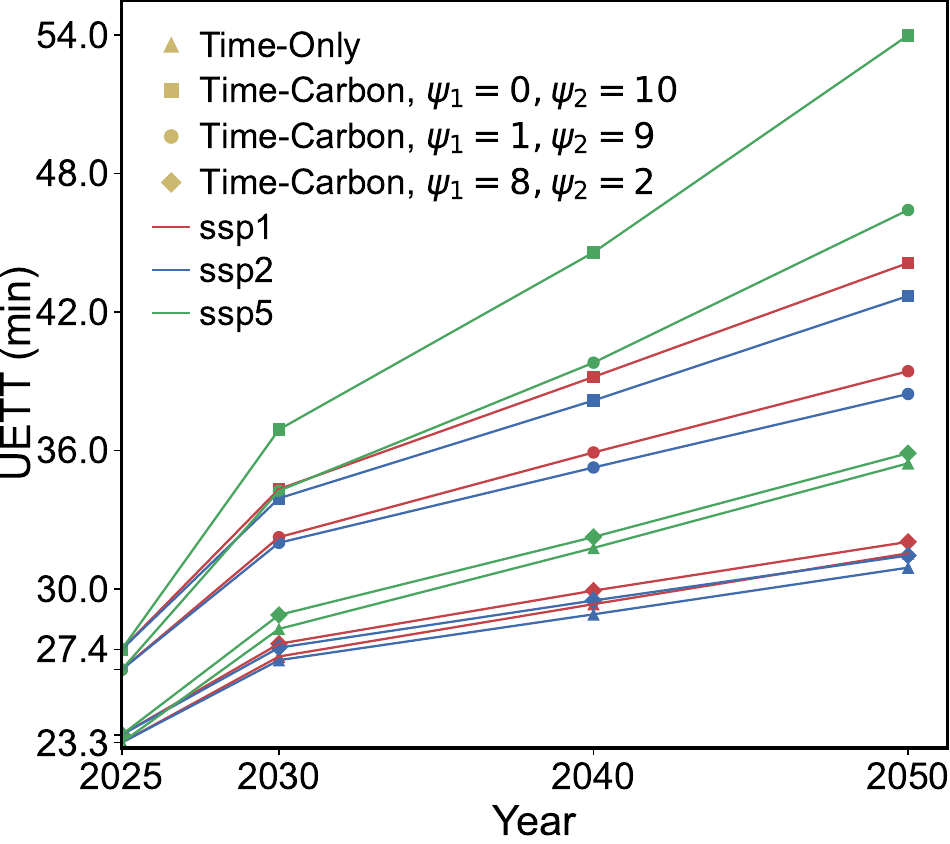}
                        \caption{San Francisco}
                        \label{fig:uett_sf}
                    \end{subfigure}%
                    \hfill%
                    \begin{subfigure}[b]{0.33\textwidth}
                        \centering
                        \includegraphics[width=\linewidth, keepaspectratio]{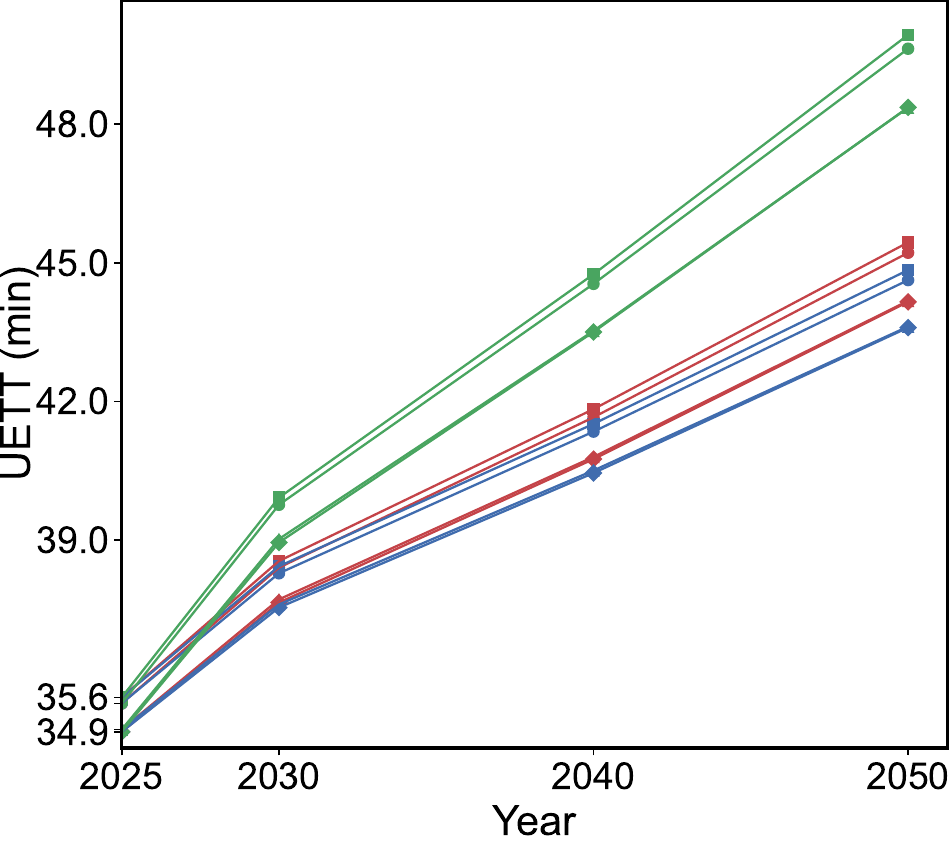}
                        \caption{Dallas}
                        \label{fig:uett_dallas}
                    \end{subfigure}%
                    \hfill%
                    \begin{subfigure}[b]{0.33\textwidth}
                        \centering
                        \includegraphics[width=\linewidth, keepaspectratio]{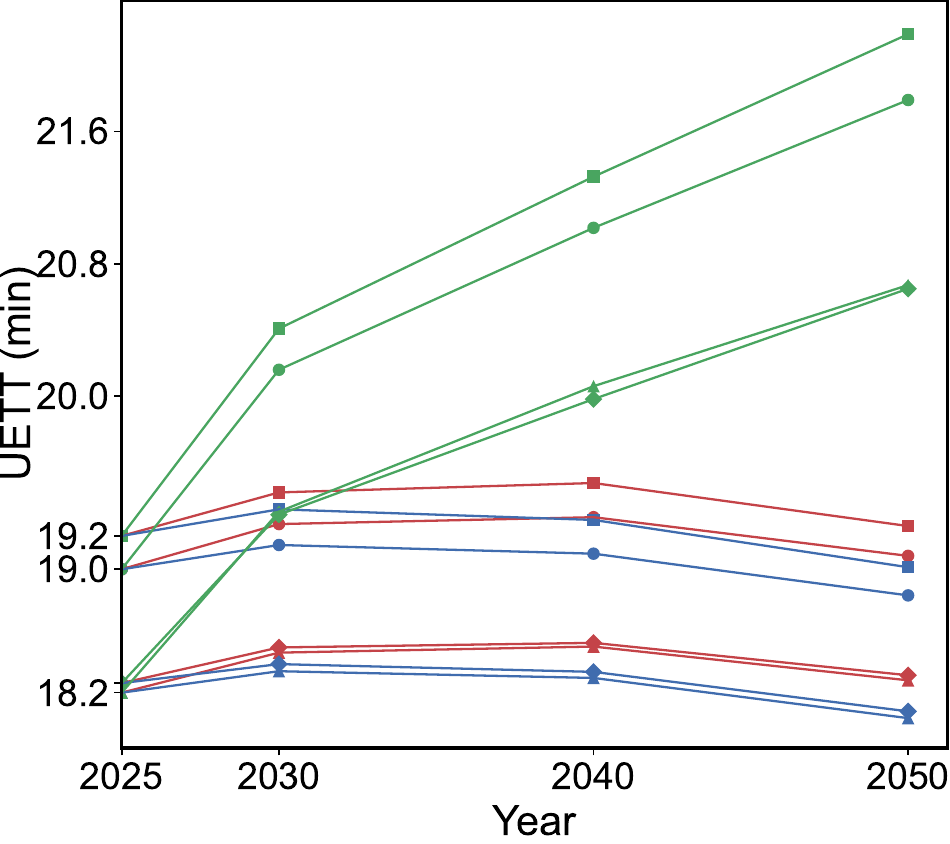}
                        \caption{Pittsburgh}
                        \label{fig:uett_pittsburgh}
                    \end{subfigure}%
                    \\[-1.5pt]
                    \begin{subfigure}[b]{0.33\textwidth}
                        \centering
                        \includegraphics[width=\linewidth, keepaspectratio]{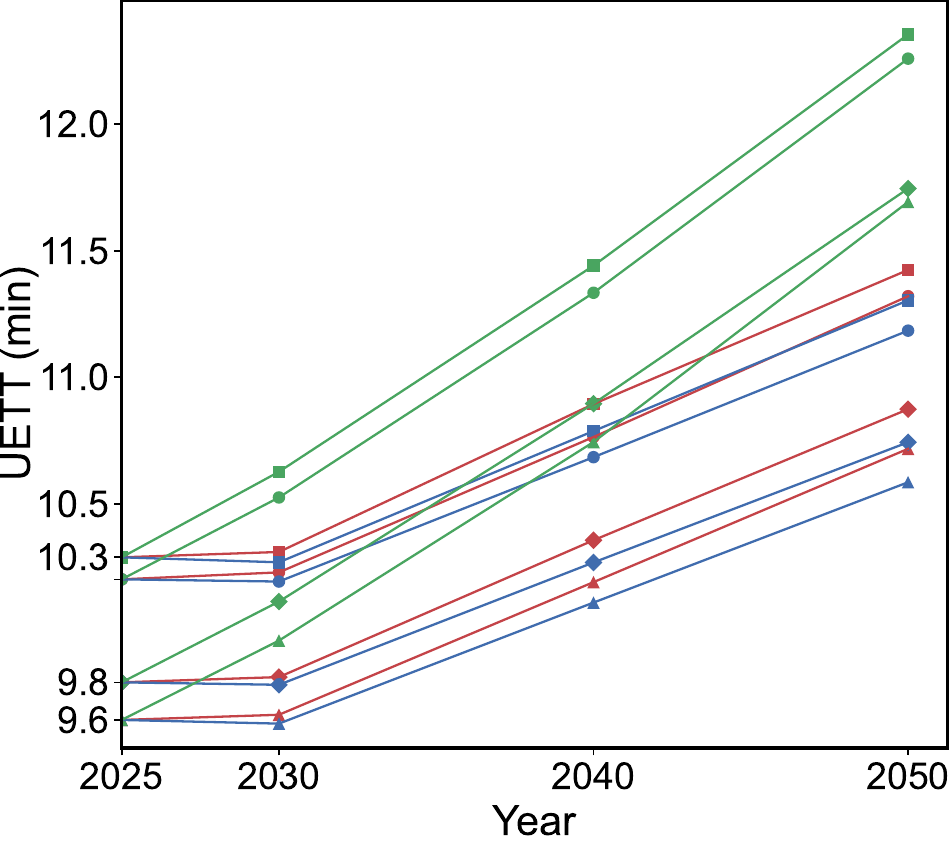}
                        \caption{Miami}
                        \label{fig:uett_miami}
                    \end{subfigure}%
                    \hfill%
                    \begin{subfigure}[b]{0.33\textwidth}
                        \centering
                        \includegraphics[width=\linewidth, keepaspectratio]{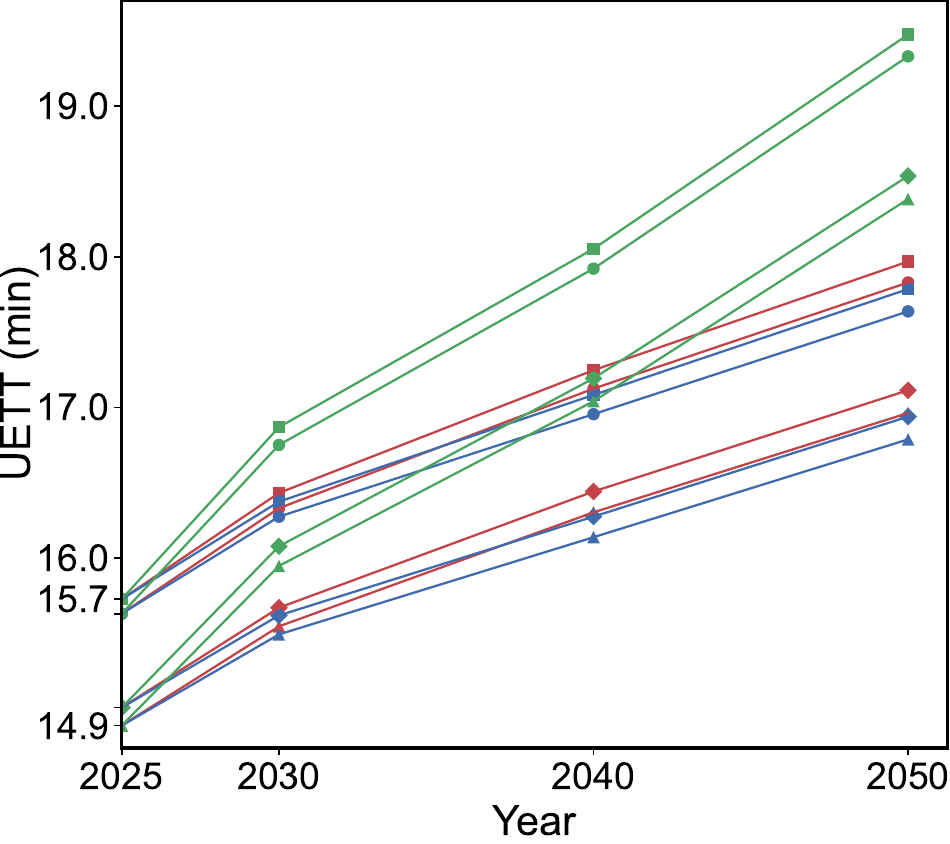}
                        \caption{Denver}
                        \label{fig:uett_denver}
                    \end{subfigure}%
                    \hfill%
                    \begin{subfigure}[b]{0.33\textwidth}
                        \centering
                        \includegraphics[width=\linewidth, keepaspectratio]{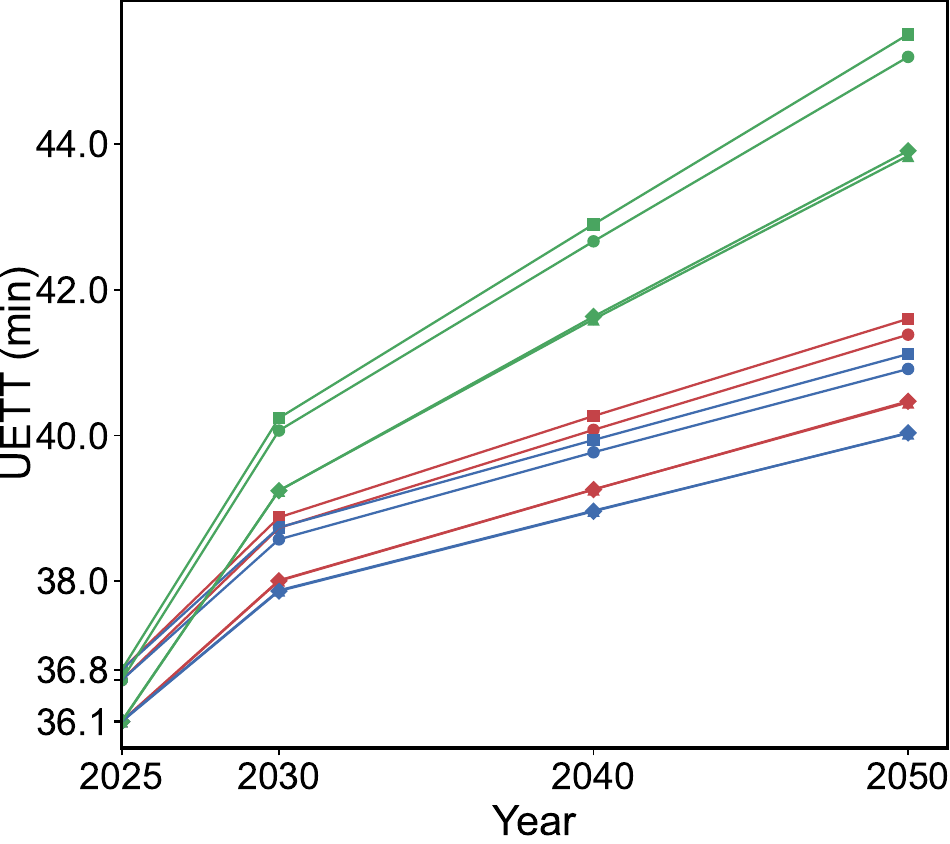}
                        \caption{Philadelphia}
                        \label{fig:uett_philadelphia}
                    \end{subfigure}%
                    \vspace{-0.7em}
                    \caption{UETT (min) for the six representative U.S. cities across all years, SSPs, and UE models.}
                    \label{fig:UETT}
                    \vspace{-1.0em} 
                \end{figure}


                UETT exhibits a positive scaling with annual population,
                aligning with the proof of Theorem~\ref{thm:short_distance},
                where population growth leads to rising traffic flows, which in turn results in increased travel time.
                Moreover, consistent with Section~\ref{sec:4.1.1}, San Francisco also shows the most significant changes in UETT.
                As shown in Table~\ref{tab:UETT table}, when $\psi_{1}=1$ and $\psi_{2}=9$,
                its trade-off percentage $P_{t}$ ranges from a minimum of 13.6\% in (baseline, 2025) to
                a maximum of 31.0\% in (SSP5, 2050).
                Additionally, for all cities, the $P_{t}$ increases as population grows,
                except for Miami, where it decreases.
                However, this does not mean that eco-routing in Miami is efficient in terms of UETT, as the CP may be increasing.
                \begin{table}[!htb]
                \centering
                \footnotesize
                \begin{threeparttable}
                    \caption{UETT (min) for the six representative U.S. cities under the Time-Only and Time-Carbon ($\psi_{1}=1$ and $\psi_{2}=9$) UE models.}
                    \label{tab:UETT table}
                    \setlength{\tabcolsep}{3pt}
                    \begin{tabularx}{0.98\textwidth}{c @{} c *{10}{C}}
                        \toprule
                        \multirow{2}{*}{City} & & baseline & \multicolumn{3}{c}{SSP1} & \multicolumn{3}{c}{SSP2} & \multicolumn{3}{c}{SSP5} \\
                        \addlinespace[-3pt]
                        \cmidrule(lr){3-3} \cmidrule(lr){4-6} \cmidrule(lr){7-9} \cmidrule(lr){10-12}
                        \addlinespace[-3pt]
                         & & 2025 & 2030 & 2040 & 2050 & 2030 & 2040 & 2050 & 2030 & 2040 & 2050 \\
                        \midrule
                        \multirow{3}{*}{San Francisco}  & TO    & 23.3  & 27.1 & 29.3 & 31.5 & 26.9 & 28.9 & 30.9 & 28.3 & 31.8 & 35.4  \\
                                                        & TC  & 26.5  & 32.2 & 35.9 & 39.4 & 32.0 & 35.3 & 38.4 & 34.2 & 39.8 & 46.4  \\
                                                        & $P_{t}$      &  13.6\%  & 19.1\% & 22.4\% & 25.0\% & 18.9\% & 22.0\% & 24.3\% & 21.2\% & 25.3\% & 31.0\%  \\[3pt]
                        \multirow{3}{*}{Dallas}         & TO    & 34.9  & 37.7 & 40.8 & 44.2 & 37.6 & 40.5 & 43.6 & 39.0 & 43.5 & 48.4  \\
                                                        & TC  & 35.5  & 38.4 & 41.7 & 45.2 & 38.3 & 41.3 & 44.6 & 39.8 & 44.5 & 49.6  \\
                                                        & $P_{t}$      &  1.6\%  & 1.8\% & 2.1\% & 2.4\% & 1.8\% & 2.1\% & 2.3\% & 1.9\% & 2.3\% & 2.6\%  \\[3pt]
                        \multirow{3}{*}{Pittsburgh}     & TO    & 18.2  & 18.5 & 18.5 & 18.3 & 18.3 & 18.3 & 18.1 & 19.3 & 20.1 & 20.7 \\
                                                        & TC  & 19.0  & 19.2 & 19.3 & 19.0 & 19.1 & 19.0 & 18.8 & 20.2 & 21.0 & 21.8  \\
                                                        & $P_{t}$      &  4.1\%  & 4.2\% & 4.2\% & 4.1\% & 4.2\% & 4.1\% & 4.1\% & 4.4\% & 4.8\% & 5.4\%  \\[3pt]
                        \multirow{3}{*}{Miami}          & TO    & 9.65  & 9.67 & 10.2 & 10.7 & 9.63 & 10.1 & 10.6 & 10.0 & 10.7 & 11.7  \\
                                                        & TC  & 10.2  & 10.2 & 10.8 & 11.3 & 10.2 & 10.7 & 11.2 & 10.5 & 11.3 & 12.3  \\
                                                        & $P_{t}$      &  5.8\%  & 5.8\% & 5.6\% & 5.6\% & 5.8\% & 5.7\% & 5.7\% & 5.7\% & 5.5\% & 4.8\%  \\[3pt]
                        \multirow{3}{*}{Denver}         & TO    & 14.9  & 15.5 & 16.3 & 17.0 & 15.5 & 16.1 & 16.8 & 15.9 & 17.0 & 18.4  \\
                                                        & TC  & 15.6  & 16.3 & 17.1 & 17.8 & 16.3 & 17.0 & 17.6 & 16.8 & 17.9 & 19.3  \\
                                                        & $P_{t}$      &  5.0\%  & 5.0\% & 5.0\% & 5.1\% & 5.0\% & 5.1\% & 5.1\% & 5.0\% & 5.2\% & 5.2\%  \\[3pt]
                        \multirow{3}{*}{Philadelphia}   & TO    & 36.1  & 38.0 & 39.3 & 40.5 & 37.9 & 39.0 & 40.0 & 39.2 & 41.6 & 43.8  \\
                                                        & TC  & 36.6  & 38.7 & 40.1 & 41.4 & 38.6 & 39.8 & 40.9 & 40.1 & 42.7 & 45.2  \\
                                                        & $P_{t}$      &  1.6\%  & 1.9\% & 2.1\% & 2.3\% & 1.8\% & 2.1\% & 2.2\% & 2.1\% & 2.6\% & 3.1\%  \\
                        \bottomrule
                    \end{tabularx}
                \end{threeparttable}
                \vspace{-1.0em}
                \end{table}

            \subsubsection{UETL}
            \label{sec:4.1.3}
                UETL is significantly different under various UE models, as shown in Fig.~\ref{fig:UETL},
                indicating that traveler route choices are fundamentally distinct.
                Under the Time-Carbon UE models, UETL is lower than that under the Time-Only UE model,
                suggesting that travelers using eco-routing choose shorter routes to minimize the time-carbon cost.
                Moreover, for a given year and SSP scenario, UETL decreases as the weight ratio $\psi_{1}:\psi_{2}$ declines.
                All of these findings align with Theorem~\ref{thm:short_distance},
                which states that reducing travel distance is critical for emission reduction when transitioning from time-routing to eco-routing.
                \begin{figure}[!htb]
                    \captionsetup[subfigure]{skip=1pt}
                    \centering
                    \begin{subfigure}[b]{0.33\textwidth}
                        \centering
                        \includegraphics[width=\linewidth, keepaspectratio]{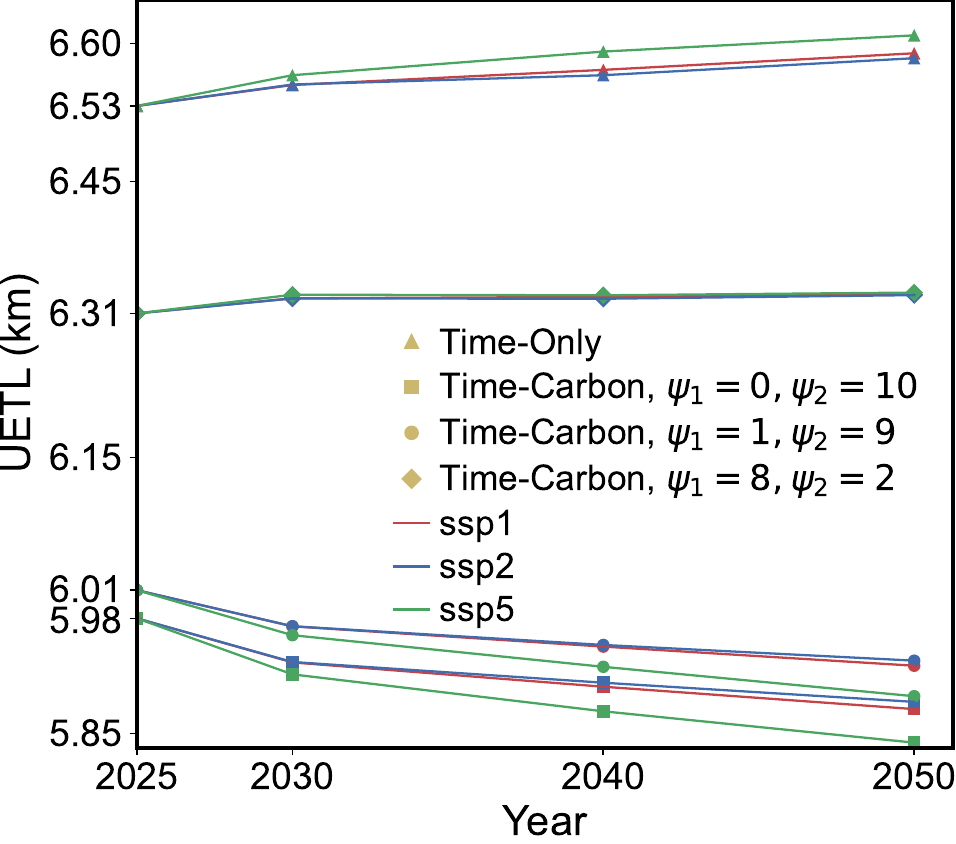}
                        \caption{San Francisco}
                        \label{fig:uetl_sf}
                    \end{subfigure}%
                    \hfill%
                    \begin{subfigure}[b]{0.33\textwidth}
                        \centering
                        \includegraphics[width=\linewidth, keepaspectratio]{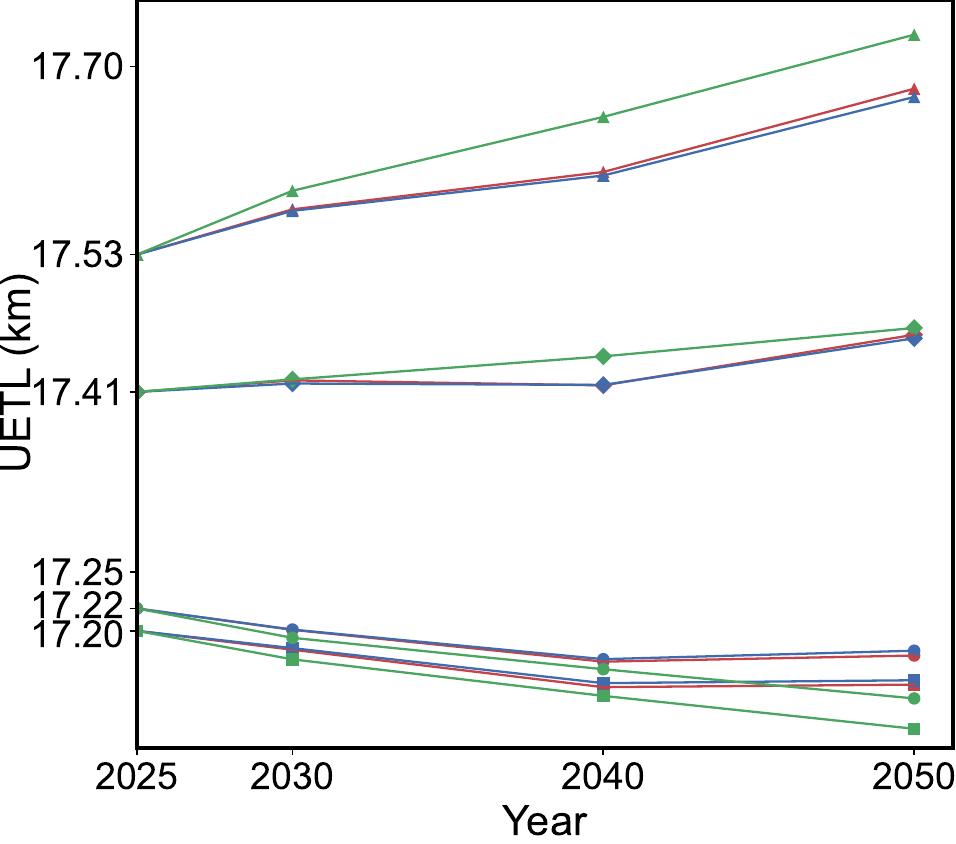}
                        \caption{Dallas}
                        \label{fig:uetl_dallas}
                    \end{subfigure}%
                    \hfill%
                    \begin{subfigure}[b]{0.33\textwidth}
                        \centering
                        \includegraphics[width=\linewidth, keepaspectratio]{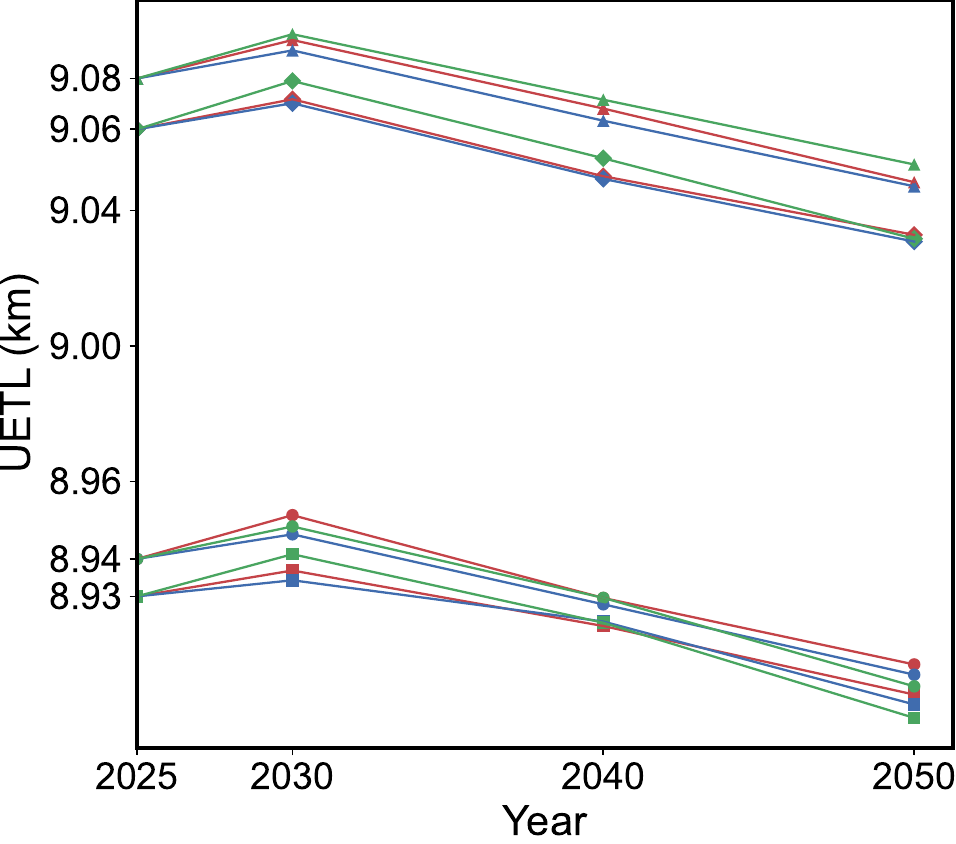}
                        \caption{Pittsburgh}
                        \label{fig:uetl_pittsburgh}
                    \end{subfigure}%
                    \\[-1.5pt]
                    \begin{subfigure}[b]{0.33\textwidth}
                        \centering
                        \includegraphics[width=\linewidth, keepaspectratio]{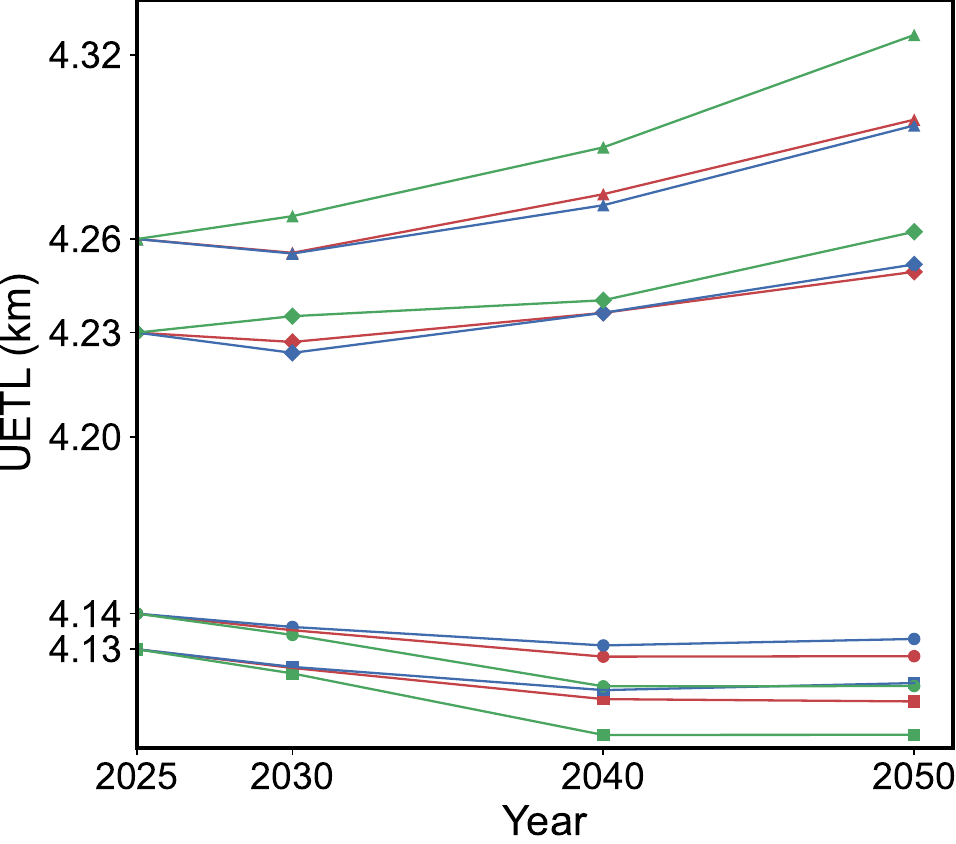}
                        \caption{Miami}
                        \label{fig:uetl_miami}
                    \end{subfigure}%
                    \hfill%
                    \begin{subfigure}[b]{0.33\textwidth}
                        \centering
                        \includegraphics[width=\linewidth, keepaspectratio]{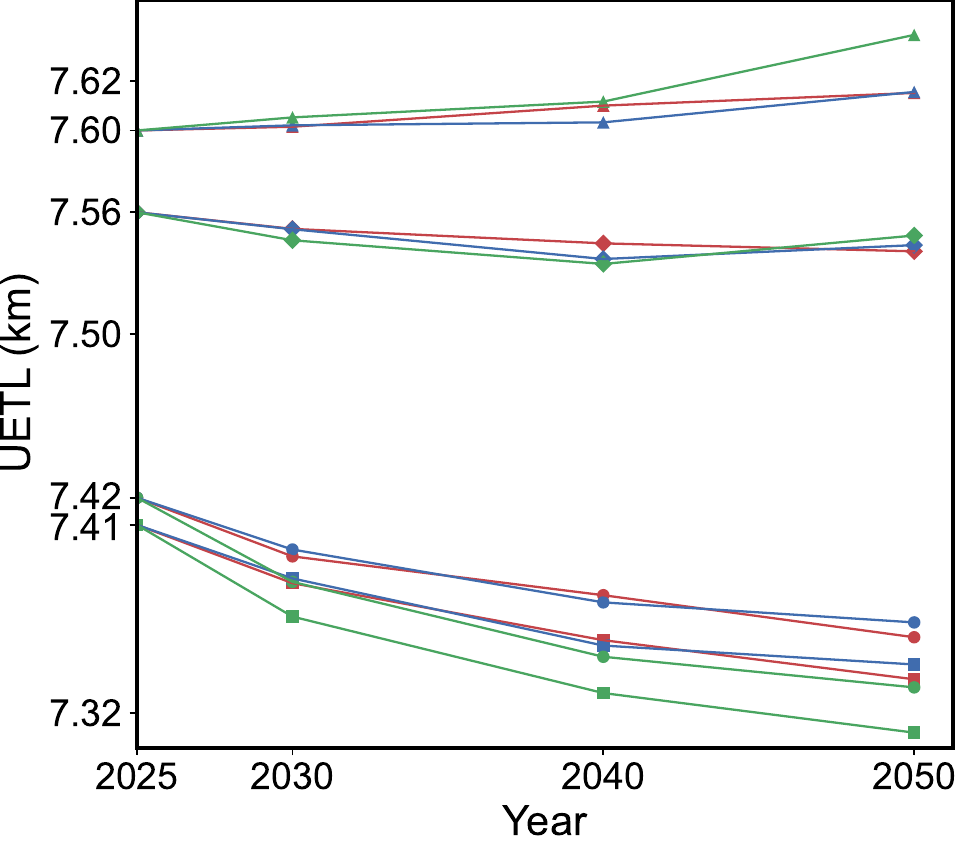}
                        \caption{Denver}
                        \label{fig:uetl_denver}
                    \end{subfigure}%
                    \hfill%
                    \begin{subfigure}[b]{0.33\textwidth}
                        \centering
                        \includegraphics[width=\linewidth, keepaspectratio]{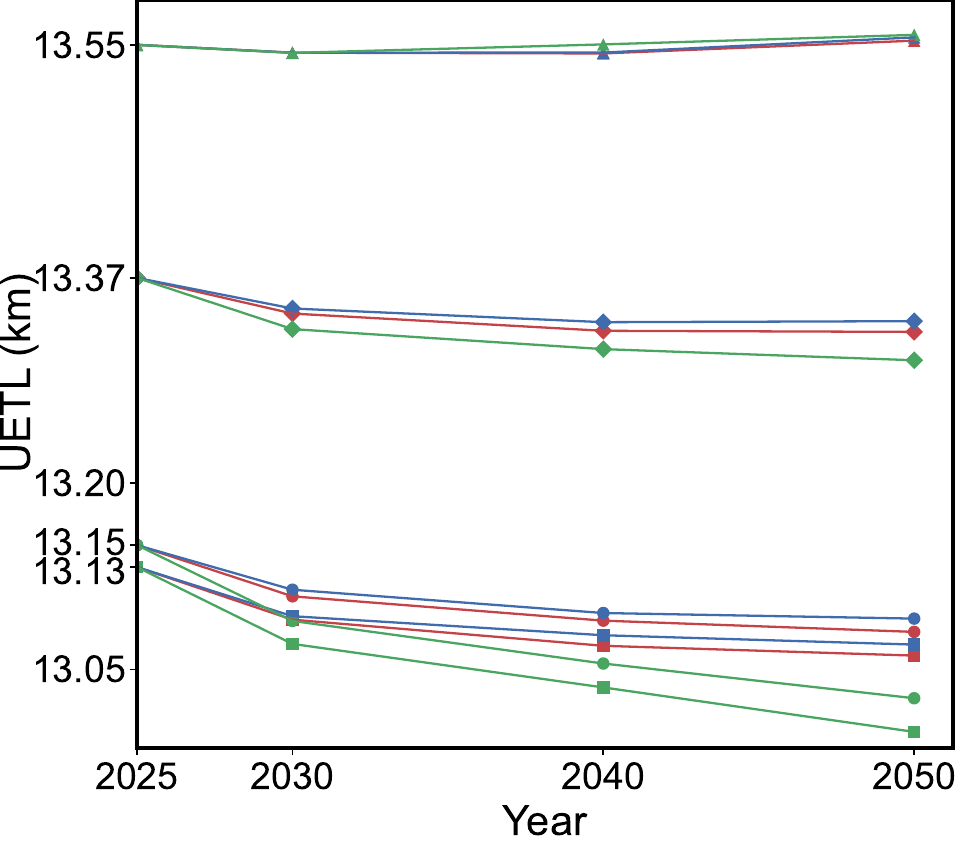}
                        \caption{Philadelphia}
                        \label{fig:uetl_philadelphia}
                    \end{subfigure}%
                    \vspace{-0.7em}
                    \caption{UETL (km) for the six representative U.S. cities across all years, SSPs, and UE models.}
                    \label{fig:UETL}
                    \vspace{-1.0em}
                \end{figure}

                For Time-Carbon UE models, the trends of UETL over time are different across various weight ratios.
                We first analyze the results with ratios of $\psi_{1}=1$, $\psi_{2}=9$ and $\psi_{1}=0$, $\psi_{2}=10$.
                Within a given city and SSP scenario, as population grows over time,
                UETL increases under the Time-Only UE model,
                whereas it decreases under the Time-Carbon UE models.
                This indicates that with population growth,
                travelers using time-routing increasingly choose longer routes;
                conversely, travelers using eco-routing increasingly select shorter routes,
                consistent with Theorem~\ref{thm:short_distance}
                that reducing travel distance is also critical for emission mitigation when employing eco-routing under population growth.
                Additionally, for the Time-Only UE model, minor deviations from the general UETL trend are observed,
                such as in Pittsburgh (2040) and (SSP5, 2050), Miami (SSP1, 2030), and Philadelphia (2030).
                These deviations typically occur in (SSPs, years) with small population growth,
                and are attributed to this slight growth and changes in the OD demand distribution $\pi_{y}^{s, t}$.
                With greater population growth, these exceptions are no longer observed;
                for example, Miami (SSP5, 2030), having a higher population size than (SSP1, 2030), follows the general UETL trend.
                Similarly, the few exceptions under the Time-Carbon UE model are attributed to the same factors,
                such as in Dallas (SSP1/2, 2050), where the exception is eliminated under (SSP5, 2050).
                These exceptions indicate that while changes in $\pi_{y}^{s, t}$ do influence UETL,
                their impacts are typically masked by the effects of population growth.
                The influence of $\pi_{y}^{s, t}$ becomes discernible only when population growth is minimal,
                implying that population acts as the dominant factor driving UETL and travel route choice.
                \begin{table}[!htb]
                \centering
                \footnotesize
                \begin{threeparttable}
                    \caption{UETL (km) for the six representative U.S. cities under Time-Only and Time-Carbon ($\psi_{1}=1$ and $\psi_{2}=9$) UE models.}
                    \label{tab:UETL table}
                    \setlength{\tabcolsep}{3pt}
                    \begin{tabularx}{0.98\textwidth}{c @{} c *{10}{C}}
                        \toprule
                        \multirow{2}{*}{City} & & baseline & \multicolumn{3}{c}{SSP1} & \multicolumn{3}{c}{SSP2} & \multicolumn{3}{c}{SSP5} \\
                        \addlinespace[-3pt]
                        \cmidrule(lr){3-3} \cmidrule(lr){4-6} \cmidrule(lr){7-9} \cmidrule(lr){10-12}
                        \addlinespace[-3pt]
                         & & 2025 & 2030 & 2040 & 2050 & 2030 & 2040 & 2050 & 2030 & 2040 & 2050 \\
                        \midrule
                        \multirow{2}{*}{San Francisco}  & TO    & 6.53  & 6.55 & 6.57 & 6.59 & 6.56 & 6.57 & 6.58 & 6.57 & 6.59 & 6.61  \\
                                                        & TC  & 6.01  & 5.97 & 5.94 & 5.92 & 5.97 & 5.95 & 5.93 & 5.96 & 5.92 & 5.89  \\[3pt]
                        \multirow{2}{*}{Dallas}         & TO    & 17.53  & 17.57 & 17.61 & 17.68 & 17.57 & 17.60 & 17.67 & 17.59 & 17.66 & 17.73  \\
                                                        & TC  & 17.22  & 17.20 & 17.17 & 17.18 & 17.20 & 17.17 & 17.18 & 17.19 & 17.16 & 17.14  \\[3pt]
                        \multirow{2}{*}{Pittsburgh}     & TO    & 9.08  & 9.09 & 9.07 & 9.05 & 9.09 & 9.07 & 9.05 & 9.09 & 9.07 & 9.05 \\
                                                        & TC  & 8.94  & 8.95 & 8.93 & 8.91 & 8.94 & 8.92 & 8.90 & 8.95 & 8.93 & 8.90  \\[3pt]
                        \multirow{2}{*}{Miami}          & TO    & 4.262  & 4.258 & 4.276 & 4.300 & 4.258 & 4.273 & 4.298 & 4.269 & 4.291 & 4.326  \\
                                                        & TC  & 4.145  & 4.139 & 4.131 & 4.131 & 4.140 & 4.135 & 4.136 & 4.138 & 4.122 & 4.122  \\[3pt]
                        \multirow{2}{*}{Denver}         & TO    & 7.60  & 7.60 & 7.61 & 7.61 & 7.60 & 7.60 & 7.61 & 7.60 & 7.61 & 7.64  \\
                                                        & TC  & 7.42  & 7.39 & 7.38 & 7.36 & 7.40 & 7.37 & 7.36 & 7.38 & 7.35 & 7.33  \\[3pt]
                        \multirow{2}{*}{Philadelphia}   & TO    & 13.553  & 13.547 & 13.546 & 13.557 & 13.547 & 13.547 & 13.559 & 13.547 & 13.554 & 13.561  \\
                                                        & TC  & 13.150  & 13.109 & 13.089 & 13.080 & 13.114 & 13.095 & 13.091 & 13.089 & 13.055 & 13.027  \\
                        \bottomrule
                    \end{tabularx}
                \end{threeparttable}
                \vspace{-1.0em}
                \end{table}
                When $\psi_{1}=8$ and $\psi_{2}=2$, it is difficult to identify a general UETL trend,
                as it is unstable across cities and years.
                This phenomenon arises because the relative weight of the carbon cost is too small,
                causing the dominance of time cost versus carbon cost in the generalized cost to vary by city and year.
                Within some cities and year intervals, the UETL increases as population grows,
                aligning with the Time-Only UE model (time cost dominates);
                within others, it decreases, consistent with the Time-Carbon UE model (carbon cost dominates).
                Consequently, given the behavioral inertia of travelers,
                the practical implementation of eco-routing necessitates appropriate $\psi_{1}$ and $\psi_{2}$ to ensure carbon cost dominance,
                thereby guaranteeing the stable evolution of the lowest carbon routes under population growth.


        \subsection{Impact of Population Growth}
        \label{sec:4.3}
            In this section, we aggregate population data for the six U.S. cities across all years and SSPs
            to analyze urban dynamics with respect to overall population change and evaluate the efficiency of eco-routing.
            To ensure the dominance of carbon cost while accounting for time cost,
            subsequent analysis primarily considers the cost weights $\psi_{1}=1$ and $\psi_{2}=9$, unless otherwise noted.


            \subsubsection{Impact of Population Growth on Urban Metrics}
            \label{sec:4.3.1}
                By selecting the (baseline, 2025) population and urban metrics as the reference,
                we obtain $G_{\mathrm{baseline}, \mathrm{2025}}$ to represent the percentage population growth, and calculate the change ratios of urban metrics.
                As shown in Fig.~\ref{fig:population change}, the change ratios of carbon emissions, UETL, and UETT all exhibit expected trends.
                Moreover, although the urban metrics are influenced by changes in population and OD demand distribution $\pi_{y}^{s, t}$,
                their respective trends remain consistent across most cities,
                indicating that population growth is the dominant driver of urban dynamics.
                These observations collectively corroborate the analysis in Section~\ref{sec:4.1}.
                \begin{figure}[!htb]
                    \captionsetup[subfigure]{skip=1pt}
                    \centering
                    \begin{subfigure}[b]{0.33\textwidth}
                        \centering
                        \includegraphics[width=\linewidth, keepaspectratio]{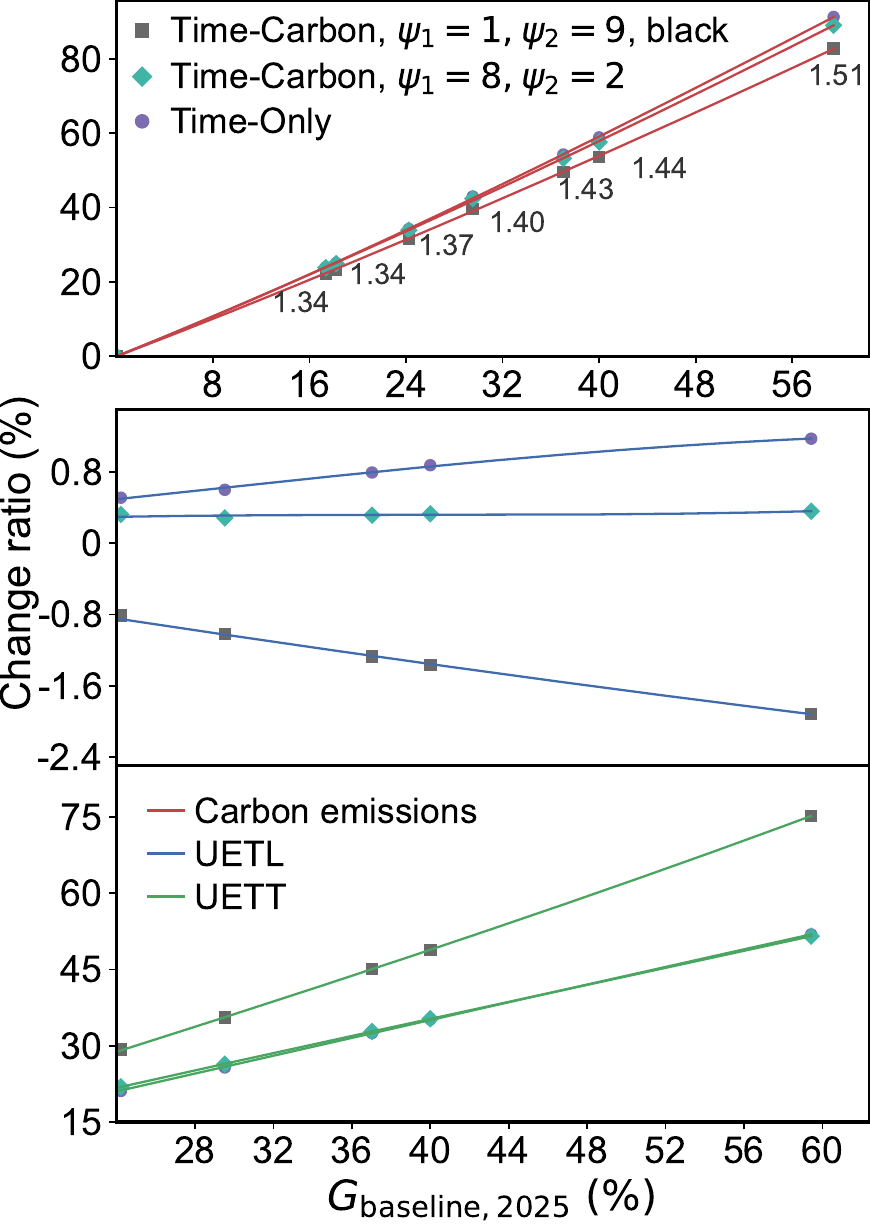}
                        \caption{San Francisco}
                        \label{fig:pop_sf}
                    \end{subfigure}%
                    \hfill%
                    \begin{subfigure}[b]{0.33\textwidth}
                        \centering
                        \includegraphics[width=\linewidth, keepaspectratio]{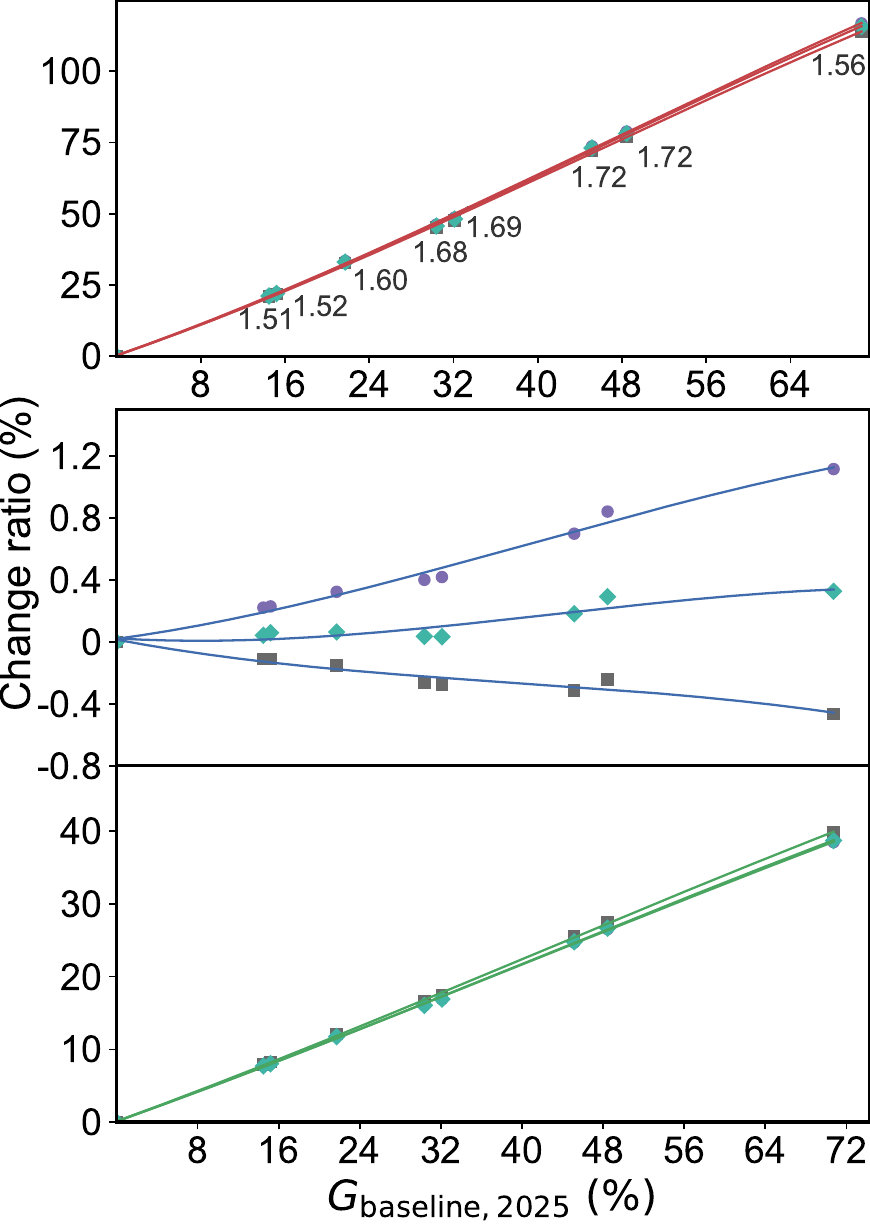}
                        \caption{Dallas}
                        \label{fig:pop_dallas}
                    \end{subfigure}%
                    \hfill%
                    \begin{subfigure}[b]{0.33\textwidth}
                        \centering
                        \includegraphics[width=\linewidth, keepaspectratio]{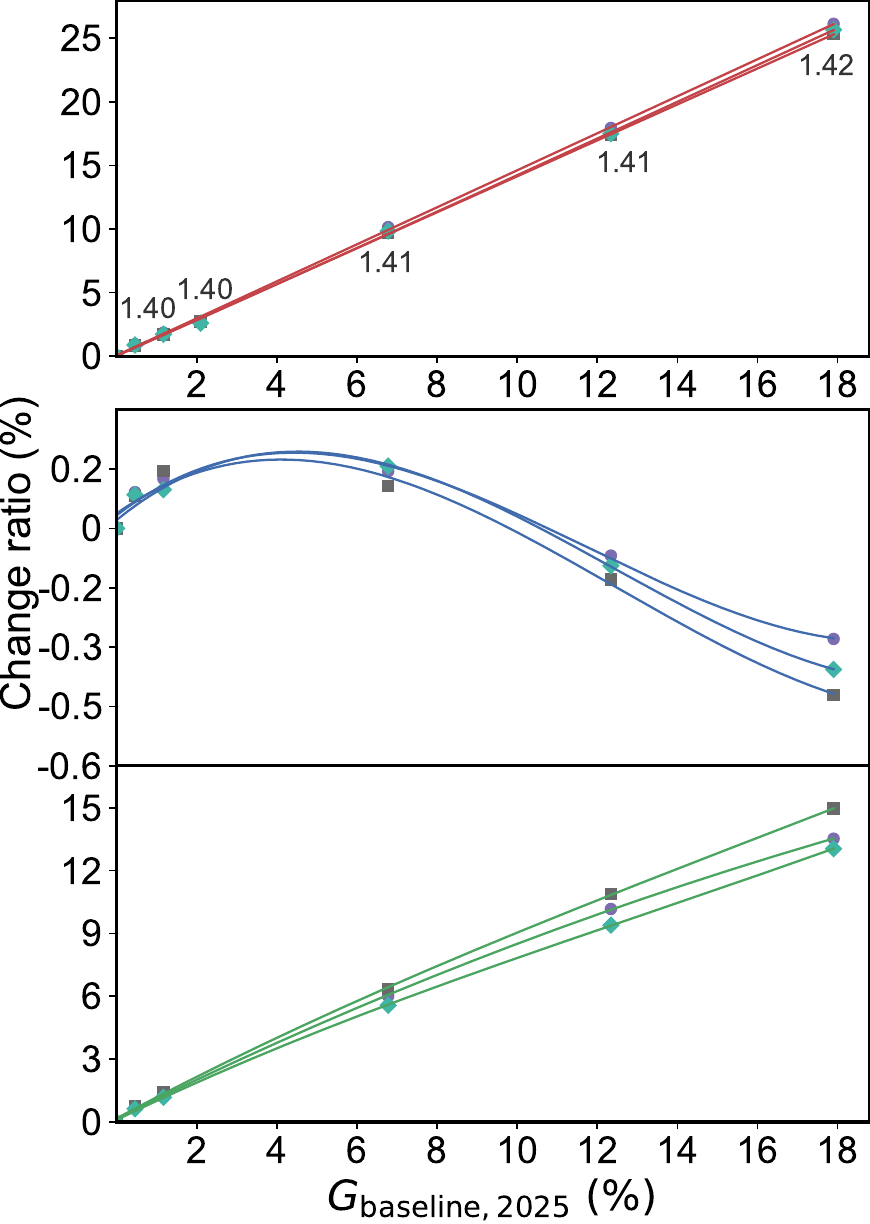}
                        \caption{Pittsburgh}
                        \label{fig:pop_pittsburgh}
                    \end{subfigure}%
                    \\[-1.5pt]
                    \begin{subfigure}[b]{0.33\textwidth}
                        \centering
                        \includegraphics[width=\linewidth, keepaspectratio]{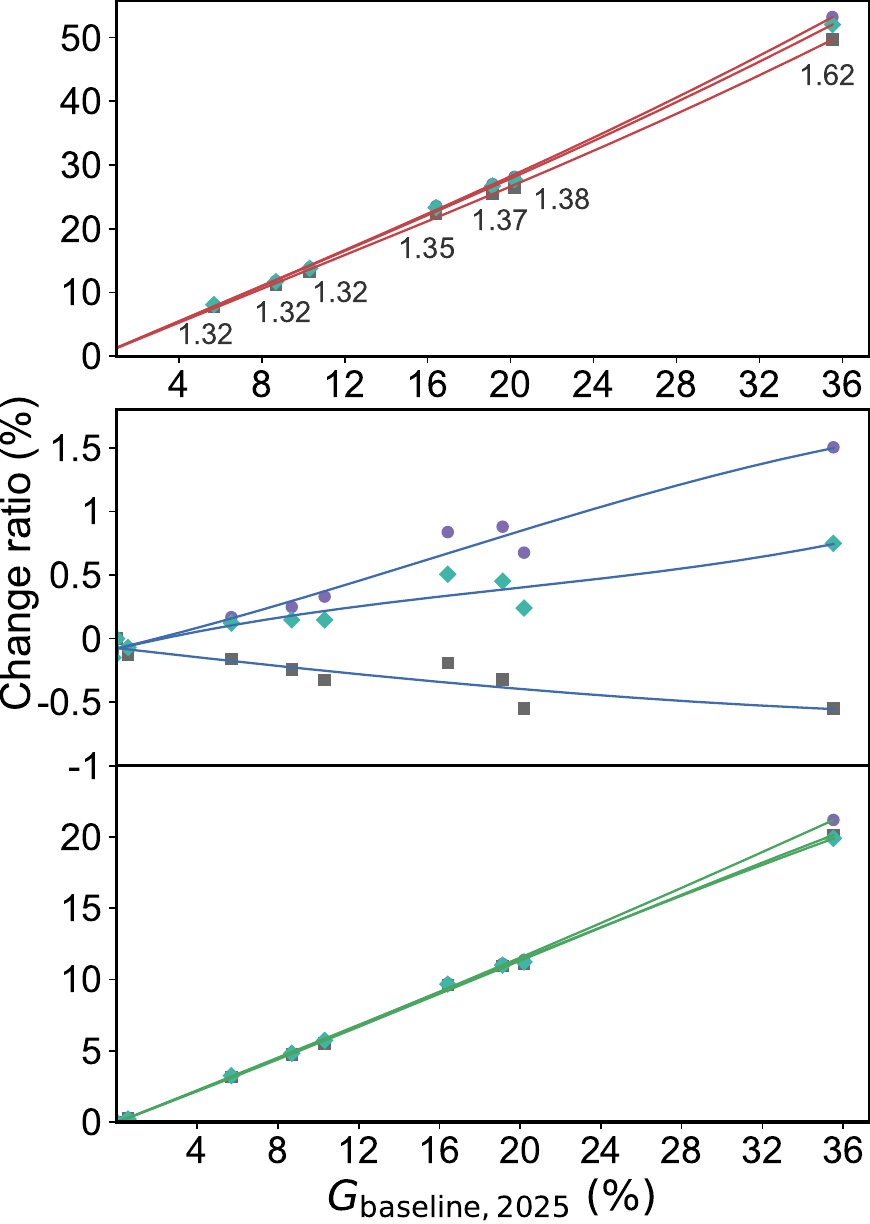}
                        \caption{Miami}
                        \label{fig:pop_miami}
                    \end{subfigure}%
                    \hfill%
                    \begin{subfigure}[b]{0.33\textwidth}
                        \centering
                        \includegraphics[width=\linewidth, keepaspectratio]{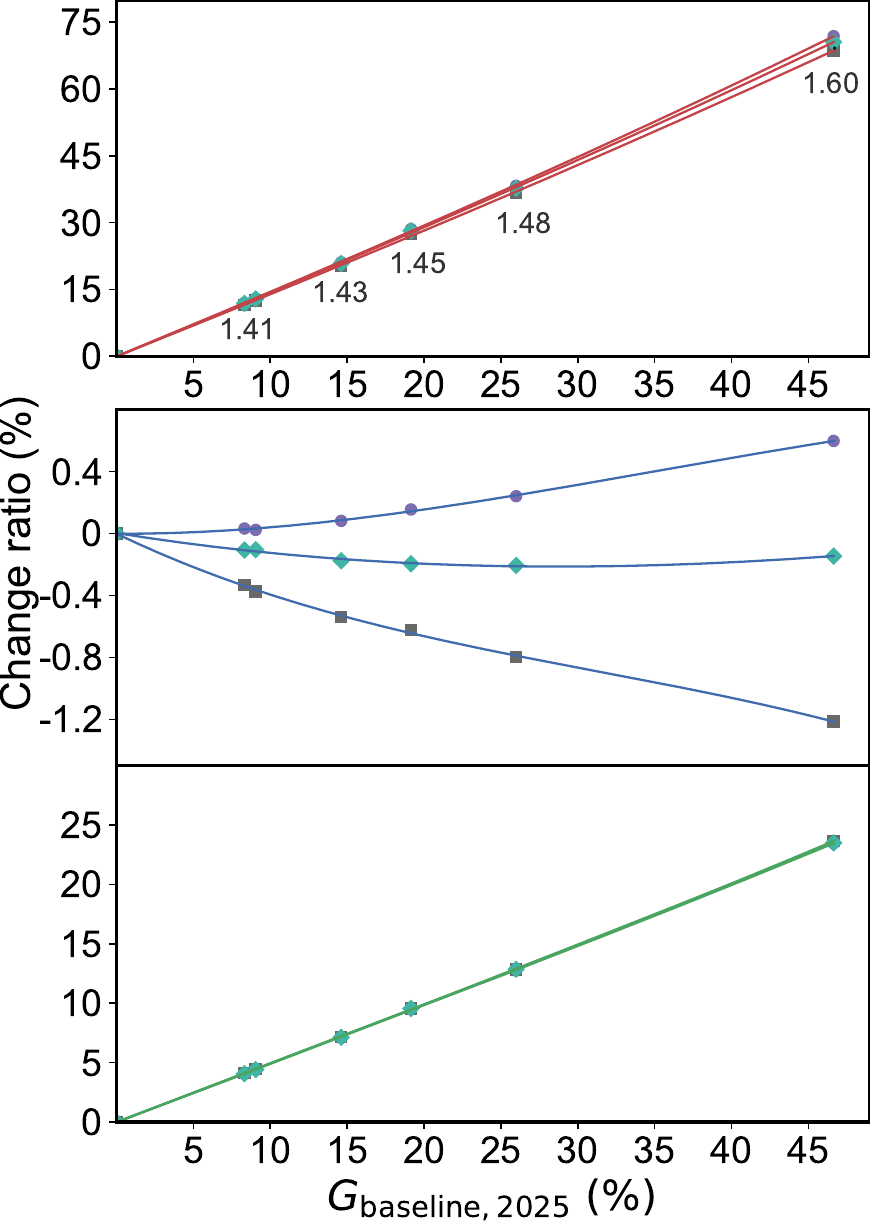}
                        \caption{Denver}
                        \label{fig:pop_denver}
                    \end{subfigure}%
                    \hfill%
                    \begin{subfigure}[b]{0.33\textwidth}
                        \centering
                        \includegraphics[width=\linewidth, keepaspectratio]{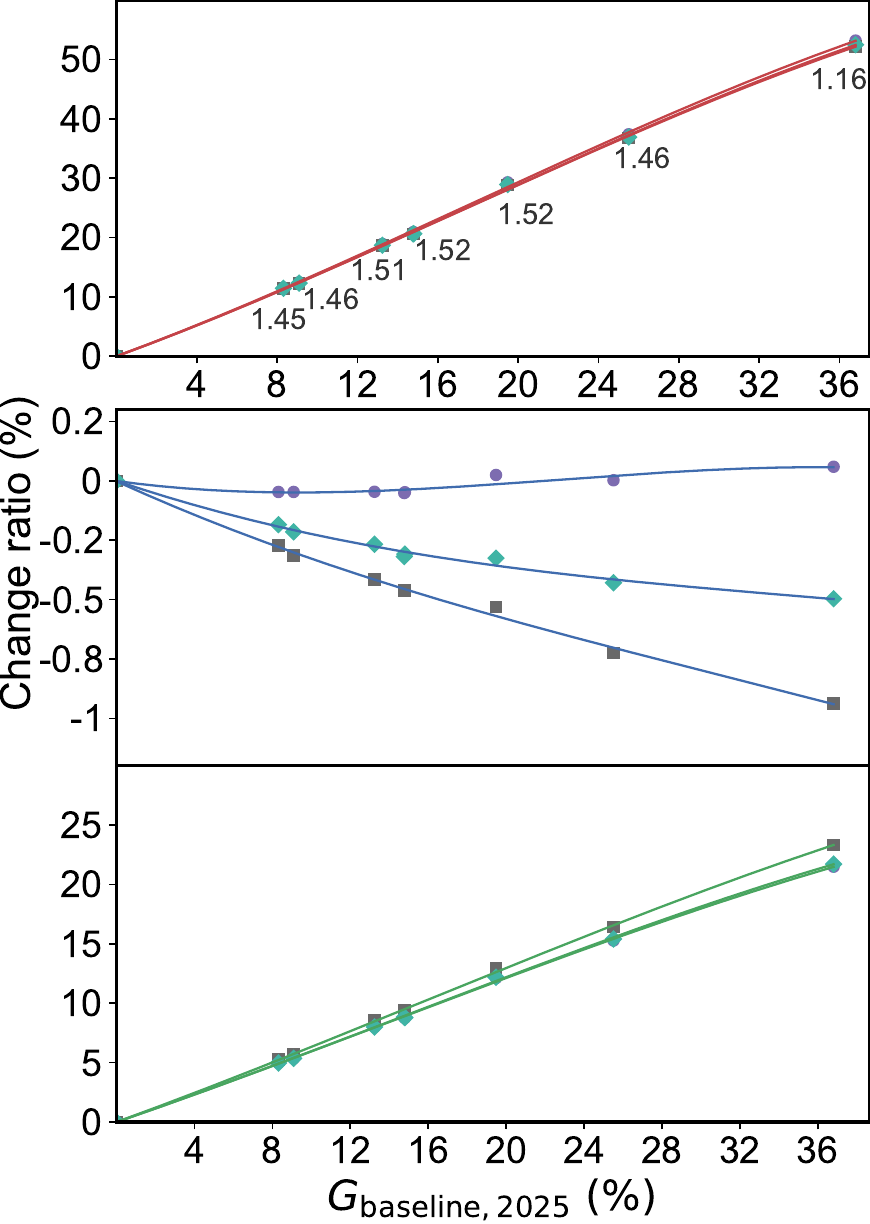}
                        \caption{Philadelphia}
                        \label{fig:pop_philadelphia}
                    \end{subfigure}%
                    \vspace{-0.7em}
                    \caption{Change ratios of carbon emissions (top), UETT (middle), and UETL (bottom)
                        with respect to population growth percentage $G_{\mathrm{baseline}, \mathrm{2025}}$ for the six U.S. cities.
                   }
                    \label{fig:population change}
                    \vspace{-1.0em} 
                \end{figure}

                Across most cities, under the Time-Only UE model, both UETL and UETT increase as the population grows.
                In contrast, under the Time-Carbon UE model, UETL decreases while UETT increases,
                with the percentage increase in UETT greater than that under the Time-Only UE model.
                Comparing these results demonstrates that under population growth,
                travelers using time-routing prioritize maintaining travel speed despite an increase in travel distance,
                whereas those using eco-routing prioritize reducing travel distance despite a significant decline in travel speed.
                These findings indicate that the carbon reduction benefits from reducing travel length outweigh the benefits from maintaining travel speed,
                aligning with Theorem~\ref{thm:short_distance}
                that reducing travel distance is significantly more critical than maintaining travel speed for mitigating emissions driven by population growth.
                However, this does not imply that travel speed is unimportant;
                rather, it reveals that under population growth, links along routes selected by eco-routing,
                especially those on short routes, will inevitably become increasingly congested,
                thereby constraining carbon reduction potential, as is visually evident in Section~\ref{sec:4.2}.
                Therefore, to reduce emissions and construct low-carbon road networks,
                it is crucial to improve the jobs-housing balance, enhance spatial proximity, and promote polycentric development
                to counteract urban sprawl, thereby shortening travel distance.
                Simultaneously, strategies should also focus on improving the short routes to relieve the increasing congestion
                (Remark~\ref{re:construction strategies}).

                For all cities, the slopes of the emission change ratio curves are similar across all UE models,
                indicating that eco-routing cannot fundamentally slow down the upward trend of emissions driven by population growth.
                This finding corroborates Corollary~\ref{cor:eco_coefficient}
                that the scaling order of emissions with population remains invariant regardless of routing strategies.
                Moreover, as established in Corollary~\ref{cor:finite_expansion}, the implementation of bounded capacity expansion also fails to alter the scaling order.
                In addition, while emissions and population exhibit a global linear relationship (Theorem~\ref{thm:scaling}),
                from a local perspective within the examined population range,
                emissions in most cities scale superlinearly with population (Remark~\ref{re:Local scaling}).
                Only a minority of cases, such as Dallas and Philadelphia within specific population intervals, display a sublinear relationship.
                Given the generality of these findings, the results hold true even for networks resulting from new road construction.
                Therefore, although routing and construction strategies can reduce absolute emissions,
                they cannot alter the superlinear scaling order of emissions with population.
                To fundamentally mitigate future emissions,
                policies should extend beyond them to manage the increase in total demand,
                and to follow a sustainable SSP characterized by slower population growth.



            \subsubsection{Impact of Population Growth on Eco-routing Efficiency}
            \label{sec:4.3.3}
                In this section, we evaluate the efficiency of eco-routing based on the optimization levels.
                As shown in Fig.~\ref{fig:pop_carbon},
                San Francisco transitions from Level 1 through 3 to 4;
                Dallas shifts from Level 4 to 1;
                Pittsburgh remains in Level 4;
                Miami moves from Level 2 to 1 and subsequently to 3;
                Denver transitions from Level 4 through 3 to 1;
                and Philadelphia shifts from Level 2 to 4.
                Furthermore, the variations in optimization level for the six U.S. cities under different SSPs are illustrated in Fig.~\ref{fig:op_ssp}.
                Some cities, such as Dallas, Pittsburgh, and Denver, are already in inefficient optimization levels across all SSPs in 2025.
                Other cities, including San Francisco, Miami, and Philadelphia,
                begin in an efficient Level 1 or 2 but drop to Level 3 or 4 as population grows;
                notably, Miami and Philadelphia are projected to sustain Level 1 or 2 through 2050 under the SSP1 or SSP2 scenarios.
                \begin{figure}[!htb]
                    \captionsetup[subfigure]{skip=1pt}
                    \centering
                    \begin{subfigure}[b]{0.33\textwidth}
                        \centering
                        \includegraphics[width=\linewidth, keepaspectratio]{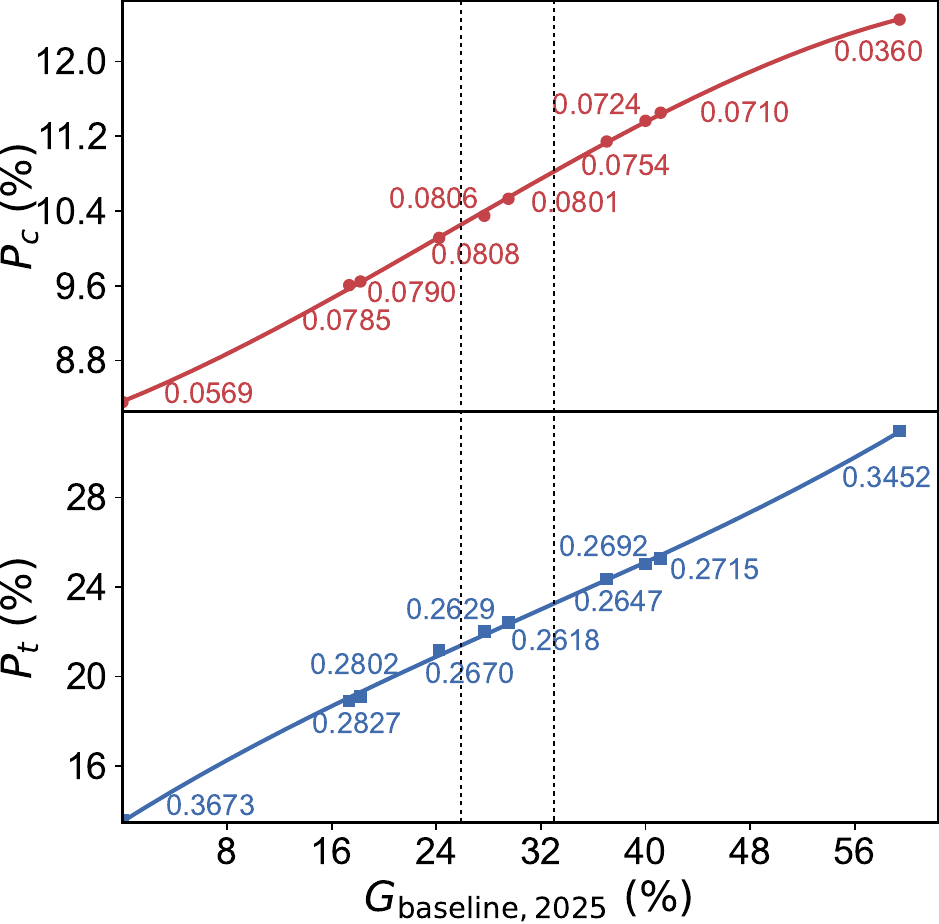}
                        \caption{San Francisco}
                        \label{fig:better_sf}
                    \end{subfigure}%
                    \hfill%
                    \begin{subfigure}[b]{0.33\textwidth}
                        \centering
                        \includegraphics[width=\linewidth, keepaspectratio]{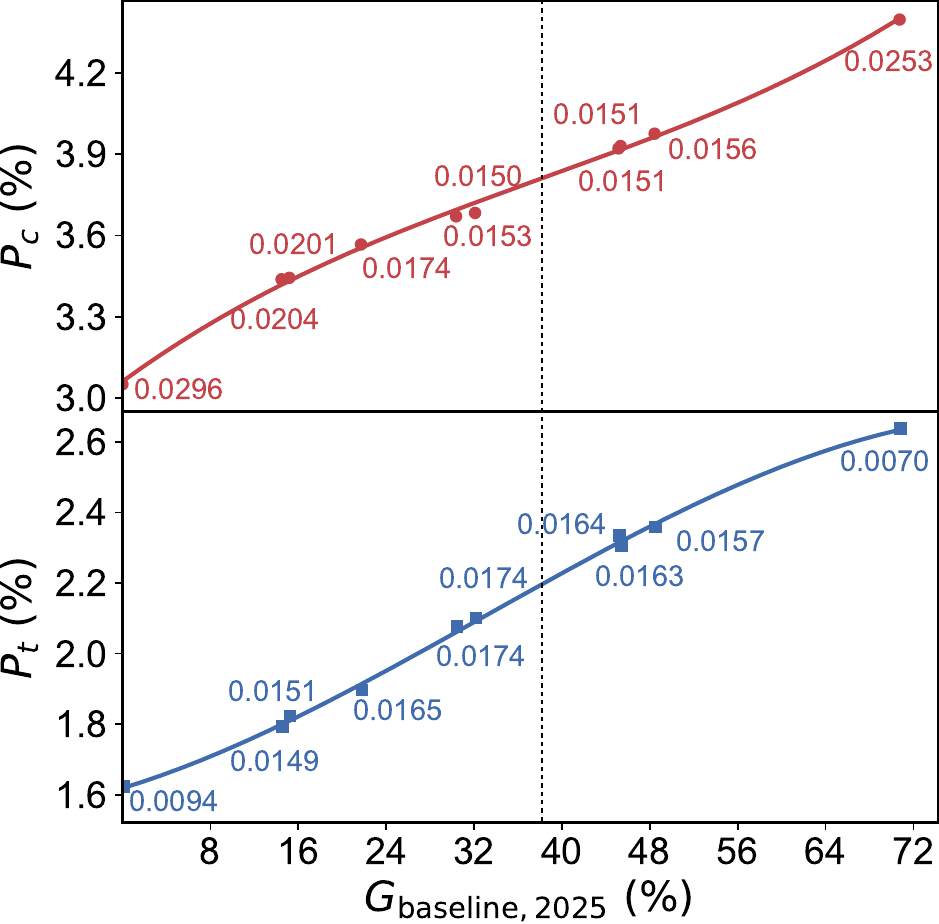}
                        \caption{Dallas}
                        \label{fig:better_dallas}
                    \end{subfigure}%
                    \hfill%
                    \begin{subfigure}[b]{0.33\textwidth}
                        \centering
                        \includegraphics[width=\linewidth, keepaspectratio]{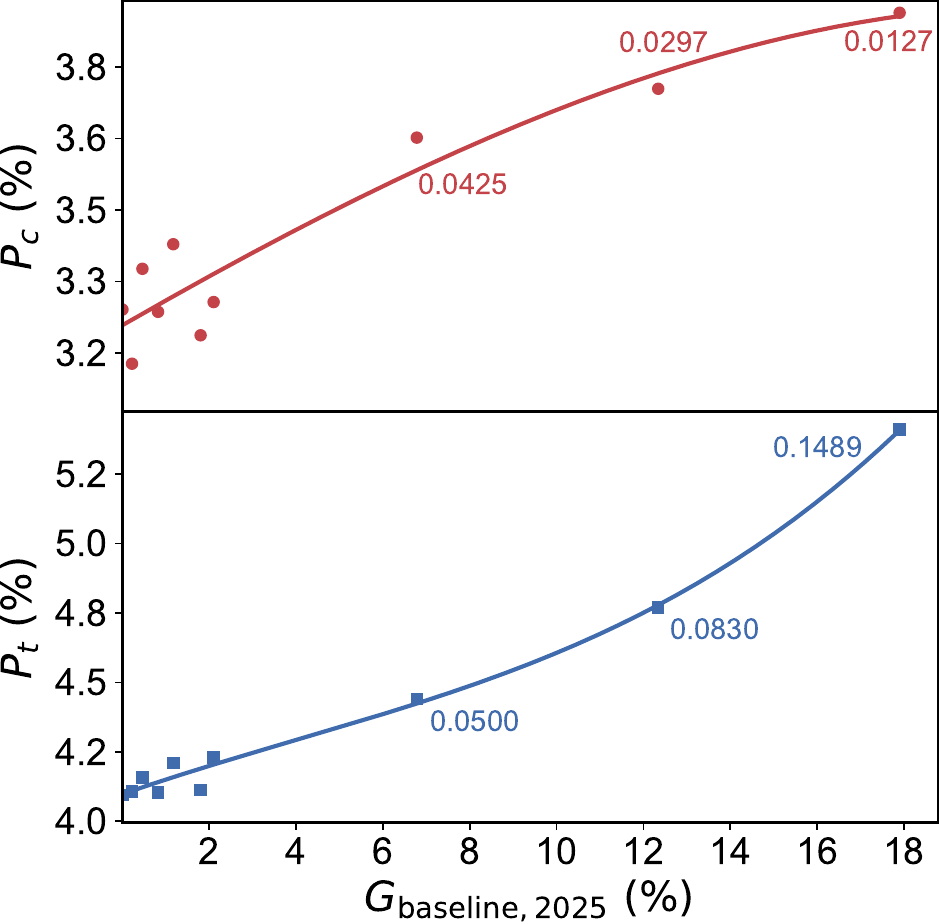}
                        \caption{Pittsburgh}
                        \label{fig:better_pittsburgh}
                    \end{subfigure}%
                    \\[-1.5pt]
                    \begin{subfigure}[b]{0.33\textwidth}
                        \centering
                        \includegraphics[width=\linewidth, keepaspectratio]{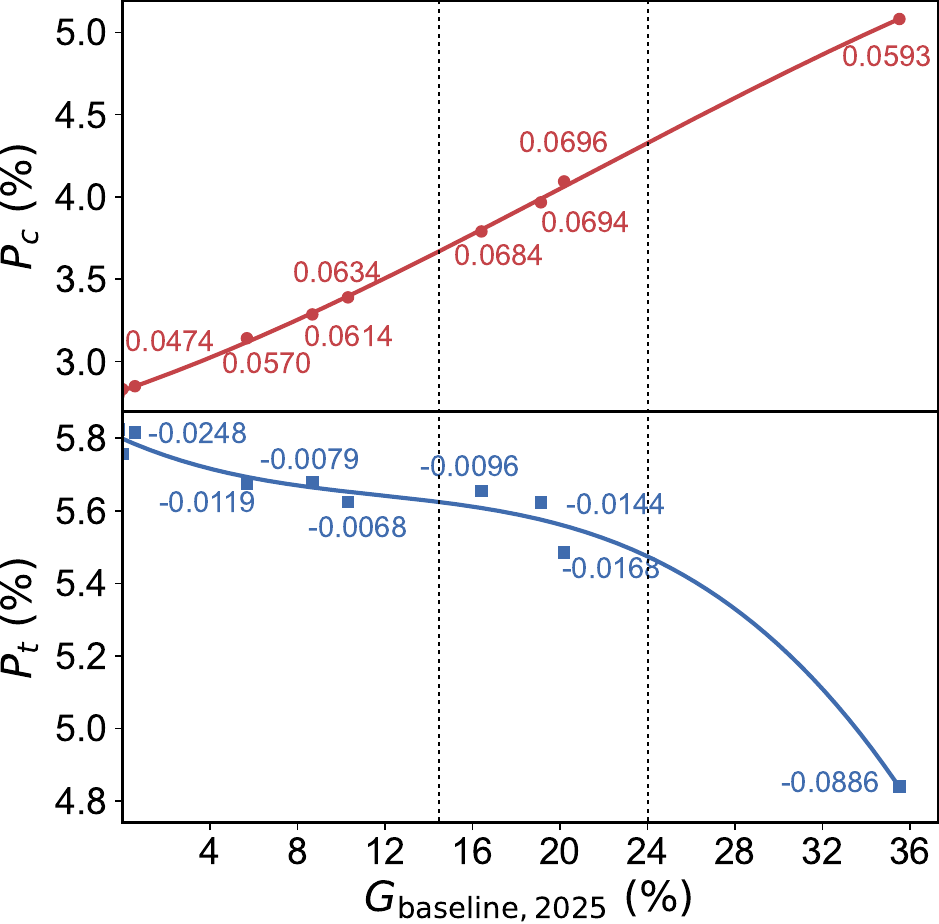}
                        \caption{Miami}
                        \label{fig:better_miami}
                    \end{subfigure}%
                    \hfill%
                    \begin{subfigure}[b]{0.33\textwidth}
                        \centering
                        \includegraphics[width=\linewidth, keepaspectratio]{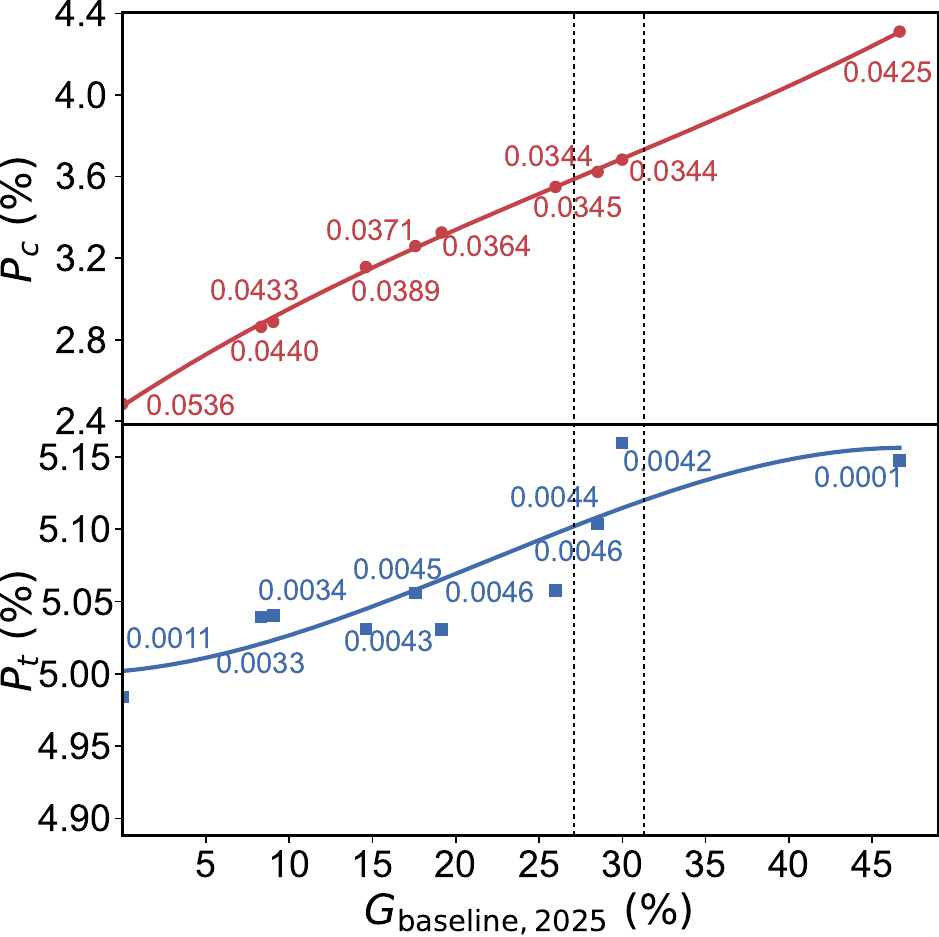}
                        \caption{Denver}
                        \label{fig:better_denver}
                    \end{subfigure}%
                    \hfill%
                    \begin{subfigure}[b]{0.33\textwidth}
                        \centering
                        \includegraphics[width=\linewidth, keepaspectratio]{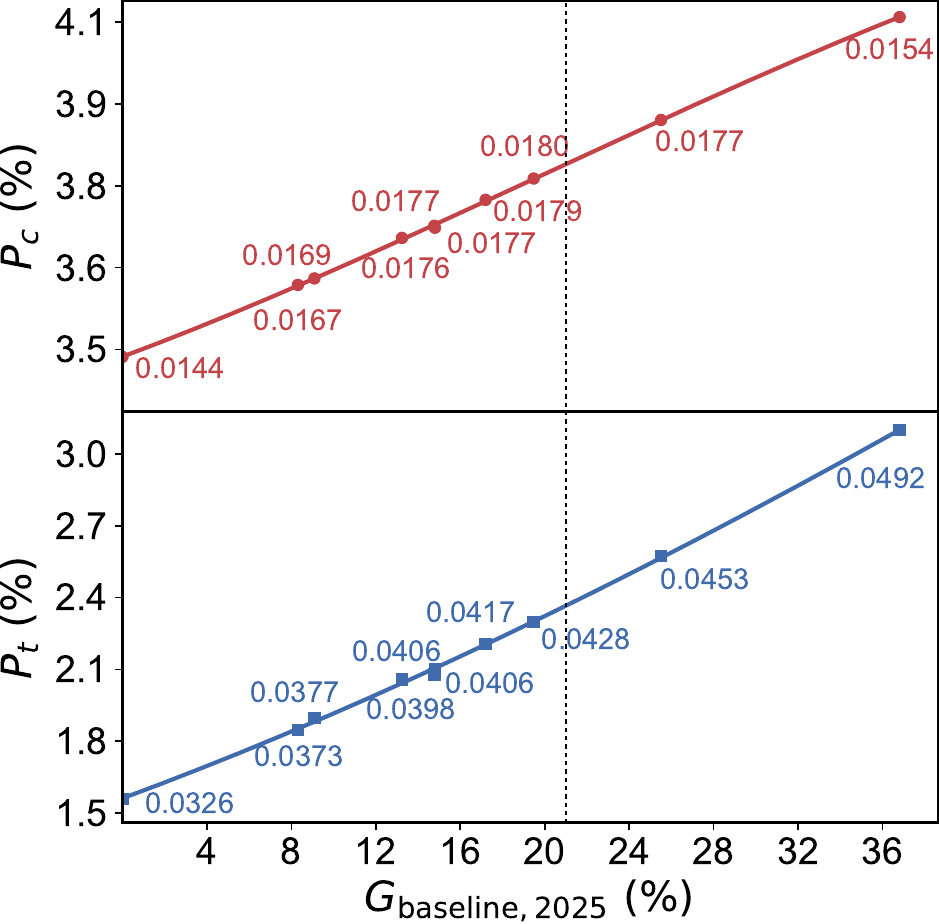}
                        \caption{Philadelphia}
                        \label{fig:better_philadelphia}
                    \end{subfigure}%
                    \vspace{-0.7em}
                    \caption{Optimization levels for the six U.S. cities
                        (vertical dashed lines indicate the boundaries between optimization levels).}
                    \label{fig:pop_carbon}
                    \vspace{-1.0em}
                \end{figure}

                \begin{figure}[!htb]
                    \captionsetup[subfigure]{skip=1pt}
                    \centering
                    \begin{subfigure}[b]{0.33\textwidth}
                        \centering
                        \includegraphics[width=\linewidth, keepaspectratio]{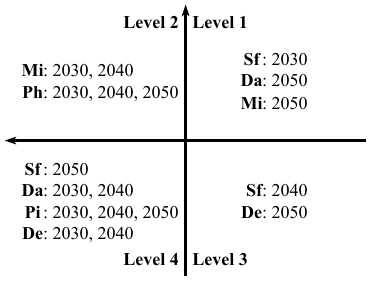}
                        \caption{SSP1 scenario}
                        \label{fig:op_ssp1}
                    \end{subfigure}%
                    \hfill%
                    \begin{subfigure}[b]{0.33\textwidth}
                        \centering
                        \includegraphics[width=\linewidth, keepaspectratio]{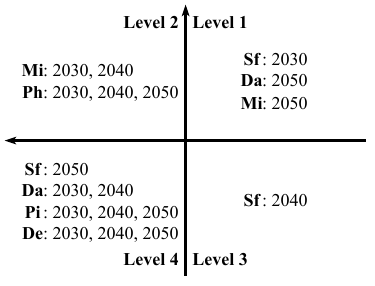}
                        \caption{SSP2 scenario}
                        \label{fig:op_ssp2}
                    \end{subfigure}%
                    \hfill%
                    \begin{subfigure}[b]{0.33\textwidth}
                        \centering
                        \includegraphics[width=\linewidth, keepaspectratio]{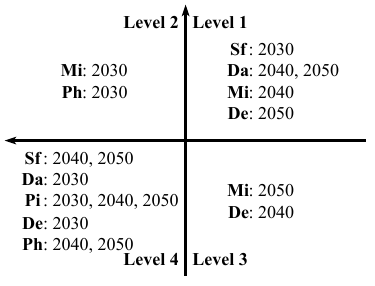}
                        \caption{SSP5 scenario}
                        \label{fig:op_ssp5}
                    \end{subfigure}%
                    \vspace{-0.7em}
                    \caption{Variations in optimization level for the six U.S. cities under different SSPs.
                        (The classification of optimization levels is based on Fig.~\ref{fig:Optimization Level}.
                        Sf: San Francisco; Da: Dallas; Pi: Pittsburgh; Mi: Miami; De: Denver; and Ph: Philadelphia.)}
                    \label{fig:op_ssp}
                    \vspace{-1.0em}
                \end{figure}

                These observations show that the eco-routing efficiency exhibits three distinct patterns:
                in some cities, eco-routing is already inefficient if the city resides in Level 3 or 4;
                in others, eco-routing is currently efficient but is projected to become inefficient under population growth,
                as its efficiency declines once the optimization level deviates from Level 1 or 2;
                and in a few cities, eco-routing maintains long-term high efficiency,
                provided that future development follows the sustainable SSP1 or SSP2 scenarios.
                Therefore, given that most cities either already operate at or tend to shift toward inefficient optimization levels,
                sole reliance on eco-routing for emission mitigation is insufficient under population growth.
                Complementary strategies should be adopted to directly reduce emissions and travel time, as well as to maintain efficient optimization levels.

        \subsection{Carbon Bottlenecks}
        \label{sec:4.4}
            \subsubsection{VOC}
            \label{sec:4.2}

                    In this section, we utilize VOC to analyze the urban dynamics from a microscopic road network perspective.
                    When travelers shift from time-routing to eco-routing,
                    as illustrated in Fig.~\ref{fig:01 VOCs map b},
                    links along shorter routes experience an increase in VOC;
                    conversely, links on longer routes experience a decrease.
                    For other cities, similar patterns are observed in Fig.~\ref{fig:2050 ssp2 six cities},
                    all of which corroborate Theorem~\ref{thm:short_distance} and the findings in Section~\ref{sec:4.1.3}.
                    Furthermore, there are a set of specific links that exhibit a significant increase in VOC,
                    and their spatial distribution remains nearly invariant across the three analysis years.
                    When using eco-routing under population growth,
                    as shown in Fig.~\ref{fig:01 VOCs map a},
                    in addition to a general increase in VOC observed across almost all road segments,
                    there are also specific links on short paths exhibiting a significant increase
                    while their spatial distribution remains highly consistent across years.
                    Moreover, as illustrated in Fig.~\ref{fig:more VOCs maps},
                    although these specific links are already characterized by high VOC and low speed,
                    their traffic flow continues to surge as the population grows.
                    These observations corroborate our analysis in Section~\ref{sec:4.3.1},
                    confirming that under population growth, links along routes selected by eco-routing,
                    especially those on short paths, will inevitably become increasingly congested.
                    \begin{figure}[!htb]
                        \captionsetup[subfigure]{skip=1pt}
                        \centering
                        \begin{subfigure}[b]{\textwidth}
                            \centering
                            \begin{minipage}[b]{0.33\textwidth}
                                \includegraphics[width=\linewidth]{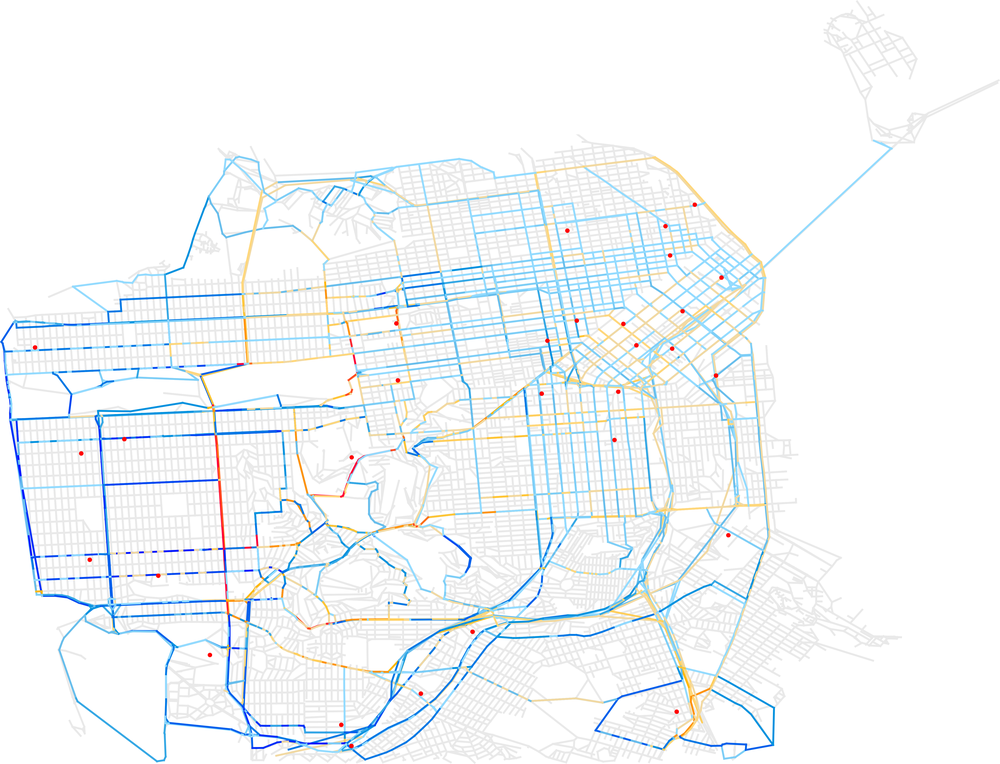}
                            \end{minipage}%
                            \hfill%
                            \begin{minipage}[b]{0.33\textwidth}
                                \includegraphics[width=\linewidth]{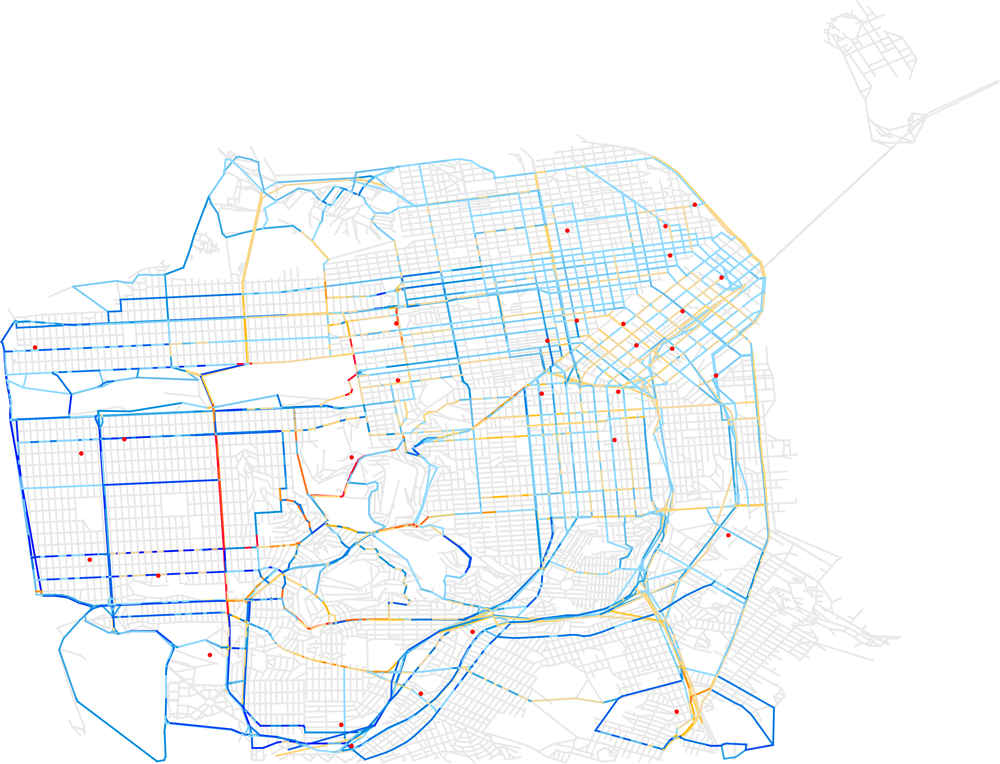}
                            \end{minipage}%
                            \hfill%
                            \begin{minipage}[b]{0.33\textwidth}
                                \includegraphics[width=\linewidth]{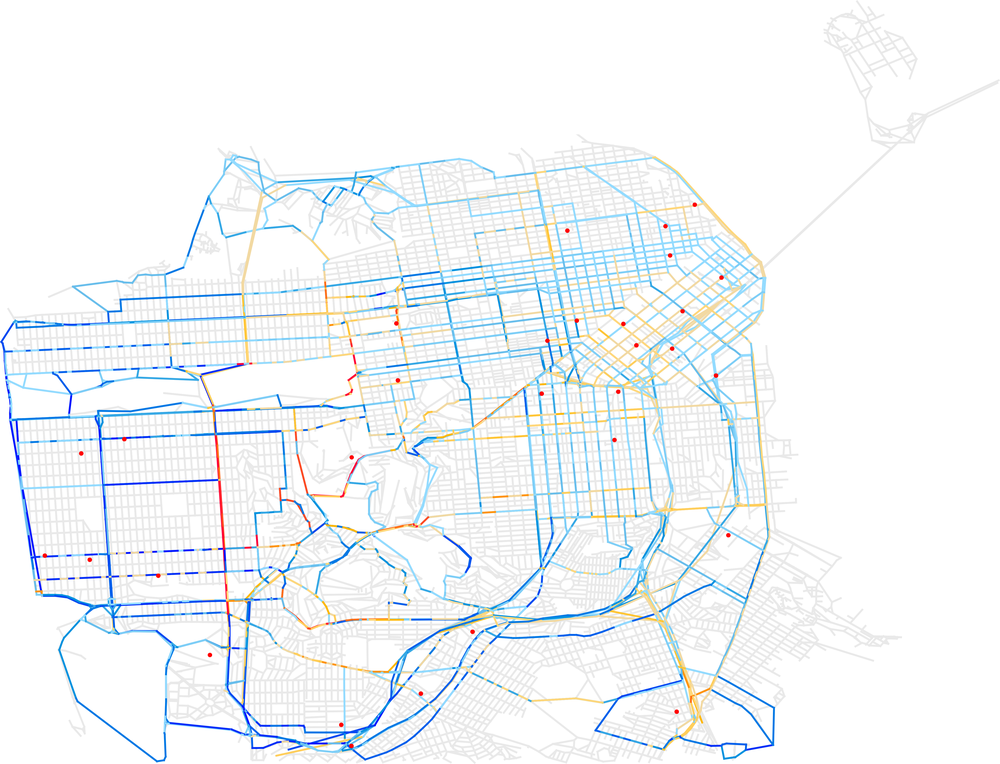}
                            \end{minipage}%
                        \end{subfigure}%
                        \par\vspace{-0.1pt}
                        \begin{subfigure}[b]{0.90\textwidth}
                            \centering
                            \includegraphics[width=\linewidth]{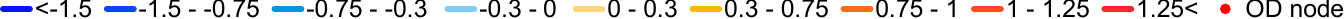}
                            \caption{2030 (left), 2040 (middle), and 2050 (right)}
                            \label{fig:01 VOCs map b}
                        \end{subfigure}%
                        \\
                        \begin{subfigure}[b]{\textwidth}
                            \centering
                            \begin{minipage}[b]{0.33\textwidth}
                                \includegraphics[width=\linewidth]{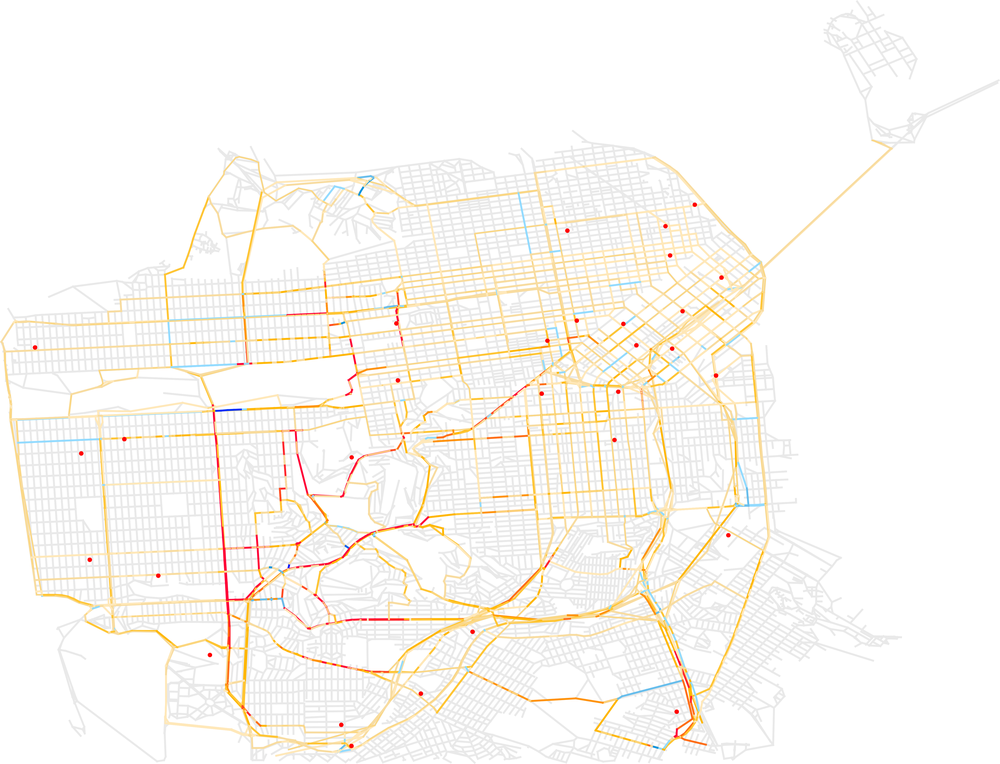}
                            \end{minipage}%
                            \hfill%
                            \begin{minipage}[b]{0.33\textwidth}
                                \includegraphics[width=\linewidth]{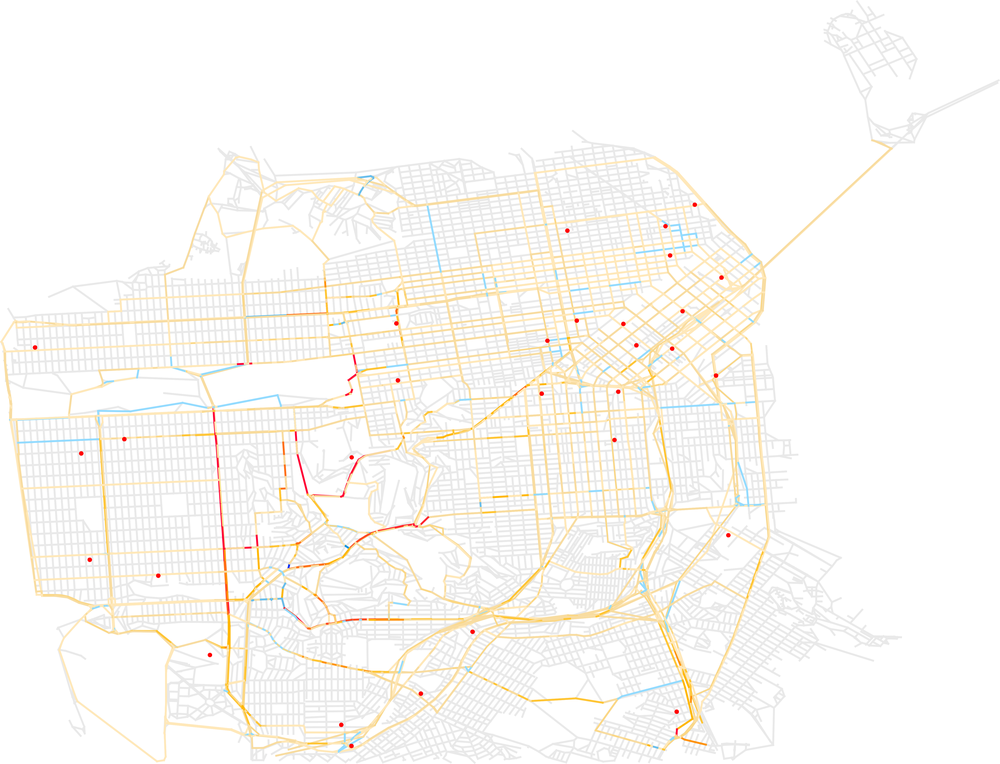}
                            \end{minipage}%
                            \hfill%
                            \begin{minipage}[b]{0.33\textwidth}
                                \includegraphics[width=\linewidth]{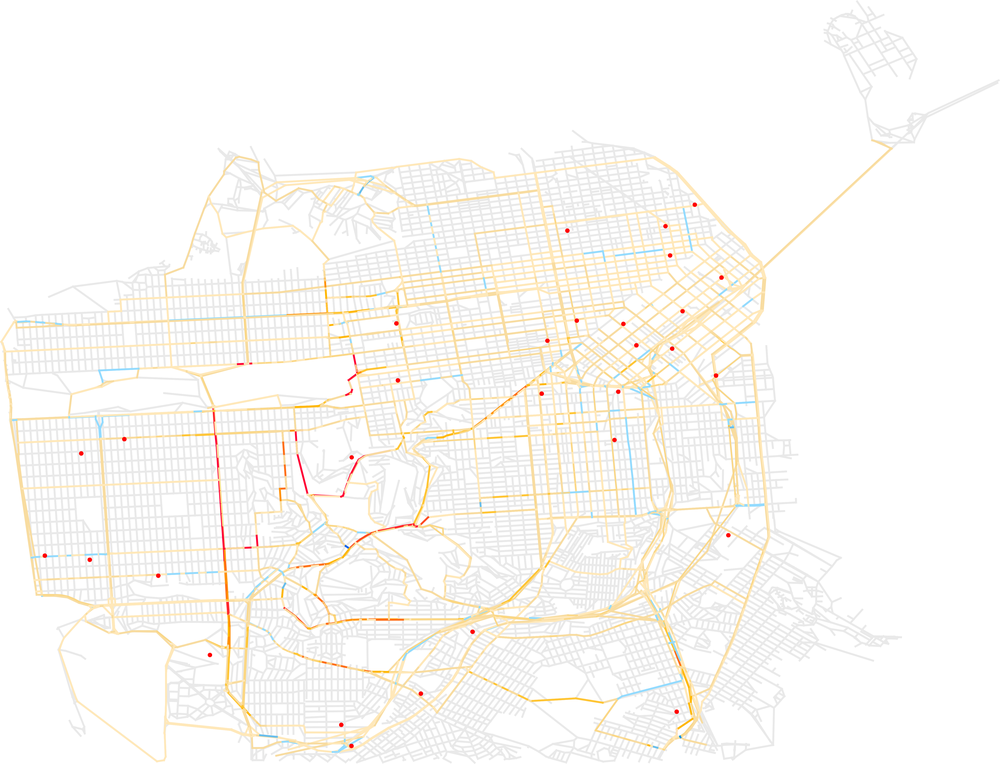}
                            \end{minipage}%
                        \end{subfigure}%
                        \par\vspace{-0.1pt}
                        \begin{subfigure}[b]{0.95\textwidth}
                            \centering
                            \includegraphics[width=\linewidth]{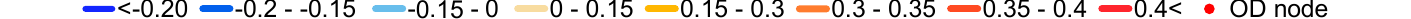}
                            \caption{2030 (left), 2040 (middle), and 2050 (right)}
                            \label{fig:01 VOCs map a}
                        \end{subfigure}%
                        \vspace{-0.7em}
                        \caption{Maps of VOC difference in San Francisco under the SSP2 scenario.
                                (a) The VOC difference between different UE models,
                                calculated as the VOC under the Time-Carbon UE model ($\psi_{1}=1$, $\psi_{2}=9$) minus that under the Time-Only UE model.
                                (b) The VOC difference between different years under the Time-Carbon UE model ($\psi_{1}=1$, $\psi_{2}=9$),
                                calculated as the VOC in specific year minus that in corresponding previous analysis years.
                                (OD nodes represent the top 15 origins and destinations ranked by demand volume.)}
                        \label{fig:VOCs map}
                        \vspace{-1.0em}
                    \end{figure}

                    \begin{figure}[!htb]
                        \captionsetup[subfigure]{skip=1pt}
                        \centering
                        \begin{subfigure}[b]{0.33\textwidth}
                            \centering
                            \includegraphics[width=\linewidth, keepaspectratio]{01_2050_ssp2}
                            \caption{San Francisco}
                            \label{fig:voc_sf}
                        \end{subfigure}%
                        \hfill%
                        \begin{subfigure}[b]{0.33\textwidth}
                            \centering
                            \includegraphics[width=\linewidth, keepaspectratio]{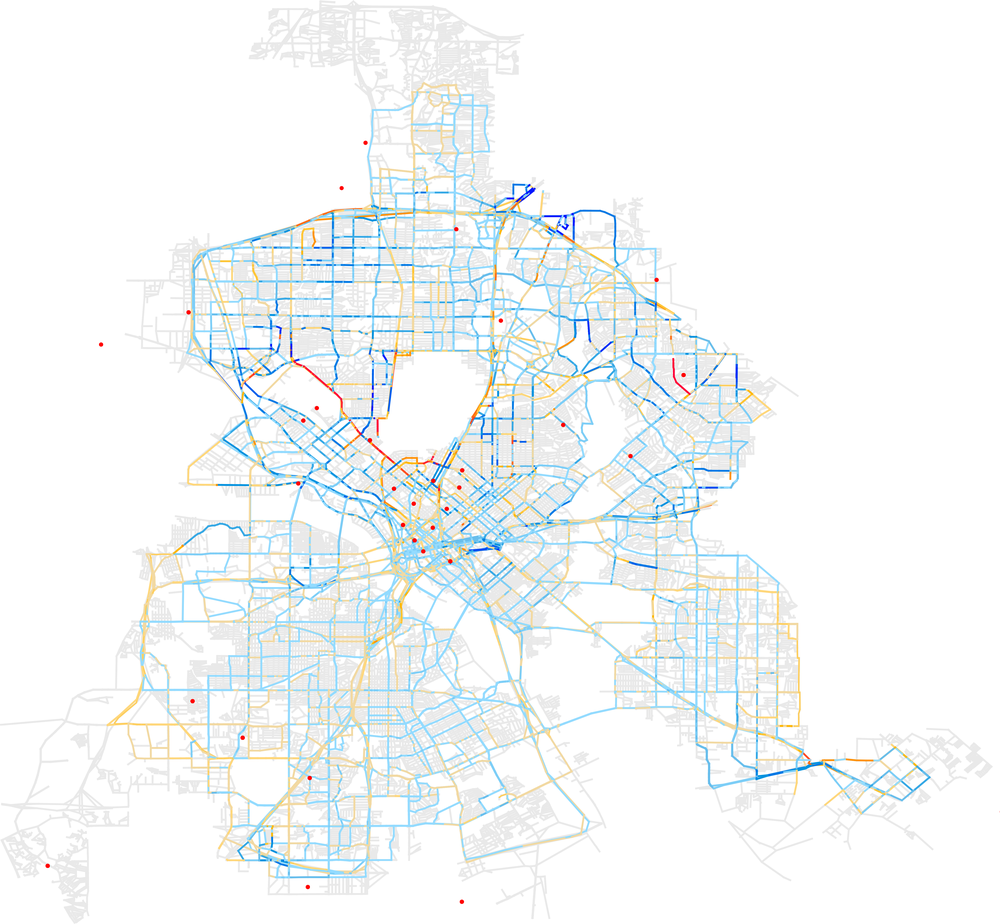}
                            \caption{Dallas}
                            \label{fig:voc_dallas}
                        \end{subfigure}%
                        \hfill%
                        \begin{subfigure}[b]{0.33\textwidth}
                            \centering
                            \includegraphics[width=\linewidth, keepaspectratio]{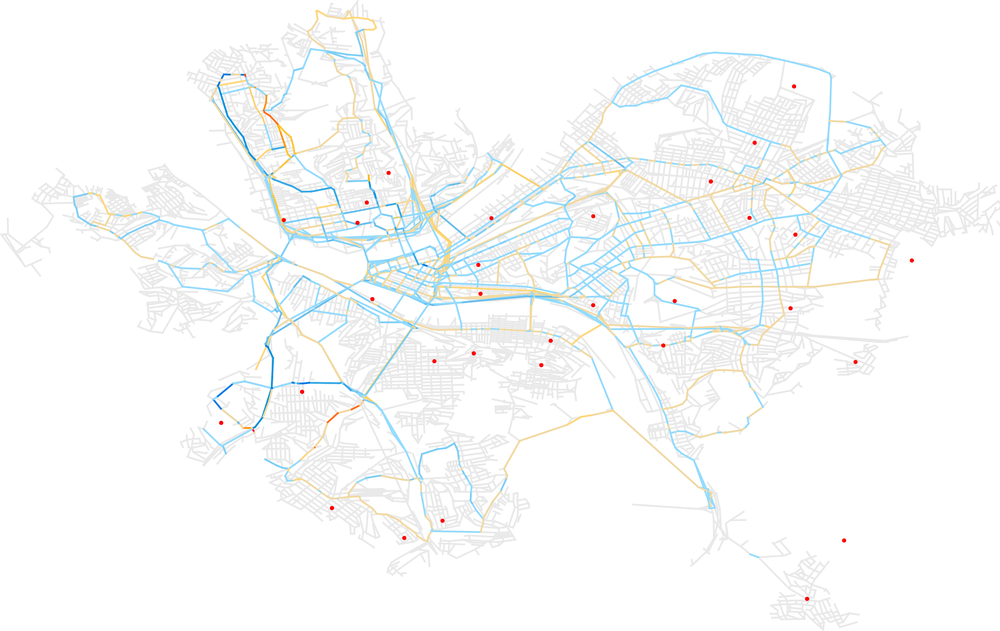}
                            \caption{Pittsburgh}
                            \label{fig:voc_pittsburgh}
                        \end{subfigure}%
                        \\[-1.5pt]
                        \begin{subfigure}[b]{0.33\textwidth}
                            \centering
                            \includegraphics[width=\linewidth, keepaspectratio]{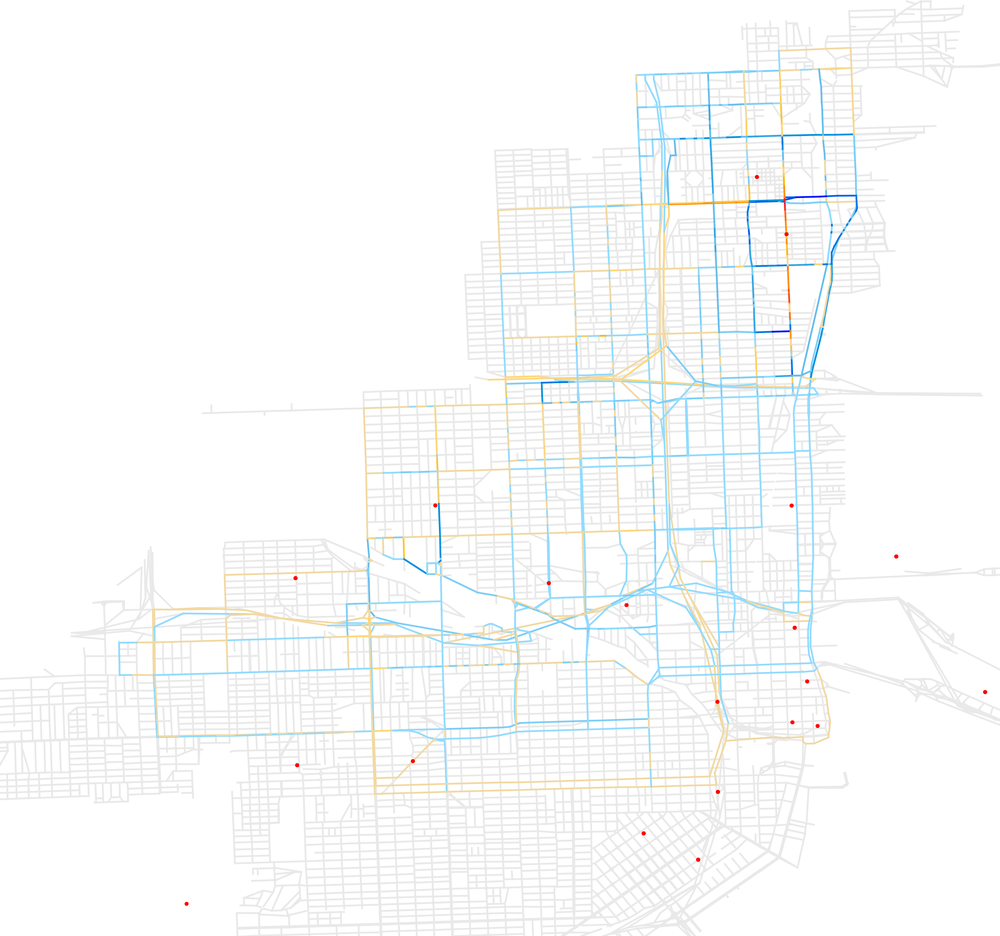}
                            \caption{Miami}
                            \label{fig:voc_miami}
                        \end{subfigure}%
                        \hfill%
                        \begin{subfigure}[b]{0.33\textwidth}
                            \centering
                            \includegraphics[width=\linewidth, keepaspectratio]{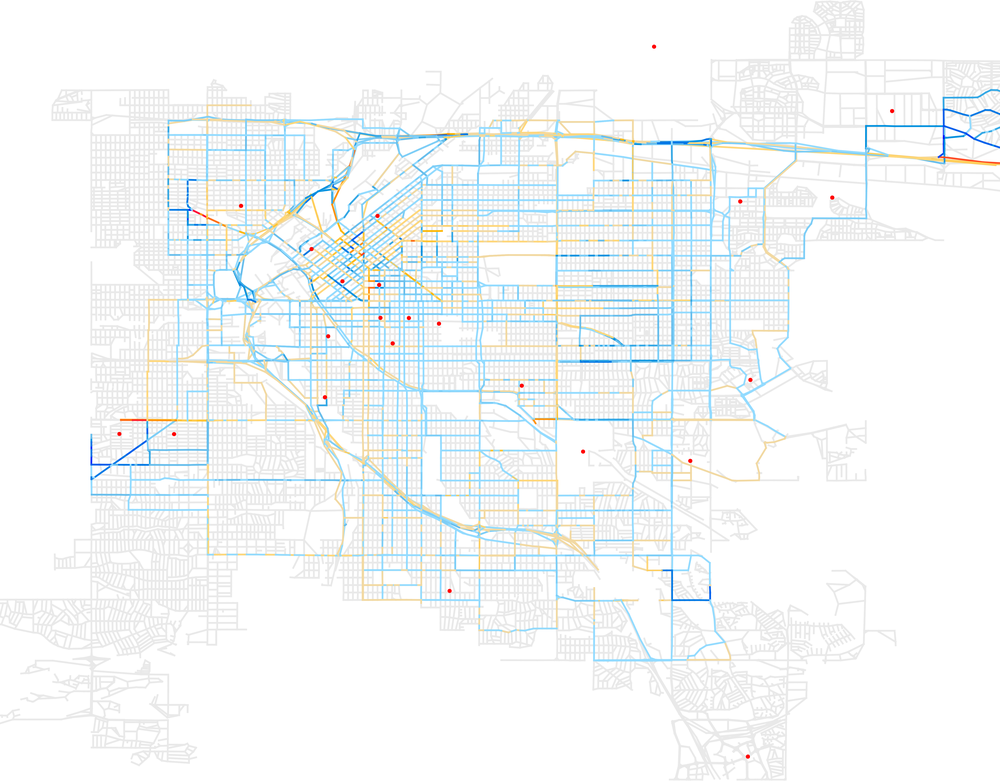}
                            \caption{Denver}
                            \label{fig:voc_denver}
                        \end{subfigure}%
                        \hfill%
                        \begin{subfigure}[b]{0.33\textwidth}
                            \centering
                            \includegraphics[width=\linewidth, keepaspectratio]{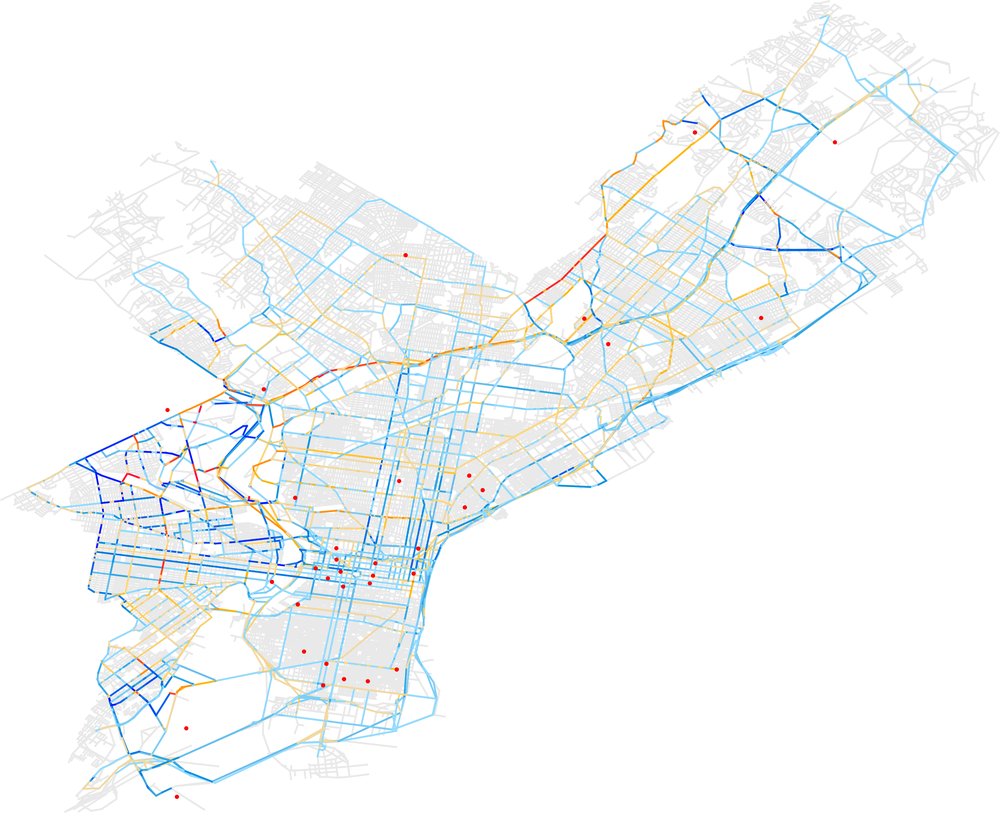}
                            \caption{Philadelphia}
                            \label{fig:voc_philadelphia}
                        \end{subfigure}%
                        \par\vspace{-7.6pt}
                        \begin{subfigure}[b]{0.90\textwidth}
                            \centering
                            \includegraphics[width=\linewidth]{legend1}
                        \end{subfigure}%
                        \vspace{-0.7em}
                        \caption{Maps of VOC difference between Time-Only and Time-Carbon UE ($\psi_{1}=1$, $\psi_{2}=9$) models
                                for the six U.S. cities under the SSP2 scenario in 2050.}
                        \label{fig:2050 ssp2 six cities}
                        \vspace{-1.0em}
                    \end{figure}

                    The spatial distribution of the two sets of specific links observed in Fig.~\ref{fig:VOCs map} exhibit high similarity.
                    In the first set (Fig.~\ref{fig:01 VOCs map b}), under the Time-Carbon UE model,
                    the VOC on these links increases significantly relative to the Time-Only UE model,
                    driving a reduction in carbon emissions.
                    In the second set (Fig.~\ref{fig:01 VOCs map a}), the VOC increases annually,
                    contributing to the mitigation of emissions driven by population growth.
                    Both results indicate that these links are key for reducing emissions.
                    However, under eco-routing, they already exhibit high VOC and will become increasingly congested as population grows.
                    On the one hand, severe congestion results in longer travel time and lower speed,
                    which increases the fuel consumption factor $h(v)$,
                    thereby offsetting part of the emission savings derived from reducing route length.
                    On the other hand, such congestion limits further reduction in travel distance,
                    directly constraining the potential for emission mitigation.
                    Consequently, we refer to the links from the second set as carbon bottlenecks, and
                    future strategies should focus on improving these bottlenecks.

                    Identifying carbon bottlenecks requires setting a significance threshold to distinguish links exhibiting a significant increase in VOC,
                    as well as specifying the year and SSP scenario.
                    Specifically, for (SSP2, 2050) with a threshold of 0.35,
                    we first calculate the difference in VOC between 2040 and 2050 under the Time-Carbon UE model and SSP2 scenario,
                    and then select links where the difference exceeds 0.35 as the bottleneck links.
                    Carbon bottlenecks are characterized by high VOC that will inevitably increase with future population growth,
                    implying that they exert a high and rising carbon influence within the road network.
                    Therefore, as suggested by Proposition~\ref{prop:zipf},
                    strategies to expand the capacity of these links
                    can yield substantial and long-term emission reduction benefits (Remark~\ref{re:dynamic carbon influence}).

            \subsubsection{Efficient Strategy: Improving Carbon Bottlenecks}
            \label{sec:4.4.1}

                In this section, we explore the effectiveness of expanding the capacity of carbon bottlenecks.
                Specifically, for all cities, we identify carbon bottlenecks using a significance threshold of 0.35,
                based on (SSP2, 2050) and Time-Carbon UE model with $\psi_{1}=1$ and $\psi_{2}=9$.
                However, no bottlenecks are found in Pittsburgh under the defined parameters,
                even when the threshold is moderately lowered.
                In the San Francisco network, as shown in Table~\ref{tab:Expanding table}, carbon bottleneck links constitute only 0.46\% of the total links,
                yet expanding this small number of links yields significant dual benefits for travel time and emissions,
                verifying Proposition~\ref{prop:zipf}.
                On the one hand, adding capacity reduces travel time and the fuel consumption factor $h(v)$,
                mitigating the issue where congestion offsets the emission benefits gained from the reduced route length.
                On the other hand, the expanded bottlenecks attract more vehicles to the shorter routes,
                thereby reducing their travel distance and emissions (Remark~\ref{re:redistribution}).
                Specifically, with a 50\% expansion, the reduction in emissions ranges from 1.2\% to 1.9\%,
                while the reduction in UETT spans from 16.5\% to 21.9\%.
                With a 100\% expansion, emission reductions range from 2.2\% to 3.6\%,
                and UETT reductions range from 23.8\% to 31.6\%.
                In stark contrast, the other links account for 99.54\% of the network segments.
                Although expanding this vast number of links by 20\% achieves a slightly larger emission reduction than expanding bottleneck links,
                it falls significantly behind in reducing UETT.
                These results collectively confirm that
                expanding carbon bottlenecks is a long-term efficient strategy (Remark~\ref{re:dynamic carbon influence}).
                \begin{table}[!htb]
                \centering
                \footnotesize
                \begin{threeparttable}
                    \caption{Carbon emissions (t) and UETT (min) under the Time-Carbon UE model ($\psi_{1}=1$, $\psi_{2}=9$)
                        in San Francisco with different expansion strategies.}
                    \label{tab:Expanding table}
                    \setlength{\tabcolsep}{3pt}
                    \begin{tabularx}{0.98\textwidth}{c @{} c *{10}{C}}
                        \toprule
                        \multirow{2}{*}{} & \multirow{2}{*}{Expansion} & baseline & \multicolumn{3}{c}{SSP1} & \multicolumn{3}{c}{SSP2} & \multicolumn{3}{c}{SSP5} \\
                        \addlinespace[-3pt]
                        \cmidrule(lr){3-3} \cmidrule(lr){4-6} \cmidrule(lr){7-9} \cmidrule(lr){10-12}
                        \addlinespace[-3pt]
                         & & 2025 & 2030 & 2040 & 2050 & 2030 & 2040 & 2050 & 2030 & 2040 & 2050 \\
                        \midrule
                        \multirow{3}{*}{Emissions} & 0\%     & 323  & 398 & 450 & 495 & 394 & 439 & 482 & 425 & 501 & 589  \\[3pt]
                                                   & 20\%    & 305  & 377 & 428 & 472 & 374 & 418 & 459 & 404 & 477 & 563  \\
                                                   & (99.54\%) & 5.4\%  & 5.1\% & 4.9\% & 4.7\% & 5.1\% & 4.9\% & 4.7\% & 5.0\% & 4.7\% & 4.5\%  \\[3pt]
                                                   & 50\%    & 316  & 391 & 443 & 489 & 388 & 432 & 475 & 418 & 494 & 582  \\
                                                   & (0.46\%)  & 1.9\%  & 1.6\% & 1.5\% & 1.4\% & 1.6\% & 1.5\% & 1.4\% & 1.6\% & 1.4\% & 1.2\%  \\[3pt]
                                                   & 100\%   & 311  & 385 & 437 & 482 & 382 & 426 & 469 & 412 & 487 & 576  \\
                                                   & (0.46\%)  & 3.6\%  & 3.1\% & 2.9\% & 2.6\% & 3.1\% & 2.9\% & 2.7\% & 3.0\% & 2.6\% & 2.2\%  \\[4pt]
                        \multirow{3}{*}{UETT}      & 0\%     & 26.5  & 32.2 & 35.9 & 39.4 & 32.0 & 35.3 & 38.4 & 34.2 & 39.8 & 46.4 \\[3pt]
                                                   & 20\%    & 23.7  & 28.8 & 32.3 & 35.6 & 28.5 & 31.6 & 34.6 & 30.7 & 35.9 & 41.9 \\
                                                   & (99.54\%) & 10.7\%  & 10.6\% & 10.2\% & 9.8\% & 10.8\% & 10.2\% & 9.9\% & 10.4\% & 9.8\% & 9.7\%  \\[3pt]
                                                   & 50\%    & 22.1  & 26.3 & 29.0 & 31.5 & 26.1 & 28.5 & 30.8 & 27.7 & 31.7 & 36.3 \\
                                                   & (0.46\%)  & 16.5\%  & 18.4\% & 19.3\% & 20.2\% & 18.4\% & 19.3\% & 19.9\% & 19.0\% & 20.4\% & 21.9\%  \\[3pt]
                                                   & 100\%   & 20.2  & 23.6 & 25.8 & 27.8 & 23.5 & 25.4 & 27.3 & 24.8 & 28.0 & 31.7 \\
                                                   & (0.46\%)  & 23.8\%  & 26.8\% & 28.2\% & 29.4\% & 26.7\% & 28.0\% & 29.1\% & 27.7\% & 29.5\% & 31.6\%  \\
                        \bottomrule
                    \end{tabularx}
                    \begin{tablenotes}[flushleft]
                        \footnotesize
                        \item where (0.46\%) refers to expanding carbon bottlenecks;
                            (99.54\%) refers to expanding links in the road network excluding carbon bottlenecks.
                    \end{tablenotes}
                \end{threeparttable}
                \vspace{-1.0em}
                \end{table}

                In Fig.~\ref{fig:expand1.b}, the slopes of carbon emissions curves are virtually identical across all expansion percentages,
                indicating that bounded capacity expansion fails to alter the scaling order of emissions with population,
                aligning with Corollary~\ref{cor:finite_expansion}.
                Additionally, as shown in Figs.~\ref{fig:expand1.c} and~\ref{fig:expand2},
                for all cities,
                both OP and CP trends remain nearly unchanged after expansion.
                Consequently, the strategy of expanding carbon bottlenecks does not alter the optimization level.
                Managers can implement this strategy by converting on-street parking into traffic lanes,
                as these bottlenecks typically feature such parking.
                Moreover, given the lower traffic demand during off-peak periods,
                a dynamic curbside management approach can be adopted to flexibly reallocate curb space
                for either parking or traffic lanes based on temporal traffic needs.
                Regarding road capacity expansion subject to a budget constraint, Proposition~\ref{prop:shadow_price} implies that
                the optimal scheme prioritizes allocating the budget to bottlenecks with the highest carbon influence per unit cost.
                Upon budget exhaustion, the carbon influence per unit cost equalizes across all expanded links and corresponds to the shadow price.
                In summary, the strategy of expanding carbon bottlenecks is practically feasible,
                effectively reducing both emissions and travel time without compromising the efficiency of eco-routing.

                \begin{figure}[!htb]
                    \captionsetup[subfigure]{skip=1pt}
                    \centering
                    \hspace{0.12\textwidth}
                    \begin{subfigure}[b]{0.33\textwidth}
                        \centering
                        \includegraphics[width=\linewidth, keepaspectratio]{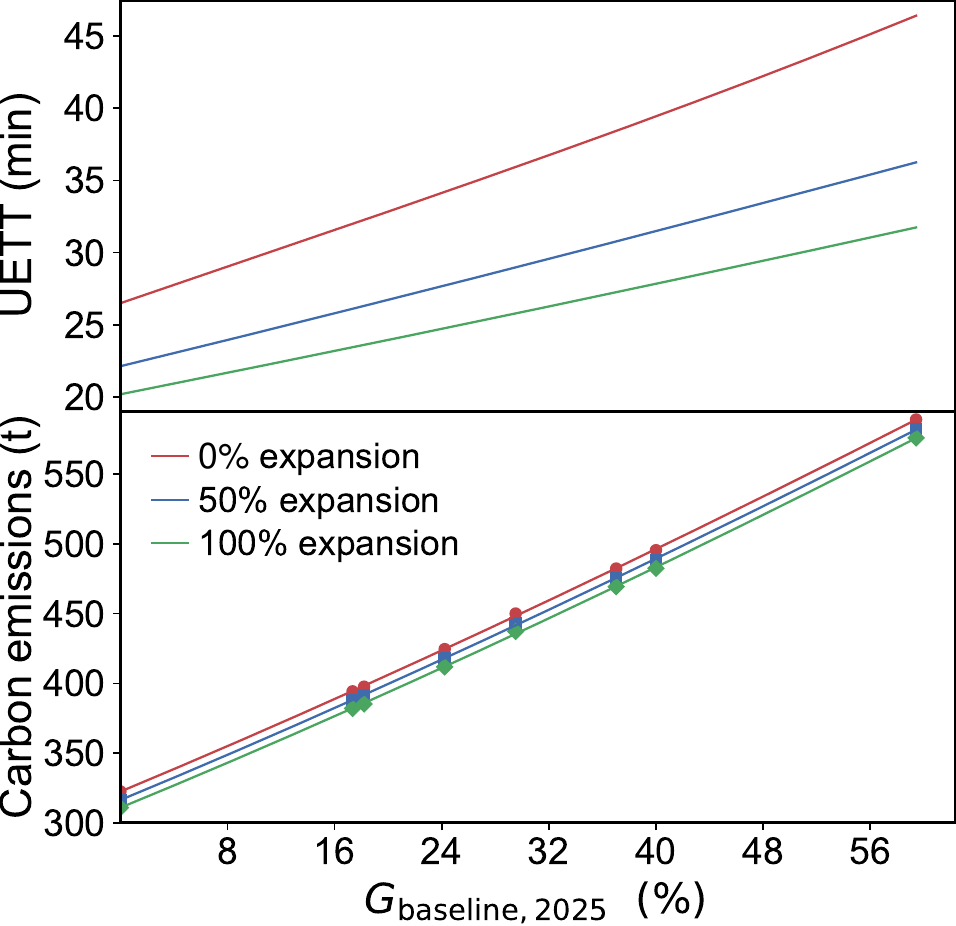}
                        \caption{}
                        \label{fig:expand1.b}
                    \end{subfigure}%
                    \hfill%
                    \begin{subfigure}[b]{0.33\textwidth}
                        \centering
                        \includegraphics[width=\linewidth, keepaspectratio]{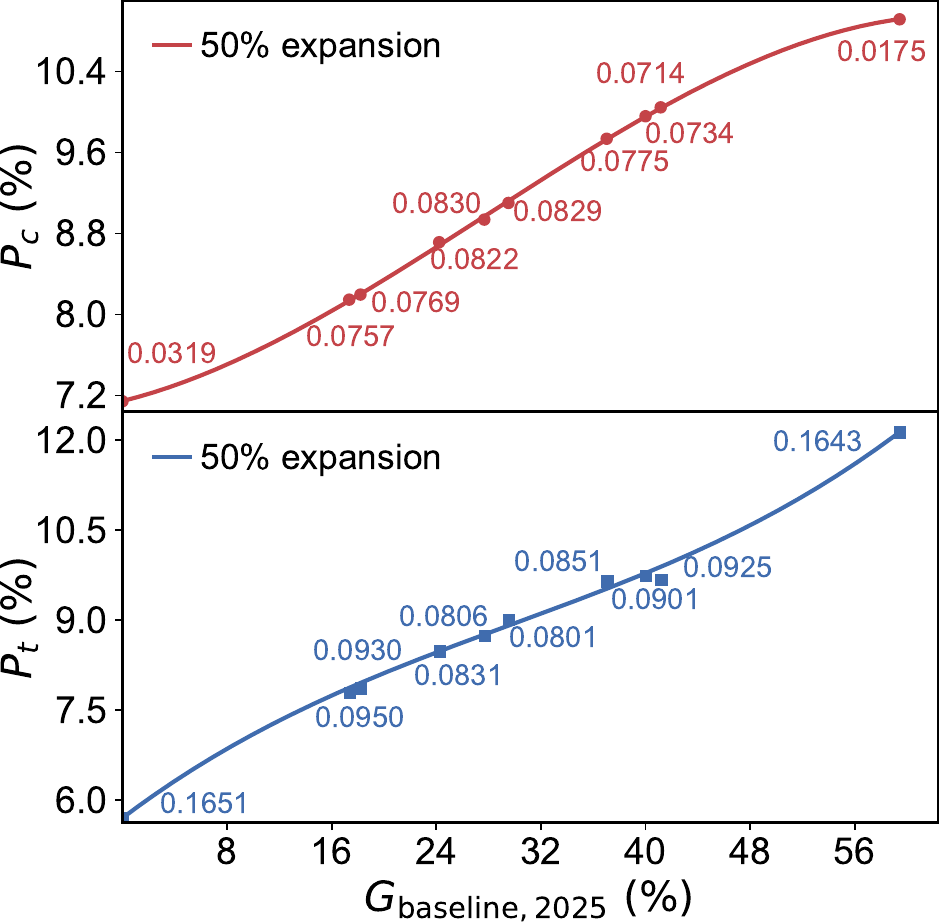}
                        \caption{}
                        \label{fig:expand1.c}
                    \end{subfigure}%
                    \hspace{0.12\textwidth}
                    \vspace{-0.7em}
                    \caption{Impact of expanding carbon bottlenecks in San Francisco.
                            (a) Carbon emissions (t) and UETT (min) with 0\%, 50\%, and 100\% expansion.
                            (b) Optimization levels with 50\% expansion.}
                    \label{fig:expand1}
                    \vspace{-1.0em}
                \end{figure}

                \begin{figure}[!htb]
                    \captionsetup[subfigure]{skip=1pt}
                    \centering
                    \begin{subfigure}[b]{0.33\textwidth}
                        \centering
                        \includegraphics[width=\linewidth, keepaspectratio]{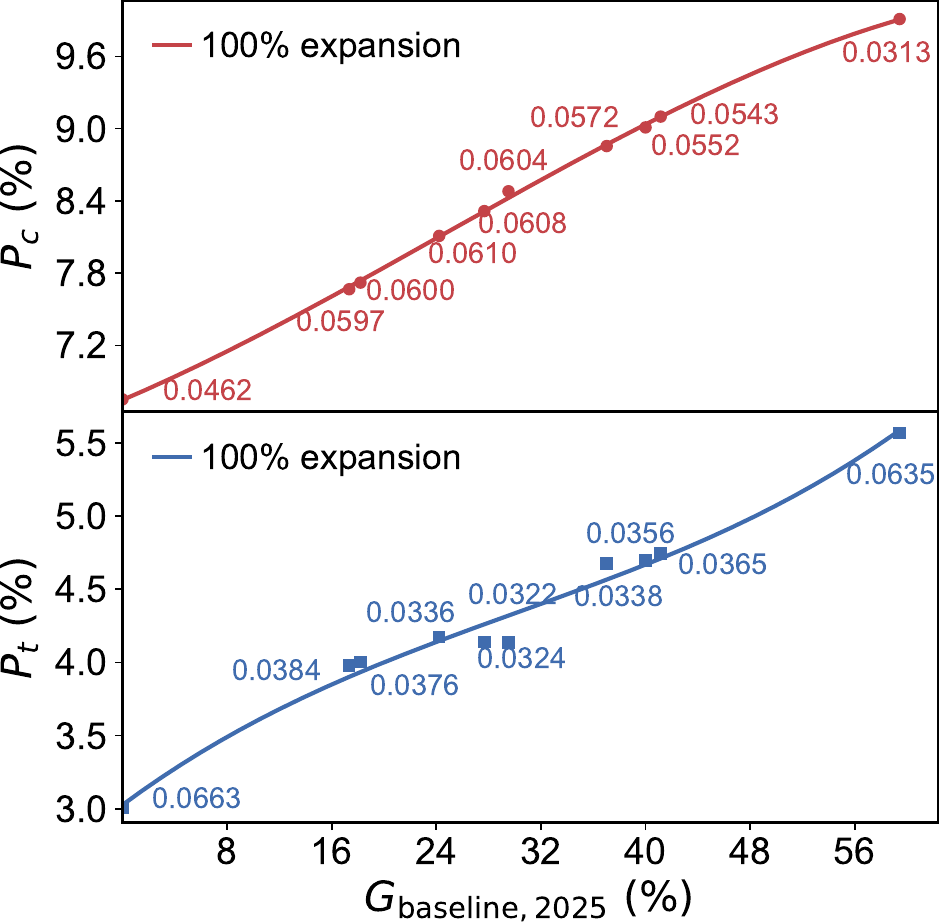}
                        \caption{San Francisco}
                        \label{fig:expand_sf}
                    \end{subfigure}%
                    \hspace{0.05\textwidth}%
                    \begin{subfigure}[b]{0.33\textwidth}
                        \centering
                        \includegraphics[width=\linewidth, keepaspectratio]{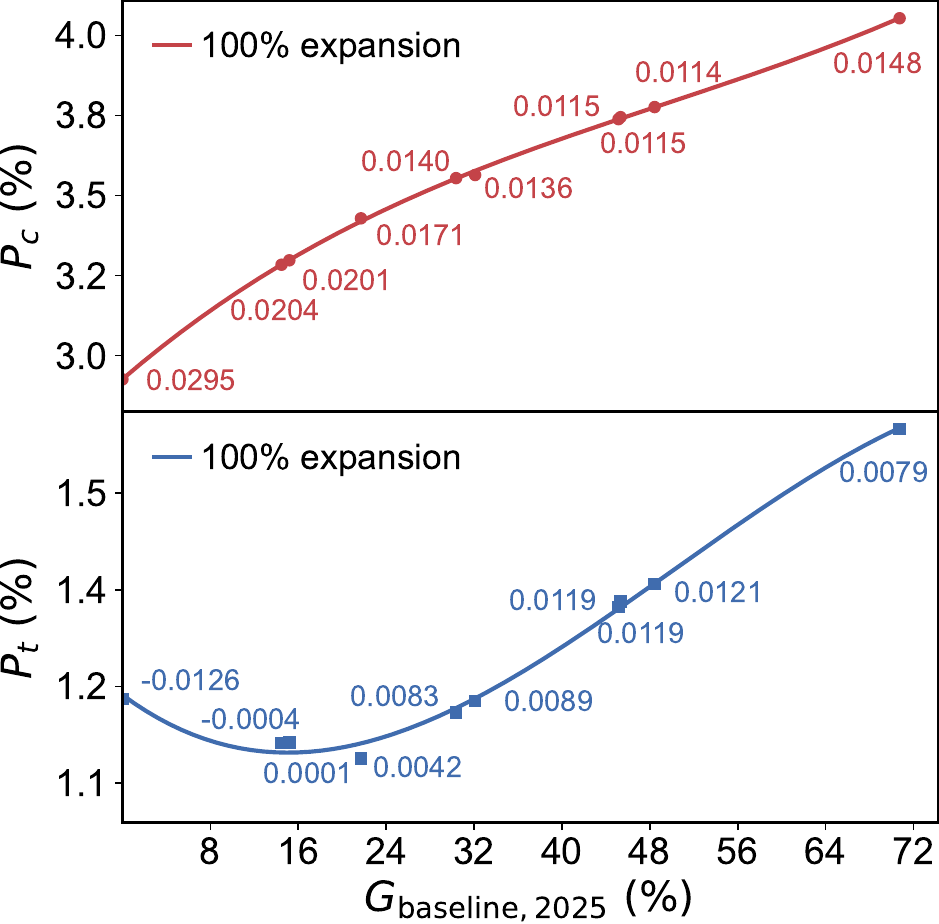}
                        \caption{Dallas}
                        \label{fig:expand_dallas}
                    \end{subfigure}%
                    \\[-1.5pt]
                    \begin{subfigure}[b]{0.33\textwidth}
                        \centering
                        \includegraphics[width=\linewidth, keepaspectratio]{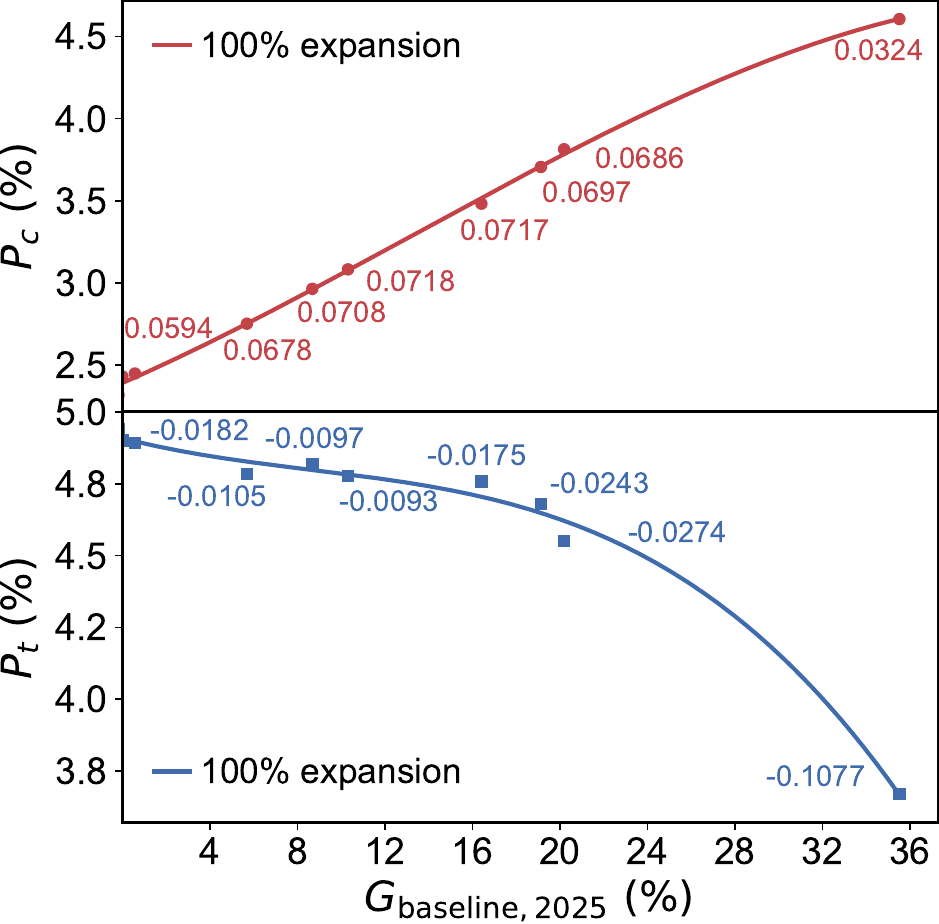}
                        \caption{Miami}
                        \label{fig:expand_miami}
                    \end{subfigure}%
                    \hfill%
                    \begin{subfigure}[b]{0.33\textwidth}
                        \centering
                        \includegraphics[width=\linewidth, keepaspectratio]{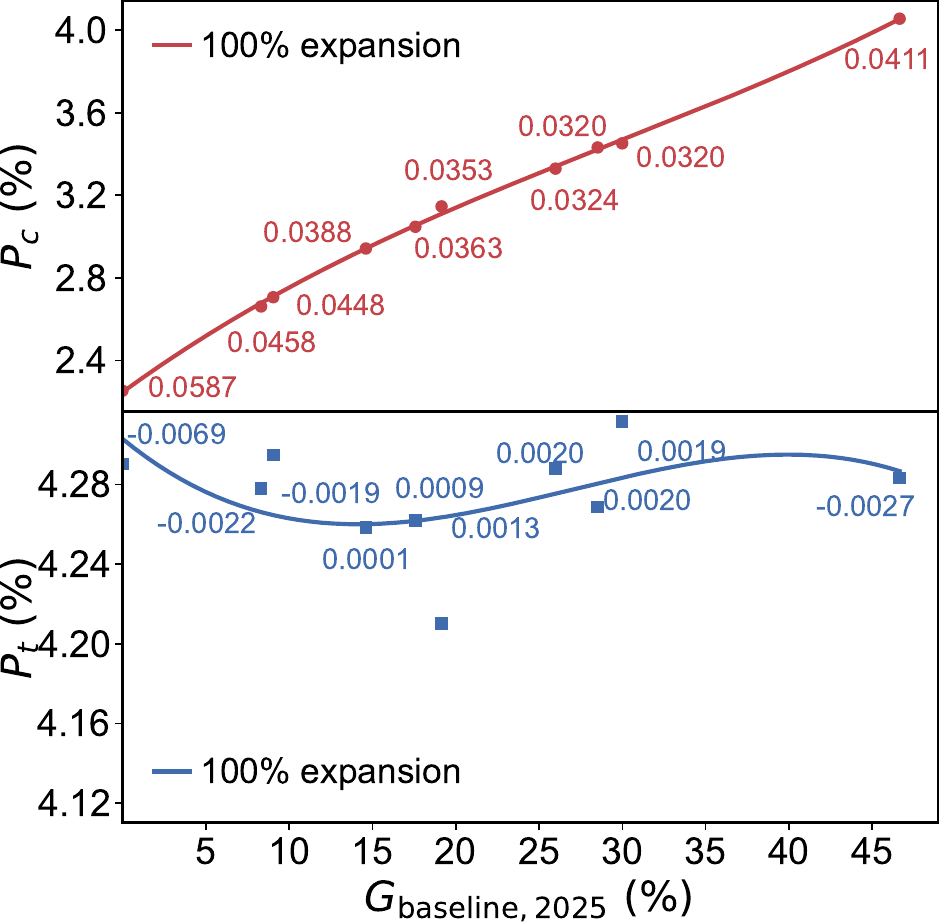}
                        \caption{Denver}
                        \label{fig:expand_denver}
                    \end{subfigure}%
                    \hfill%
                    \begin{subfigure}[b]{0.33\textwidth}
                        \centering
                        \includegraphics[width=\linewidth, keepaspectratio]{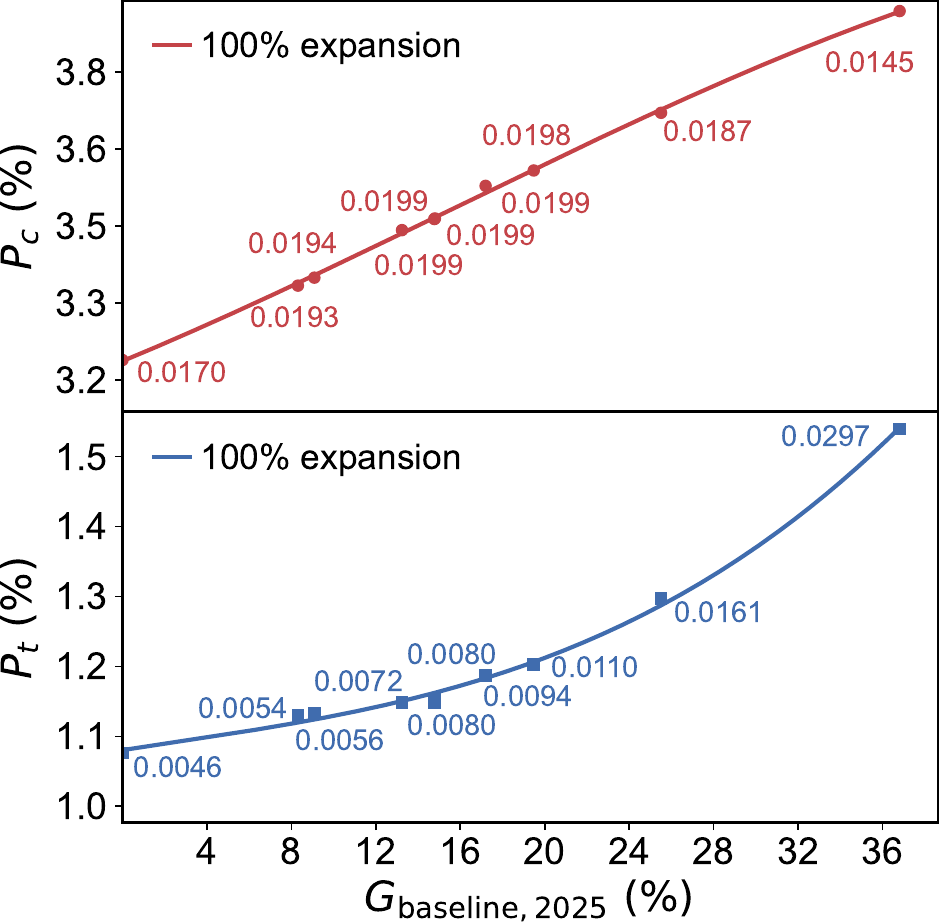}
                        \caption{Philadelphia}
                        \label{fig:expand_philadelphia}
                    \end{subfigure}%
                    \vspace{-0.7em}
                    \caption{Optimization levels for five cities with 100\% expansion.}
                    \label{fig:expand2}
                    \vspace{-1.0em}
                \end{figure}


            \subsubsection{Induced Demand}
            \label{sec:4.4.2}
                Although the strategy of expanding carbon bottlenecks is proven to be efficient,
                it inevitably leads to the issue of induced demand, which may compromise the optimization level.
                This induced demand primarily manifests as increased demand on existing OD pairs,
                which can be decomposed into variations in $\pi_{y}^{s, t}$ and population growth.
                As established in Sections~\ref{sec:4.1} and~\ref{sec:4.3.1},
                population growth acts as the dominant driver of changes in urban metrics.
                Therefore, it is reasonable to infer that the impact of induced demand on the optimization level is functionally equivalent to
                an additional increase in population, preserving the trend of optimization level.

                To validate this inference, we first calculate the changes in travel cost for each OD pair before and after expansion
                under both Time-Only and Time-Carbon UE models.
                Subsequently, we identify the top 10\% of OD pairs exhibiting the largest decrease in travel cost
                and increase their demand volume by 10\% and 20\% to simulate the induced demand.
                As shown in Fig.~\ref{fig:induce demand},
                the shapes of all corresponding curves remain similar before and after introducing the induced demand;
                specifically, the curves with 10\% and 20\% increase manifest as leftward shifts relative to the curve with 0\% increase.
                Regarding the trend of optimization level,
                this result implies that the impact of induced demand is indeed functionally equivalent to additional population growth,
                thereby corroborating the inference.
                Moreover, given that the dominance of population growth is consistent across cities,
                this finding is generalizable to diverse road networks.
                Consequently, the strategy of expanding carbon bottlenecks does not alter the trend of cities' optimization level,
                even when accounting for potential induced demand.
                \begin{figure}[!htb]
                    \captionsetup[subfigure]{skip=1pt}
                    \centering
                    \hspace{0.12\textwidth}
                    \begin{subfigure}[b]{0.33\textwidth}
                        \centering
                        \includegraphics[width=\linewidth, keepaspectratio]{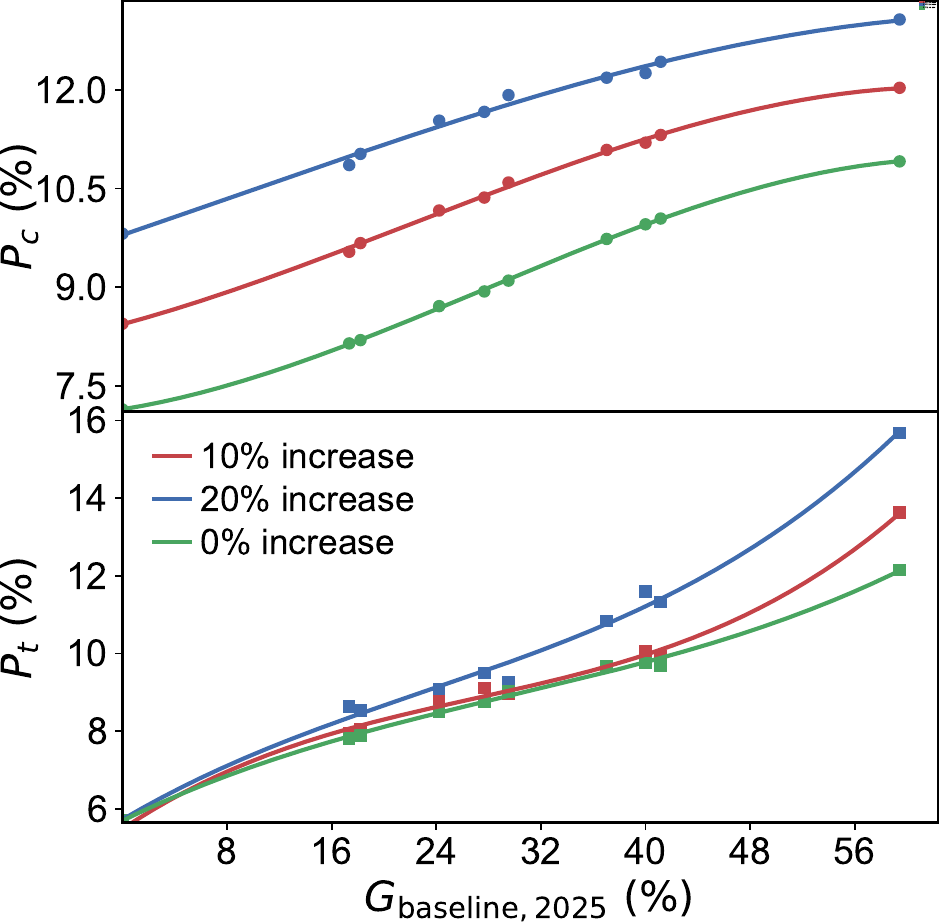}
                        \caption{50\% expansion}
                        \label{50_induced}
                    \end{subfigure}%
                    \hfill%
                    \begin{subfigure}[b]{0.33\textwidth}
                        \centering
                        \includegraphics[width=\linewidth, keepaspectratio]{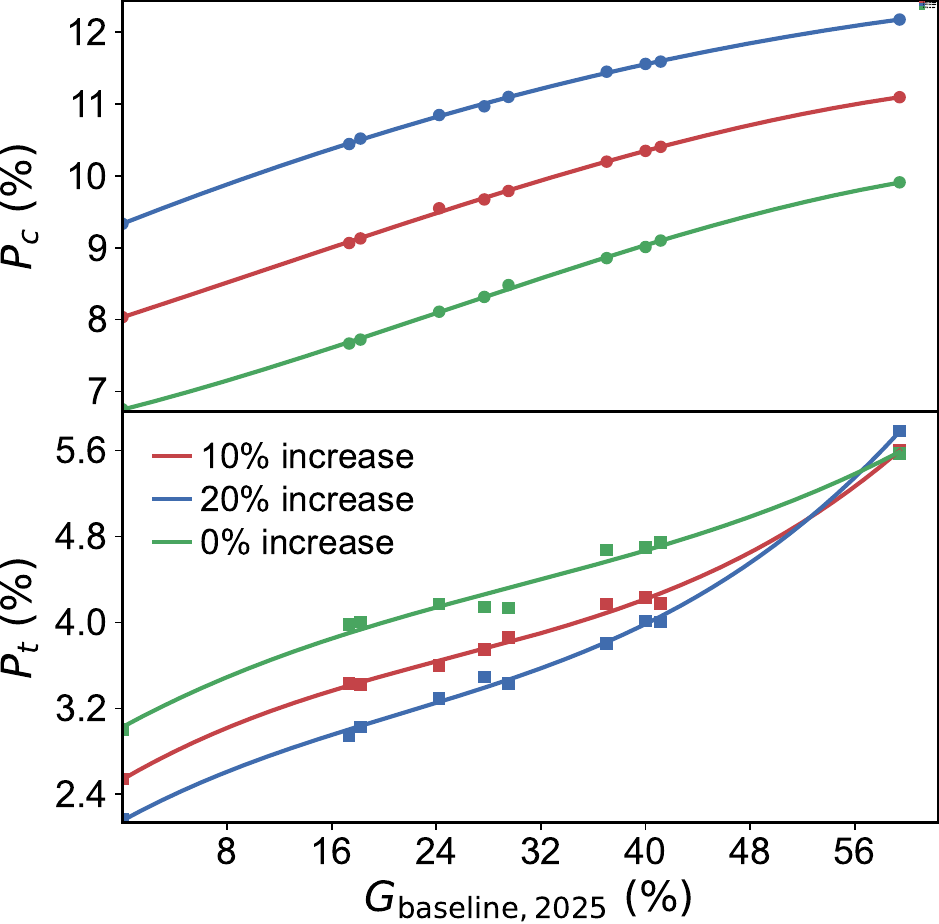}
                        \caption{100\% expansion}
                        \label{100_induced}
                    \end{subfigure}%
                    \hspace{0.12\textwidth}%
                    \vspace{-0.7em}
                    \caption{Optimization levels for San Francisco with 50\% and 100\% expansion, specifically accounting for the induced demand.}
                    \label{fig:induce demand}
                    \vspace{-1.0em}
                \end{figure}

    \section{Conclusion}
    \label{sec:5}
        This study develops an analytical framework that integrates the demand forecasting method
        with the Time-Only and Time-Carbon UE models
        to simulate future network traffic under different routing strategies.
        We apply the bi-conjugate Frank-Wolfe algorithm to these models and
        establish the existence, uniqueness, and equivalence of their optimal solutions.
        Using real-world road networks, commuting OD demand, and population projections under various SSPs for the six representative U.S. cities,
        we reveal that population growth is the dominant driver of urban dynamics and eco-routing efficiency:
        (i) While eco-routing mitigates urban transport emissions,
        emissions in most cities scale superlinearly with population regardless of routing and construction strategies,
        as the scaling order is an intrinsic property of road networks.
        (ii) The efficiency of eco-routing exhibits distinct patterns:
        in some cities, eco-routing is already inefficient;
        in others, while currently efficient, it is projected to become inefficient under population growth;
        and in only a few cities, it retains long-term high efficiency, exclusively under SSP1 or SSP2 scenarios.
        (iii) Under population growth,
        travelers using eco-routing increasingly select shorter routes despite significant decreases in travel speed, thereby giving rise to carbon bottlenecks.
        This indicates that the carbon reduction benefits from reducing travel distance outweigh those from maintaining travel speed.
        Furthermore, the strategy of expanding the small number of bottlenecks (0.46\%) is efficient and practically feasible,
        reducing emissions (3\%) and travel time (28\%) without compromising the efficiency of eco-routing,
        even when accounting for potential induced demand.


        Based on the findings, this study offers policy recommendations for emission mitigation, infrastructure development, and low-carbon urban transport planning.
        Specifically:
        (i) Implement eco-routing and improve carbon bottlenecks.
        Both strategies have proven effective in emissions mitigation.
        Imposing a carbon tax on fuel serves as a practical implementation of city-scale eco-routing,
        necessitating an appropriate tax rate to ensure the stable evolution of the lowest carbon routes under population growth.
        Concurrently, a dynamic curbside management approach can be applied to the carbon bottlenecks
        to reallocate curb space for either parking or lanes based on temporal traffic needs.
        (ii) Optimize urban spatial functionality.
        Given the critical role of reducing travel distance in emission mitigation,
        policies should prioritize improving the jobs-housing balance and promoting polycentric development
        to counteract urban sprawl and shorten the average trip length.
        (iii) Manage gasoline car travel demand.
        As the scaling order of emissions with population remains invariant to routing and construction strategies,
        the fundamental approach to mitigate emissions lies in managing the growing demand for gasoline car travel,
        including expanding public transit, incentivizing carpooling, and following a sustainable SSP characterized by slower population growth.
        Furthermore, introducing electric vehicles is expected to significantly reduce emissions.
        They displace gasoline vehicles, generating zero direct emissions,
        and tend to bypass road segments heavily utilized by gasoline vehicles,
        thereby yielding synergistic emission reduction benefits.


        A limitation of this study is that the travel model only refers to cars,
        and future research needs to conduct a more comprehensive analysis by
        incorporating carpooling, public transit, and other factors.
        Besides, the UE traffic assignment faces challenges under substantial population growth,
        as it may underestimate the severity of congestion.
        Moreover, the strategy of expanding carbon bottlenecks leads to induced demand;
        although we simulated this impact, unforeseen effects may arise in the long run due to complex demographic dynamics
        such as residential relocation and the emergence of new OD pairs; nevertheless, in the short run, this strategy remains efficient.
        In addition, the OD demand data utilized in this study reflects the peak-hour traffic conditions,
        which accounts for the majority of daily urban transport emissions.
        However, conclusions may vary slightly during off-peak periods;
        for instance, the spatial distribution of carbon bottlenecks may shift during other temporal segments.
        With the availability of data for additional time periods,
        further research could extend to different intervals.

     \section*{Declaration of Interest}
         The authors declare no conflict of interest.
    \section*{Data Availability}
        The processed data and implementation code will be made available upon publication.
        The Frank-Wolfe algorithm converges to a unique equilibrium state due to strict convexity (Section~\ref{sec:3.3.2}),
        ensuring reproducibility regardless of initialization.

     \section*{Author Contributions}
         Zhiheng Shi: Conceptualization, Data Curation, Formal Analysis, Methodology, Validation, Visualization, Writing - Original Draft, Writing - review \& editing;
         Xiaohan Xu: Data Curation, Formal Analysis, Resources, Supervision, Writing - Original Draft, Writing - review \& editing;
         Wei Ma: Methodology, Resources, Supervision, Writing - Original Draft, Writing - review \& editing;
         Kairui Feng: Methodology, Resources, Supervision, Writing - Original Draft, Writing - review \& editing;
         Bin He: Supervision, Writing - Original Draft.

     \section*{ Acknowledgement }
         Kairui Feng is supported by Shanghai Municipal Science and Technology Commission Explorers Program (24TS1401600) and Xiaomi Foundation.
         Kairui Feng and Bin He are supported by the National Natural Science Foundation of China (Grant No. 62088101).
         This support is gratefully acknowledged.



    \appendix
    \section{Additional Maps}
        \setcounter{figure}{0}
        Additional maps under the Time-Carbon UE model ($\psi_1=1$ and $\psi_2=9$) are shown in Fig.~\ref{fig:more VOCs maps}.
        \vspace{-0.5em}
        \newpage
        \begin{figure}[H]
            \captionsetup[subfigure]{skip=1pt}
            \centering
            \begin{subfigure}{1.0\linewidth}
                \setlength{\tabcolsep}{0pt}
                \centering
                \begin{tabular}{p{0.24\linewidth} p{0.24\linewidth} p{0.24\linewidth} p{0.24\linewidth}}
                    \myimage{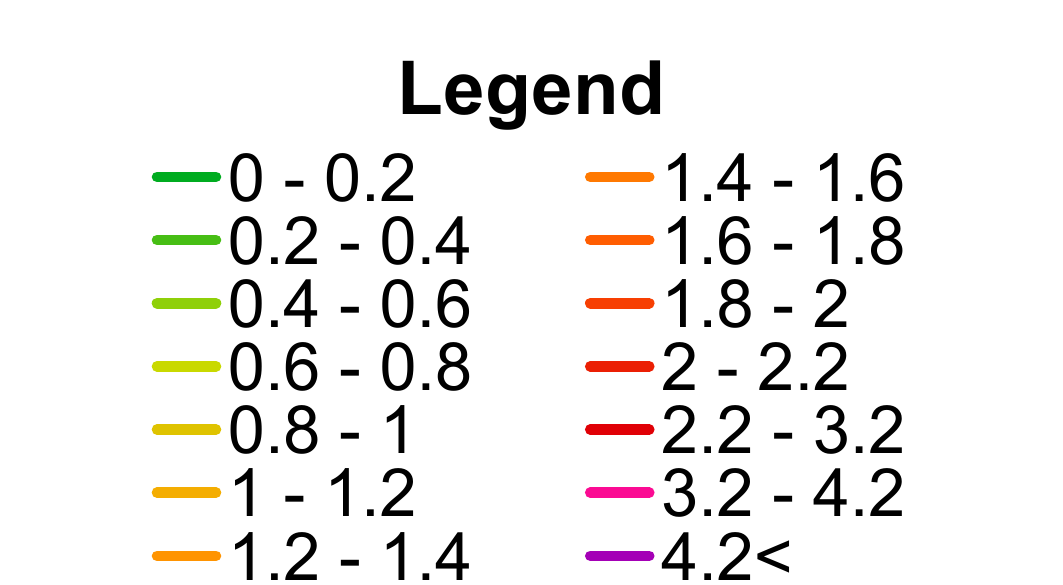}    & \myimage{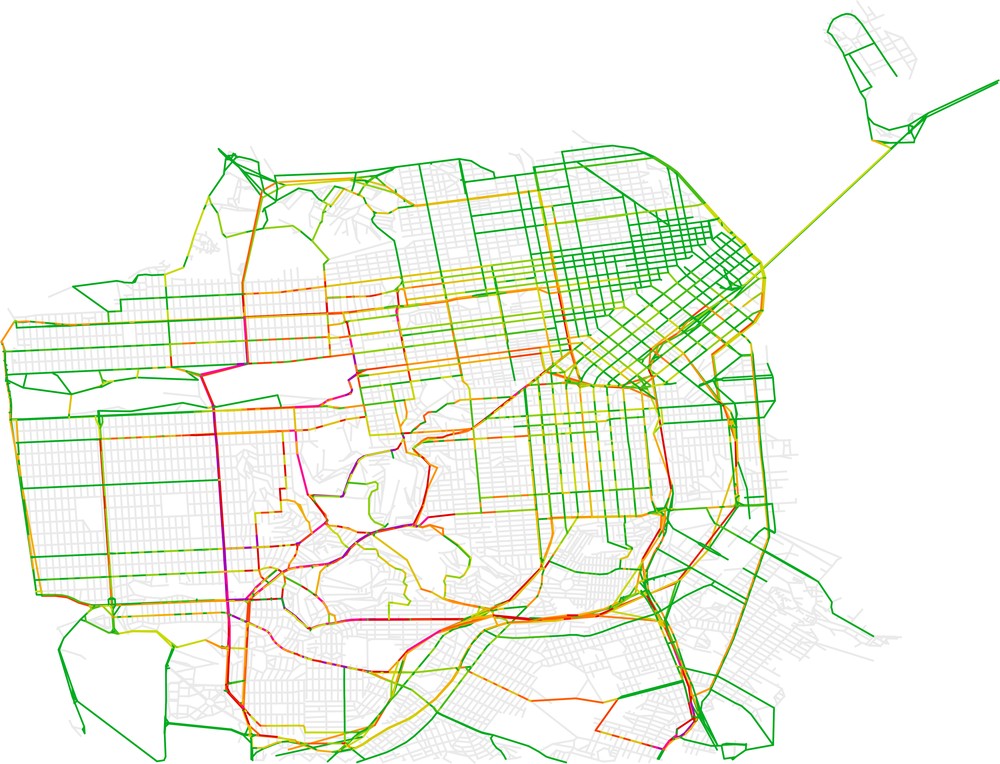} & \myimage{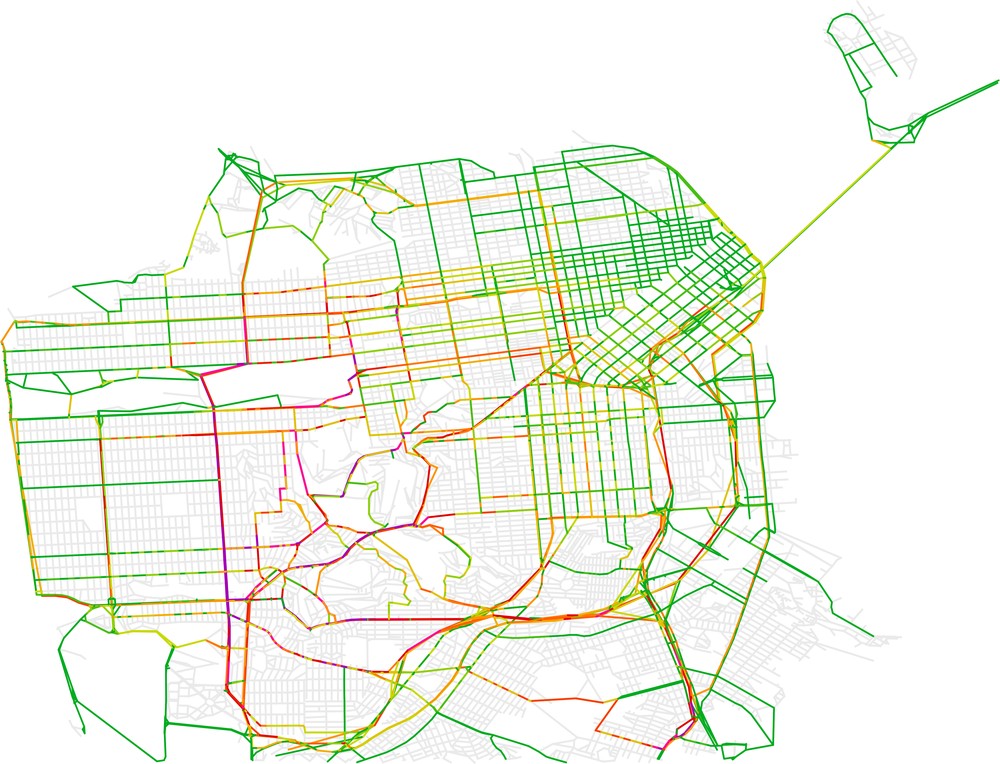} & \myimage{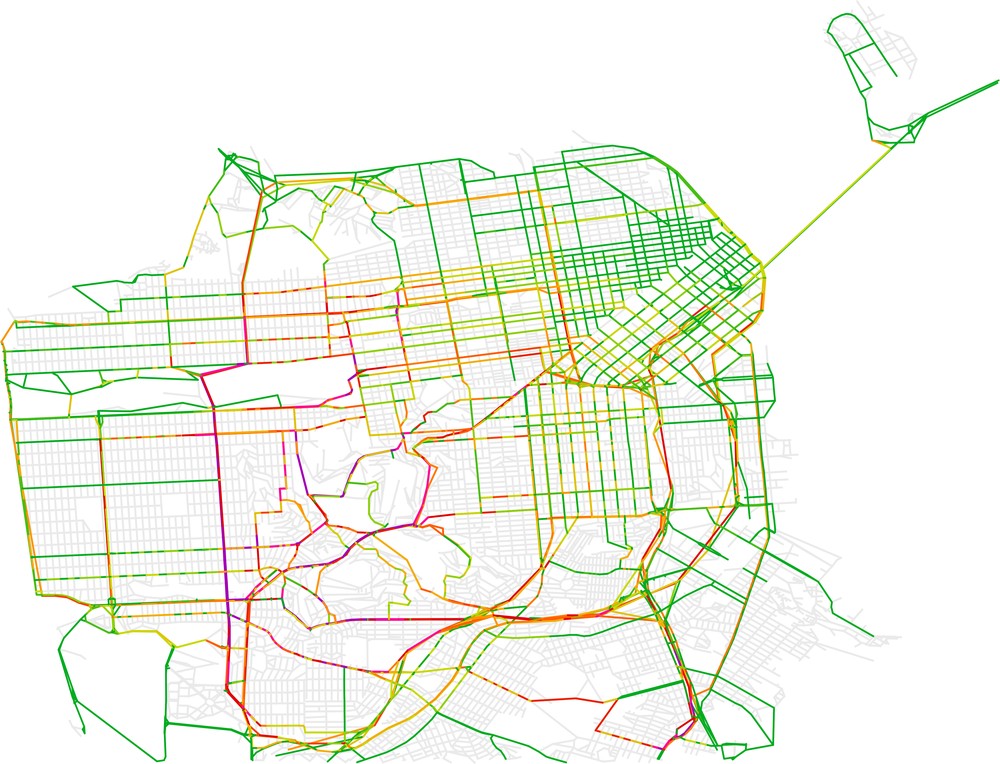} \\
                    \myimage{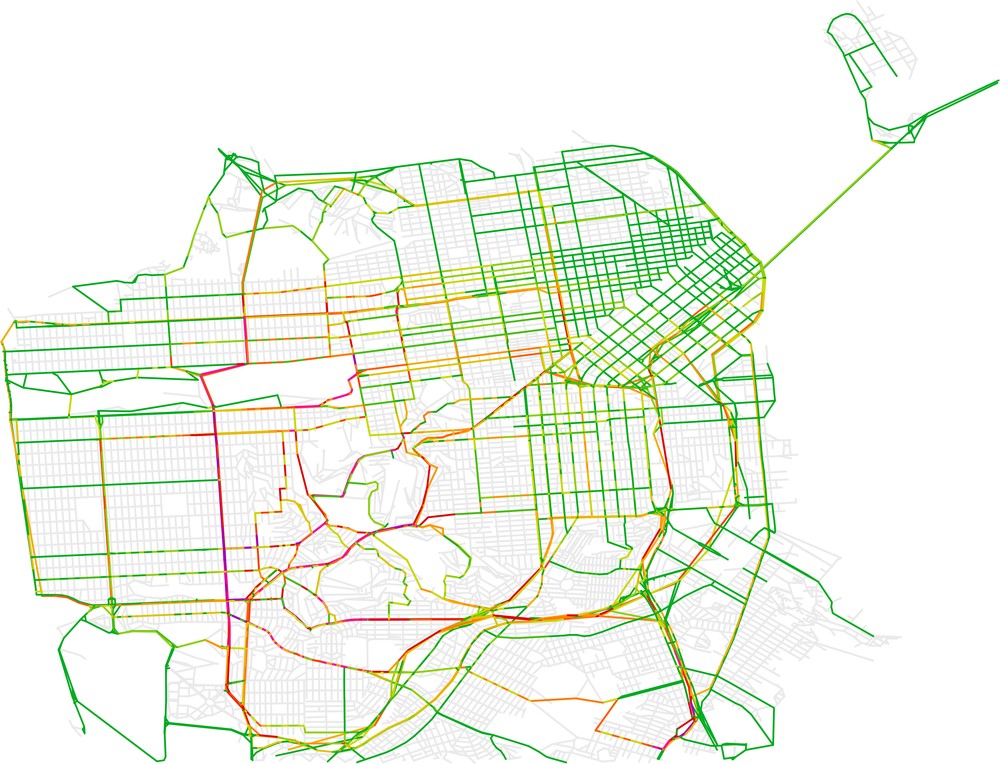}  & \myimage{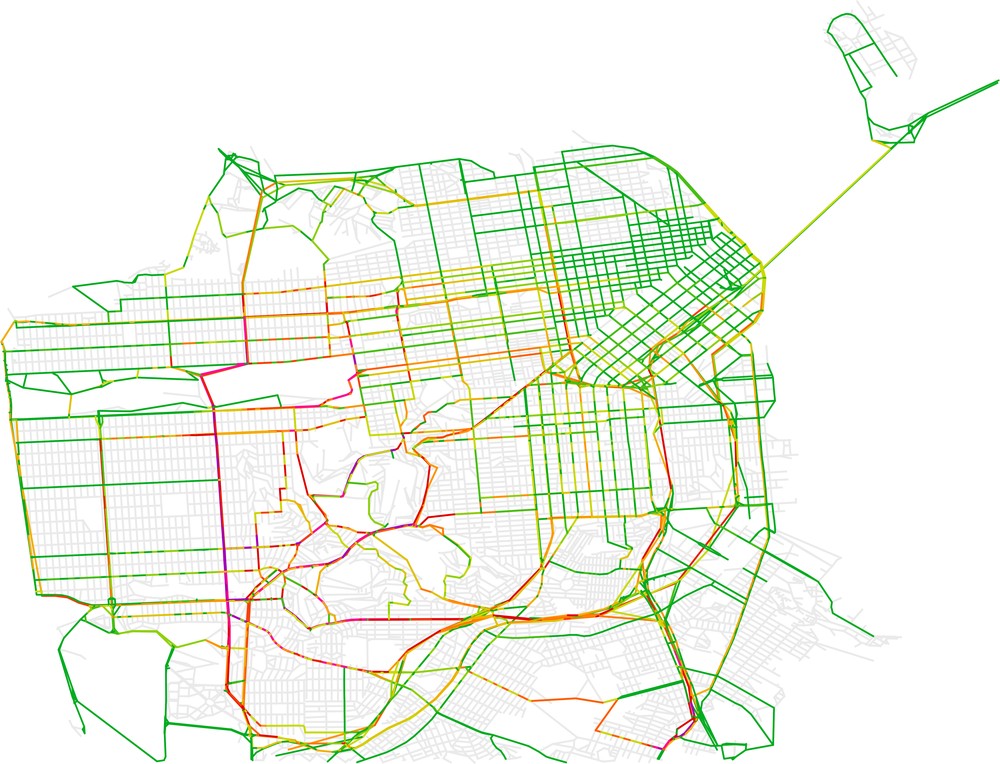} & \myimage{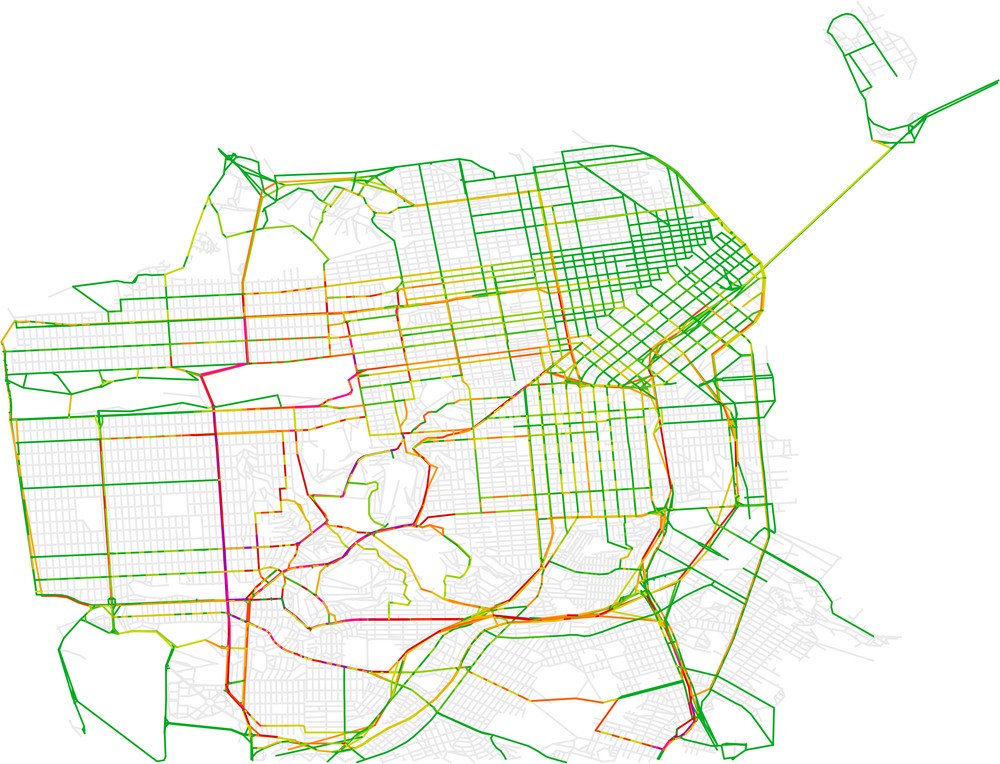} & \myimage{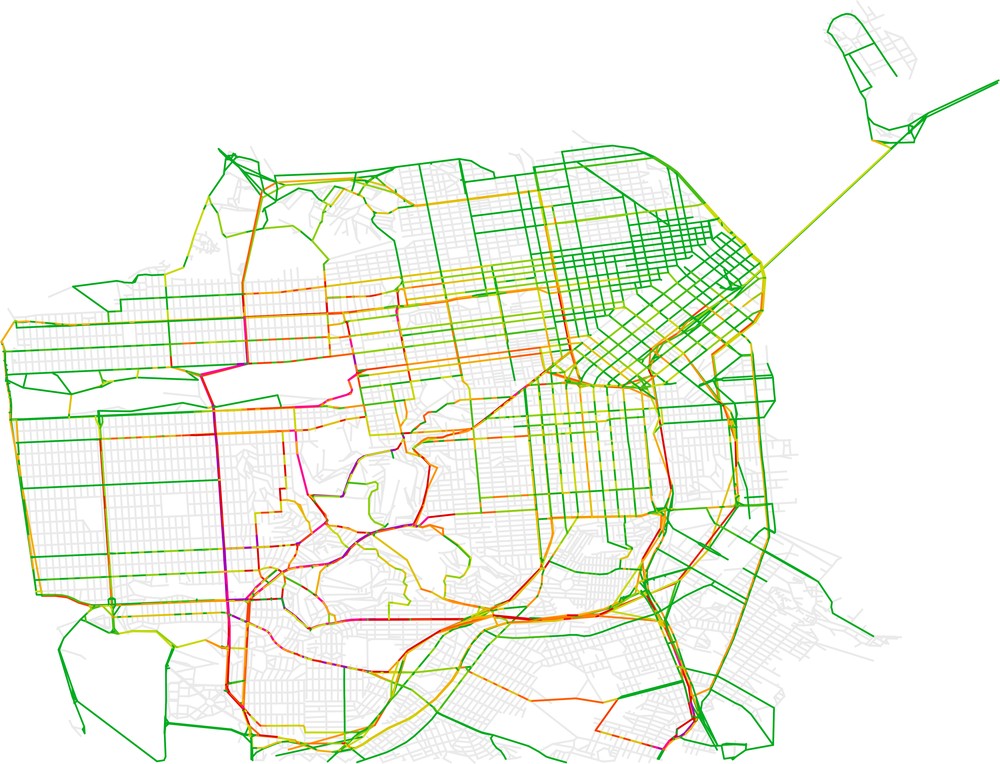} \\
                    \myimage{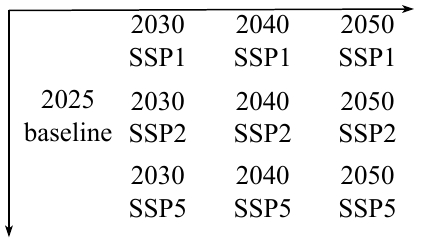}              & \myimage{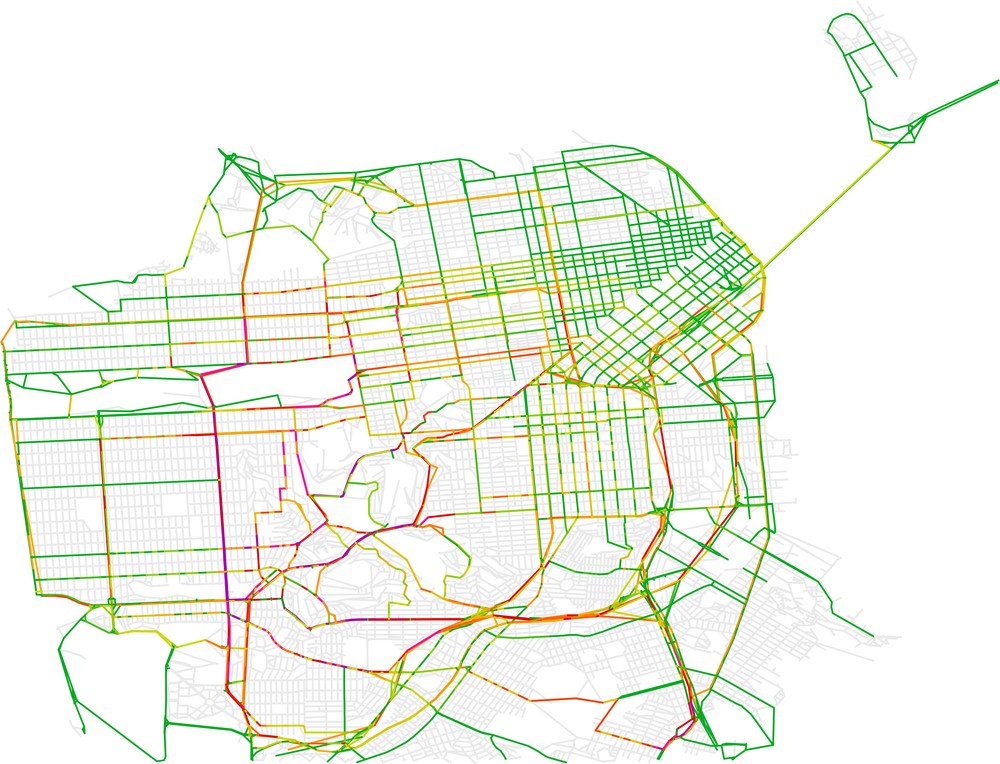} & \myimage{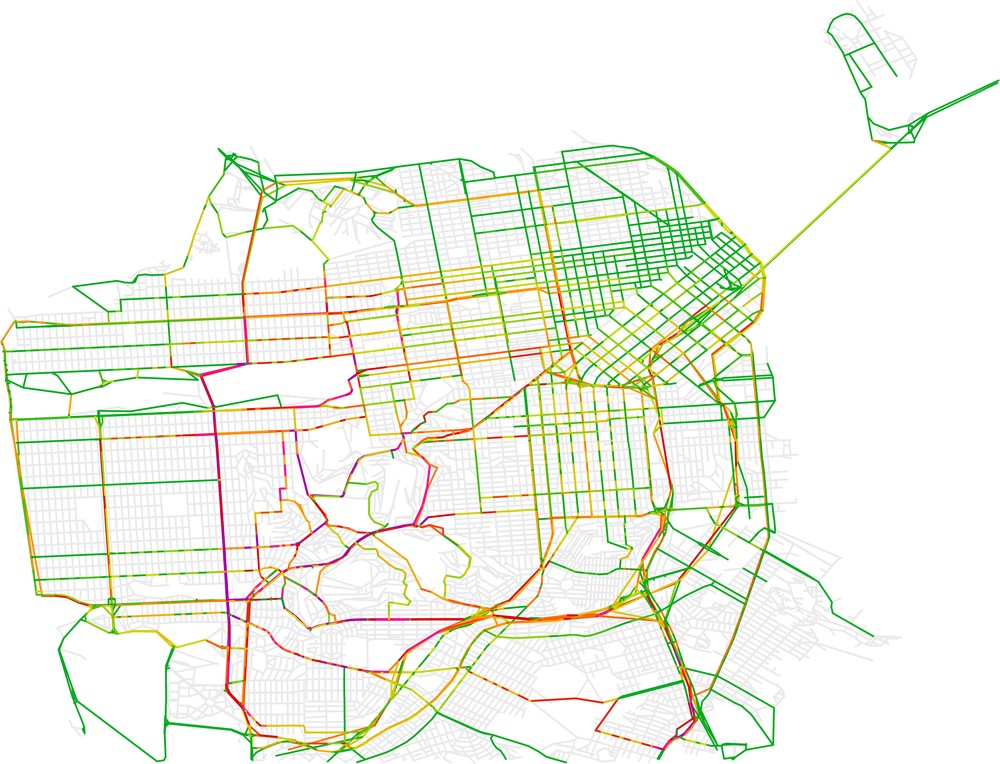} & \myimage{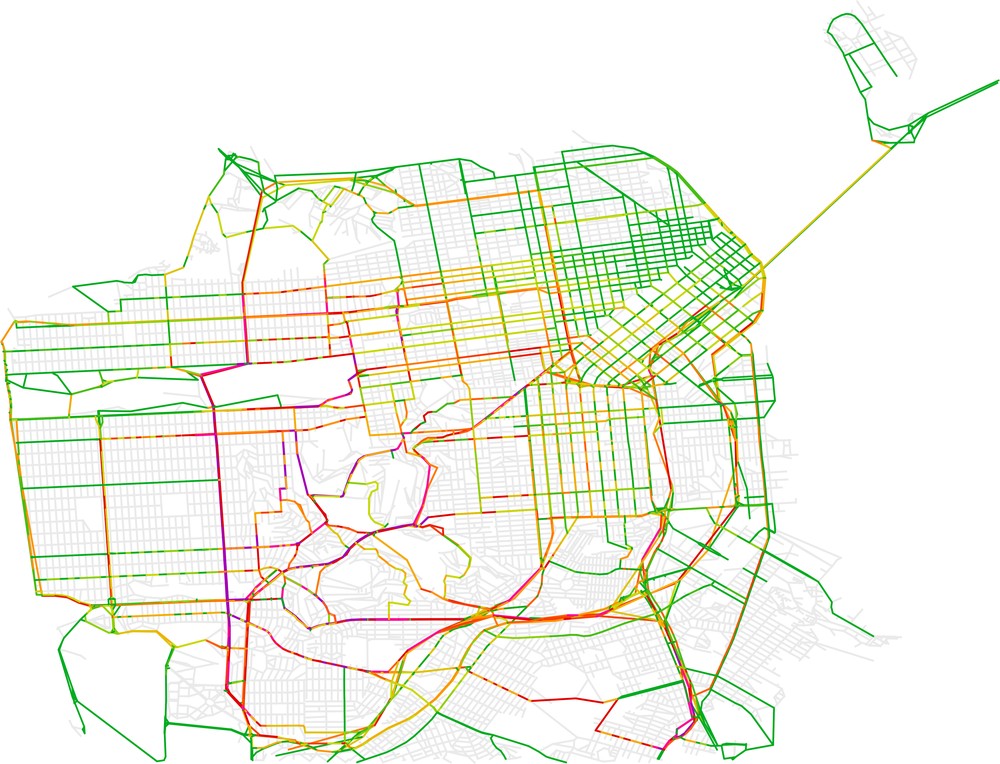} \\
                \end{tabular}
                \caption{San Francisco}
                \label{fig:more_sf}
            \end{subfigure}%
            \vspace{0.1em}
            \begin{subfigure}{1.0\linewidth}
                \setlength{\tabcolsep}{0pt}
                \centering
                \begin{tabular}{p{0.24\linewidth} p{0.24\linewidth} p{0.24\linewidth} p{0.24\linewidth}}
                                                    & \myimage{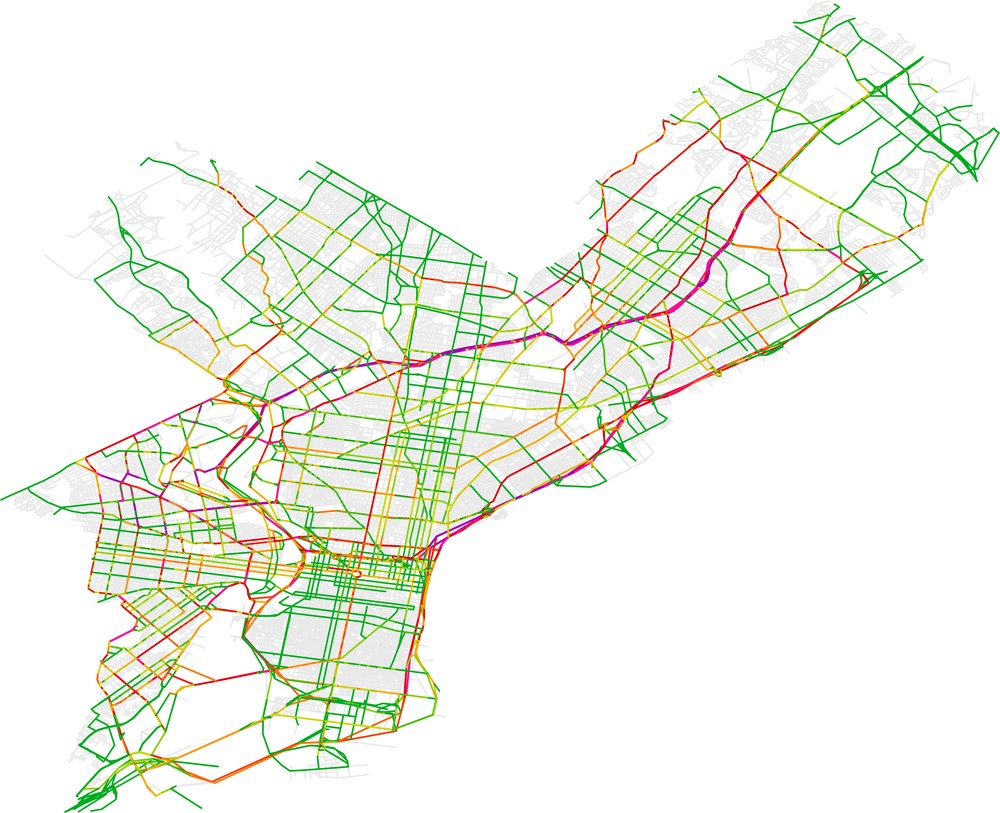} & \myimage{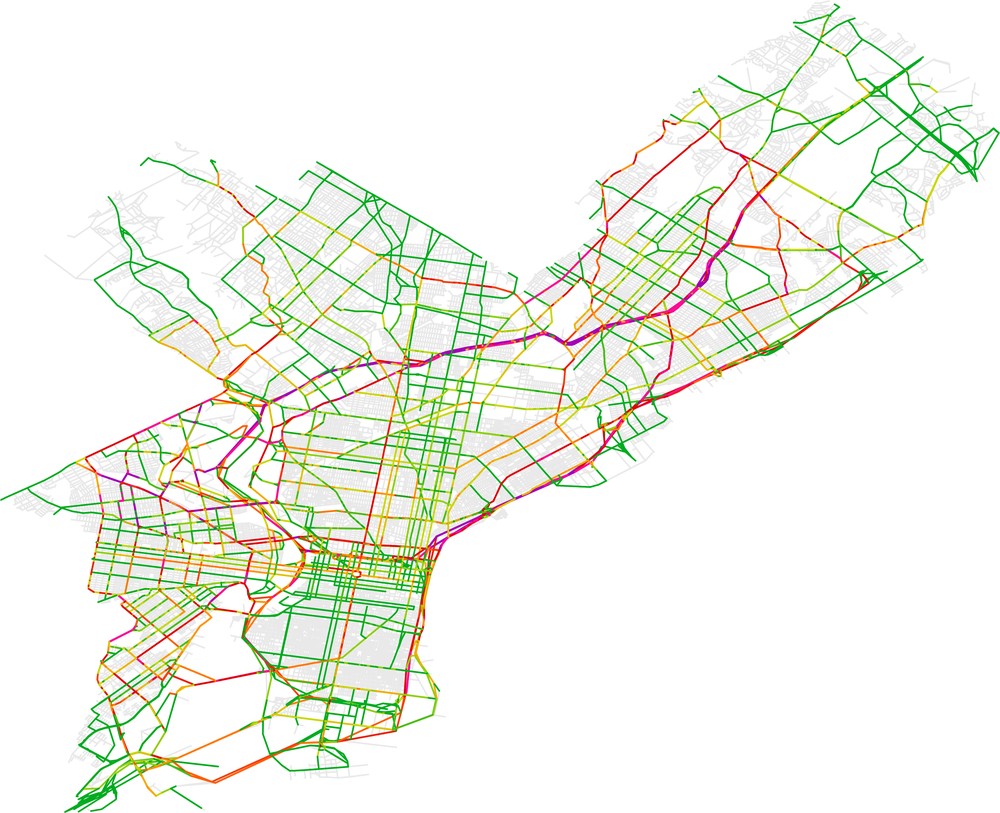} & \myimage{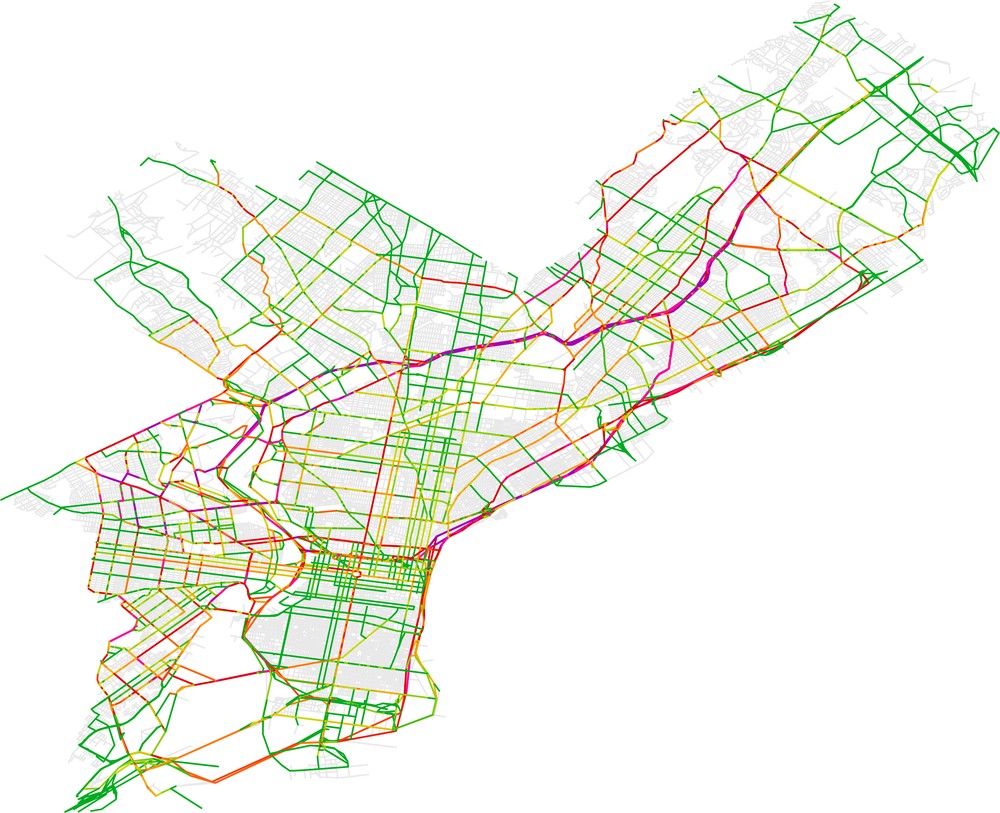} \\
                    \myimage{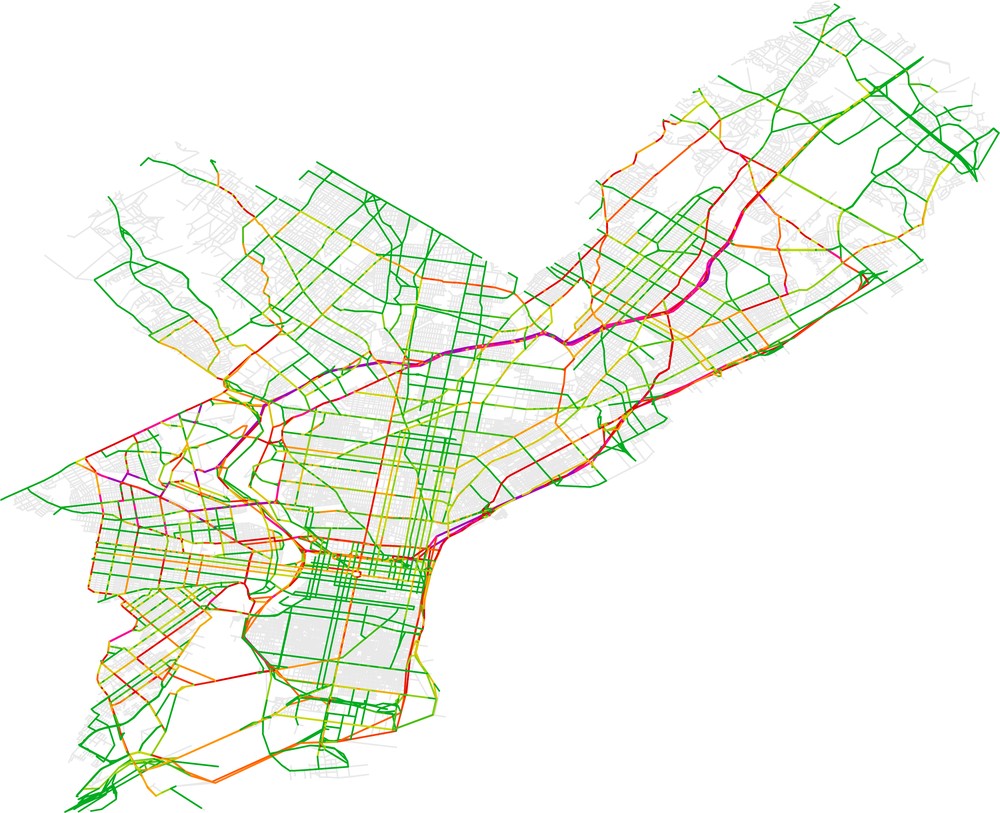}  & \myimage{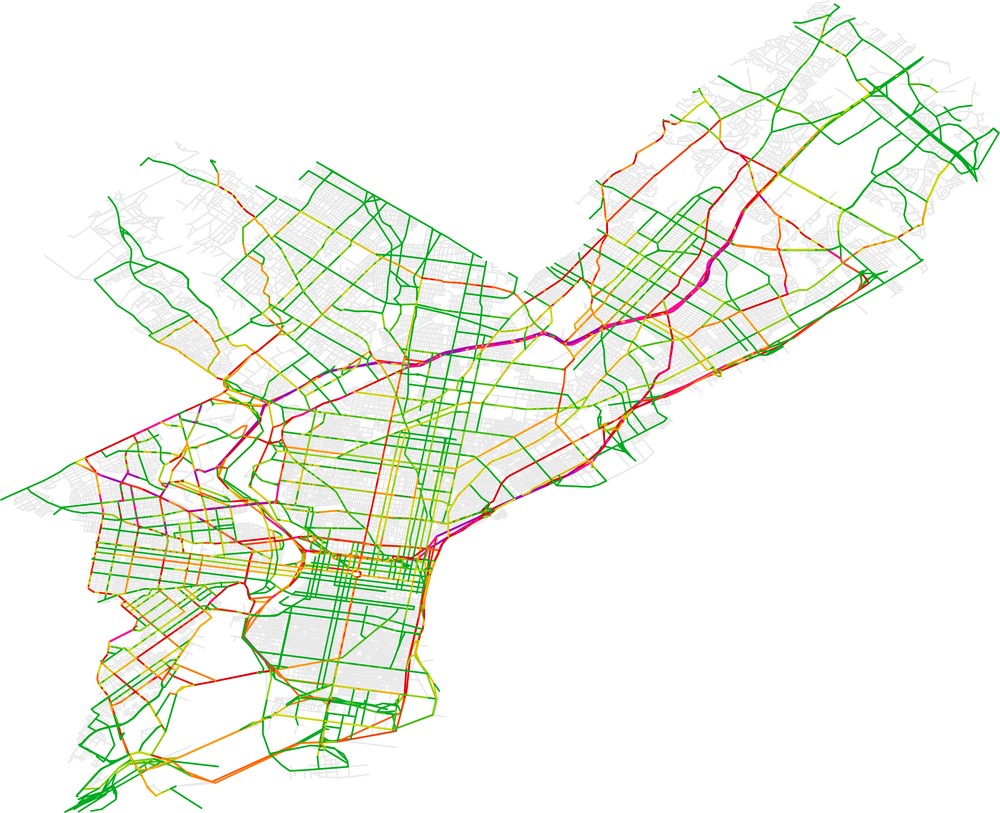} & \myimage{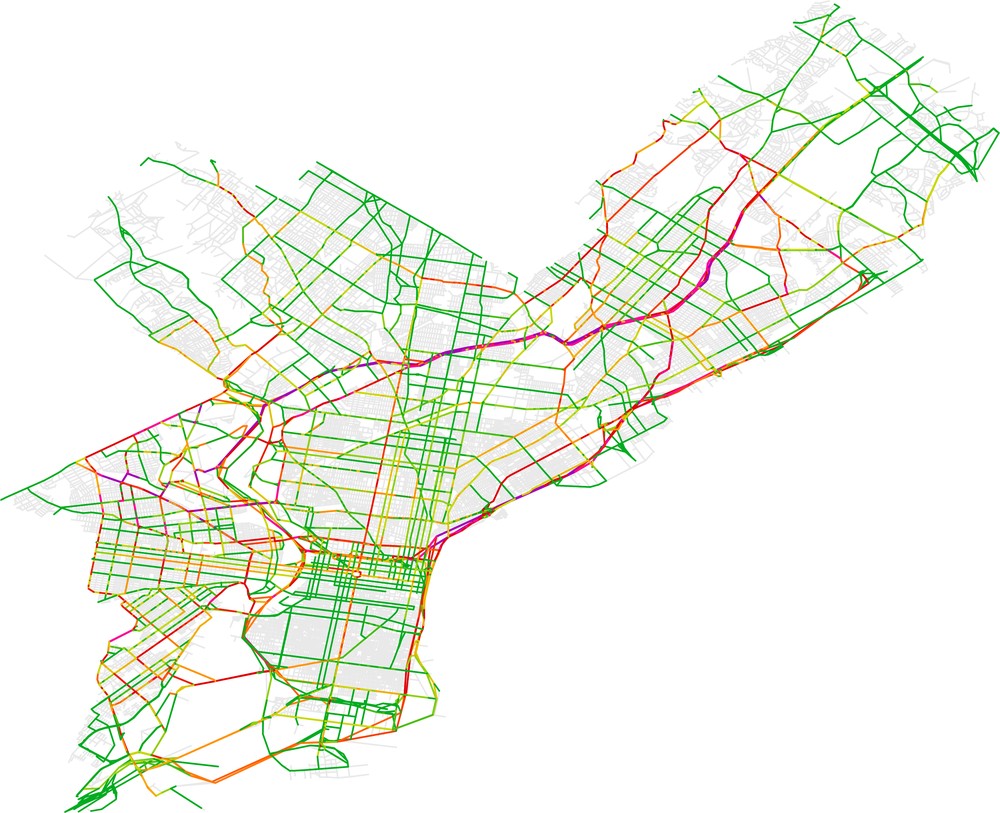} & \myimage{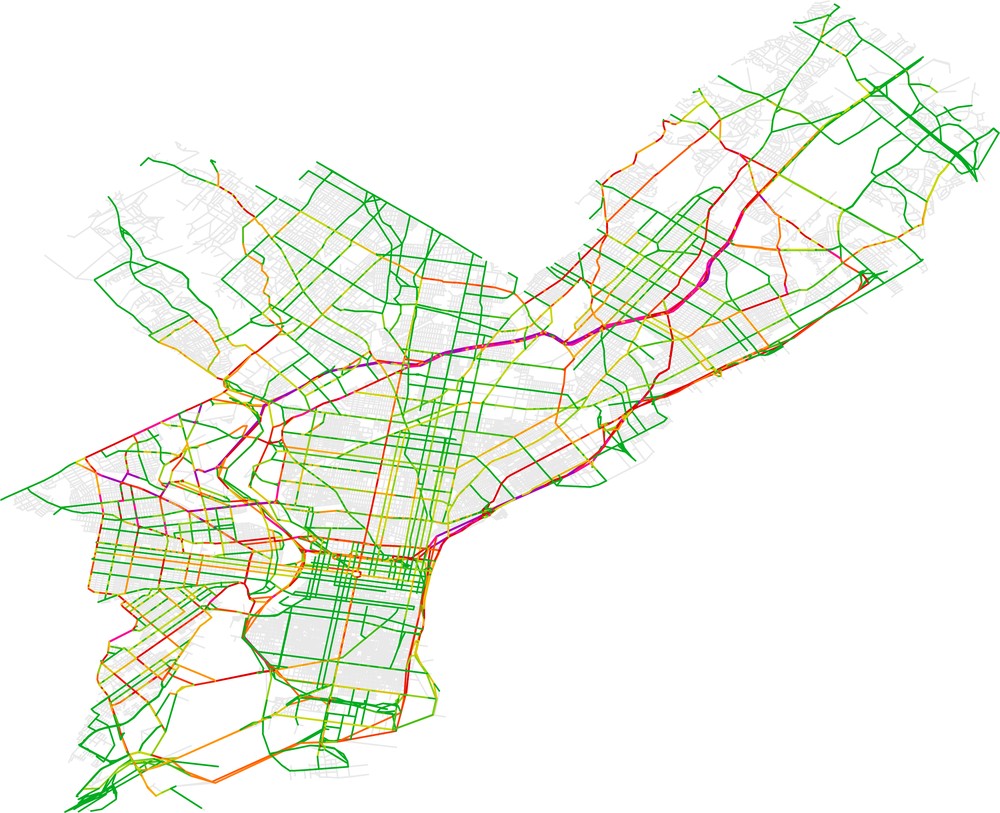} \\
                                    & \myimage{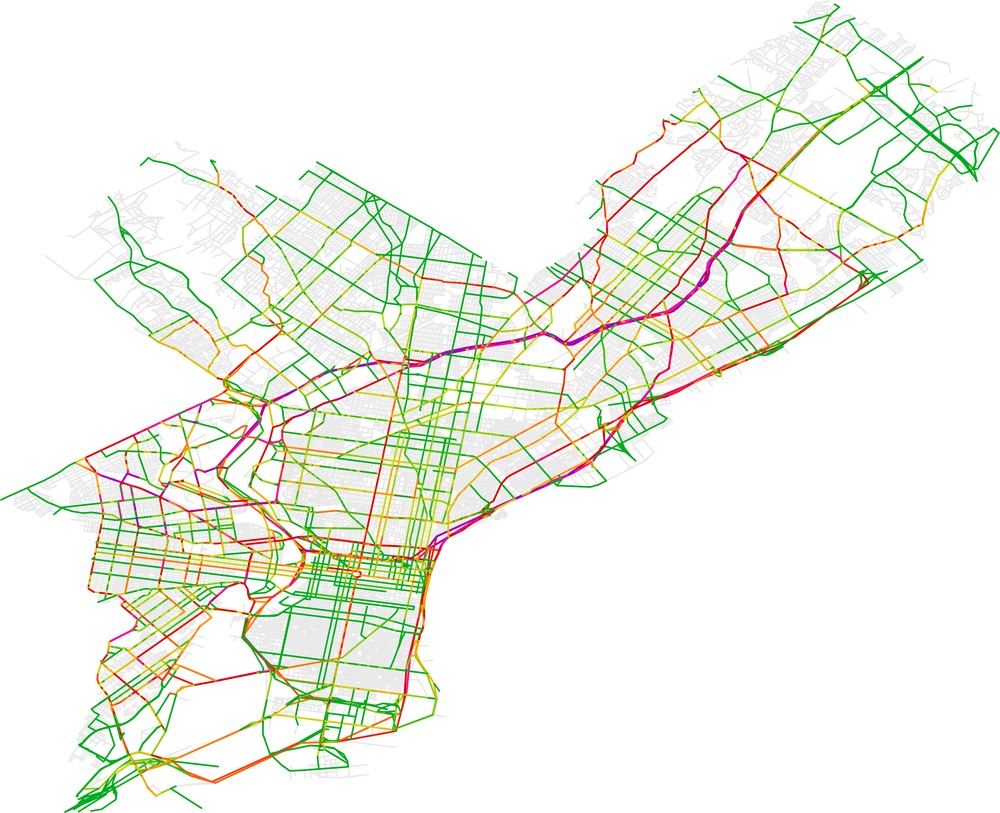} & \myimage{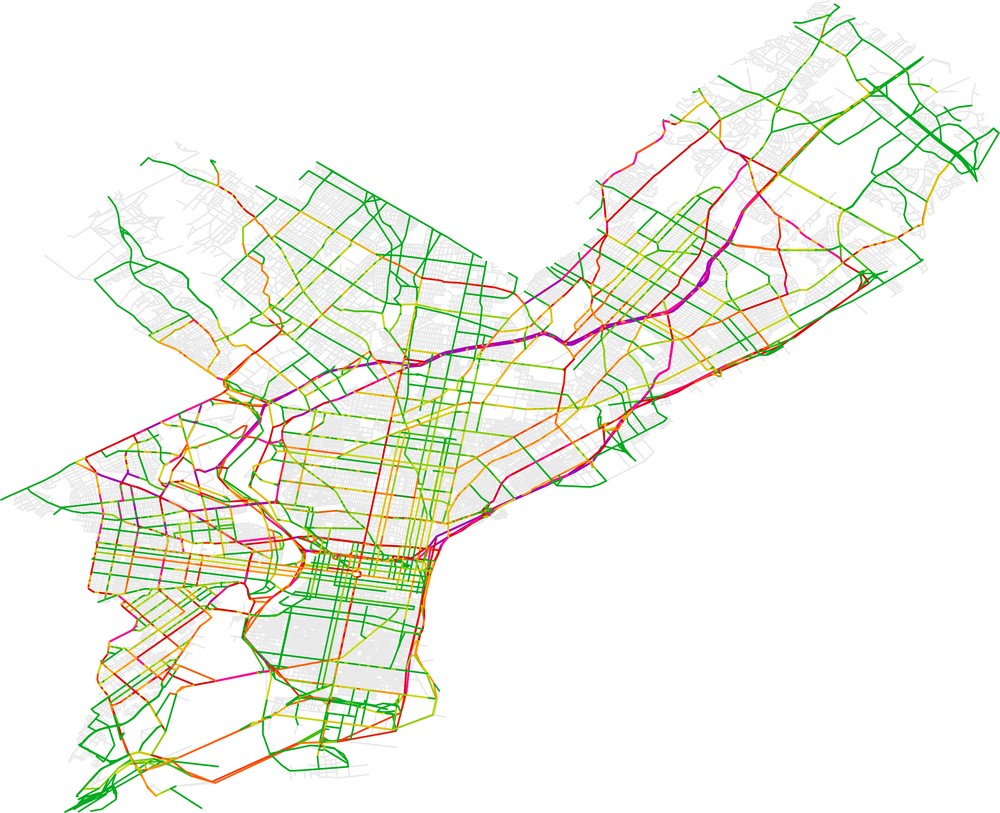} & \myimage{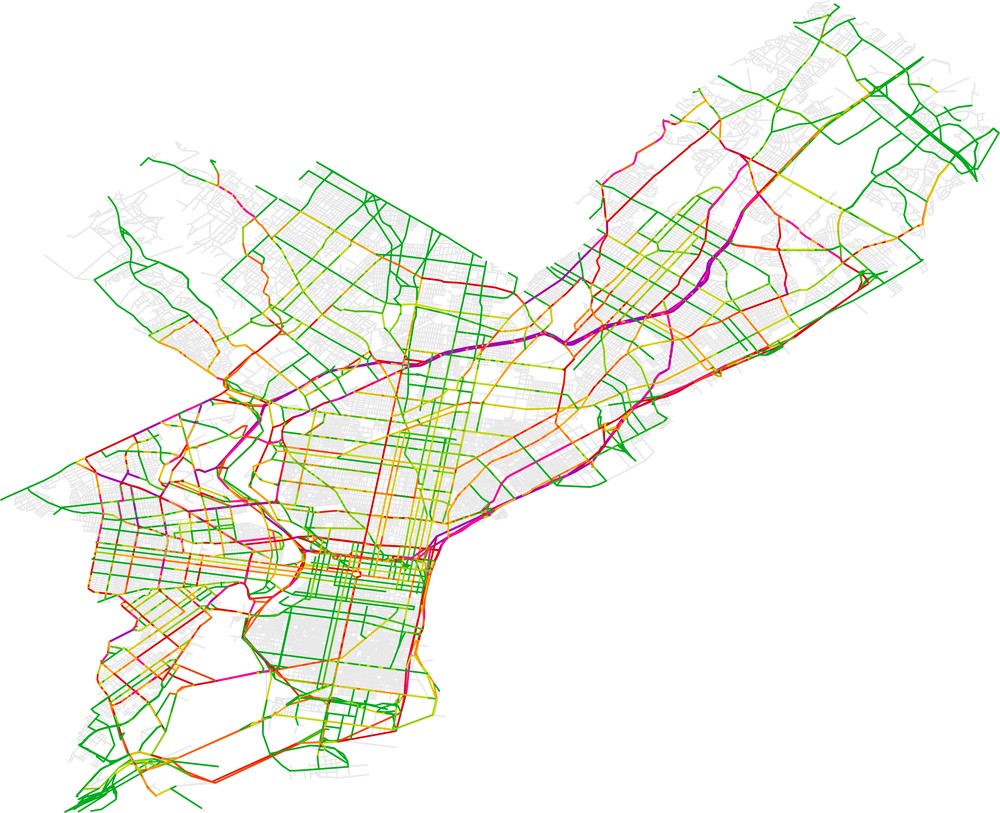} \\
                \end{tabular}
                \caption{Philadelphia}
                \label{fig:more_philadelphia}
            \end{subfigure}%
            \vspace{-0.7em}
            \caption{Maps of VOC across cities, years and SSPs under the Time-Carbon UE model ($\psi_{1}=1, \psi_{2}=9$).}
            \label{fig:more VOCs maps}
        \end{figure}
        \newpage
        \begin{figure}[H]
            \ContinuedFloat
            \captionsetup[subfigure]{skip=1pt}
            \centering
            \begin{subfigure}{1.0\linewidth}
                \setlength{\tabcolsep}{0pt}
                \centering
                \begin{tabular}{p{0.245\linewidth} p{0.245\linewidth} p{0.245\linewidth} p{0.245\linewidth}}
                                                    & \myimage{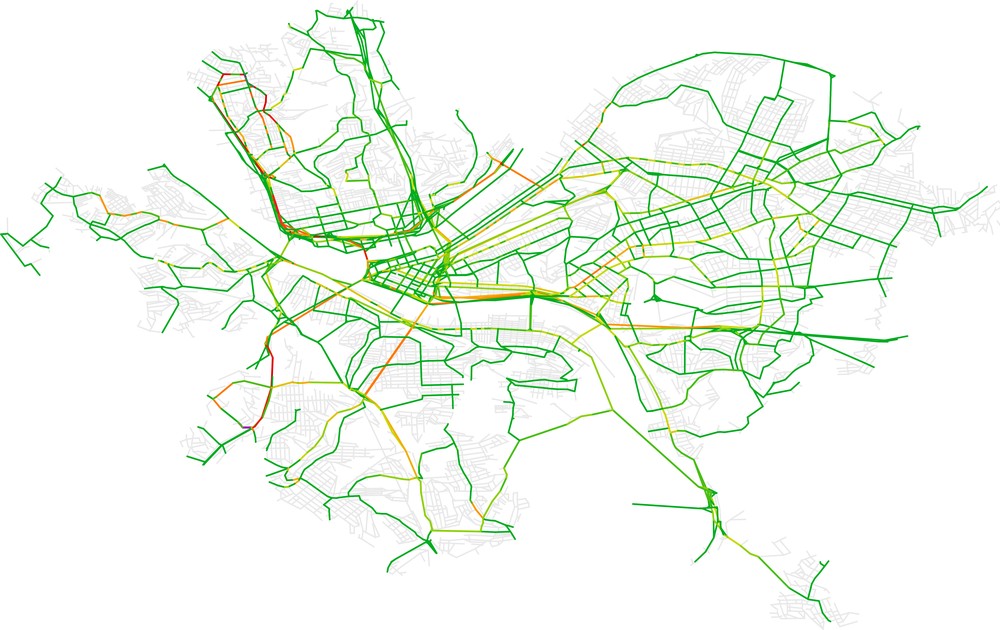} & \myimage{with03_2030_ssp2} & \myimage{with03_2030_ssp2} \\
                    \myimage{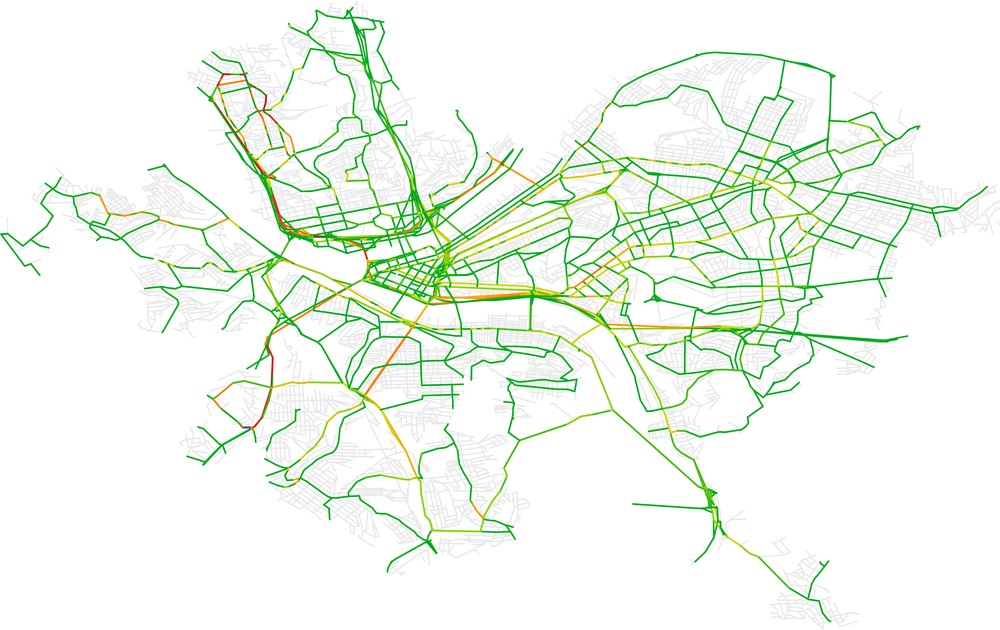}  & \myimage{with03_2030_ssp2} & \myimage{with03_2030_ssp2} & \myimage{with03_2030_ssp2} \\
                    {\footnotesize (c) Pittsburgh}                  & \myimage{with03_2030_ssp2} & \myimage{with03_2030_ssp2} & \myimage{with03_2030_ssp2} \\
                \end{tabular}
                \caption{Pittsburgh}
                \label{fig:more_pittsburgh}
            \end{subfigure}%
            \vspace{0.1em}
            \begin{subfigure}{1.0\linewidth}
                \setlength{\tabcolsep}{0pt}
                \centering
                \begin{tabular}{p{0.245\linewidth} p{0.245\linewidth} p{0.245\linewidth} p{0.245\linewidth}}
                                                    & \myimage{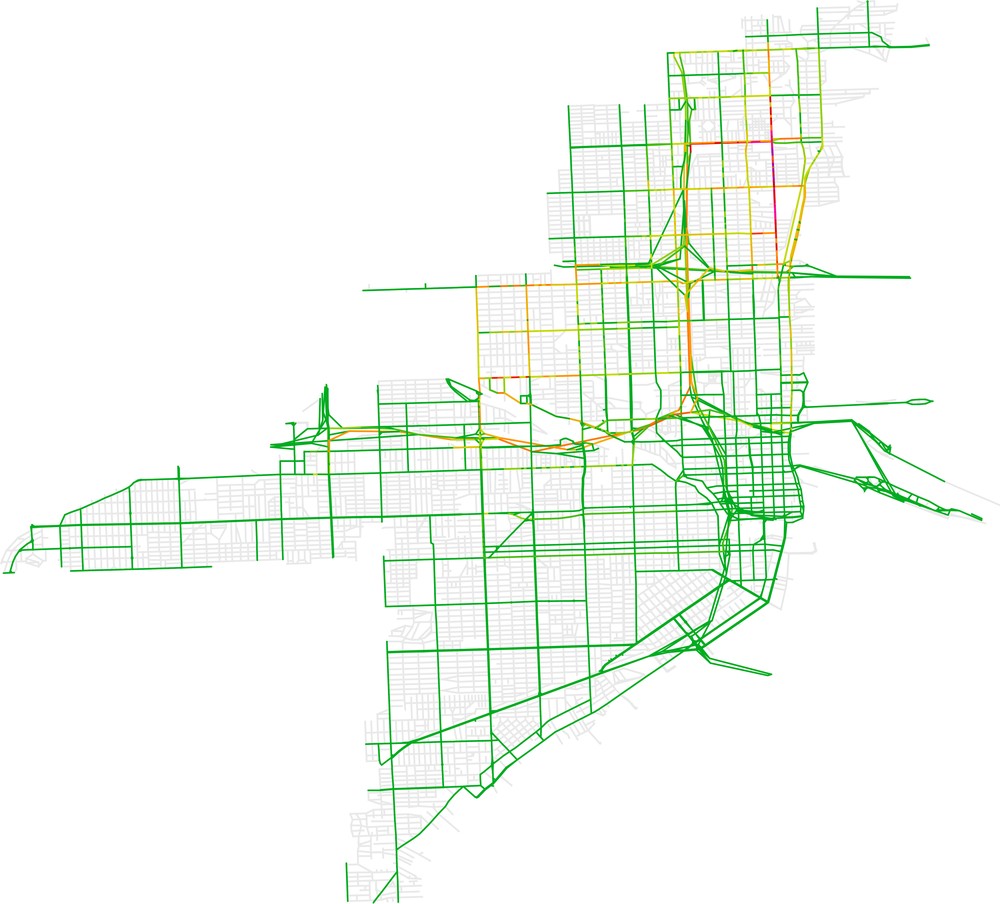} & \myimage{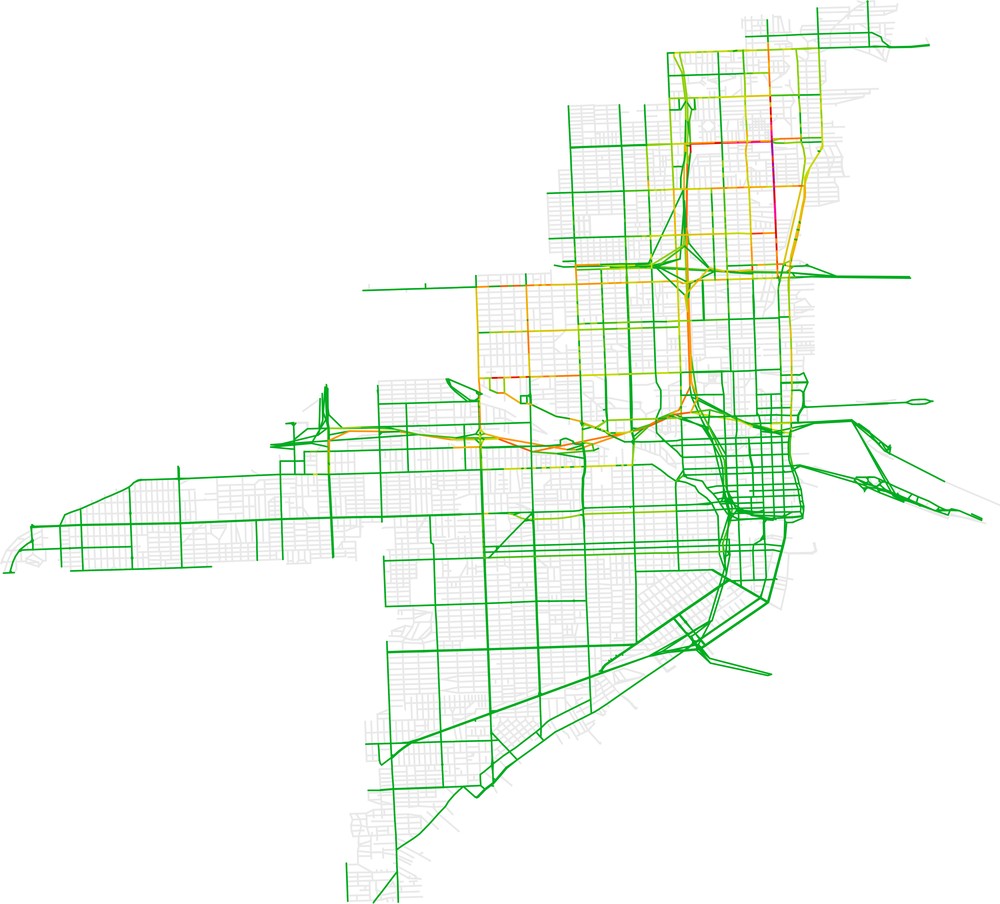} & \myimage{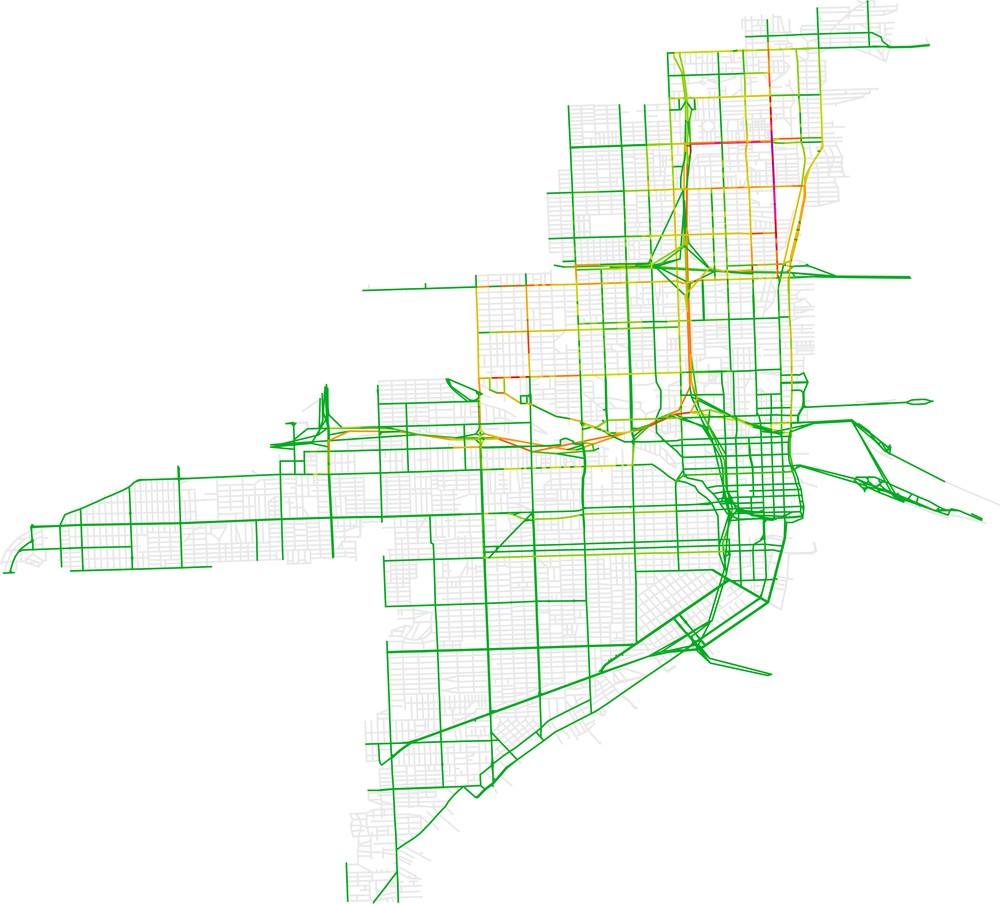} \\
                    \myimage{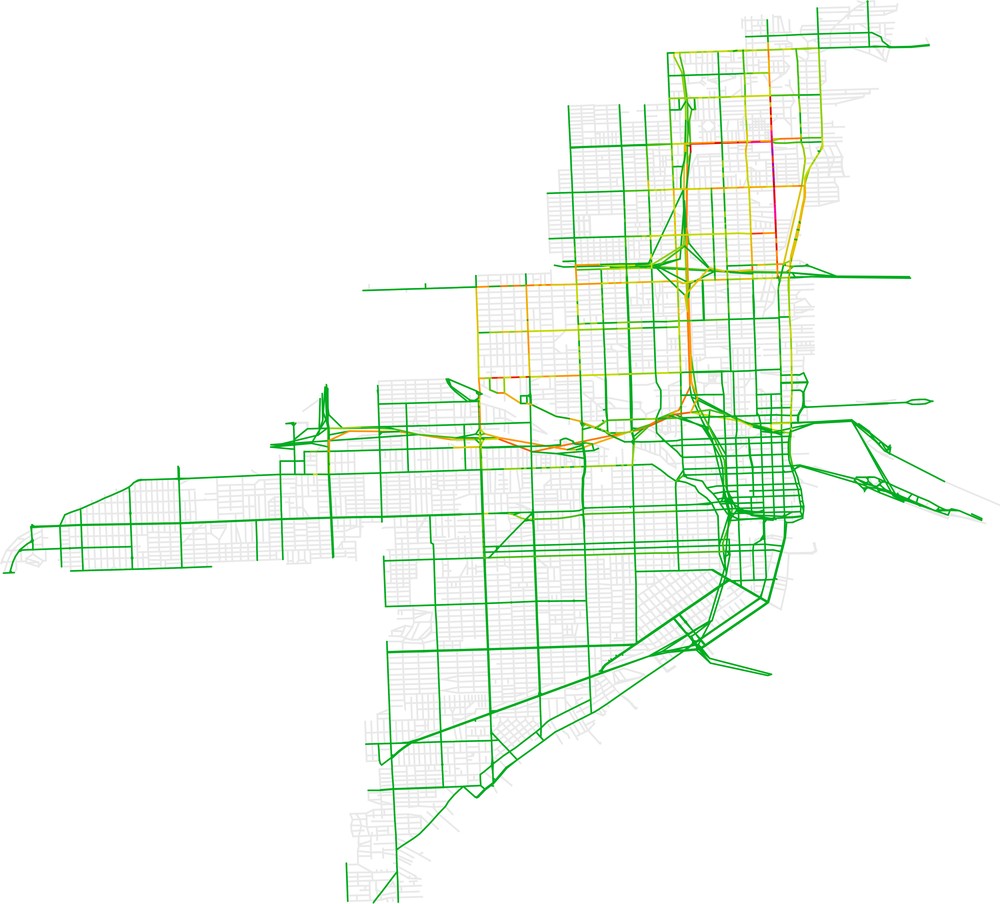}  & \myimage{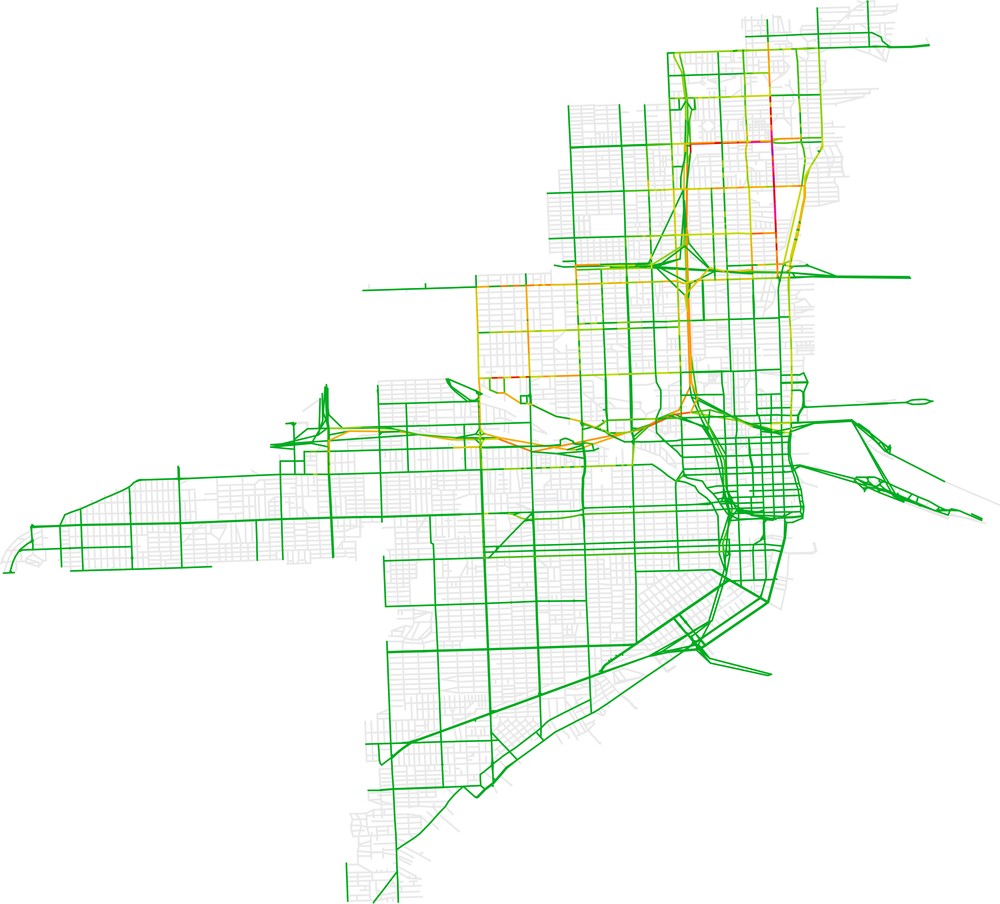} & \myimage{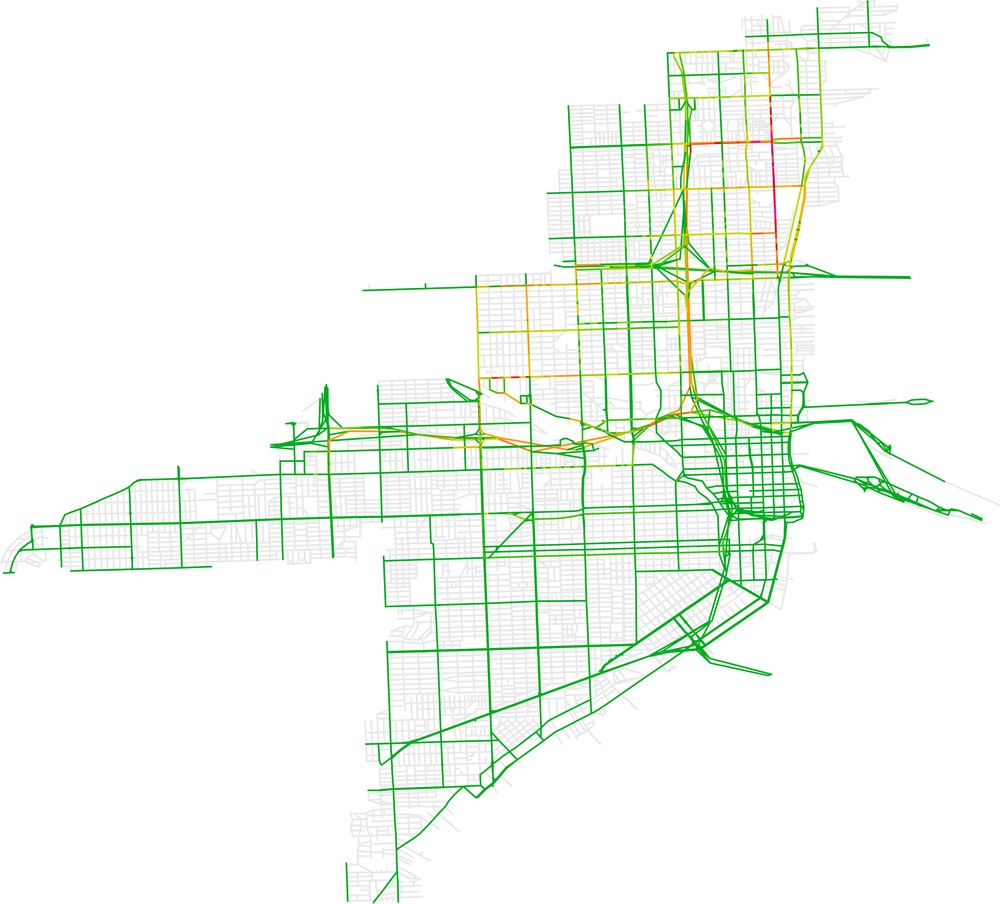} & \myimage{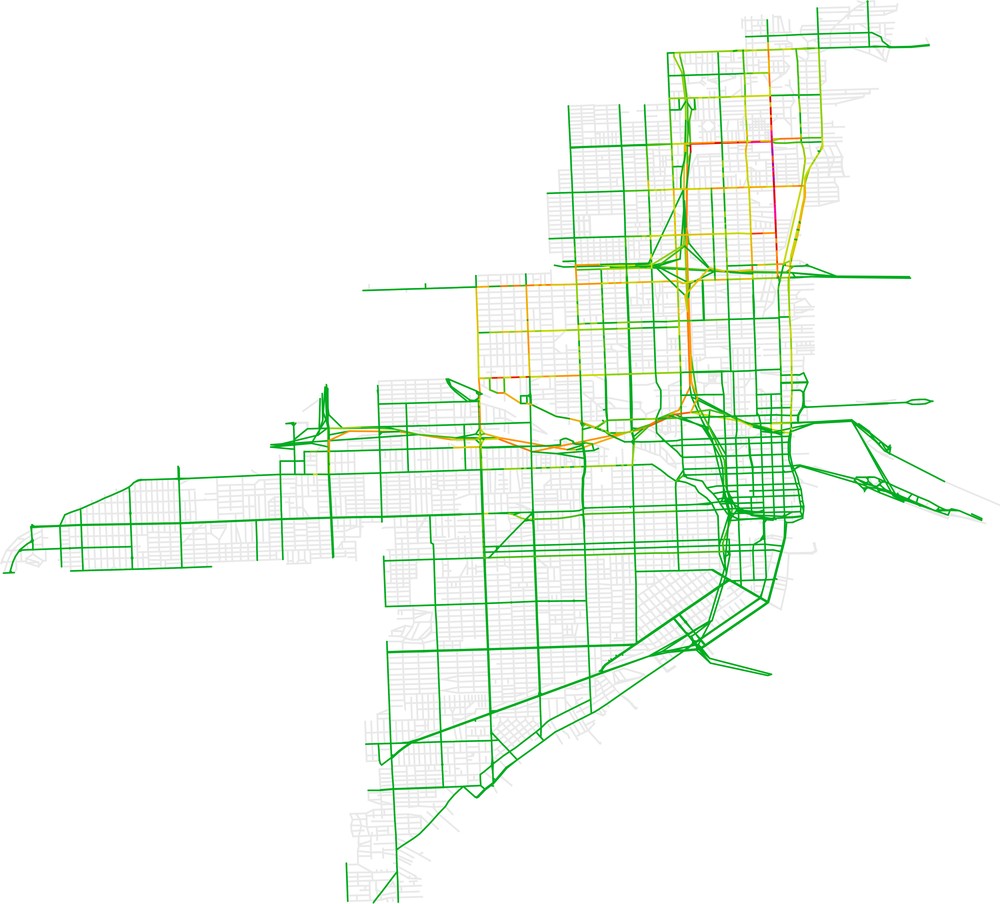} \\
                                           & \myimage{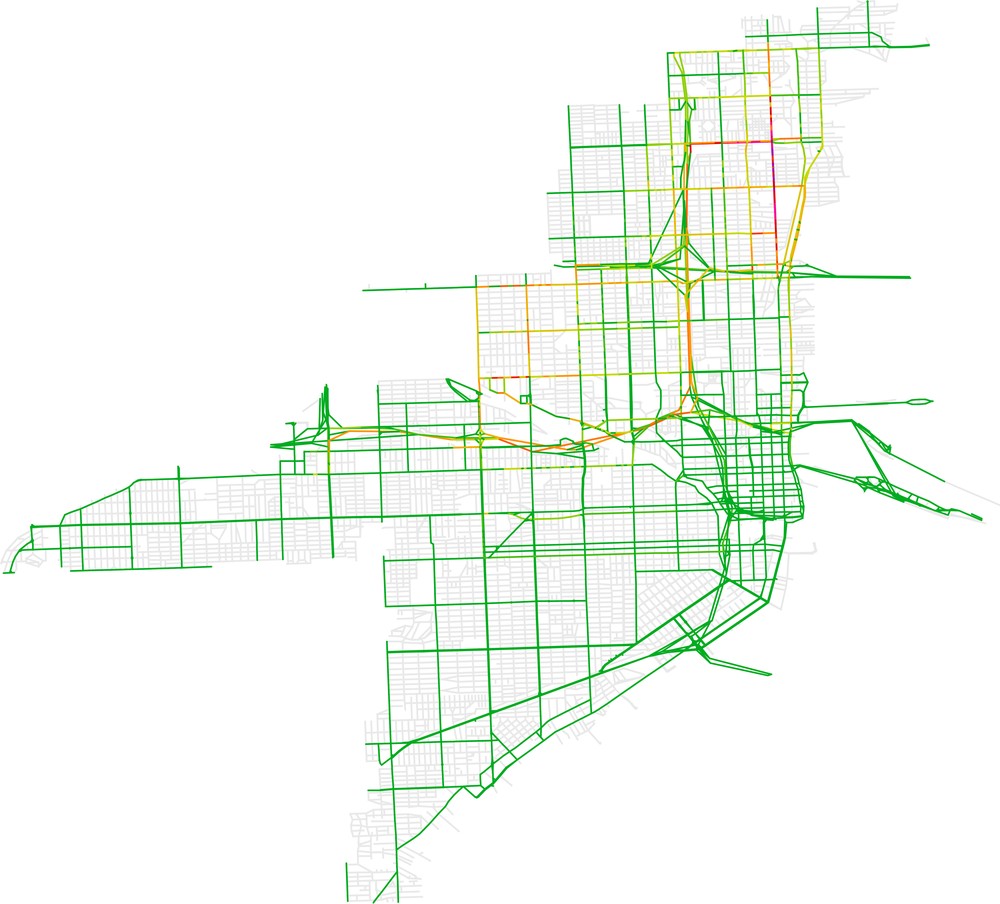} & \myimage{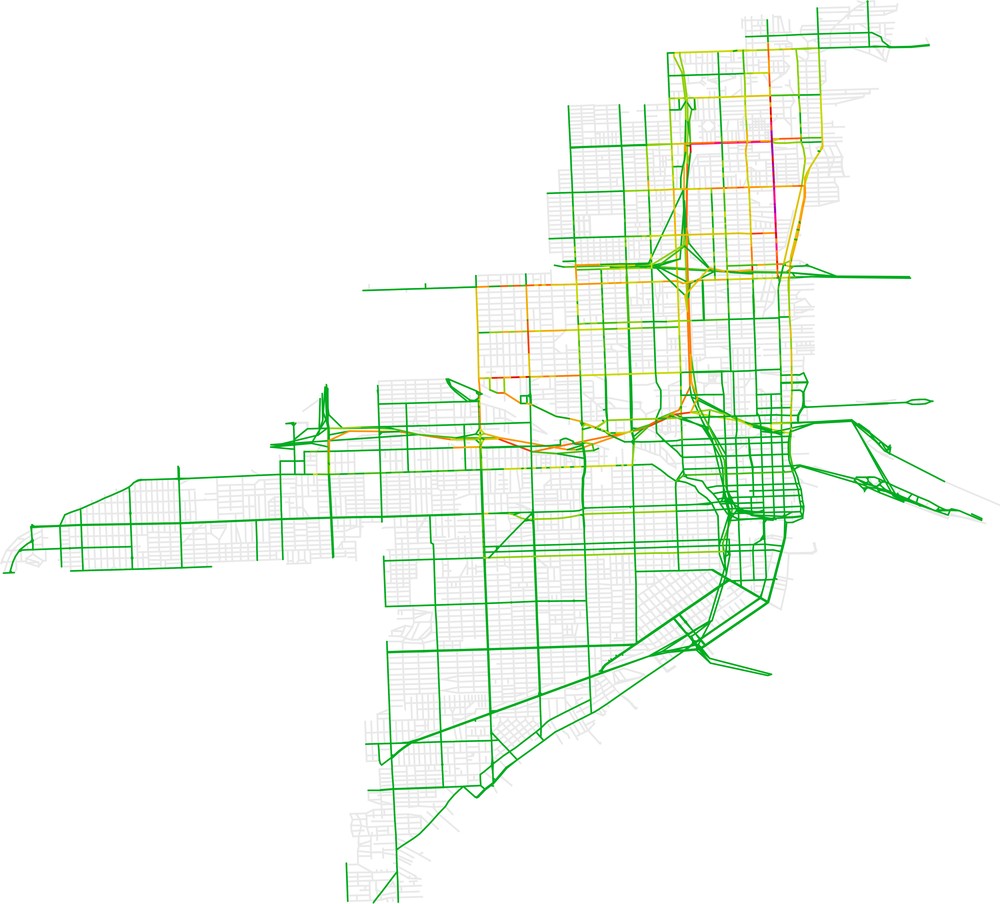} & \myimage{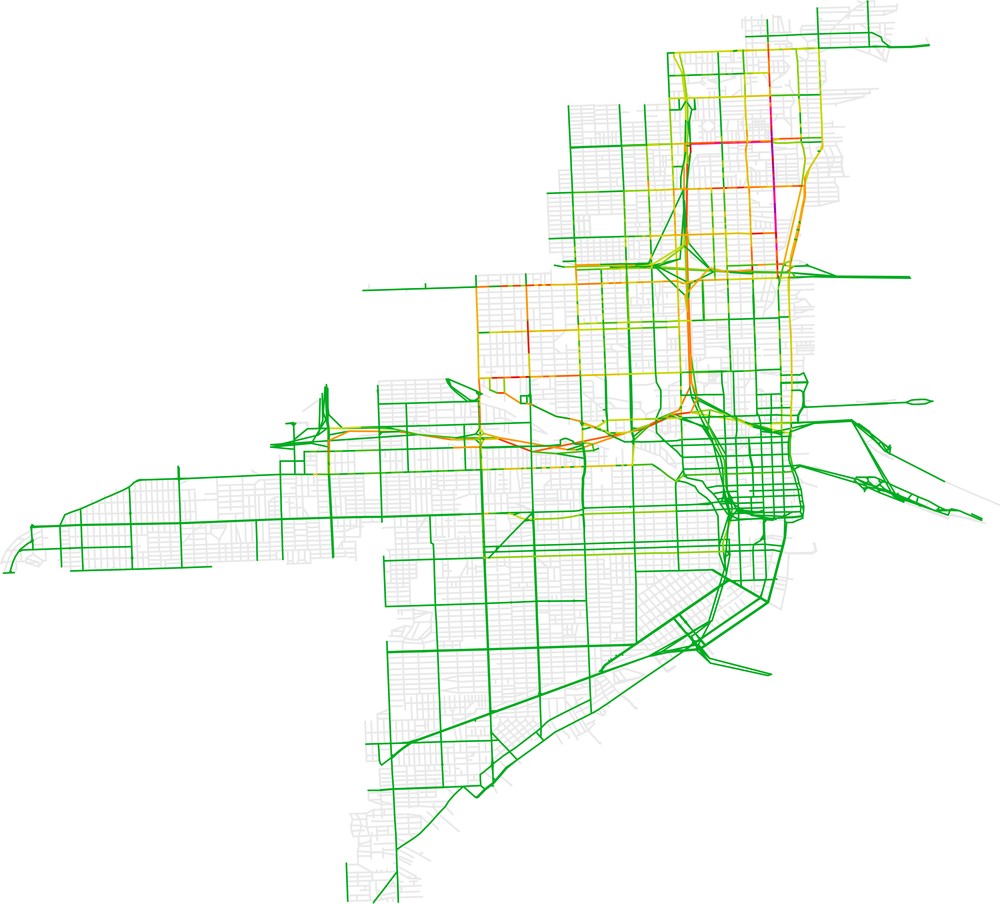} \\
                \end{tabular}
                \caption{Miami}
                \label{fig:more_miami}
            \end{subfigure}%
            \vspace{-0.7em}
            \caption[]{(continued).}
        \end{figure}
        \newpage
        \begin{figure}[H]
            \ContinuedFloat
            \captionsetup[subfigure]{skip=1pt}
            \centering
            \begin{subfigure}{1.0\linewidth}
                \setlength{\tabcolsep}{0pt}
                \centering
                \begin{tabular}{p{0.245\linewidth} p{0.245\linewidth} p{0.245\linewidth} p{0.245\linewidth}}
                                                    & \myimage{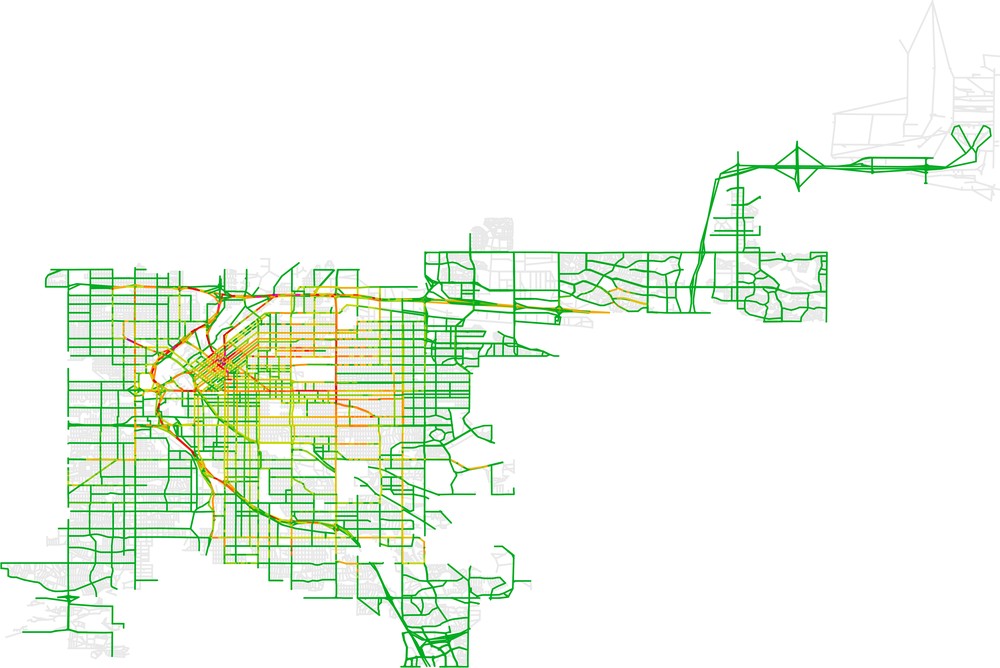} & \myimage{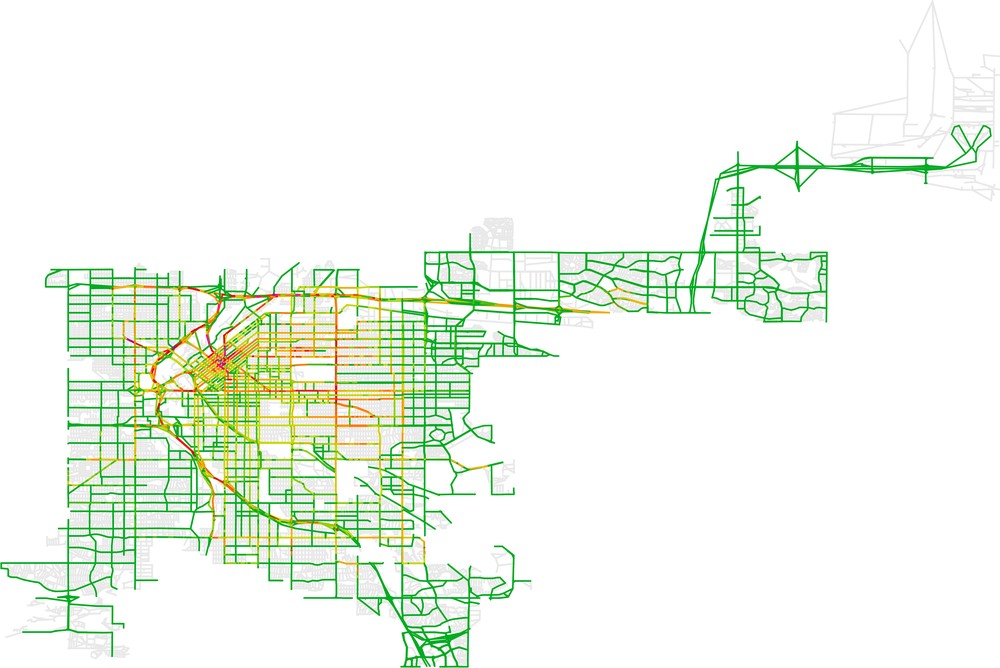} & \myimage{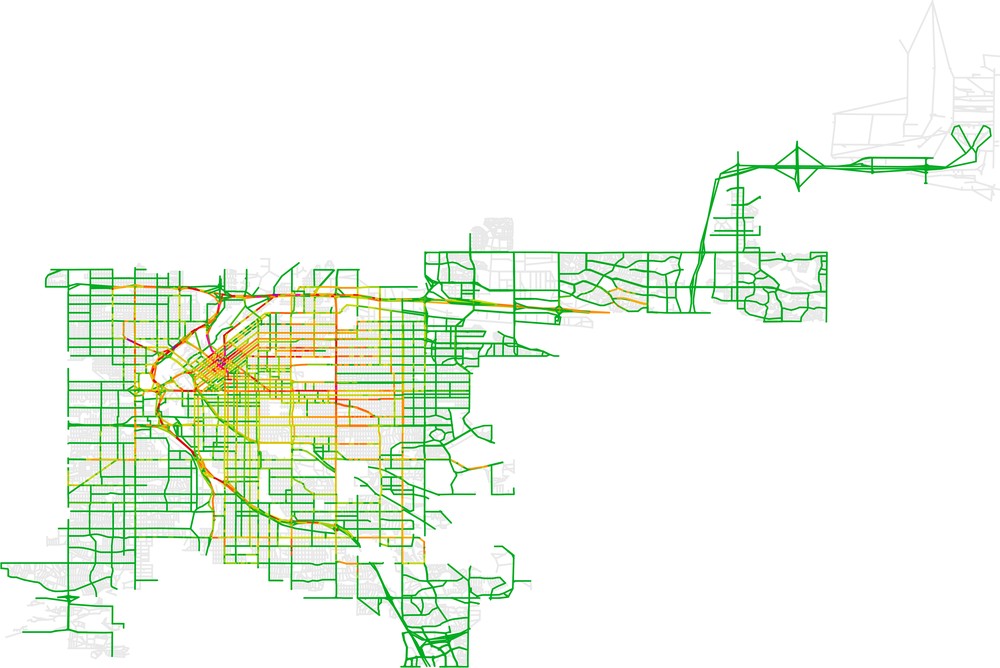} \\
                    \myimage{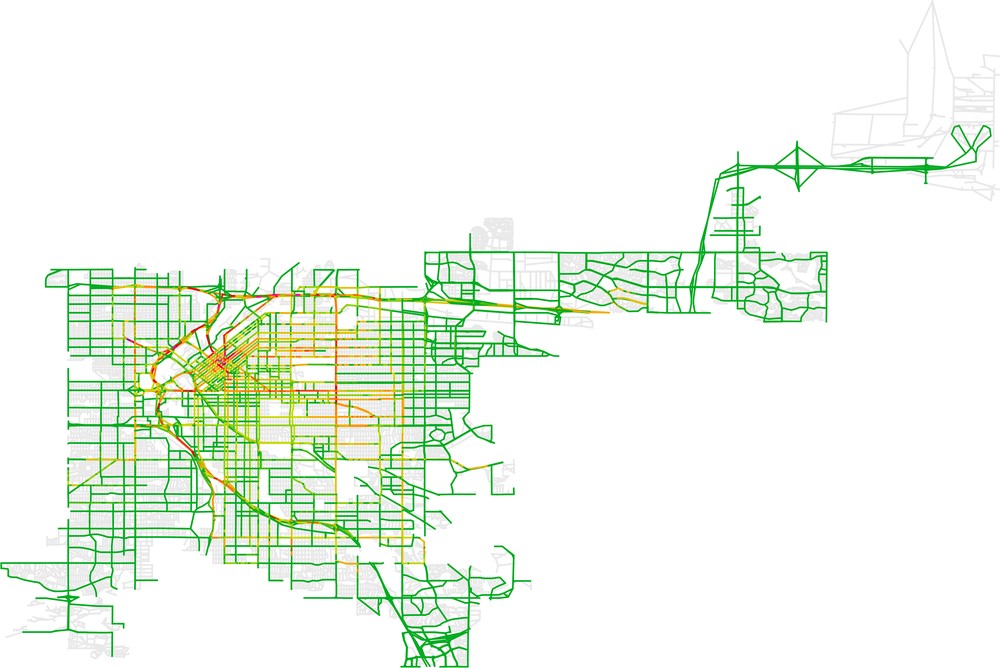}  & \myimage{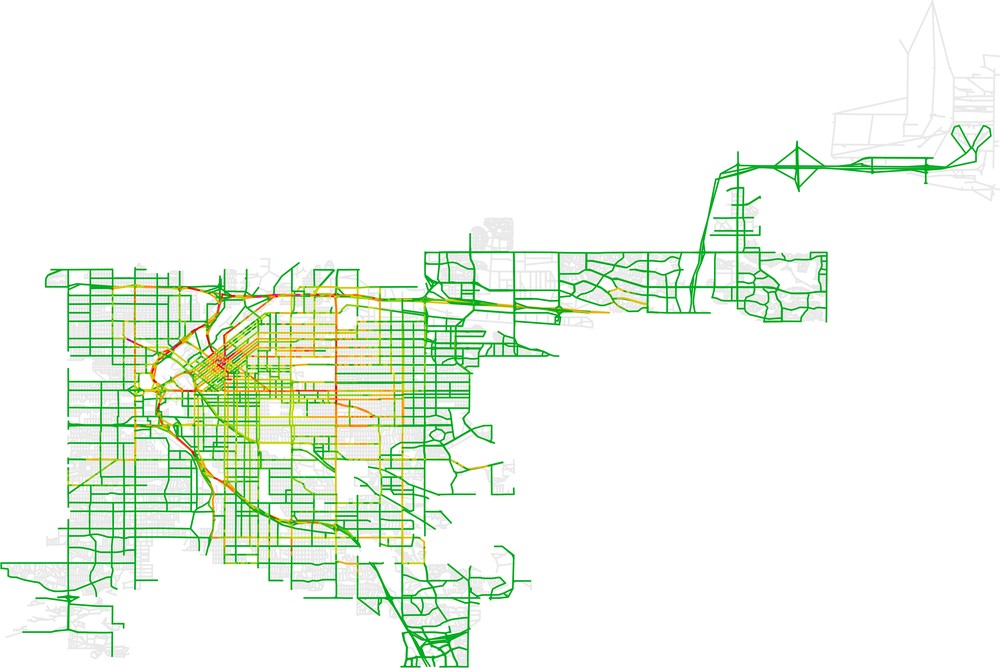} & \myimage{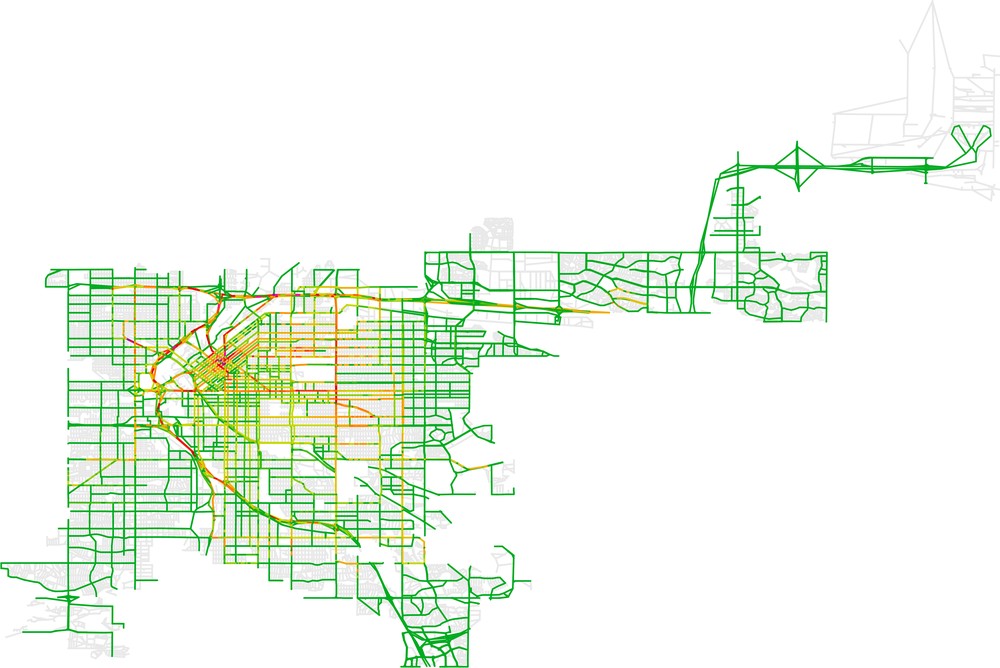} & \myimage{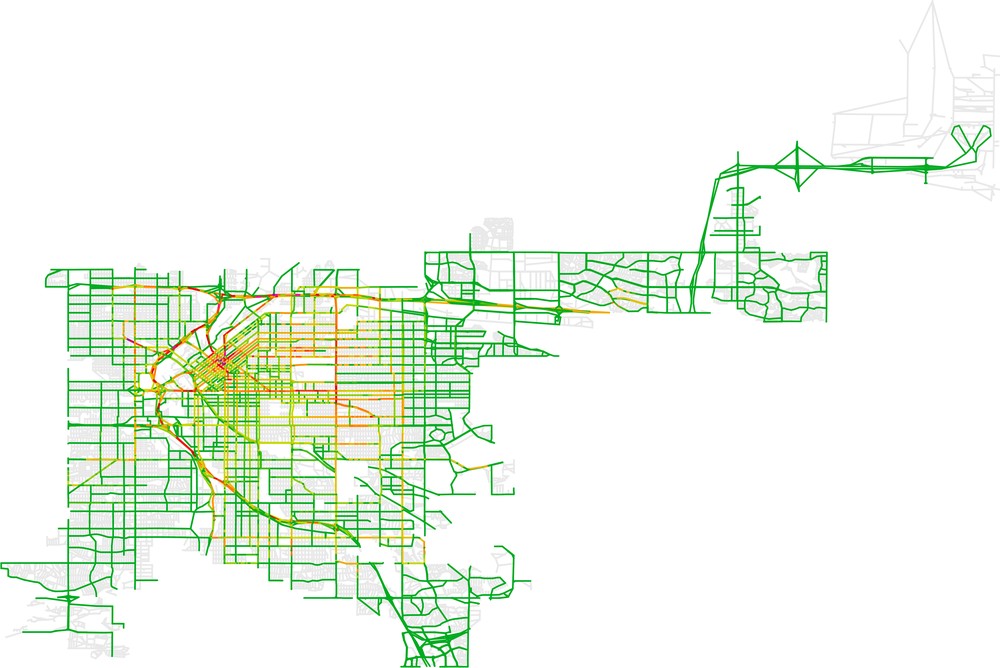} \\
                                                    & \myimage{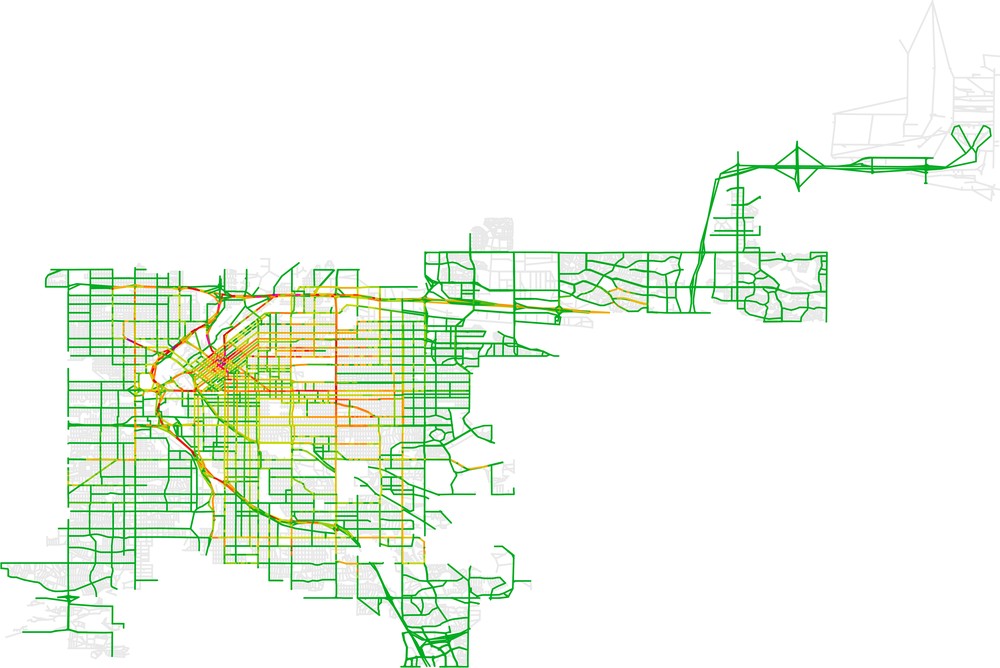} & \myimage{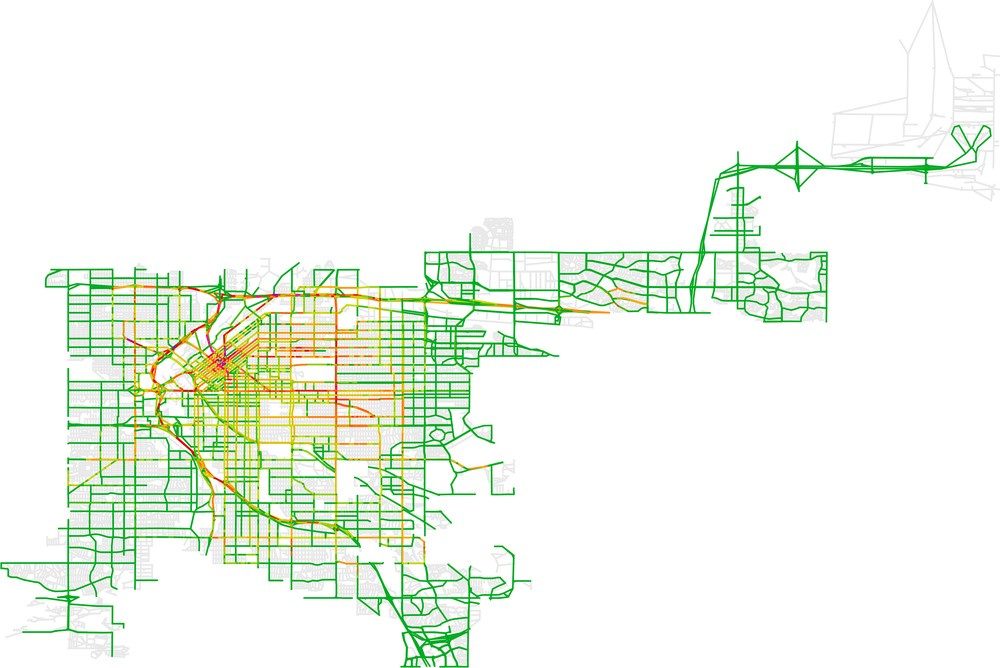} & \myimage{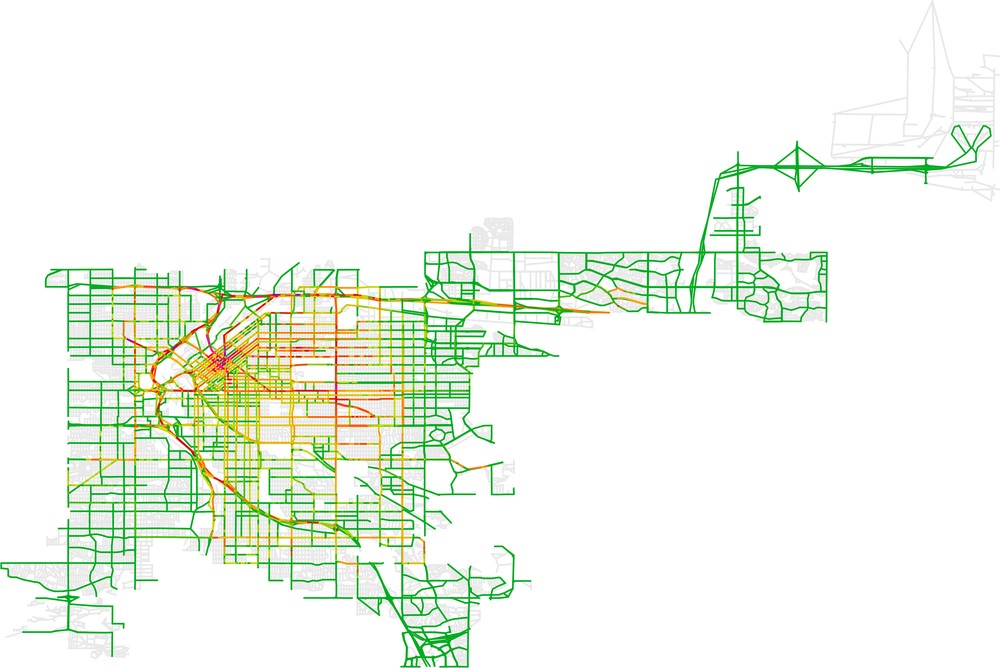} \\
                \end{tabular}
                \caption{Denver}
                \label{fig:more_denver}
            \end{subfigure}%
            \vspace{0.1em}
            \begin{subfigure}{1.0\linewidth}
                \setlength{\tabcolsep}{0pt}
                \centering
                \begin{tabular}{p{0.245\linewidth} p{0.245\linewidth} p{0.245\linewidth} p{0.245\linewidth}}
                                                    & \myimage{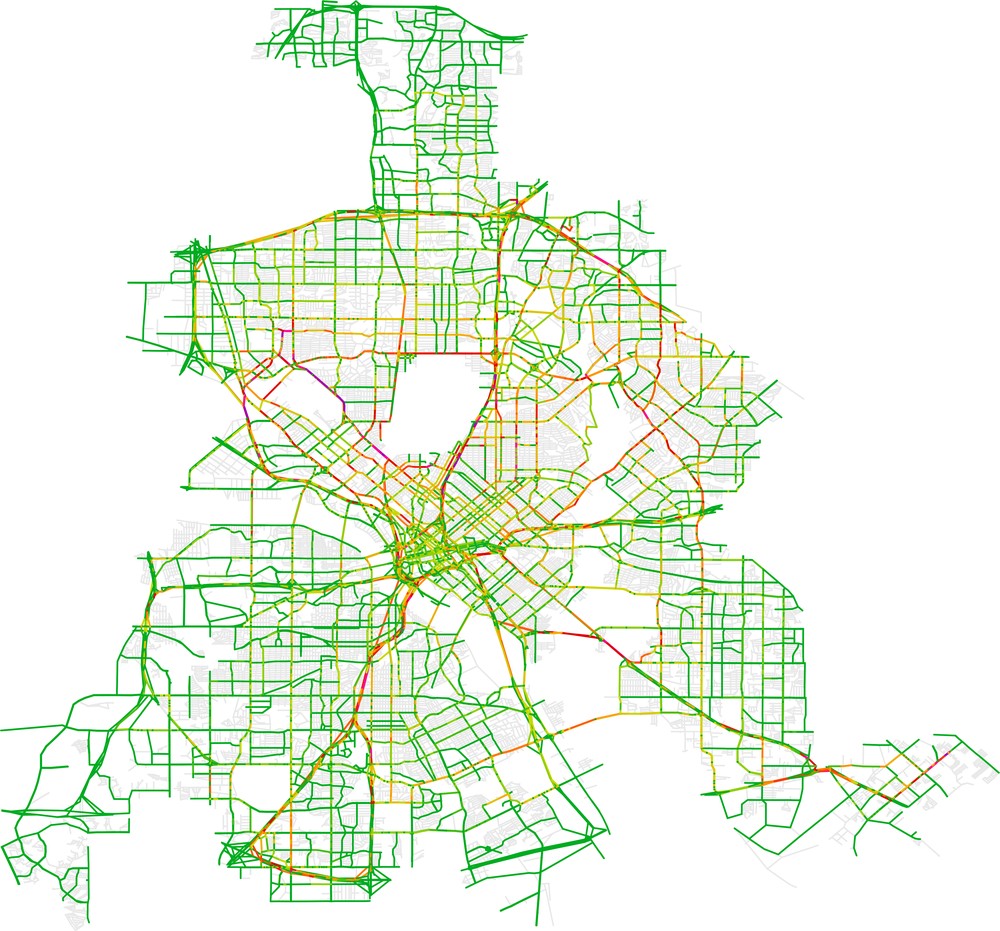} & \myimage{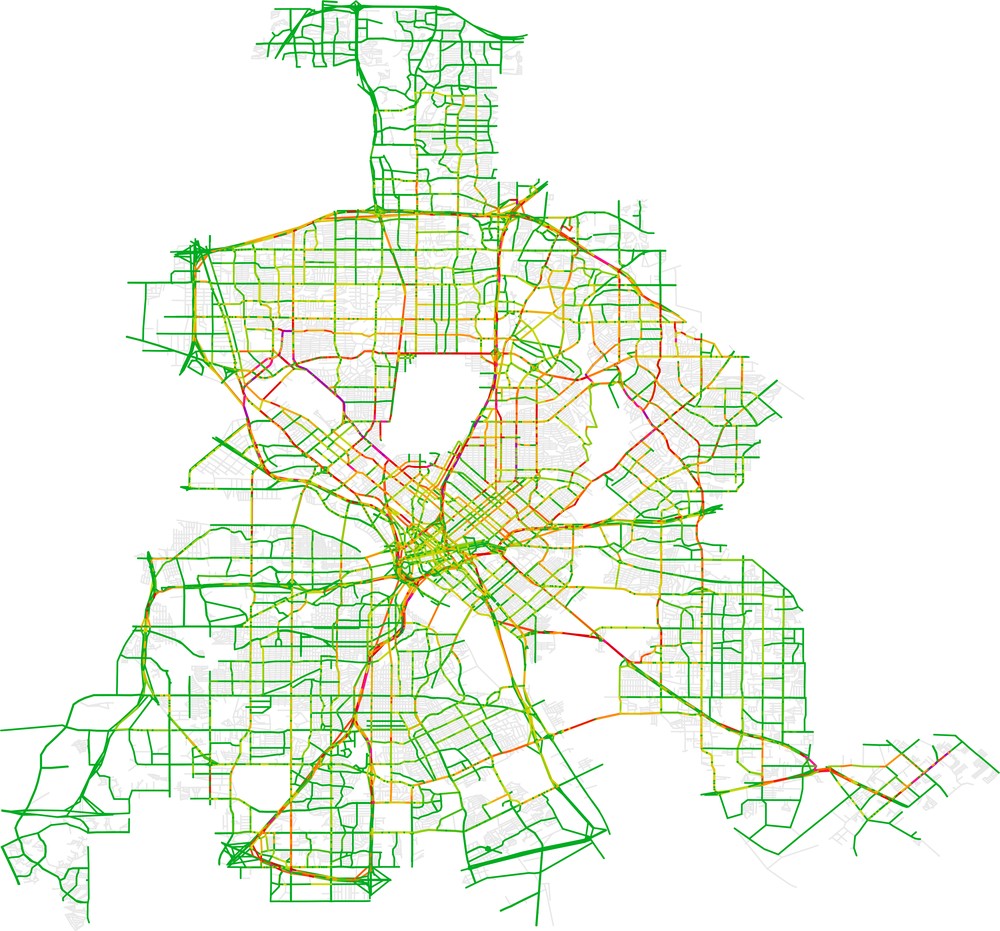} & \myimage{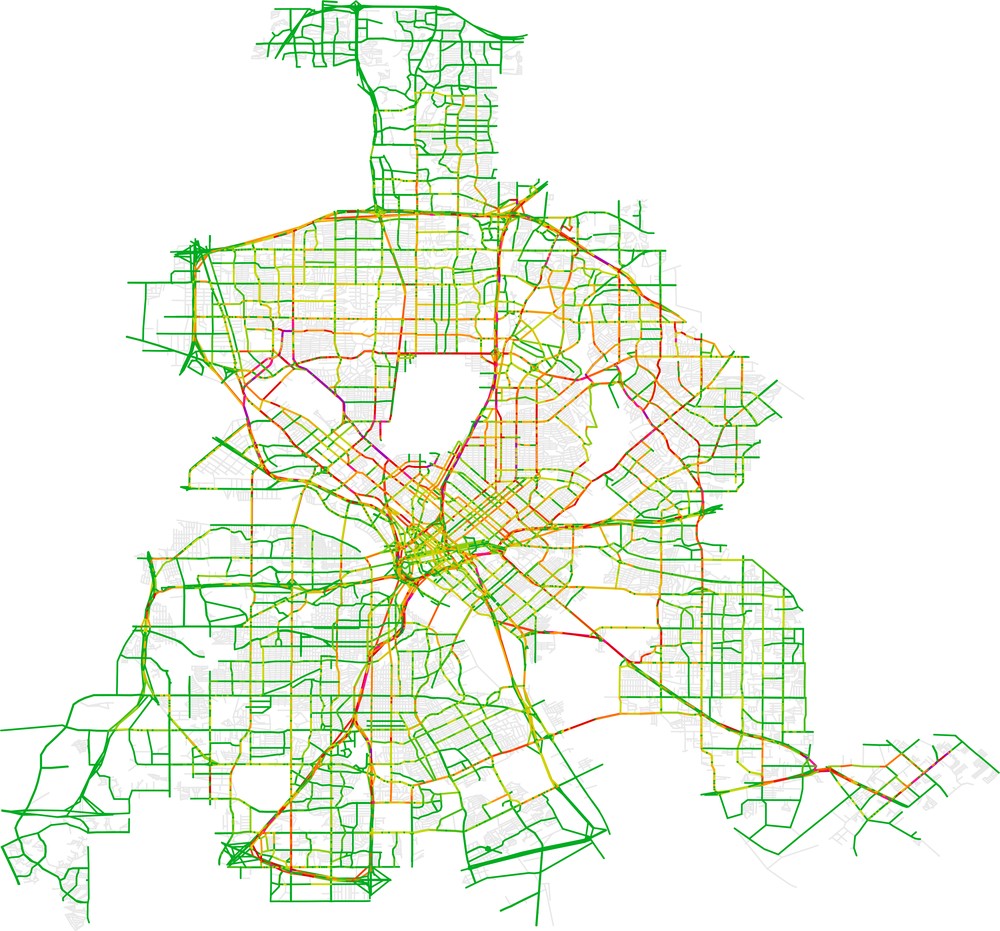} \\
                    \myimage{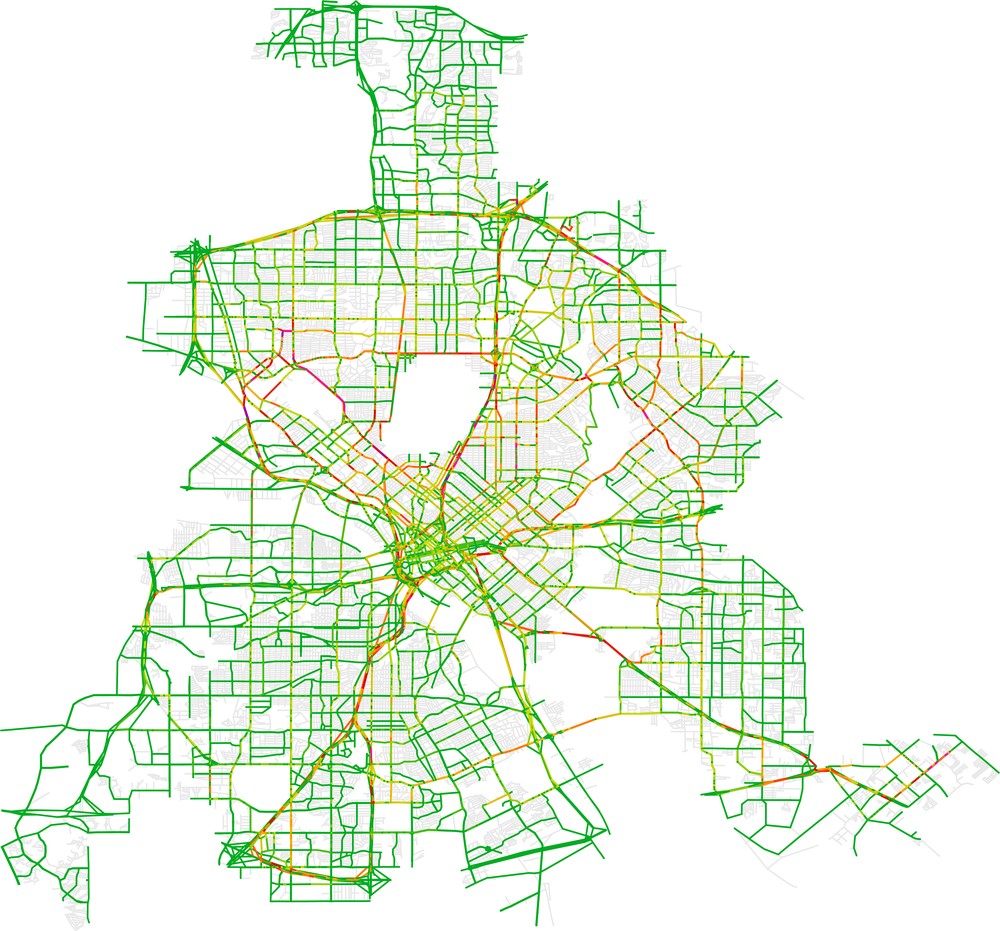}  & \myimage{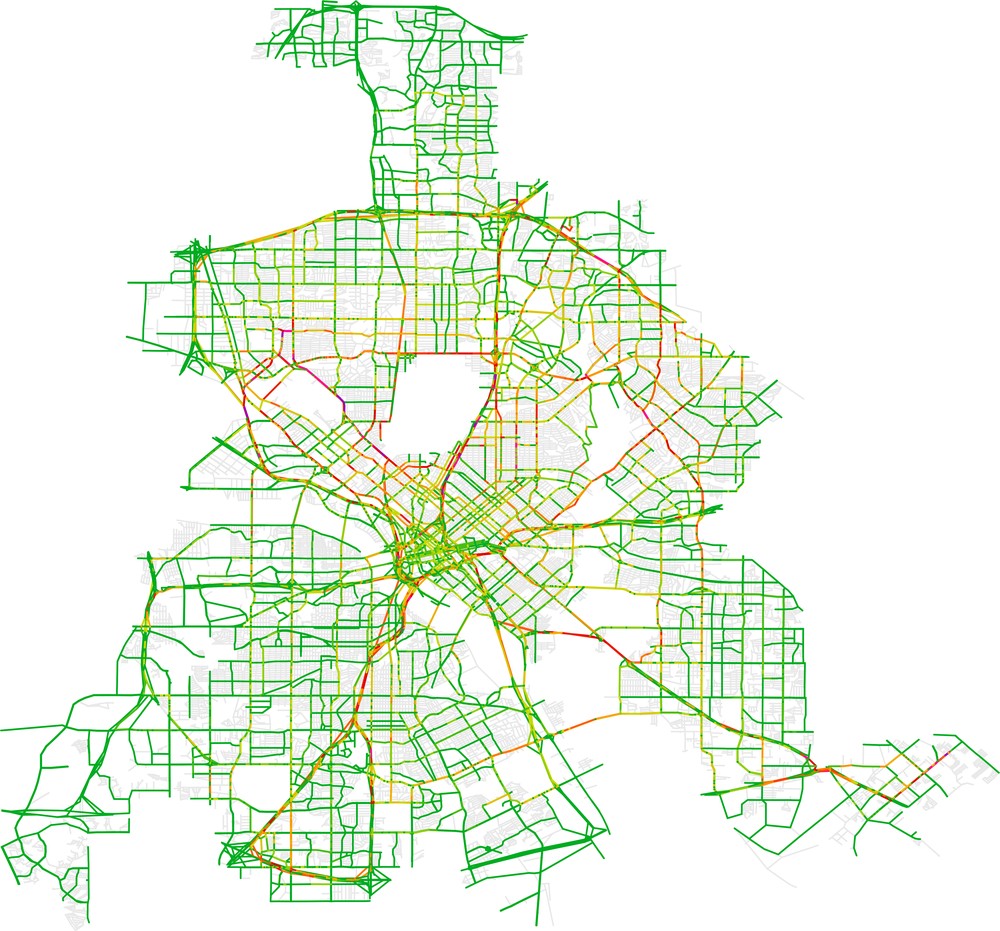} & \myimage{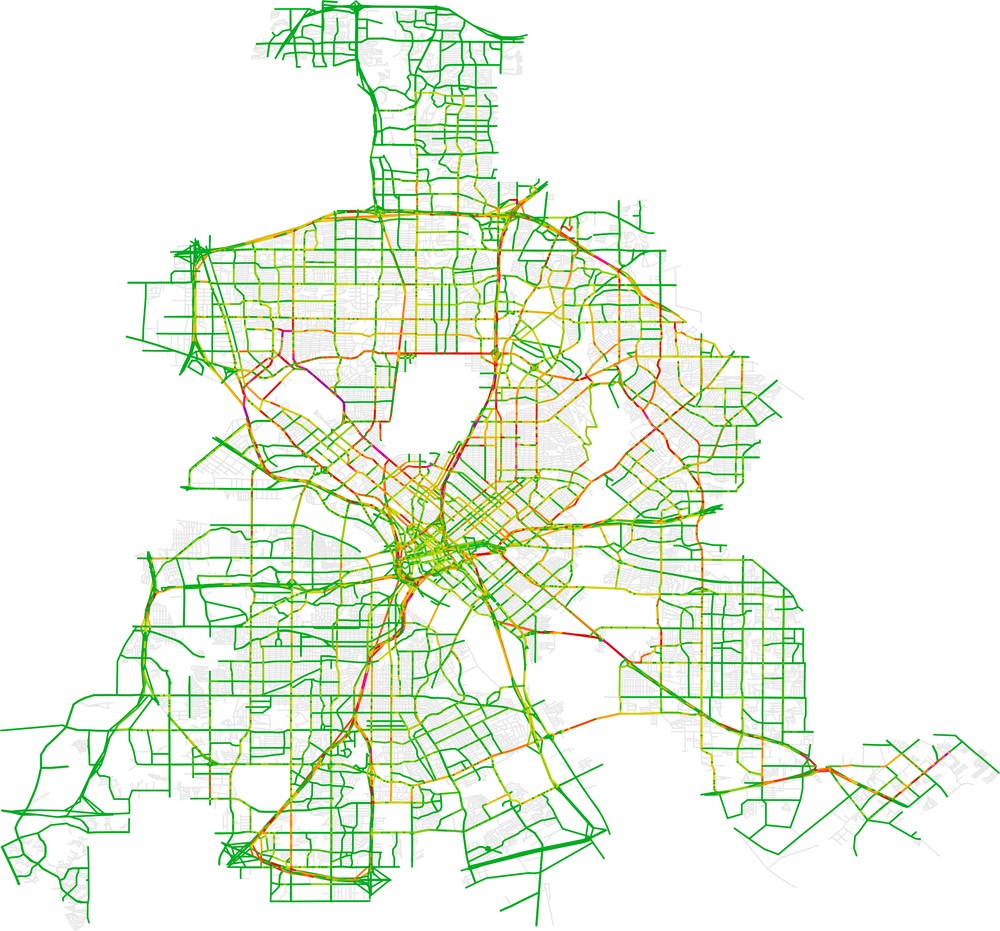} & \myimage{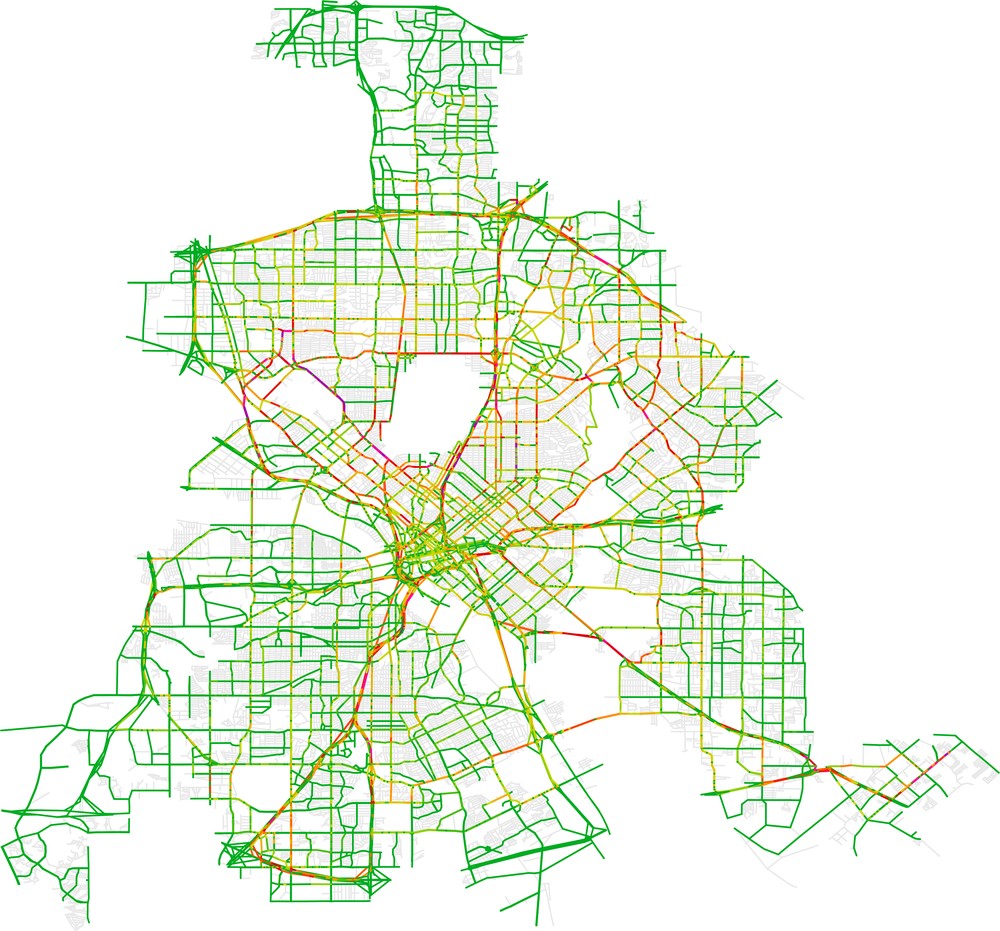} \\
                                          & \myimage{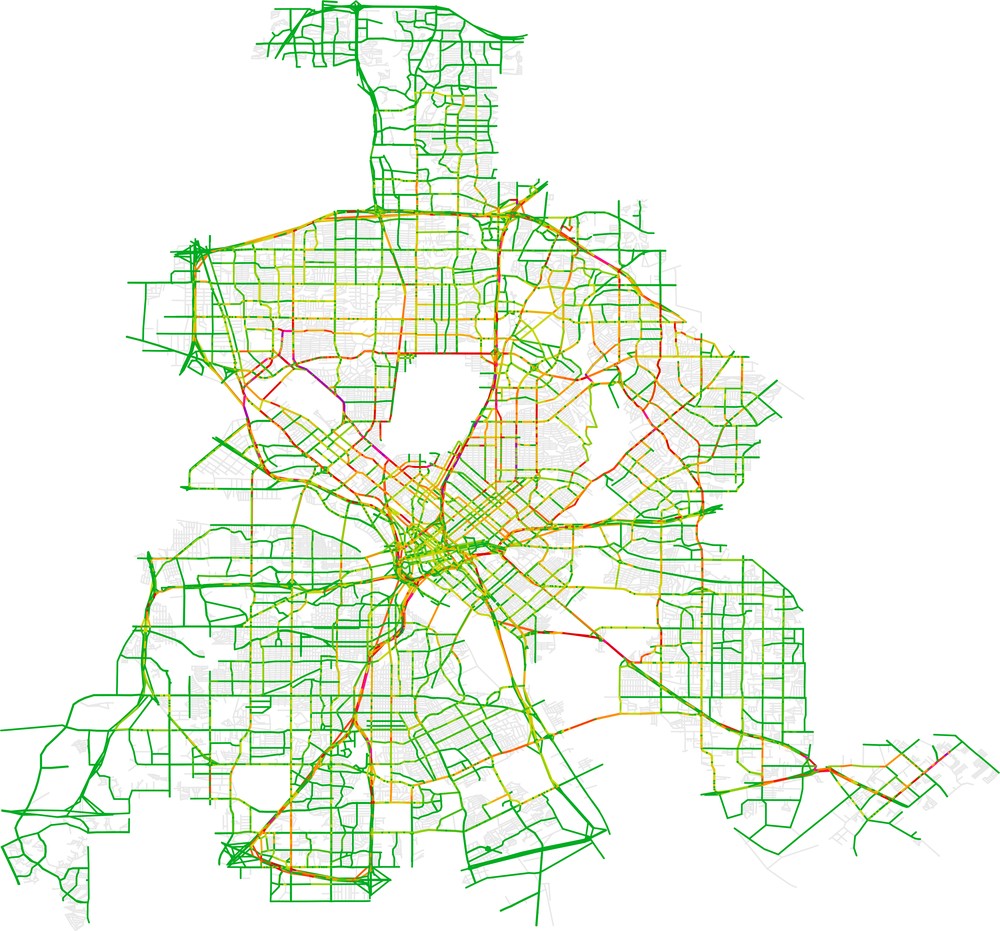} & \myimage{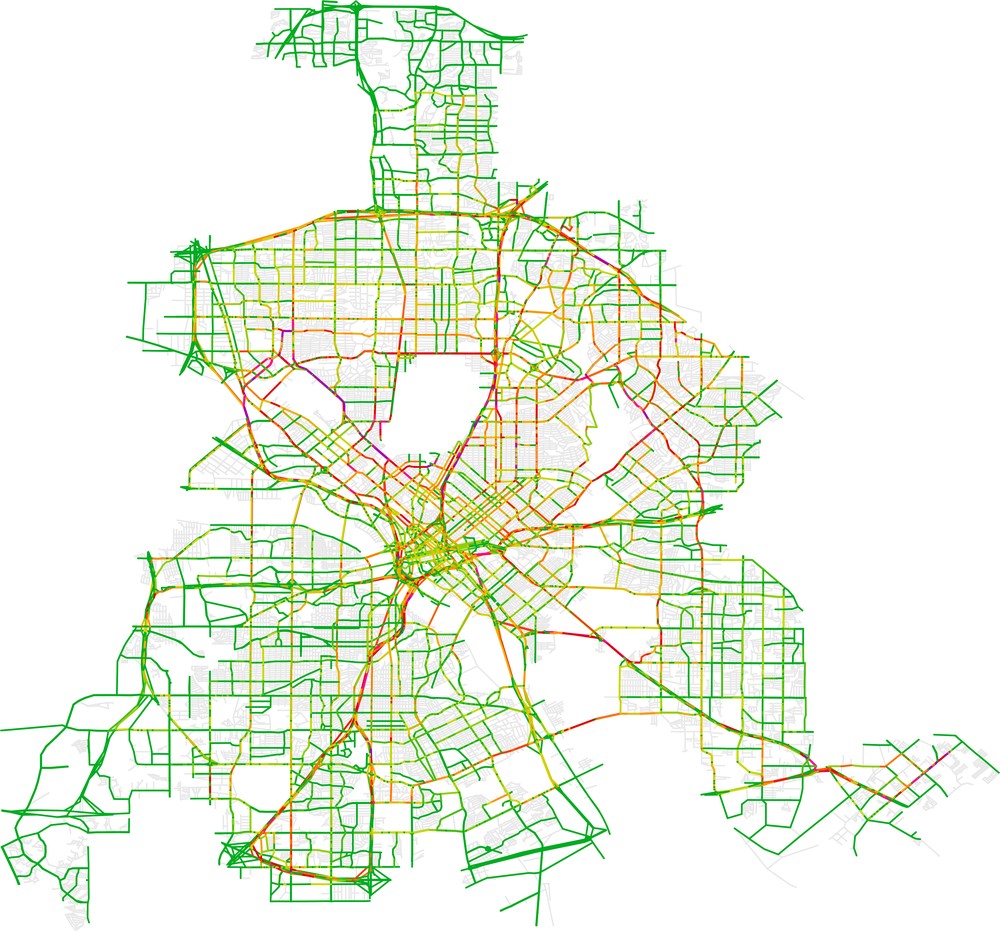} & \myimage{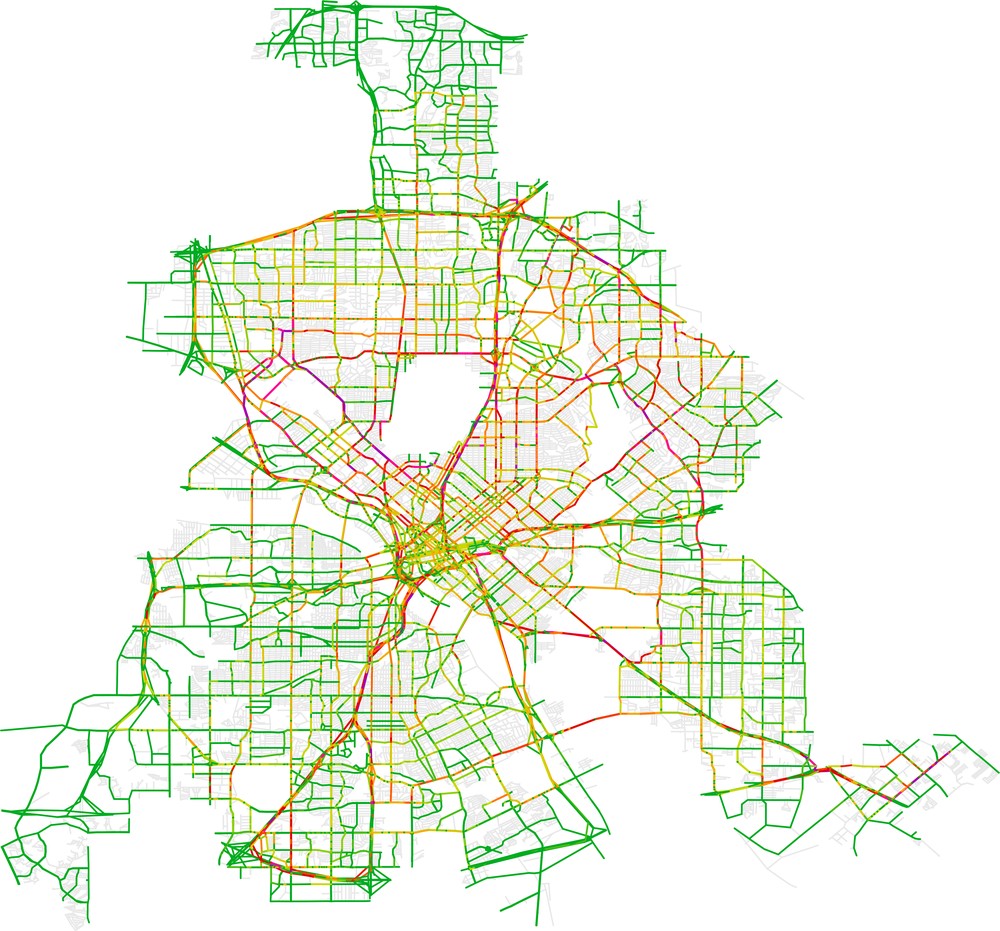} \\
                \end{tabular}
                \caption{Dallas}
                \label{fig:more_dallas}
            \end{subfigure}%
            \vspace{-0.7em}
            \caption[]{(continued).}
        \end{figure}

    \newpage
    \begingroup
        \sloppy   
        \bibliography{sample}
    \endgroup

\end{document}